\documentclass[a4paper,10pt,leqno]{amsart}
\usepackage{amsmath}
\usepackage{amsthm}
\usepackage{amssymb}
\usepackage{stmaryrd}
\usepackage{url}
\usepackage{here}
\usepackage[all]{xy}
\usepackage{comment}
\usepackage{mathtools}
\usepackage{ascmac}
\usepackage{ulem}
\usepackage{framed}
\usepackage{graphicx}
\usepackage{nicematrix}
\usepackage{lmodern}
\usepackage{caption}
\usepackage{subcaption}
\usepackage{footnote}
\usepackage{amscd}
\usepackage{mathdots}
\usepackage{titlesec}
\usepackage{titletoc}
\mathtoolsset{showonlyrefs}
\usepackage{enumitem}
\usepackage{bbm}

\usepackage{tikz}

\usepackage{xcolor}
\definecolor{lgray}{RGB}{240,240,240}
\definecolor{webgreen}{rgb}{0,.5,0}
\definecolor{webbrown}{rgb}{.6,0,0}
\definecolor{RoyalBlue}{cmyk}{1, 0.50, 0, 0}
\usepackage[colorlinks=true, breaklinks=true, urlcolor=webbrown, linkcolor=RoyalBlue, citecolor=webgreen, backref=page]{hyperref}

\usepackage{newtxtext,newtxmath}
\usepackage{newtxtext}

\usepackage{mathabx}

\titlecontents{section} 
[0pt]                 
{\bfseries}            
{\thecontentslabel\quad} 
{}                    
{\titlerule*[0.5em]{.}\contentspage} 

\titlecontents{subsection} 
[5pt]                 
{}            
{\thecontentslabel\quad} 
{}                    
{\titlerule*[0.5em]{.}\contentspage} 

\titlecontents{subsubsection} 
[10pt]                 
{}            
{\thecontentslabel\quad} 
{}                    
{\titlerule*[0.5em]{.}\contentspage} 

\titlecontents{paragraph} 
[15pt]                 
{}            
{\thecontentslabel\quad} 
{}                    
{\titlerule*[0.5em]{.}\contentspage} 

\titlespacing*{\section}{0pt}{5.5ex plus 1ex minus .2ex}{4.3ex plus .2ex}

\numberwithin{equation}{subsection}

\titleformat{\section}{\centering\normalfont\large\scshape}{\thesection}{1em}{}

\titleformat{\subsection}{\normalfont\normalsize\bfseries}{\thesubsection}{1em}{}

\titleformat{\subsubsection}{\normalfont\normalsize\bfseries}{\thesubsubsection}{1em}{}

\titleformat{\paragraph}{\normalfont\normalsize\it}{\theparagraph}{1em}{}

\theoremstyle{plain}
\newtheorem{thm}{Theorem}[section]
\newtheorem{lem}[thm]{Lemma}
\newtheorem{prop}[thm]{Proposition}

\newtheorem{RHP}{Riemann-Hilbert Problem}

\theoremstyle{definition}

\newtheorem*{IMY}{Theorem (Its-Miyahara-Yattselev)}

\theoremstyle{remark}
\newtheorem*{rem}{Remark}

\newtheoremstyle{case}{}{}{}{}{}{:}{ }{}
\theoremstyle{case}

\usepackage[most]{tcolorbox}
\tcbsetforeverylayer{autoparskip}
\usepackage{varwidth}
\tcbuselibrary{skins}
\tcbuselibrary{theorems}
\tcolorboxenvironment{myitemize}{blanker, breakable, before skip=6pt, after skip=6pt, borderline west={1pt}{0pt}{webbrown}}
\tcbset{colback=lgray,colframe=webbrown,boxrule=0.2mm,breakable,before skip=6pt,after skip=6pt,left=3pt,right=3pt}

\makeatletter
\renewenvironment{cases}[1][l]{\matrix@check\cases\env@cases{#1}}{\endarray\right.}
\def\env@cases#1{%
  \let\@ifnextchar\new@ifnextchar
  \left\lbrace\def\arraystretch{1.2}%
  \array{@{}#1@{\quad}l@{}}}
\makeatother

\newcommand{\C}{\mathbb C}   
\newcommand{\R}{\mathbb R}

\newcommand{\Z}{\mathbb Z}

\renewcommand{\S}{\mathfrak{S}}

\renewcommand{\a}{\boldsymbol{a}}
\newcommand{\z}{\boldsymbol{z}}

\newcommand{\ii}{\mathrm{i}}

\newcommand{\aj}{\mathsf{a}}
\newcommand{\kk}{\mathsf{k}}
\newcommand{\KK}{\mathsf{K}}

\newcommand{\stkout}[1]{\ifmmode\text{\sout{\ensuremath{#1}}}\else\sout{#1}\fi}

\DeclareMathOperator{\Arg}{Arg}

\DeclareMathOperator{\dist}{dist}

\DeclareMathOperator{\res}{res}
\DeclareMathOperator{\sign}{sgn}
\DeclareMathOperator{\Ai}{Ai}

\DeclareMathOperator{\sd}{sd}

\renewcommand{\O}{\mathcal{O}}
\renewcommand{\Re}{\operatorname{Re}}
\renewcommand{\Im}{\operatorname{Im}}

\begin{document}

\title[Transition asymptotics for the sinh-Gordon PIII]{Transition asymptotics for the real solutions of the sinh-Gordon Painlev\'e III equation}

\author[Miyahara]{Kenta Miyahara}
\address{(K. Miyahara) Department of Mathematical Sciences, 
Indiana University Indianapolis,
402 N. Blackford St., Indianapolis, IN 46202 USA}
\email{kemiya@iu.edu}

\author[Yattselev]{Maxim L. Yattselev}
\address{(M.L. Yattselev) Department of Mathematical Sciences, 
Indiana University Indianapolis,
402 N. Blackford St., Indianapolis, IN 46202 USA}
\email{maxyatts@iu.edu}

\begin{abstract}
We consider solutions of the sinh-Gordon Painlev\'e~III equation
\[
u_{xx} + \frac{1}{x} u_x = \sinh u
\]
that are real on \( (0,\infty)\). They are parametrized by the monodromy parameter \( p\in\overline\C \), \( |p|>1 \), and an additional real parameter \( s^\R \) when \( p=\infty \). Our previous joint work with A.~Its described the asymptotic behavior of these solutions as \( x\to\infty \). Here, we describe the transition as \( x, p\to \infty \), \( 2\Im(p)=-s^\R\), between singular solutions (\( |p|<\infty \)) and smooth solutions (\( p=\infty \)). In short, if we parametrize \( |p|^2 = 1 + e^{2\varkappa x}\), then the smooth exponential asymptotics of the solutions extends to the region \( \varkappa>1 \), with a change of the leading order term at \( \varkappa=2\); at \( \varkappa=1\) the exponential behavior transitions into an elliptic asymptotics, which holds for all \( 0<\varkappa<1\); as \( \varkappa \) decays to zero, elliptic asymptotics degenerates into trigonometric one, which holds for all \( p \) fixed.

\smallskip

\noindent
\textit{Keywords}: Painlev\'e III equation; sinh-Gordon equation; isomonodromic deformation; Riemann-Hilbert problem; steepest-descent method
\end{abstract}

\date{\today}

\let\ds\displaystyle

\maketitle 

\setcounter{tocdepth}{3}
\tableofcontents

\setlength{\parskip}{6pt}

\section{Introduction and Main Results}

This paper is a continuation of our previous work \cite{IMY}, joint with Alexander Its, where we studied the asymptotic behavior of {\it real} solutions of the sinh-Gordon Painlev\'e III equation
\begin{align}
    u_{xx} + \frac{1}{x} u_x = \sinh u.
    \label{sinh-gordon PIII}
\end{align}
Our interest stems from the fact that \eqref{sinh-gordon PIII} is the simplest tt*-Toda equation of Cecotti and Vafa \cite{CV1, CV2, CV3} whose analysis has been of much recent interest \cite{GuLi12, GuLi14, GIL1, GIL2, GIL3, GIL4, mcintosh2024}.

We understand reality in the following sense. It is known that solutions of \eqref{sinh-gordon PIII}  can have movable logarithmic singularities and that going around any logarithmic singularity (via the complex plane) of a solution of \eqref{sinh-gordon PIII} results in the addition of \( 2 \pi \ii \) to the solution. So, when we speak of real solutions, we ignore all the acquired additions of \( 2\pi \ii\) (since \eqref{sinh-gordon PIII} does not change when \( u(x) \) is replaced by \( u(x) + 2\pi \ii n \), \( n\in \Z \), thus obtained functions are still local solutions of \eqref{sinh-gordon PIII}).

In \cite{IMY} we were interested in the asymptotic behavior of the real solutions of \eqref{sinh-gordon PIII} on \( \R_{>0}:=(0,\infty) \) as \( x\to 0^+\) and \( x\to+\infty\). There, we have explained that these solutions can be parametrized by a parameter \( p \) from the extended complex plane with \( |p|>1\) and an additional real parameter \( s^\R \) in the case \(p=\infty\). In terms of the asymptotic behavior at infinity, we have obtained the following result \cite[Theorem~2.1]{IMY}\footnote{We refer the reader to our previous work \cite{IMY} where the relation of Theorem~\hyperref[IMY main thm]{IMY} with the results of the earlier works of Novokshenov \cite{Novok4} and Kitaev \cite{Kitaev2} obtained by the isomonodromic deformation method is discussed in detail.}.

\begin{tcolorbox}[colback=white,colframe=gray]
\begin{IMY}
\label{IMY main thm}
\it
To each finite \( p \), \( |p|>1 \), there corresponds one real solution of the sinh-Gordon Painlev\'e III equation \eqref{sinh-gordon PIII} on \( \R_{>0} \) and
\begin{align}
\label{asymp for B neq 0}
    u(x;p) = \ln \left(\frac{1 - \sin \left( x - \kappa\ln x + \phi \right)+\O(x^{-1})}{1 + \sin \left( x - \kappa\ln x + \phi \right)+\O(x^{-1})} \right), 
\end{align}
as $x \to +\infty$, where the error term is uniform in \( x \),
\begin{equation}
\label{kappa phi}
\kappa := \frac{1}{2 \pi} \ln ( |p|^2 - 1 ) \quad \text{and} \quad 
\phi := -\kappa \ln 4 + \arg \left( p\,\Gamma\left(\frac12+\ii\kappa\right) \right).
\end{equation}
To \( p=\infty \) there corresponds a one-parameter family of real solutions of the sinh-Gordon Painlev\'e III equation \eqref{sinh-gordon PIII} on \( \R_{>0} \) such that 
\begin{align}
    u(x;\infty,s^\R) = \frac{2s^\R}\pi K_0(x) + \O\left(\frac{e^{-2x}}{\sqrt{x}} \right),
    \label{exp decay}
\end{align}
as \( x\to+\infty \), where \( s^\R\in\R \) is the parametrizing (Stokes) parameter and \( K_0(x) \) is the modified Bessel function of the second kind.
\end{IMY}
\end{tcolorbox}

Formula \eqref{asymp for B neq 0} implies that solutions \( u(x;p) \) for finite \( p \) possess logarithmic singularities while solutions \( u(x;\infty,s^\R) \) are singularities free for all \( x \) large enough (entirely smooth on \(\R_{>0}\) when \(|s^\R|\leq 2\), see \cite{GIL2,GIL4}). In this work, we are interested in the transition asymptotics describing the evolution of singular solutions towards smooth solutions as \( x\to\infty \) and \( p \to \infty \). More precisely, given real parameters \( s^\R \) and \( \varkappa>0 \), we are interested in \( p=p(\varkappa,x)\) such that
\begin{equation}
\label{asymp regime}
\Im(p(\varkappa,x)) = -\tfrac12s^\R \quad \text{and} \quad |p(\varkappa,x)|^2 = 1+ e^{2\varkappa x}
\end{equation}
(notice that \( \pi\kappa=\varkappa x\) in the notation of \eqref{kappa phi}; also, given \( |s^\R|\), it must holds that \( e^{\varkappa x}\geq \frac12 |s^\R| \) assuming that \( |\Re(p)|\geq 1\)). It is not obvious from the statement of the above theorem that the imaginary part of \( p \) must be kept constant while \( |p| \) tends to infinity. This choice will be made clear further below in Section~\ref{ss:dimp}, see \eqref{mon data}, where we explain that another natural way to parametrize real solutions is with the monodromy data \eqref{mon surface} of the system  (Lax pair) of matrix differential equations \eqref{Lax pair sinh-Gordob PIII}. That parametrization shows that \( \Im(p)\) and \( |p|\) can be viewed as independent parameters. Hence, we keep one of them fixed and send the other to infinity.

Investigations of this type are not new. The study of the gap probabilities in the bulk of the eigenvalue spectrum of Hermitian random matrices drawn from the Gaussian Unitary Ensemble and related questions naturally lead to the Fredholm determinant 
\[
D_{\sin}(s, \gamma) = \det \left( \mathcal I - \gamma \, \mathcal{K}_{\sin} \right), \quad K_{\sin} (\lambda, \mu) = \frac{\sin (\lambda - \mu)}{\pi (\lambda - \mu)},
\]
where \( \mathcal{K}_{\sin} \) is an integral operator on $L^2(-s,s)$ with the kernel \( K_{\sin} (\lambda, \mu) \) and \( \gamma \in [0, 1] \) is a parameter. By the 1990s, the large $s$ asymptotics of \( D_{\sin}(s, \gamma) \) were fully and rigorously found for each fixed $\gamma$. It was shown that $D_{\sin}(s, \gamma)$ decays exponentially for each $\gamma \in [0, 1)$, whereas $D_{\sin}(s, 1)$ admits a super-exponential decay. Dyson conjectured that the transition between these two modes happens when one considers the double scaling regime
\[
\gamma = 1 - e^{-2\varkappa_{\sin} s}, \quad \varkappa_{\sin}>0,
\]
and that the transition asymptotic formulae involve Jacobi theta functions. This was only recently verified by Bothner-Deift-Its-Krasovsky using Riemann-Hilbert analysis of integrable operators in their ongoing seminal work \cite{BDIK1, BDIK2, BDIK3, BDIK4}.

In a similar vein, Bothner \cite{Bothner17} has studied the Fredholm determinant \( D_{\mathrm{Ai}}(x, \gamma) \), \( \gamma>0 \), corresponding to an integral operator on \( L^2(x,\infty) \) with the Airy kernel, specifically, its asymptotics as \( x\to-\infty \). This determinant admits an explicit expression
\begin{align}
D_{\mathrm{Ai}}(x, \gamma) = \exp \left[ - \int_x^\infty (y-x) u^2(y; \gamma) dy \right],
\end{align}
where \( u(x; \gamma) \) is the solution of the Painlev\'e II equation
\begin{align}
u_{xx} = x u + 2u^3
\end{align}
such that \( \gamma^{-\frac12}u(x; \gamma) \) possesses Airy asymptotics as \( x\to+\infty\). As with the large $s$ behavior of $D_{\sin}(s, \gamma)$, asymptotics of \( u(x; \gamma) \) as \( x\to-\infty \) for each \( \gamma \) fixed was already known \cite{AblowitzSegur77,AblowitzSegur81,HatingsMcLeod,Kapaev92}. The  Ablowitz-Segur solutions ($\gamma < 1$) and the Hastings-McLeod solution ($\gamma = 1$) are smooth near $-\infty$, while Kapaev solutions ($\gamma > 1$) are singular oscillatory. Bothner has studied the transition from the Ablowitz-Segur and the Kapaev solutions to the Hastings-McLeod one in the asymptotic regime
\[
\gamma = 1 \pm e^{-\varkappa_\mathrm{Ai}(-x)^{3/2}}\, \quad \varkappa_\mathrm{Ai}>0.
\]
Moreover, in \cite{Bothnergap}, Bothner has studied the above double scaling limits of \( D_{\mathrm{Ai}}(x, \gamma) \) as well as of \( D_{\mathrm{Bes}}(s, \gamma) \) (a Fredholm determinant corresponding to an integral operator with the Bessel kernel) directly, i.e., without going through the Painlev\'e functions.

\begin{figure}[ht!]
    \centering
    \includegraphics[width=0.68\linewidth]{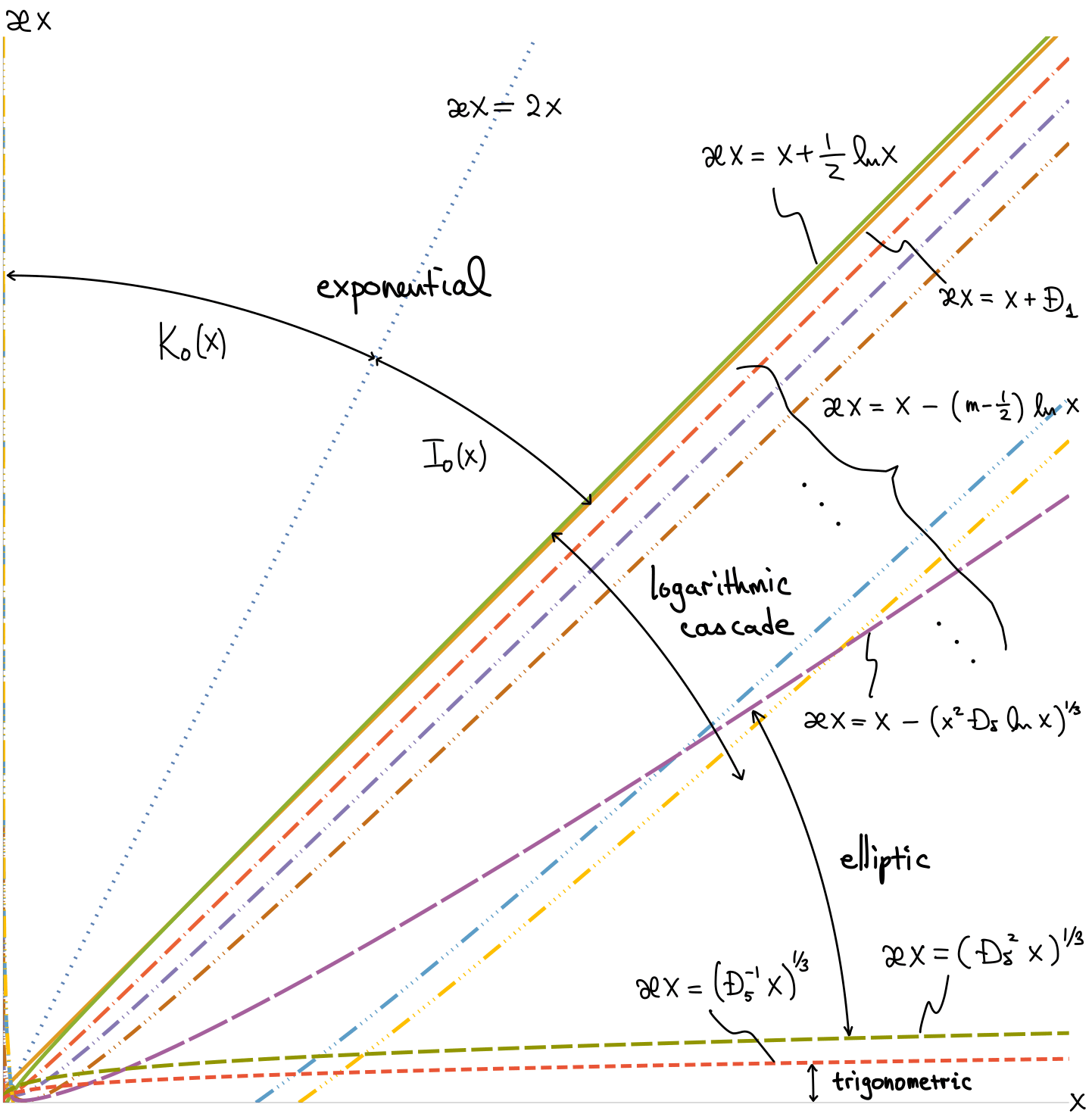}
    \caption{The type of behavior of \( u(x;p(\varkappa,x)) \) for each choice of \( (\varkappa,x)\).}
    \label{fig:results}
\end{figure}

As was the case in all of the above-mentioned works, the precise statements of our result are rather technical, but can be summarized as follows.
We show that the asymptotic formula \eqref{exp decay} extends to the region \( \varkappa>2\) of the \((x,\varkappa x)\)-plane  (the first part of the exponential region on Figure~\ref{fig:results}). At \( \varkappa=2\), the exponential nature of the decay remains the same, but the leading order term changes from \( K_0(x)\) to \( I_0(x) \), the modified Bessel function of the first kind. This behavior persists for all \( 1<\varkappa<2\) (the second part of the exponential region on Figure~\ref{fig:results}). These results are described in Theorem~\ref{thm:1}. At \( \varkappa=1\), exponential decay changes into the behavior similar to \eqref{asymp for B neq 0}, but with the sine function replaced by a certain Jacobi elliptic function. This transition happens via ``logarithmic cascades'', see Theorem~\ref{thm:2}-\ref{thm:3} and the ``logarithmic cascade'' region of Figure~\ref{fig:results}. Elliptic asymptotics holds for each \( 0<\varkappa<1\), see Theorem~\ref{thm:4}, which corresponds to the ``elliptic'' region on Figure~\ref{fig:results}. Finally, when \( \varkappa \) decays to zero, elliptic asymptotics degenerates into \eqref{asymp for B neq 0}, see Theorem~\ref{thm:5}. This is the ``trigonometric'' region of Figure~\ref{fig:results}. 

Figure~\ref{fig:results} can also be understood as follows. Solution \( u(x,p) \) for a fixed parameter \( p \) corresponds to a horizontal line on Figure~\ref{fig:results}. Hence, for a very large fixed \( p \), we see that such a solution first resembles an exponentially decaying one with no singularities. However, as \( x \) increases, singularities start to appear, but seemingly not regularly enough to be described by an elliptic or trigonometric function (this corresponds to the horizontal line passing through the logarithmic cascades region). Next, for even larger \( x \), the behavior of the solution is described with the help of elliptic functions, and eventually it becomes trigonometric as in \eqref{asymp for B neq 0}.

Throughout the paper, we always assume that \eqref{asymp regime} takes place. For a non-negative quantity \( q \), we let \( q_+:=\max\{1,q\} = \exp\{\ln^+q\} \). We further set \( \varepsilon :=\sign(\Re(p)) \).

Our first theorem describes asymptotics of the real solutions of the sinh-Gordon Painlev\'e III equation \eqref{sinh-gordon PIII} in the exponential region \( \varkappa>1 \).

\begin{tcolorbox}[colback=white,colframe=gray]
\begin{thm}
\label{thm:1}
There exists a large enough constant \( D_1>0 \) such that for all 
\begin{align}
\label{D1}
\varkappa_1 x \geq D_1 + \ln |s^\R|_+ , \quad \varkappa_1:=\min\{1,\varkappa-1\},
\end{align}
and all \( x\geq1 \), it holds that
\begin{equation}
\label{asymptotic formula 1}
u(x,p(\varkappa,x)) = \frac2\pi s^\R K_0(x) - 2\varepsilon e^{-\varkappa x} I_0(x) + \O\left(|s^\R|_+ x^{-\frac12}e^{-2\varkappa_1 x}\right),
\end{equation}
where \( I_0(x)\) and \( K_0(x) \) are the modified Bessel functions of the first and second kind, respectively, and the constant in the error term \( \O(\cdot) \) is absolute. If \( s^\R=0 \), \( \varkappa_1 \) can be replaced by \( \varkappa-1 \) in \eqref{D1} and \eqref{asymptotic formula 1}.
\end{thm}
\end{tcolorbox}

Of course, condition \eqref{D1} simply says that \( x\geq D_1\) when \( \varkappa\geq2 \), while it becomes \( \varkappa x\geq x+D_1 \) when \( \varkappa\in(1,2) \), which is a simultaneous condition on \( \varkappa \) and \( x \).

\begin{tcolorbox}[colback=white,colframe=gray]
\begin{rem}
Recall that \( K_0(x)\) and \( I_0(x) \) possess the following asymptotic expansions:
\[
\frac1\pi K_0(x) \sim \frac{e^{-x}}{\sqrt{2\pi x}} \sum_{n=0}^\infty (-1)^n \frac{a_n}{x^n}, \quad I_0 (x) \sim \frac{e^x}{\sqrt{2\pi x}} \sum_{n=0}^\infty \frac{a_n}{x^n}, \quad a_n = \frac{((1/2)_n)^2}{2^n n!},
\]
where \( (\cdot)_n\) is the Pochhammer symbol (rising factorial), see \cite[(10.40.1-2)]{NIST:DLMF}. Hence, for each fixed \( \varkappa \) and all \( x \) sufficiently large, it holds that
\[
u(x,p(\varkappa,x)) \sim e^{-\varkappa_1x}\sqrt{\frac2{\pi x}} \sum_{n=0}^\infty \frac{a_n}{x^n}\begin{cases}
(-1)^ns^\R, & \varkappa>2, \\
((-1)^ns^\R - \varepsilon), & \varkappa=2, \\
(- \varepsilon), & 1<\varkappa<2,
\end{cases}
\]
assuming \( s^\R\neq 0 \), where we have omitted the error terms as they are smaller than the terms in the above expansions.
\end{rem}
\end{tcolorbox}

Theorem~\ref{thm:1} allows us to take \( \varkappa \) as close to \( 1 \) as we wish at a price of increasing \( x \) since we always must have \( (\varkappa-1)x\geq D_1\). Naturally, this condition does not allow us to take \( \varkappa=1 \). The latter case is addressed by the following theorem.

\begin{tcolorbox}[colback=white,colframe=gray]
\begin{thm}
\label{thm:2} 
For any non-negative integer \( m\geq0 \), there exists a constant \( D_2(m)>0 \) such that for all
\begin{align}
\label{D2}
x\geq D_2(m) + 8\ln |s^\R|_+  \quad \text{and} \quad \varkappa x = x - (m+\nu)\ln x, \quad  |\nu| \leq \tfrac12,
\end{align}
it holds that
\begin{equation}
\label{asymptotic formula 2}
u(x,p(\varkappa,x)) = 2\varepsilon \ln\left(\frac{1 + \chi_m x^{-\nu-\frac12} + \chi_{m+1}^{-1} x^{\nu-\frac12} + \O_m(x^{-1})}{1 - \chi_m x^{-\nu-\frac12} - \chi_{m+1}^{-1} x^{\nu-\frac12} + \O_m(x^{-1})}\right),
\end{equation}
where \( \chi_m = \sqrt{\frac\pi 2} \frac{(-4)^m}{(m-1)!}\) (it is zero when \( m=0 \)).
\end{thm}
\end{tcolorbox}

In the integrable systems literature, asymptotic formulae of type \eqref{asymptotic formula 2} are known as {\it logarithmic cascades} since the coefficient next to the leading power of \( x \) changes each time a curve \( x-(m+\frac12)\ln x\) is crossed.

\begin{tcolorbox}[colback=white,colframe=gray]
\begin{rem}
Any half-integer constant in front of \( \ln x \) can represent as \( m+\tfrac12 = (m+1)-\tfrac12\). If we plug these representations into \eqref{asymptotic formula 2}, we obtain
\[
2\varepsilon \ln\left( \frac{1 + \chi_{m+1}^{-1} + \O_m(x^{-1})}{1 - \chi_{m+1}^{-1} + \O_m(x^{-1})} \right) \quad \text{and} \quad 2\varepsilon \ln\left( \frac{1 + \chi_{m+1} + \O_m(x^{-1})}{1 -\chi_{m+1} + \O_m(x^{-1})} \right).
\]
Clearly, both formulae lead to the same asymptotic statement in \( x \). Of course, it is important to observe here that the error term does depend on \( m \).
\end{rem}
\end{tcolorbox}

\begin{tcolorbox}[colback=white,colframe=gray]
\begin{rem}
Formula \eqref{asymptotic formula 2} points to the appearance of the first (after the period of exponential behavior) singularity of \( u(x,p(\varkappa,x)) \). The matters are complicated by the dependence on \( D_2(m) \) that restricts the range of applicability of \eqref{asymptotic formula 2} and implicitly appears in the error term. However, ignoring these complications, \eqref{asymptotic formula 2} says that there must be a singularity around
\[ 
|\chi_{m+1}|^{\frac12-\nu}
\]
for all \( m \) large and \( \nu \) sufficiently close to \( \frac12\) (it is a zero of either numerator or denominator in \eqref{asymptotic formula 2} depending on the parity of \( m \)).
\end{rem}
\end{tcolorbox}

Formula \eqref{asymptotic formula 2} contains only two leading terms of the asymptotic expansion of the solution \( u(x,p(\varkappa,x))\) while \eqref{asymptotic formula 1} gives us a full expansion. In theory, our technique can be extended to obtain more terms, but in practice it is computationally unfeasible except in the case \( m=0 \), which is addressed in the following theorem.

\begin{tcolorbox}[colback=white,colframe=gray]
\begin{thm}
\label{thm:3}
For each natural number \( n \), there exists a constant \( D_3(n)>0 \) such that for all
\begin{align}
\label{D3}
x\geq D_3(n) + 8\ln |s^\R|_+ \quad \text{and} \quad \varkappa x= x- \nu \ln x, \quad |\nu|\leq \tfrac12, 
\end{align}
it holds that
\begin{equation}
    \label{asymptotic formula 3}
    u(x,p(\varkappa,x)) = 2\varepsilon\ln\left( \frac{P_n(-x^{\nu}) + \O_n\big(x^{\nu-n-\frac12}\big)}{P_n(x^{\nu}) + \O_n\big(x^{\nu-n-\frac12}\big)}\right),
\end{equation}
where \( P_n(t) \) is a polynomial of degree \( n \) defined by
\[
P_n(t) := \det(I - t \, C_nL_n),
\]
where \( C_n=[c_{n,k+m-1}]_{k,m=1}^{2n-1} \) (\( k \) is the row index and \( m \) is the column one), 
\begin{equation}
    \label{definition Cn}
c_{n,k} := \frac{\ii^k}{\sqrt{2\pi x}} \sum_{m = \left\lceil \frac{k-1}{2} \right\rceil}^{n-1} \left( -\frac2x\right)^m \left(\frac12\right)_m \binom{m+\frac12}{2m+1-k}    
\end{equation}
for each \( k=1,\ldots,2n-1\) and \( c_{n,k}:=0\) if \( k\geq 2n\), and
\begin{equation}
    \label{definition Ln}
    L_n := \left[ \frac{(-1)^m}{(2\ii)^{m+k-1}}\binom{-m}{k-1} \right]_{k,m=1}^{2n-1}.
\end{equation}
Moreover, if we write
\[
P_n(t) = 1 + p_{n,1}t + p_{n,2}t^2 + \cdots + p_{n,2n-1}t^{2n-1},
\]
then \( p_{n,1} = \O\big(x^{-\frac12}\big)\), \( p_{n,2k} = \O(x^{-k(k+1)}) \) and \( p_{n,2k+1} = \O(x^{-\frac12-k(k+2)}) \), \( k\geq 1 \).
\end{thm}
\end{tcolorbox}

Clearly,  \( P_n(t) \) is simply the reciprocal of the characteristic polynomial of \( C_nL_n \).

\begin{tcolorbox}[colback=white,colframe=gray]
\begin{rem}
Regions \eqref{D1} and \eqref{D3} overlap for \( x+D_1 \leq \varkappa x\leq x+\tfrac12\ln x\). 
Hence, we can readily see by comparing \eqref{asymptotic formula 1} and \eqref{asymptotic formula 3} that
\[
p_{n,1} = -\mathrm{tr}(C_nL_n) = \frac1{2\sqrt{2\pi x}} \sum_{m=0}^{n-1}\frac{a_m}{x^m},
\]
which also can be verified by a direct computation of the trace. 
\end{rem}
\end{tcolorbox}

\begin{tcolorbox}[colback=white,colframe=gray]
\begin{rem}
It follows from the previous remark that 
\[
u(x) = 2\varepsilon \ln \left( \frac{1 - \tfrac{x^{\nu-\frac12}}{2\sqrt{2\pi}} + \cdots + \O_n\big(x^{\nu-n-\frac12}\big)}{1 + \tfrac{x^{\nu-\frac12}}{2\sqrt{2\pi }} + \cdots + \O_n\big(x^{\nu-n-\frac12}\big)}\right),
\]
which matches exactly with \eqref{asymptotic formula 2} (recall that \( \chi_0=0 \)), where \( \cdots \) represent an increasing with \( n \) number of terms of the expansion. That is, Theorem~\ref{thm:3}, in theory, gives a full asymptotic expansion of \( u(x) \) as compared to Theorem~\ref{thm:2}, but only up to the first Stokes line \( \varkappa x = x - \tfrac12\ln x\). 
\end{rem}
\end{tcolorbox}

Next, we turn our attention to the case \( \varkappa\in(0,1)\). We show below in Lemma~\ref{lem:varkappa} that to each \( \varkappa\in(0,1) \) there corresponds a unique \( \alpha \in(0,\tfrac\pi2) \) such that
\begin{align}
\label{defn of alpha}
\varkappa = \frac{1}{2} \int_0^\alpha \sqrt{2(\cos 2\theta - \cos 2\alpha)} \, d\theta.
\end{align}
 Moreover, \( \alpha(\varkappa)\) is continuous and increasing with range \( (0,\tfrac\pi2) \). In fact, as in  Lemma~\ref{lem:varkappa}, definition \eqref{defn of alpha} can be rewritten as
\begin{align}
\label{kappa as elliptic}
\varkappa = E(\sin \alpha) - \cos^2 \alpha \, K(\sin \alpha),
\end{align}
where \(K(\cdot) \) and \( E(\cdot)\) are the complete elliptic integrals of the first and second kind:
\begin{equation}
\label{Elliptic integrals}
K(k) = \int_0^1 \frac{1}{\sqrt{(1 - y^2)(1 - k^2 y^2)}} \, dy \quad \text{and} \quad E(k) = \int_0^1 \sqrt{\frac{1 - k^2 y^2}{1 - y^2}} \, dy.
\end{equation}
\begin{figure}
    \centering
    \includegraphics[width=0.7\linewidth]{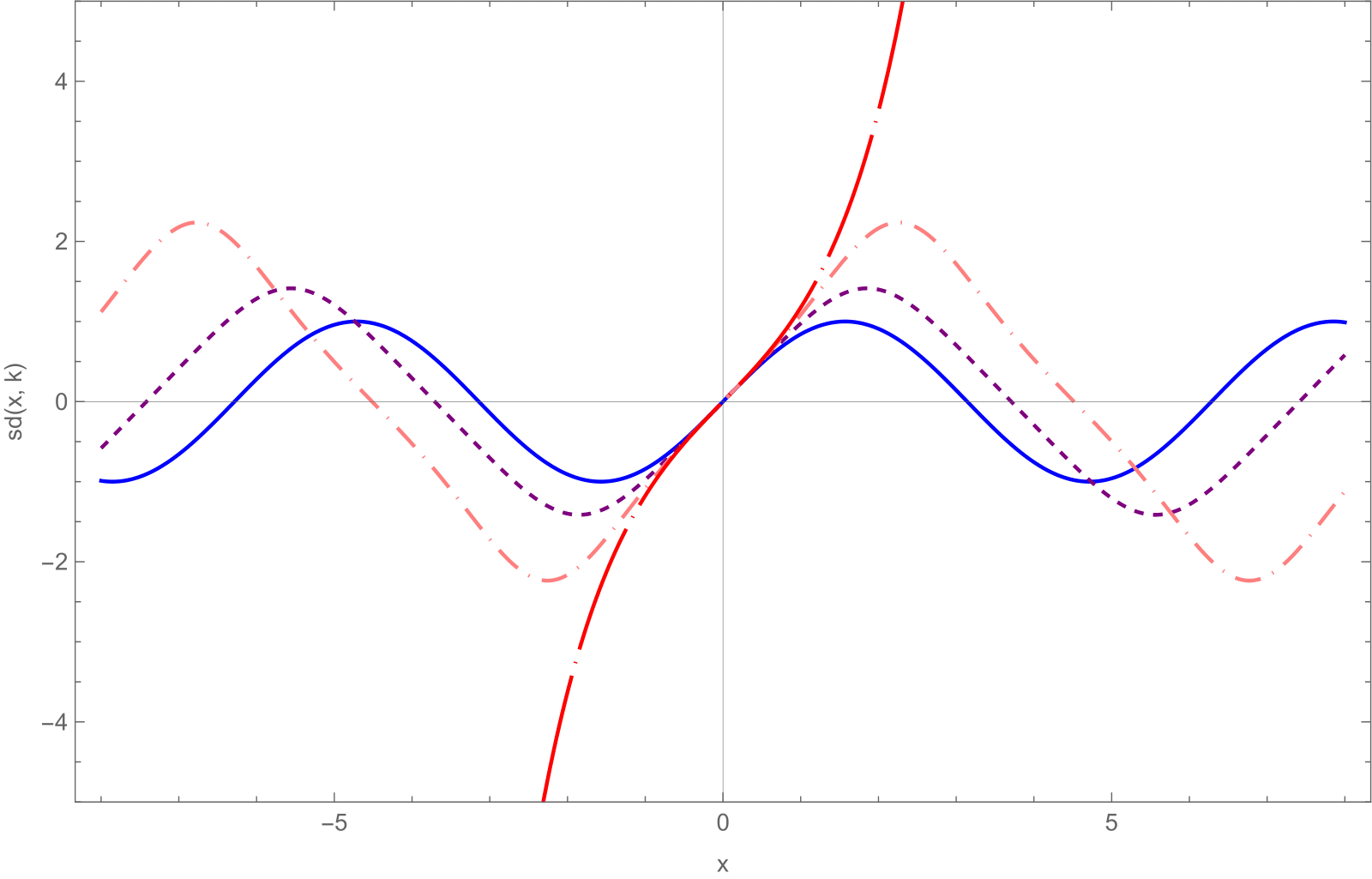}
    \caption{Graphs of $y = \sd(x,k)$ for $k=0, 0.5, 0.8, 1$; $y = \sd(x,0) = \sin(x)$ (blue solid curve); $y = \sd(x,0.5)$ (purple dashed curve); $y = \sd(x,0.8)$ (pink dash-dotted curve); $y = \sd(x,1) = \sinh(x)$ (red long-dash-dotted curve).}
    \label{fig:y=sd(x,k)}
\end{figure}
\noindent
Recall also that \( \sd(\cdot,\kk)\) is one of the Jacobian elliptic functions corresponding to the elliptic modulus \( \kk \), where, to shorten our notation, we set
\begin{align}
\label{ks}
\kk := \sin\alpha, \quad \kk^\prime := \cos\alpha, \quad \KK := K(\kk), \quad \text{and} \quad \KK^\prime := K(\kk^\prime),
\end{align}
which are, of course, functions of \( \varkappa \). The function \( \sd(\cdot,\kk)\) is odd and periodic on the real line with period \( 4\KK \), see \cite[Table 22.4.1]{NIST:DLMF}. It is increasing from \( -1/\kk^\prime \) to \( 1/\kk^\prime \) on \( (-\KK,\KK) \) and then symmetrically decreasing on \( (\KK,3\KK) \), see Figure~\ref{fig:y=sd(x,k)}.

\begin{tcolorbox}[colback=white,colframe=gray]
\begin{thm}
\label{thm:4}
Let \( \varkappa\in(0,1)\) and, as before, \( \pi\kappa = \varkappa x = \tfrac12\ln ( |p(\varkappa,x)|^2 - 1 ) \). Given \( \delta\in(0,\tfrac12)\), define
\begin{equation}
\label{Z delta}    
\mathcal Z_\delta := \left\{ (\varkappa,x)~\big|~\left|\frac{x-2\kappa\KK^\prime}{2\KK}-n-\frac12\right|<\delta,~n\in\Z \right\}.
\end{equation}
Then, there exists a large enough constant \( D_4>0 \) such that for all \((\varkappa,x)\notin\mathcal Z_\delta \) satisfying
\begin{align}
\label{D4}
(D_\delta^2x)^{\frac13} + \ln |s^\R|_+ \leq \varkappa x \leq x - ( x^2 D_\delta \ln x)^{\frac13}
\end{align}
and \( x\geq e \), where \( D_\delta := \delta^{-2} D_4 \), it holds that
\begin{equation}
\label{asymptotic formula 4}
u(x,p(\varkappa,x)) = \varepsilon \ln \left(\frac{1 - \kk^\prime \sd(x - 2\kappa \KK^\prime,\kk) + E_\varkappa(x)}{1 + \kk^\prime\sd(x - 2\kappa \KK^\prime,\kk) + E_\varkappa(x)}\right),
\end{equation}
where \( x \delta^2\KK^2\kk^3(\kk^\prime)^6 |E_\varkappa(x)| \leq C \) for some absolute constant \( C \).
\end{thm}
\end{tcolorbox}

\begin{tcolorbox}[colback=white,colframe=gray]
\begin{rem}
When \( x-2\kappa\KK^\prime = (2n+1)\KK \), it holds that \( \kk^\prime \sd(x - 2\kappa \KK^\prime,\kk) = (-1)^n \). Thus, the fraction in \eqref{asymptotic formula 4} is close to either zero or infinity at such points. This makes it more transparent why one might need to introduce the set \( \mathcal Z_\delta \) and why the error term could be expected to depend on \( \delta \) (the distance to singularities). Of course, the same formation of singularities can be seen in \eqref{asymp for B neq 0} and \eqref{asymptotic formula 5} further below around the points where the sine function is equal to \( \pm1 \). However, in those formulae the error terms hold uniformly for all \( x \) large and no analogue of the set \( \mathcal Z_\delta \) is needed. The difference stems from employed techniques: one uses global parametrix that utilizes theta functions and the other uses dressing method with polynomial prefactor.
\end{rem}
\end{tcolorbox}

\begin{tcolorbox}[colback=white,colframe=gray]
\begin{rem}
Assuming \eqref{D4} and \( x\geq e \), we show further below in Lemma~\ref{lem:K alpha} that \( x \delta^2\KK^2\kk^3(\kk^\prime)^6 \geq c_0 D_4 \) for some absolute constant \( c_0 \). Thus, it holds that
\[
|E_\varkappa(x)| \leq C/(c_0 D_4),
\]
which can be made small by choosing \( D_4 \) large. However, if \( \epsilon\leq \varkappa \leq 1-\epsilon \) for some \( \epsilon>0 \), then \( \KK^2\kk^3(\kk^\prime)^6 \geq c_\epsilon \) for some constant \( c_\epsilon \), in which case
\[
E_\varkappa(x) = \O_{\delta,\epsilon}\big(x^{-1}\big).
\]
\end{rem}
\end{tcolorbox}

\begin{tcolorbox}[colback=white,colframe=gray]
\begin{rem}
Trivially, every given point in the region \( 0<\varkappa<1\), \( e\leq x \) lies on some curve \( \varkappa x=x-(m+\nu)\ln x \) for an appropriate (and mostly very large) choice of \( m+\nu\). This makes \eqref{asymptotic formula 2} meaningful only for very large \( x \) as the error term in \eqref{asymptotic formula 2} depends on \( m \). It seems then that \( x \) should be so large that there is no overlap of this part of the curve \( \varkappa x=x-(m+\nu)\ln x \) with the region in \eqref{D4}. This makes direct comparison between \eqref{asymptotic formula 2} and \eqref{asymptotic formula 4} impossible, which is certainly a significant drawback of the logarithmic cascades asymptotics.
\end{rem}
\end{tcolorbox}

\begin{tcolorbox}[colback=white,colframe=gray]
\begin{rem}
Recall \eqref{kappa phi}, where now \( \kappa=\kappa(\varkappa,x) \) and \( \phi=\phi(\varkappa,x) \). When  
\[
\pi\kappa = \varkappa x = D x^{\frac13},
\]
we show in Theorem~\ref{thm:sd-sine}, further below in Appendix~\ref{app:elliptic}, that
\begin{align}
\label{4 to 5 reduction}
\kk^\prime \sd(x - 2\kappa \KK^\prime,\kk) = \varepsilon\sin \left( x - \kappa \ln x + \phi \right) + \O \big( x^{-\frac13} D^2 \ln x\big)
\end{align}
as \( x\to\infty \). This expression clearly shows the leading order matching of formula \eqref{asymptotic formula 4} with the one in \eqref{asymp for B neq 0}, see also \eqref{asymptotic formula 5}.
\end{rem}
\end{tcolorbox}

Finally, the following theorem holds.

\begin{tcolorbox}[colback=white,colframe=gray]
\begin{thm}
\label{thm:5}
There exists a large enough constant \( D_5>0 \) such that for all 
\begin{align}
\label{D5}
(\varkappa x)_+ \leq \left(D_5^{-1} x\right)^{\frac13}, \quad (\varkappa x)_+ = \max\{\varkappa x,1\},
\end{align}
it holds that 
\begin{align}
\label{asymptotic formula 5}
u(x;p(\varkappa,x)) = \ln \left(\frac{1 - \sin \left( x - \kappa\ln x + \phi \right)+\O((\varkappa x)_+^3x^{-1})}{1 + \sin \left( x - \kappa\ln x + \phi \right)+\O((\varkappa x)_+^3x^{-1})} \right), 
\end{align}
where \( \kappa,\phi \) are as in \eqref{kappa phi} and the constants in the error terms \( \O(\cdot) \) are absolute.
\end{thm}
\end{tcolorbox}

Formula \eqref{asymptotic formula 5} indeed does not have \( \varepsilon\) in front of the logarithm. This is already anticipated by the presence of \( \varepsilon \) in both \eqref{asymptotic formula 4} and \eqref{4 to 5 reduction}.

\begin{tcolorbox}[colback=white,colframe=gray]
\begin{rem}
Even though formulae \eqref{asymptotic formula 4} and \eqref{asymptotic formula 5} match to the leading order (as explained in the remark preceding Theorem~\ref{thm:5}), unfortunately the region in \eqref{D4} and \eqref{D5} do not overlap as \( D_\delta^2>D_5^{-1} \). On the other hand, Theorem~\ref{thm:5} does overlap with Theorem~\hyperref[IMY main thm]{IMY} as the former does cover every \( |p|^2 \geq 2 \) fixed (i.e., \( \varkappa x\geq0 \) fixed).
\end{rem}
\end{tcolorbox}


Acknowledgments: The authors thank Thomas Bothner, Percy Deift, Alexander Its, and Igor Krasovsky for sharing the preprint \cite{BDIK4}, and also thank them, along with Peter Miller, for helpful discussions and remarks. The first author was partially supported by a scholarship from the Japan Student Services Organization and the Hokushin Scholarship Foundation.

\section{Proof of Theorem~\ref{thm:1}}

Throughout the paper, \( \sigma_1,\sigma_2,\sigma_3 \) are the Pauli matrices
\[
\sigma_1 = \begin{pmatrix}
    0 & 1 \\ 1 & 0 
\end{pmatrix},
\quad
\sigma_2 = \begin{pmatrix}
    0 & -\ii \\ \ii & 0
\end{pmatrix},
\quad
\text{and}
\quad
\sigma_3 = \begin{pmatrix}
    1 & 0 \\ 0 & -1
\end{pmatrix}.
\]
The symbols \( \C \) and \( \overline \C \) stand for the complex plane and the extended complex plane, respectively. Moreover, we set \( \C^*:=\C\setminus\{0\} \). We shall denote by \( S^\rho \) the circle centered at the origin of radius \( \rho \). Some of our results required \( x \) to be sufficiently large. Even if no restriction is mentioned, we always assume that \( x\geq 1 \). 

\subsection{Direct and Inverse Monodromy Problem}
\label{ss:dimp}

As standard in the theory of Painlev\'e transcendents, we view \eqref{sinh-gordon PIII} as a compatibility condition of an overdetermined system of linear differential equations, a Lax pair. A Lax pair for the sinh-Gordon equation was first found by Flaschka-Newell \cite{FN}, see also \cite[Chapter 15]{FIKN}. Let \( \Psi(\lambda, x) \) be a \( 2\times 2 \) matrix function on $\C^* \times \C$ such that
\begin{align}
\begin{dcases}
    \frac{d \Psi}{d \lambda} = \left( \frac{\ii x^2}{16} \sigma_3 + \frac{x u_x}{4 \lambda} \sigma_1 + \frac{\ii}{\lambda^2}Q \right) \Psi,\\
    \frac{d \Psi}{d x} = \left( \frac{u_x}{2} \sigma_1 + \frac{\ii x \lambda}{8} \sigma_3 \right) \Psi, 
\end{dcases}
    \label{Lax pair sinh-Gordob PIII}
\end{align}
where
\begin{equation}
    \label{defn of Q}
    Q = \begin{pmatrix}
        \cosh u & - \sinh u\\
        \sinh u & - \cosh u
    \end{pmatrix}.
\end{equation}
One can readily verify that the compatibility condition $\Psi_{\lambda x} = \Psi_{x \lambda}$ is equivalent to the sinh-Gordon Painlev\'e III equation \eqref{sinh-gordon PIII}.

The first equation in \eqref{Lax pair sinh-Gordob PIII} has two singular points \( \lambda=0,\infty\). Studying its Stokes sectors and formal solutions around these singularities allows one to identify the monodromy data of the equation, see for example \cite[Section~3.1]{IMY}. As it turns out, the corresponding monodromy surface is given by
\begin{align}
\label{mon surface}
\left\{ m=(A, s^{\R}, B^{\R}) \in \C \times \R \times \R \,|~~ |A|^2 = 1 + (B^\R)^2, ~~ A - \overline{A} = \ii s^{\R} B^\R  \right\}.
\end{align}
The monodromy data \( (A, s^{\R}, B^{\R}) \) can be re-parametrized by a pair \( (p,\iota) \), where $p\in\overline\C$, \( |p|>1 \) and $\iota\in\{\pm1\}$, via
\begin{equation}
\label{p parametrization}
\begin{cases}
\displaystyle s^\R = -2\Im(p), \quad B^\R = \frac{\iota}{\sqrt{|p|^2-1}}, \quad A = \overline p B^\R, & p\neq\infty, \smallskip \\
s^\R \text{ is any}, \quad B^\R =0, \quad A = \iota, & p=\infty
\end{cases}
\end{equation}
(the inverse map is given by \( p= \overline A/B^\R \), \( \iota = \sign(B^\R)\) when \( B^\R \) is non-zero and \( p=\infty \), \( \iota = A \) otherwise). 

As we explain in the last remark of \cite[Section~3.2]{IMY}, if \( u \) is a solution of \eqref{sinh-gordon PIII} corresponding to the monodromy data \( (p,\iota) \), then the solution corresponding to \( (p,-\iota) \) is \( u+2\pi\ii \). Hence, real solutions of \eqref{sinh-gordon PIII} do not depend on \( \iota \). Moreover, it is also pointed out in \cite[Appendix~A]{IMY} that the solution corresponding to the monodromy data \( (-p,-\iota) \) is given by \( -u \).  Hence, if \( \epsilon\in\{\pm1\}\), then
\[
u(x,(p,\iota)) = \epsilon u(x,(\epsilon p,1)).
\]
It readily follows from \eqref{p parametrization} that transformations \( (p,\iota) \mapsto (p,-\iota) \) and \( (p,\iota) \mapsto (-p,-\iota) \) correspond to \( (A, s^{\R}, B^{\R}) \mapsto (-A, s^{\R}, -B^{\R}) \) and \( (A, s^{\R}, B^{\R}) \mapsto (A, -s^{\R}, -B^{\R}) \), respectively. Thus, if \( \varepsilon = \sign(\Re(p)) \) and \( \iota =\sign(B^\R) \), then
\begin{equation}
\label{symmetries of u}
u\big(x,(A, s^{\R}, B^{\R})\big) = \varepsilon u\big(x,(\iota\varepsilon A, \varepsilon s^{\R}, \iota B^{\R})\big).
\end{equation}
Observe that \( \iota B^\R = |B^\R| \) and that \( \Re(\iota\varepsilon A)>0 \). Therefore, in what follows, we study only solutions corresponding to \( B^\R>0 \) and \( \Re(A)>0 \) and recover the general case from \eqref{symmetries of u}. Since the representation of the monodromy data \( m(\varkappa,x) \) via triples \( (A, s^{\R}, B^{\R}) \) is more convenient for us, we restate \eqref{asymp regime} as follows. Given real parameters \( s^\R \) and \( \varkappa>0 \), we are interested in the asymptotic behavior of the solution \( u(x;m(\varkappa,x)) \) corresponding to the monodromy data \( m(\varkappa,x)=(A, s^{\R}, B^{\R}) \) such that
\begin{equation}
    \label{mon data}
    B^\R = e^{-\varkappa x}, \quad |A|^2 = 1 + e^{-2\varkappa x}, \quad 2\Im(A) = s^\R e^{-\varkappa x}, \quad \Re(A)>0.
\end{equation}

It is well understood in the theory of integrable systems that the solution \( \Psi(\lambda,x) \) of the first equation in \eqref{Lax pair sinh-Gordob PIII} corresponding to the monodromy data \( m \) and therefore the solution \( u(x;m)\) of \eqref{sinh-gordon PIII} can be recovered via a certain Riemann-Hilbert problem, see for example \cite[Section~3.2]{IMY}.

\begin{figure}[ht!]
    \centering
    \includegraphics[width=8cm]{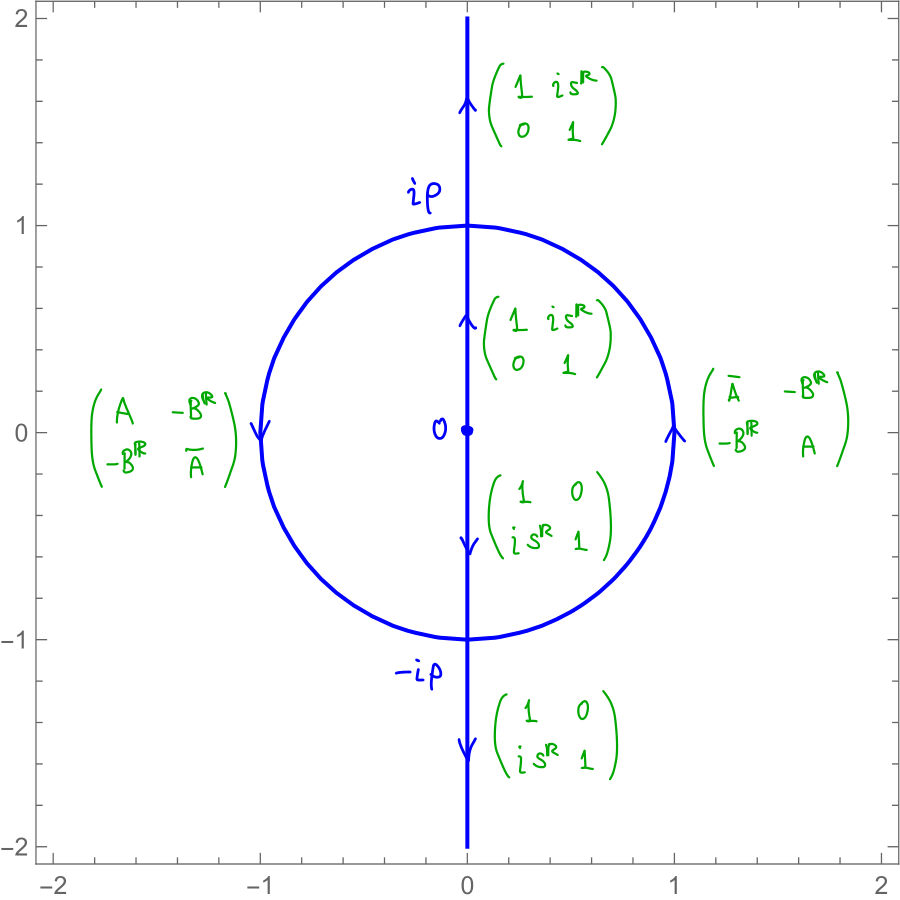}
    \caption{The contour $\Gamma_\rho$ and the jump matrices $G_{\hat{\Psi}}$ for RHP~\ref{rhp original}.}
    \label{Psi hat problem pic}
\end{figure}

\begin{RHP} 
\label{rhp original}
Find a $2 \times 2$ matrix function $\hat{\Psi}(\lambda)$ such that
\begin{enumerate}
    \item $\hat{\Psi}(\lambda)$ is analytic in $\C \setminus \Gamma_\rho$, where  $\Gamma_\rho=\ii\R\cup S^\rho$ is oriented as on Figure~\ref{Psi hat problem pic};
    \item one-sided traces $\hat{\Psi}_{\pm}(\lambda)$ defined by
    \begin{align*}
        \hat{\Psi}_{\pm}(\lambda) := \lim_{\{ \pm \text{ side of } \Gamma_\rho \} \ni \lambda' \to \lambda \in \Gamma_\rho } \hat{\Psi}(\lambda'),
    \end{align*}
    where a subscript $+$ (resp. $-$) refers to a limit to the oriented contour from its left (resp. right) side, exist a.e. on \( \Gamma_\rho \), belong to \( L^2(\Gamma_\rho) \), and satisfy
    \begin{align*}
        \hat{\Psi}_+(\lambda) = \hat{\Psi}_-(\lambda) G_{\hat{\Psi}}, \quad \lambda \in \Gamma_\rho,
    \end{align*}
    where the jump matrices $G_{\hat{\Psi}}(\lambda)$ on \( \Gamma_\rho \) are as on Figure \ref{Psi hat problem pic};
    \item it holds that
    \begin{align}
    \label{Psi hat as lambda -> 0}
        \hat{\Psi}(\lambda) = \begin{cases}
            (I + \O(1/\lambda)) e^{\frac{\ii x^2 \lambda}{16} \sigma_3} &  \text{as} \quad \lambda\to\infty, \smallskip \\
            P_0 (I + \O(\lambda)) e^{-\frac{\ii}{\lambda} \sigma_3} & \text{as} \quad \lambda\to0,
        \end{cases}
    \end{align}
    where \( P_0 := \begin{pmatrix} \cosh \frac u2 & \sinh \frac u2 \\ \sinh \frac u2 & \cosh \frac u2\end{pmatrix} \).
\end{enumerate}
\end{RHP}

That is, the monodromy data \( m \) determines the jump matrices \( G_{\hat\Psi} \) and the unique solution \( \hat\Psi(\lambda) \) of this Riemann-Hilbert problem determines the corresponding solution \( u(x;m) \) via its asymptotics at the origin.


\subsection{Renormalized Riemann-Hilbert Problem}

Take \( \rho = 4/x \) in RHP~\ref{rhp original} and apply the following gauge transformation:
\begin{align}
    Y(\lambda) = \hat{\Psi}\left( \frac{4\lambda}{x} \right) e^{-x \theta(\lambda)\sigma_3}, \quad \theta(\lambda) := \frac{\ii}{4}\left(  \lambda - \frac1\lambda \right).
    \label{defn of Y}
\end{align}
Then \( Y(\lambda) \) is the solution of the following Riemann-Hilbert problem.
\begin{RHP}
\label{RHP Y}
Find a $2 \times 2$ matrix function $Y(\lambda)$ such that
\begin{enumerate}
    \item $Y(\lambda)$ is analytic in $\C \setminus \Gamma$, where $\Gamma=\ii\R\cup S^1$ is oriented as on Figure~\ref{Y problem pic};
    \item one-sided traces $Y_\pm(\lambda)$ exist a.e. on \( \Gamma \), belong to \( L^2(\Gamma) \), and satisfy
    \begin{align*}
        Y_+(\lambda) = Y_-(\lambda) G_{Y}(\lambda), \quad \lambda \in \Gamma,
    \end{align*}
    where the jump matrices $G_{Y}(\lambda)$ on \( \Gamma \) are as on Figure~\ref{Y problem pic};
    \item it holds that
    \begin{align*}
        Y(\lambda) = \begin{cases}
            I + \O(1/\lambda)  &  \text{as} \quad \lambda\to\infty, \smallskip \\
            P_0 (I + \O(\lambda)) & \text{as} \quad \lambda\to0.
        \end{cases}
    \end{align*}
\end{enumerate}
\end{RHP}

\begin{figure}[ht!]
    \centering
    \includegraphics[width=8cm]{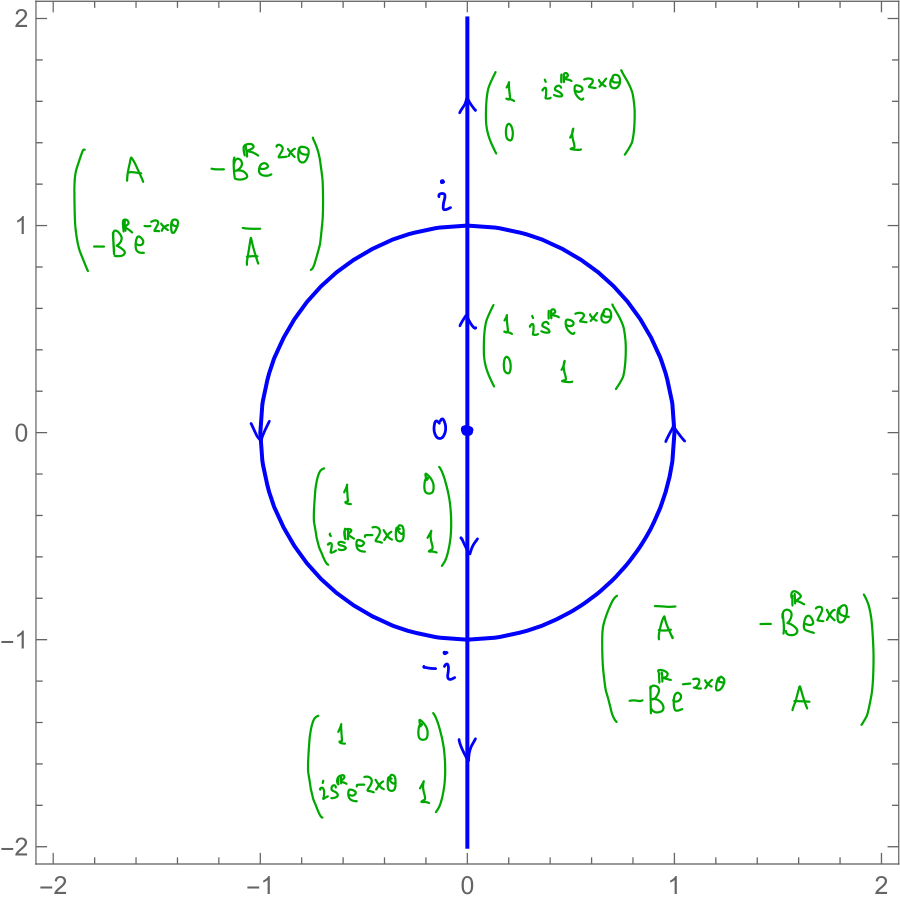}
    \caption{The jump matrices $G_Y(\lambda)$ for RHP~\ref{RHP Y}.}
    \label{Y problem pic}
\end{figure}

It holds on the imaginary axis that
\begin{equation}
\label{theta on imaginary}
4\theta(\ii y)=-(y+1/y), \quad y\in\R,
\end{equation}
while for \( \lambda\in S^1 \) we have that \( 2\theta(\lambda) =-\Im(\lambda)\) and therefore
\begin{equation}
\label{B is almost small}
B^\R e^{\pm2x\theta(\lambda)} = e^{-x(\varkappa\pm\Im(\lambda))} \leq e^{-(\varkappa-1)x}.
\end{equation}
Furthermore, recall that \( \Re(A)>0 \). It readily follows from \eqref{mon data} that
\[
\Re(A)^2 = 1 + \left(1 - \tfrac 14(s^\R)^2\right)e^{-2\varkappa x} \quad \Rightarrow \quad \left(\tfrac 14(s^\R)^2-1\right)e^{-2\varkappa x}\leq 1.
\]
Trivially, we have that
\[
|A-1| \leq |\Re(A)-1| + |\Im(A)| = \left|1 - \tfrac 14(s^\R)^2\right| \frac{e^{-2\varkappa x}}{1+\Re(A)} + \frac12 |s^\R|e^{-\varkappa x}.
\]
Hence, if \( |s^\R|\leq 2 \), then \( |A-1|< 2e^{-\varkappa x} \). On the other hand, if \( |s^\R|\geq 2 \), then
\[
|A-1| \leq \left(\sqrt{\tfrac 14(s^\R)^2-1}+\frac12 |s^\R|\right)e^{-\varkappa x} \leq |s^\R| e^{-\varkappa x}.
\]
Therefore, we can conclude that
\begin{align}
\label{A is almost 1}
|A-1| \leq 2|s^\R|_+ e^{-\varkappa x}.
\end{align}
Thus, the jump \( G_Y(\lambda) \) is close \( I \) when \eqref{D1} holds. We make this statement more precise in the next subsection.


\subsection{Asymptotic Analysis} 
\label{AAI}

Recall that \( \varkappa_1=\min\{1,\varkappa-1\}\). Throughout the paper,
\[
\|A\|_{L^2(\Gamma)}^2 := \int_{\Gamma}\|A(\lambda)\|^2|d\lambda| \quad \text{and} \quad  \|A\|_{L^\infty(\Gamma)} := \mathrm{ess}\sup_{\Gamma}\|A(\lambda)\|,
\]
where \( \|A\| \) is the Frobenius norm of a matrix \( A \).

We get from \eqref{theta on imaginary}, \eqref{B is almost small}, and \eqref{A is almost 1} that
\begin{equation}
\label{G_Y-I infinity norm}
\|G_Y-I\|_{L^\infty(\Gamma)} \leq 2|s^\R|_+ e^{-\varkappa_1x}.
\end{equation}
If \( s^\R=0\), there is no jump on the imaginary axis and the above estimate holds with \( \varkappa_1 \) replaced by \( \varkappa-1\). Observe that
\begin{equation}
\label{L2 norm estimate 1}
\frac1{2\pi}\int_{S^1} e^{\pm 4x\theta(\lambda)}|d\lambda| = \frac1{2\pi}\int_{-\pi}^\pi e^{\mp 2x\cos y}dy =  I_0(2x) \leq C_1\frac{e^{2x}}{\sqrt x},
\end{equation}
for some constant \( C_1 \), see \cite[(10.32.1)]{NIST:DLMF} and the remark right after Theorem~\ref{thm:1}, where to get the first equality, one needs to make a substitution \( \lambda \mapsto \ii\lambda \) first. Moreover, we get from \eqref{theta on imaginary} that
\begin{align}
\int_0^{\pm\ii\infty} e^{\pm 4x\theta(\lambda)}|d\lambda| & = \left(\int_0^1 + \int_1^\infty\right) e^{-x(y + \frac{1}{y})} dy = 2\int_1^\infty \frac{t e^{-2xt}}{\sqrt{t^2-1}}dt \\
& = 2e^{-2x}\int_0^\infty \frac{(y+1)e^{-2x y}}{\sqrt{y(y+2)}}dy \leq 2e^{-2x}\int_0^\infty \sqrt{\frac{y+2}{y}} e^{-2x y}dy \\
& = 2\frac{e^{-2x}}{x} \int_0^\infty \sqrt{\frac{t+2x}t} e^{-2t}dt \leq 2\frac{e^{-2x}}{\sqrt x} \int_0^\infty \sqrt{\frac{t+2}t} e^{-2t}dt,
\label{L2 norm estimate 2}
\end{align}
where to deduce the second equality we made substitutions \( y=t\pm\sqrt{t^2-1}\) depending on the range of \( y \) and used \( x\geq1 \) on the last step. It readily follows from \eqref{A is almost 1}, \eqref{L2 norm estimate 1}, and \eqref{L2 norm estimate 2} that
\begin{equation}
\label{G_Y-I L2 norm}
\| G_Y - I \|_{L^2(\Gamma)} \leq C_2 |s^\R|_+ x^{-\frac14}e^{-\varkappa_1x}
\end{equation}
for some absolute constant \( C_2 \). Again, If \( s^\R=0\), \eqref{L2 norm estimate 2} does not apply and \eqref{L2 norm estimate 1} yields that \eqref{G_Y-I L2 norm} holds with \( \varkappa_1 \) simply replaced by \( \varkappa-1\). 

Next, by the small norm theorem, see for example \cite[Theorem~8.1]{FIKN}, \eqref{G_Y-I infinity norm} implies that there exists a unique solution of RHP~\ref{RHP Y} for all $\varkappa_1 x$ large and
\begin{align}
    Y(\lambda) = I + \frac{1}{2 \pi \ii} \int_\Gamma \frac{\rho(\lambda')(G_{Y}(\lambda') - I)}{\lambda' - \lambda} d\lambda', \quad \lambda \notin \Gamma, \label{gl Y singular int eq}
\end{align}
where $\ds \rho(\lambda') = Y_-(\lambda^\prime)$ for $\lambda' \in \Gamma$. More precisely, if we denote by \( \mathcal C_\Gamma \) an operator from \( L^2(\Gamma)\) into itself that acts according to the rule
\begin{align}
\label{operator C Gamma}
\begin{aligned}
(\mathcal C_\Gamma F)(\mu) &:= \left(\mathcal C_- \left[ F(G_Y - I) \right] \right) (\mu)\\
&= \lim_{\lambda\to\mu\in \Gamma^-}\frac{1}{2 \pi \ii} \int_\Gamma \frac{F(\lambda')(G_{Y}(\lambda') - I)}{\lambda' - \lambda} d\lambda',
\end{aligned}
\end{align}
where \( \mathcal C_- \) is the Cauchy operator that sends a function into the trace on \( \Gamma^- \) of its Cauchy integral, then, by taking the limit \( \lambda\to\mu\in\Gamma^- \), we can rewrite \eqref{gl Y singular int eq} as
\[
\rho = I + \mathcal C_\Gamma \rho.
\]
In fact, the solvability of RHP~\ref{RHP Y} is equivalent to the solvability of this integral equation. Clearly, \( \|\mathcal C_\Gamma \| \leq \|\mathcal C_- \| \|G_Y-I\|_{L^\infty(\Gamma)} \)  (it is known that \( \mathcal C_- \) is a bounded operator on \( L^2(\Gamma)\), see Appendix~\ref{s:co}). Hence, according to \eqref{G_Y-I infinity norm}, we can choose \( D_1 \) in \eqref{D1} large enough so that \(\|\mathcal C_\Gamma \|\leq \tfrac12\) (the specific value \(\tfrac12\) is not important, any value in \( (0,1)\) will do). Then, using the von Neumann series, we get that
\begin{equation}
\label{von Neumann}
\rho = (\mathcal I-\mathcal C_\Gamma)^{-1}I = I + \sum_{m=1}^\infty \mathcal C_\Gamma^m I,
\end{equation}
where \( \mathcal I\) is the identity operator. In particular,
\begin{equation}
\label{gl rho_Y - I has a small norm}
\|\rho - I\|_{L^2(\Gamma)} \leq 2 \|\mathcal C_- (G_Y-I)\|_{L^2(\Gamma)} \leq C_3 |s^\R|_+ x^{-\frac14}e^{-\varkappa_1x}
\end{equation}
by \eqref{G_Y-I L2 norm} for some absolute constant \( C_3 \). 

Since \( G_{Y}(\lambda') - I \) vanishes exponentially at the origin, \( Y_+(0) = Y_-(0) =: Y(0)\) and one has that
\begin{align}
    Y(0) = I + \frac{1}{2\pi \ii} \int_\Gamma \frac{G_Y(\lambda') - I}{\lambda'} d\lambda'  + \frac{1}{2\pi \ii} \int_\Gamma \frac{(\rho(\lambda') - I)(G_Y(\lambda') - I)}{\lambda'} d\lambda.
\end{align}
Notice that the supremum norm of \( (G_Y(\lambda')-I)/\lambda'\) on \( \ii [-\tfrac12,\tfrac12]\) has an upper bound as in \eqref{G_Y-I L2 norm} because
\[
y^{-2}e^{4x\theta(\ii y)}\leq y^{-2}e^{-\frac xy} \leq 4e^{-2x} \leq 4x^{-\frac14}e^{-x}, \quad y\in\big[0,\tfrac12\big].
\]
Hence, \eqref{G_Y-I L2 norm} also holds with \( G_Y(\lambda')-I\) replaced by \( (G_Y(\lambda')-I)/\lambda'\). Thus,
\begin{align}
\label{representation Y(0)}
Y(0) = I + \frac{1}{2\pi \ii} \int_\Gamma \frac{G_Y(\lambda') - I}{\lambda'} d\lambda' +  \O\left( \frac{e^{-2\varkappa_1x}}{\sqrt x} \right)
\end{align}
by Cauchy-Schwarz inequality and \eqref{gl rho_Y - I has a small norm}, where, again, the constant in \( \O(\cdot) \) is an absolute multiple of \( |s^\R|_+ \). On the imaginary axis, we have that
\[
\frac{1}{2\pi \ii} \int_{\Gamma\cap\ii\R} \frac{G_Y(\lambda')-I}{\lambda'}d\lambda'  = \left( \frac{s^\R}{2\pi }\int_0^{\ii\infty} e^{2x\theta(\lambda^\prime)}\frac{d\lambda^\prime}{\lambda^\prime} \right) \sigma_1,
\]
where we made a substitution \( \lambda^\prime \mapsto -\lambda^\prime \) in the \((2,1)\)-entry. Using \eqref{theta on imaginary} and \cite[(10.32.9)]{NIST:DLMF} we can rewrite the above integral as
\begin{equation}
    \label{integral K0}
    \frac{s^\R}{2 \pi} \int_{0}^{\infty} e^{-\frac{x}{2}(y + \frac{1}{y})} \frac{dy}{y} = \frac{s^\R}{2\pi} \int_{-\infty}^\infty e^{-x\cosh s}ds = \frac{s^\R}\pi K_0(x).
\end{equation}
On the other hand, on the unit circle, we get that
\begin{align*}
\frac{1}{2\pi \ii} \int_{S^1} \frac{G_Y(\lambda')-I}{\lambda'}d\lambda'  & = \big(\Re(A)-1\big) I - \left(\frac{B^\R}{2\pi \ii} \int_{S^1} e^{2x\theta(\lambda^\prime)}\frac{d\lambda^\prime}{\lambda^\prime} \right) \sigma_1 \\
& = \big(\Re(A)-1\big) I - B^\R I_0(x)\sigma_1,
\end{align*}
where we again made a substitution \( \lambda^\prime \mapsto -\lambda^\prime \) in the \((2,1)\)-entry and the second equality follows from a computation identical to \eqref{L2 norm estimate 1}. Altogether, we have shown that
\begin{equation}
    Y(0) = \Re(A) I + \left( \frac{s^\R}\pi K_0(x) - B^\R I_0(x)\right)\sigma_1 + \O\left( \frac{e^{-2\varkappa_1x}}{\sqrt x} \right),
    \label{gl Y asymptotic eq 2}
\end{equation}
 where the constant in \( \O(\cdot) \) is an absolute multiple of \( |s^\R|_+ \). 
 
It immediately follows from RHP~\ref{RHP Y}(3) and the definition of \( P_0 \) in RHP~\ref{rhp original}(3) that
\begin{equation}
    Y(0) = P_0 = \begin{pmatrix}
        \cosh \frac{u}{2} & \sinh \frac{u}{2}\\
        \sinh \frac{u}{2} & \cosh \frac{u}{2}
    \end{pmatrix}. \label{gl Y asymptotic eq 3}
\end{equation}
By comparing \eqref{gl Y asymptotic eq 2} with \eqref{gl Y asymptotic eq 3}, we get from \eqref{A is almost 1} that
\[
u(x) = 2\ln \left[ 1 + \frac{s^\R}\pi K_0(x) - B^\R I_0(x) + \O\left( \frac{e^{-2\varkappa_1x}}{\sqrt x} \right) \right]
\]
for all \( \varkappa_1 x \geq D_1 + \ln|s^\R|_+ \). Since it holds that
\[
\left|\frac{s^\R}\pi K_0(x) - \varepsilon B^\R I_0(x)\right| \leq |s^\R|_+\frac{C_4}2 \frac{e^{-x}+e^{-\varkappa_1 x}}{\sqrt x} \leq |s^\R|_+C_4\frac{e^{-\varkappa_1 x}}{\sqrt x}
\]
for some absolute constant \( C_4 \), we have that
\[
u(x) = \frac{2 s^\R}\pi K_0(x) - 2 e^{-\varkappa x} I_0(x) + \O\left( \frac{e^{-2\varkappa_1x}}{\sqrt x} \right)
\]
in the considered range. If we no longer assume that \( B^\R>0 \) and \( \Re(A)>0 \) in \eqref{mon data}, then \eqref{symmetries of u} together with the above formula gives
\[
u\big(x, (A,s^\R,B^\R)\big) = \varepsilon \left( \frac{2 \varepsilon s^\R}\pi K_0(x) - 2 e^{-\varkappa x} I_0(x) \right) + \O\left( \frac{e^{-2\varkappa_1x}}{\sqrt x} \right),
\]
where the constant in \( \O(\cdot) \) is an absolute multiple of \( |s^\R|_+ \), which is exactly \eqref{asymptotic formula 1} in the considered range. Finally, notice that the remarks made right after \eqref{G_Y-I infinity norm} and \eqref{G_Y-I L2 norm} imply that when \( s^\R=0 \), the above formula holds with \( \varkappa_1 \) replaced by \( \varkappa-1\) and is valid for all \( (\varkappa-1)x\geq D_1\). This finishes the proof of Theorem~\ref{thm:1}.

\section{Proof of Theorem~\ref{thm:2}}

In this section, we continue to work with the matrix function \( Y(\lambda) \) from \eqref{defn of Y}. We keep using all the notation introduced during the proof of Theorem~\ref{thm:1}.

\subsection{Modified Riemann-Hilbert Problem}
\label{ss:3.1}

It is clear from the analysis of the previous section that the jump matrices in RHP~\ref{RHP Y}(2) are close to the identity except in the vicinity of \( \pm\ii \) when \( \varkappa=1\). To overcome this difficulty, we first isolate the parts of the jumps that cause the difficulty. To this end, let \( U_\ii \) and \( U_{-\ii} \) be disks of fixed radius \( \varrho\in(0,1)\) centered at \( \ii \) and \( -\ii \), respectively. We orient \( \partial U_\ii \) and \( \partial U_{-\ii} \) {\it clockwise}. Define \( \tilde Y(\lambda) \) to be equal to \( Y(\lambda) \) everywhere except within \( U_\ii\cap\{\Re(\lambda) < 0\} \) and \( U_{-\ii}\cap\{\Re(\lambda)>0\} \),  where we set
\begin{equation}
\label{define W tilde}
\tilde Y(\lambda) := Y(\lambda) \begin{cases}  
\begin{pmatrix} 1 & -\ii s^\R e^{2 x \theta(\lambda)} \\ 0 & 1 \end{pmatrix},
& \lambda \in U_\ii\cap\{\Re(\lambda) < 0\}, \medskip \\
\begin{pmatrix} 1 & 0 \\ -\ii s^\R e^{-2x \theta(\lambda)} & 1 \end{pmatrix},
& \lambda \in U_{-\ii}\cap\{\Re(\lambda)>0\}.
\end{cases}
\end{equation}
Clearly, this modification of \( Y(\lambda) \) moves its jump away from the imaginary axis within \( U_{\pm\ii}\) to \( \partial U_{\pm\ii} \). Moreover, using \( A-\overline A=\ii s^\R B^\R \), one can readily verify that
\[
\tilde Y_+(\lambda) = \tilde Y_-(\lambda)e^{x\theta(\lambda)\sigma_3}\begin{Bmatrix}
    \begin{pmatrix}
    \overline A & -B^\R \\ 
    -B^\R & A
    \end{pmatrix}, &\lambda \in U_\ii\cap S^1 \medskip \\
    \begin{pmatrix}
    A & -B^\R \\
    -B^\R & \overline A
    \end{pmatrix}, & \lambda \in U_{-\ii}\cap S^1 \end{Bmatrix} e^{-x\theta(\lambda)\sigma_3}.
\]
That is, the middle matrices above no longer depends on the considered half-plane. Next, using \( |A|^2=1+(B^\R)^2\), we can obtain that
\begin{align}
\begin{pmatrix}
\overline A & -B^\R \\
-B^\R & A 
\end{pmatrix} = 
\begin{pmatrix}
1/A & - B^\R \\
0 & A
\end{pmatrix}
\begin{pmatrix}
1 & 0 \\
-B^\R/A & 1
\end{pmatrix}
\end{align}
and
\begin{align}
\begin{pmatrix}
A & -B^\R \\
-B^\R & \overline A 
\end{pmatrix} =
\begin{pmatrix}
A & 0 \\
- B^\R & 1/A
\end{pmatrix}
\begin{pmatrix}
1 & -B^\R/A \\
0 & 1
\end{pmatrix}.
\end{align}
Therefore, we set
\begin{equation}
\label{define W hat}
\hat Y(\lambda) := \tilde Y(\lambda) \begin{cases} 
    \begin{pmatrix} 
    1/A & -B^\R e^{2x \theta(\lambda)} \\
    0 & A
    \end{pmatrix}, & \lambda\in U_\ii\cap \{|\lambda|>1\}, \medskip \\
    \begin{pmatrix}
    A & 0\\
    -B^\R e^{-2x \theta(\lambda)} & 1/A
    \end{pmatrix}, & \lambda\in U_{-\ii} \cap \{|\lambda|>1\},
\end{cases}
\end{equation}
and let \( \hat Y(\lambda) \) be equal to \( \tilde Y(\lambda) \) everywhere else. Clearly, \( \hat Y(\lambda) \) still satisfies RHP~\ref{RHP Y} except for RHP~\ref{RHP Y}(2), which now only holds on \( \Gamma \setminus (\overline U_\ii \cup \overline U_{-\ii}) \), while
\begin{equation}
    \label{jump Y hat on S1}
    \hat Y_+(\lambda) = \hat Y_-(\lambda) \begin{cases}
    \begin{pmatrix}
    1 & 0 \\ 
    -(B^\R/A)e^{-2x\theta(\lambda)} & 1
    \end{pmatrix}, &\lambda \in U_\ii\cap S^1 \medskip \\
    \begin{pmatrix}
    1 & -(B^\R/A)e^{2x\theta(\lambda)} \\
    0 & 1
    \end{pmatrix}, & \lambda \in U_{-\ii}\cap S^1, \end{cases}
\end{equation}
and the jump relations on \( \partial U_\ii\cup\partial U_{-\ii} \) easily follow from \eqref{define W tilde} and \eqref{define W hat}, see Figure~\ref{Y-hat problem pic}. 
\begin{figure}[ht!]
    \centering
    \includegraphics[width=12cm]{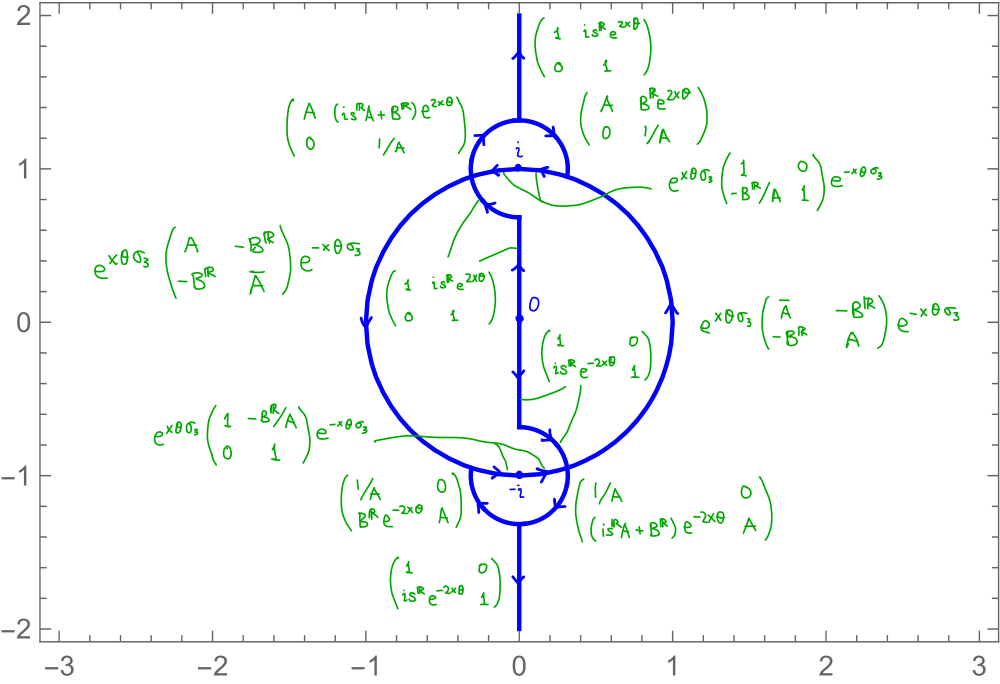}
    \caption{The contour $\hat\Gamma$ and the jump matrices \( G_{\hat Y}(\lambda) \).}
    \label{Y-hat problem pic}
\end{figure}
Thus, \( \hat Y(\lambda) \) solves the following Riemann-Hilbert problem.
\begin{RHP}
\label{RHP Y hat}
Find a $2 \times 2$ matrix function $\hat Y(\lambda)$ such that
\begin{enumerate}
    \item $\hat Y(\lambda)$ is analytic in $\C \setminus \hat\Gamma$, where contour $\hat\Gamma$ is depicted on Figure~\ref{Y-hat problem pic};
    \item one-sided traces $\hat Y_\pm(\lambda)$ exist a.e. on \( \hat\Gamma \), belong to \( L^2(\hat\Gamma) \), and satisfy
    \begin{align*}
        \hat Y_+(\lambda) = \hat Y_-(\lambda) G_{\hat Y}(\lambda), \quad \lambda \in \hat\Gamma,
    \end{align*}
    where the jump matrices $G_{\hat Y}(\lambda)$ on \( \hat\Gamma \) are as on Figure~\ref{Y-hat problem pic};
    \item it holds that
    \begin{align*}
        \hat Y(\lambda) = \begin{cases}
            I + \O(1/\lambda)  &  \text{as} \quad \lambda\to\infty, \smallskip \\
            P_0 (I + \O(\lambda)) & \text{as} \quad \lambda\to0.
        \end{cases}
    \end{align*}
\end{enumerate}
\end{RHP}

\subsection{Local Parametrices}
\label{subsec:local_par1}

In this subsection we construct matrices \( P^{(\pm \ii)}(\lambda) \) that have the same jump on \( S^1\cap U_{\pm\ii} \) as \( \hat Y(\lambda) \) and have a certain prescribed behavior on \( \partial U_{\pm\ii} \).

\subsubsection{Model Local Parametrix}

Let \( \mathrm{erfc}(\cdot)\) be the complimentary error function, see \cite[(7.2.2)]{NIST:DLMF}. Set
\begin{align}
\label{defn of He0}
He_0(z) := \begin{pmatrix}
    1 & \mathcal H_0(z) \\ 0 & 1
\end{pmatrix}, \quad z \in \C \setminus \R,
\end{align}
where \( \mathcal H_0(z) \) is a sectionally holomorphic function in \( \C\setminus\R \) given by
\begin{align}
\label{defn of H0}
\mathcal H_0(z) := \pm \tfrac12 \mathrm{erfc}(\mp \ii z)e^{-z^2}, \quad \pm\Im(z)>0.
\end{align}
It readily follows from \cite[(7.4.2)]{NIST:DLMF} that
\begin{equation}
\label{H0 jump}
\mathcal H_{0+}(s) - \mathcal H_{0-}(s) = e^{-s^2}, \quad s\in\R.
\end{equation}
Moreover, it follows from \cite[(7.12.1)]{NIST:DLMF} that \( \mathcal H_0(z) \) has the following asymptotic expansion
\begin{equation}
\label{H0 expansion}
\mathcal H_0(z) \sim \frac{-1}{2\ii \sqrt{\pi}}\sum_{k=0}^\infty \frac{(1/2)_k}{z^{2k+1}},
\end{equation}
which holds uniformly for all \( |z| \) separated away from \( 0 \). 

Furthermore, set \( \gamma_k = 2^k/k! \) for \( k\geq 0\). Given a natural number \( m>0 \), define
\[
He_m(z) = \begin{pmatrix} p_m (z) & \mathcal{H}_m(z) \\ -2\ii \sqrt{\pi} \gamma_{m-1} p_{m-1}(z) & -2\ii \sqrt{\pi} \gamma_{m-1} \mathcal{H}_{m-1}(z) \end{pmatrix}, \quad z \in \C \setminus \R,
\]
where \( p_m(z) \) is the monic Hermite orthogonal polynomial of degree \( m \):
\[
\int_\R p_m(s)s^k e^{-s^2}ds =0, \quad k\in\{0,\ldots,m-1\},
\]
(\( p_m(z)=2^{-m}H_m(z)\) in the notation of \cite[Table 18.3.1]{NIST:DLMF}) and
\begin{align}
\label{defn of Hm}
\mathcal{H}_m(z) := \frac{1}{2\pi \ii} \int_\R \frac{p_m(s) e^{-s^2}}{s - z} ds.
\end{align}
It readily follows from Sokhotski-Plemelj formulae that
\begin{equation}
\label{Hm jump}
\mathcal H_{m+}(s) - \mathcal H_{m-}(s) = p_m(s)e^{-s^2}, \quad s\in\R.
\end{equation}
Since \( p_0(z)\equiv 1\), \eqref{H0 jump} and \eqref{Hm jump} with \( m =0 \) are the same. As both expressions \eqref{defn of H0} and \eqref{defn of Hm} yield functions vanishing at infinity, one can deduce that \eqref{defn of Hm} when \( m =0 \) specializes to \eqref{defn of H0}. It is known, see \cite[Equation 18.5.13]{NIST:DLMF}, that
\begin{align}
\label{Hermite sum}
p_m(z)  = \left( \sum_{l=0}^{\left\lfloor m/2\right\rfloor} \frac{\gamma_{m-2l}}{\gamma_m}  \frac{(-1)^{l}}{l!} z^{-2l} \right) z^m.
\end{align}
Furthermore, it follows from orthogonality that
\[
\mathcal H_m(z) = -\frac{1}{2\pi \ii} \sum_{l=m}^{N-1} \frac{1}{z^{l+1}} \left( \int_\R p_m(s) \, s^l e^{-s^2} ds \right) + \frac1{z^N}\frac{1}{2\pi \ii} \int_\R \frac{s^N p_m(s) e^{-s^2}}{s - z} ds
\]
for any \( N>m\). Notice that if \( l \) and \( m \) have different parity, the integrals in the sum above are zero. Otherwise, we can compute using \cite[Equations~(12) and (13)]{BDQ} that
\begin{align}
\int_\R p_m(s) \, s^{2k+m} e^{-s^2} ds = \frac{\sqrt \pi}{\gamma_{2k+m} k!}
\end{align}
for \( k\geq0 \). Therefore, it holds that
\begin{equation}
\label{Hm expansion}
\mathcal H_m(z) \sim \frac{-1}{2\ii \sqrt{\pi}} \left(\sum_{k=0}^\infty \frac{1}{\gamma_{2k+m}k!} \frac{1}{z^{2k+1}} \right) z^{-m},
\end{equation}
uniformly for all \( |z| \) large (one can readily check that expansions \eqref{H0 expansion} and \eqref{Hm expansion} coincide when \( m=0 \)). Altogether, we see that the matrix function \( He_m(z) \) solves the following Riemann-Hilbert problem.
\begin{RHP} 
\label{RHP Hermite}
Given a non-negative integer $m$, find a $2 \times 2$ matrix function $He_m(z)$ such that
\begin{enumerate}
    \item $He_m(z)$ is analytic in $\C \setminus \R$;
    \item there exist continuous traces $He_{m\pm}(s)$ on \( \R \) that satisfy
    \begin{align}
        He_{m+}(s) = He_{m-}(s) \begin{pmatrix}
            1 & e^{-s^2} \\
            0 & 1
        \end{pmatrix}, \quad s \in \R;
    \end{align}
    \item it holds as $z \to \infty$ that
    \[
    He_m(z) = \left( I + \frac1z \begin{pmatrix}
         0 & -1/(2\ii \sqrt{\pi}\gamma_{m}) \\
        -2\ii \sqrt{\pi} \gamma_{m-1} & 0
    \end{pmatrix} + \O_m \big(z^{-2} \big) \right) z^{m\sigma_3},
    \]
    where \( \gamma_{-1}=0 \) and we point out that the coefficient matrices in parenthesis above are off diagonal next to odd powers of \( z \) and diagonal next to even powers of \( z \) by \eqref{Hermite sum} and \eqref{Hm expansion}.
\end{enumerate}
\end{RHP}

\subsubsection{Local Parametrix around \( \ii \)}

To use model parametrices \( He_m(z) \) from the previous subsection, we need local conformal maps around \( \pm \ii \). To this end, observe that
\begin{equation}
\label{theta to map}
1 \pm 2\theta(\lambda) = \pm \frac{\ii}{2\lambda}(\lambda \mp \ii)^2 =: \phi_{\pm\ii}^2(\lambda),
\end{equation}
where \( \phi_{\pm \ii}(\lambda) \) is conformal around \( \pm \ii\), in particular, in the disk \( U_{\pm\ii}\), and
\begin{equation}
\begin{cases} \label{phi branch}
\pm \phi_\ii(\lambda) > 0, & \lambda \in U_\ii\cap S^1, \quad \mp\Re(\lambda)>0, \\ 
\pm \phi_{-\ii}(\lambda) > 0, & \lambda \in U_{-\ii}\cap S^1, ~~ \pm\Re(\lambda)>0,
\end{cases}
\end{equation}
which fixes the branches of the square roots (\(\phi'_\ii(\ii)=-1/\sqrt{2}\) and \(\phi'_{-\ii}(-\ii)=1/\sqrt{2}\)). 

Recall \eqref{D2} and \eqref{mon data}. That is,
\begin{align}
\label{B thm2}
B^\R = e^{-\varkappa x} = e^{-x}x^{m+\nu}.
\end{align}
We define local parametrix around \( \ii \) by
\begin{align}
\label{P(i)}
P^{(\ii)}(\lambda) := E^{(\ii)}(\lambda) He_m\big( \sqrt x\phi_\ii(\lambda)\big) (-x^{m+\nu}/A)^{-\frac12\sigma_3}\sigma_1 , \quad \lambda\in U_\ii,
\end{align}
where \( E^{(\ii)}(\lambda) \) is an analytic matrix function in \( U_{\ii} \) given by
\begin{align}
\label{Ei}    
E^{(\ii)}(\lambda) := \sigma_1 \left(-\frac{x^\nu}A\right)^{\frac12\sigma_3} \left( \phi_\ii(\lambda) \frac{\lambda+\ii}{\lambda-\ii}\right)^{-m\sigma_3}, \quad \lambda \in U_\ii,
\end{align}
where the square roots of \( -1/A \) in \eqref{P(i)} and \eqref{Ei} must coincide. Then, one can readily show that $P^{(\ii)}(\lambda)$ satisfies the following Riemann-Hilbert problem.

\begin{RHP}
\label{RHP Y loc i}
Find a $2 \times 2$ matrix function $P^{(\ii)}(\lambda)$ such that
\begin{enumerate}
    \item $P^{(\ii)}(\lambda)$ is analytic in $U_{\ii} \setminus S^1$;
    \item one-sided traces $P^{(\ii)}_\pm(\lambda)$ are continuous on \( U_{\ii} \cap S^1 \) and satisfy
    \begin{align}
        P^{(\ii)}_+(\lambda) = P^{(\ii)}_-(\lambda) \begin{pmatrix} 1 & 0 \\ -(B^\R/A) e^{-2x\theta(\lambda)} & 1 \end{pmatrix};
    \end{align}
    \item it holds for $\lambda \in \partial U_\ii$ that
    \begin{align}
        P^{(\ii)}(\lambda) = T_i^{(\ii)}(\lambda)\big( I + \O_m\big(x^{-1}\big) \big) \left(\frac{\lambda+\ii}{\lambda-\ii}\right)^{m\sigma_3}
    \end{align}
    for \( i\in\{1,2\} \), where
\begin{align}
    T_1^{(\ii)}(\lambda) = \begin{pmatrix}
    1 & \displaystyle \frac{t^{(\ii)}_{12}(\lambda;x)}{\lambda - \ii} \smallskip \\
    \displaystyle \frac{t^{(\ii)}_{21}(\lambda;x)}{\lambda - \ii} & \displaystyle 1 + \frac{t_{21}^{(\ii)}(\lambda;x) t_{12}^{(\ii)}(\lambda;x)}{(\lambda - \ii)^2}
\end{pmatrix}
\end{align}  
and
\begin{align}
    T_2^{(\ii)}(\lambda) = \begin{pmatrix}
    \displaystyle 1 + \frac{t_{21}^{(\ii)}(\lambda;x) t_{12}^{(\ii)}(\lambda;x)}{(\lambda - \ii)^2} & \displaystyle \frac{t^{(\ii)}_{12}(\lambda;x)}{\lambda - \ii} \smallskip \\
    \displaystyle \frac{t^{(\ii)}_{21}(\lambda;x)}{\lambda - \ii} & 1
\end{pmatrix},
\end{align}  
while $t^{(\ii)}_{12}(\lambda;x)$ and $t^{(\ii)}_{21}(\lambda;x)$ are holomorphic functions on $U_\ii$ given by
\begin{align}
\label{t i is small}
\begin{cases}
    t^{(\ii)}_{12}(\lambda;x) & \displaystyle  = 2\ii \sqrt{\pi} \gamma_{m-1} A \left( \frac{\phi_\ii(\lambda)}{\lambda - \ii}\right)^{2m-1} \frac{(\lambda + \ii)^{2m}}{x^{\nu+\frac12}}, \smallskip \\
    t^{(\ii)}_{21}(\lambda;x) & \displaystyle = \frac{1}{2\ii\sqrt{\pi}} \frac{1}{\gamma_m} \frac{1}{A} \left( \frac{\phi_\ii(\lambda)}{\lambda - \ii}\right)^{-2m-1} \frac{x^{\nu-\frac12}}{(\lambda + \ii)^{2m}}.
\end{cases}
\end{align}
\end{enumerate}
\end{RHP}

We only need to explain RHP~\ref{RHP Y loc i}(3). Assuming \( \lambda\in\partial U_{\ii} \), we get from RHP~\ref{RHP Hermite}(3), \eqref{P(i)}, and \eqref{Ei} that
\begin{multline}
P^{(\ii)}(\lambda) = \Bigg( I + \frac{1}{\sqrt{x} \phi_{\ii}(\lambda)} E^{(\ii)}(\lambda) \begin{pmatrix}
         0 & -1/(2\ii \sqrt{\pi}\gamma_{m}) \\
        -2\ii \sqrt{\pi} \gamma_{m-1} & 0
    \end{pmatrix} E^{(\ii)}(\lambda)^{-1}\\
    + \O_m \big(x^{-1}\big) \Bigg) \left(\frac{\lambda+\ii}{\lambda-\ii}\right)^{m\sigma_3},
\end{multline}
where we also used the fact that \( E^{(\ii)}(\lambda) = \O( x^{\frac12|\nu|}) \) and that the next term in the expansion is diagonal. Thus, it holds that
\begin{align} 
\label{P i on the bdry}
P^{(\ii)}(\lambda) = \left( I + \frac1{\lambda-\ii}\begin{pmatrix}
    0 & t^{(\ii)}_{12}(\lambda;x) \\
    t^{(\ii)}_{21}(\lambda;x) & 0
\end{pmatrix} + \O_m\big(x^{-1}\big) \right) \left(\frac{\lambda+\ii}{\lambda-\ii}\right)^{m\sigma_3}.
\end{align}
Observe that it holds uniformly in  \( \partial U_1 \) that
\begin{align}
\label{product of ts}
t^{(\ii)}_{12}(\lambda;x)t^{(\ii)}_{21}(\lambda;x) = \frac m{2x} \left( \frac{\lambda-\ii}{\phi_\ii(\lambda)} \right)^2.
\end{align}
Thus, it only remains to notice that each \( T_i^{(\ii)}(\lambda) \) has unit determinant and
\[
T_i^{(\ii)}(\lambda)^{-1} \left( I + \frac1{\lambda-\ii}\begin{pmatrix}
    0 & t^{(\ii)}_{12}(\lambda;x) \\
    t^{(\ii)}_{21}(\lambda;x) & 0
\end{pmatrix} \right) = I + \O_m\left(\frac 1x\right).
\]

\subsubsection{Local Parametrix around \( -\ii \)}

The construction around \( -\ii \) is essentially identical. In fact, we simplify the presentation even further by setting
\begin{align}
\label{P-i}
P^{(-\ii)}(\lambda) := \sigma_1 P^{(\ii)}(-\lambda) \sigma_1, \quad \lambda\in U_{-\ii}.
\end{align}
Then, $P^{(-\ii)}(\lambda)$ satisfies the following Riemann-Hilbert problem.
\begin{RHP}
\label{RHP Y loc -i}
Find a $2 \times 2$ matrix function $P^{(-\ii)}(\lambda)$ such that
\begin{enumerate}
    \item $P^{(-\ii)}(\lambda)$ is analytic in $U_{-\ii} \setminus S^1$;
    \item one-sided traces $P^{(-\ii)}_\pm(\lambda)$ are continuous on \( U_{-\ii} \cap S^1 \) and satisfy
    \begin{align}
        P^{(-\ii)}_+(\lambda) = P^{(-\ii)}_-(\lambda) \begin{pmatrix} 1 & -(B^\R/A) e^{2x\theta(\lambda)} \\ 0 & 1 \end{pmatrix};
    \end{align}
    \item it holds for $\lambda \in \partial U_{-\ii}$ that
    \begin{align}
        P^{(-\ii)}(\lambda) = T_i^{(-\ii)}(\lambda)\big( I + \O_m\big(x^{-1}\big) \big) \left(\frac{\lambda+\ii}{\lambda-\ii}\right)^{m\sigma_3},
    \end{align}
    for \( i\in\{1,2\} \), where \( T_i^{(-\ii)}(\lambda) = \sigma_1 T_i^{(\ii)}(-\lambda) \sigma_1 \).
\end{enumerate}
\end{RHP}

\subsection{Small Norm Problem}

Recall that matrices \( P^{(\pm\ii)}(\lambda) \) are holomorphic in \( U_{\pm\ii}\), but \( T_i^{(\pm\ii)}(\lambda) \) have poles at the centers of the respective disks. Put
\begin{align}
\label{taus}
\tau_{12} = t^{(\ii)}_{12}(\ii;x) \quad \text{and} \quad \tau_{21} = t^{(\ii)}_{21}(\ii;x).
\end{align}
Observe that \( \tau_{12}\tau_{21} = mx^{-1} \) by \eqref{product of ts} and since \( |\phi_\ii^\prime(\ii)|=1/\sqrt 2\). Hence, one of these constants is less than one in absolute value when \( x>m \). We look for the solution of RHP~\ref{RHP Y hat} in the following form:
\begin{align}
\label{defn of Z hat}    
\hat Y(\lambda) =: \hat Z(\lambda) \begin{cases}
T_i^{(\pm\ii)}(\lambda)^{-1}P^{(\pm\ii)}(\lambda), & \lambda \in U_{\pm\ii}, \smallskip \\    
\displaystyle \left(\frac{\lambda+\ii}{\lambda-\ii}\right)^{m\sigma_3}, & \text{otherwise},
\end{cases}
\end{align}
where we take \( i=1 \) if \( |\tau_{12}|<1 \) and \( i=2 \) otherwise. Then, \( \hat Z(\lambda) \) must solve the following meromorphic Riemann-Hilbert problem.

\begin{RHP}
\label{RHP hat Z}
Find a $2 \times 2$ matrix function $\hat Z(\lambda)$ such that
\begin{enumerate}
    \item $\hat Z(\lambda)$ is analytic in $\C \setminus (\Gamma_Z\cup\{\pm\ii\})$, where $\Gamma_Z$ is depicted on Figure~\ref{hat Z problem pic};
    \item one-sided traces $\hat Z_{\pm}(\lambda)$ exist a.e. on \( \Gamma_Z \), belong to \( L^2(\Gamma_Z) \), and satisfy
    \[
        \hat Z_{+}(\lambda) = \hat Z_{-}(\lambda) G_Z(\lambda), 
    \]
    where the jump matrices \(G_Z(\lambda)\) are as on Figure~\ref{hat Z problem pic};
    \item it holds that \( \hat Z(\lambda) = I + \O(1/\lambda) \) as \( \lambda\to\infty \);
    \item $\hat Z(\lambda)$ has double poles at \( \ii \) and \( -\ii \) with principal parts such that the products \( \hat Z(\lambda)T_i^{(\pm\ii)}(\lambda)^{-1} \) are analytic in the respective disks \( U_{\ii} \) and \( U_{-\ii} \).
\end{enumerate}
\end{RHP}

\begin{figure}[ht!]
    \centering
    \includegraphics[width=8cm]{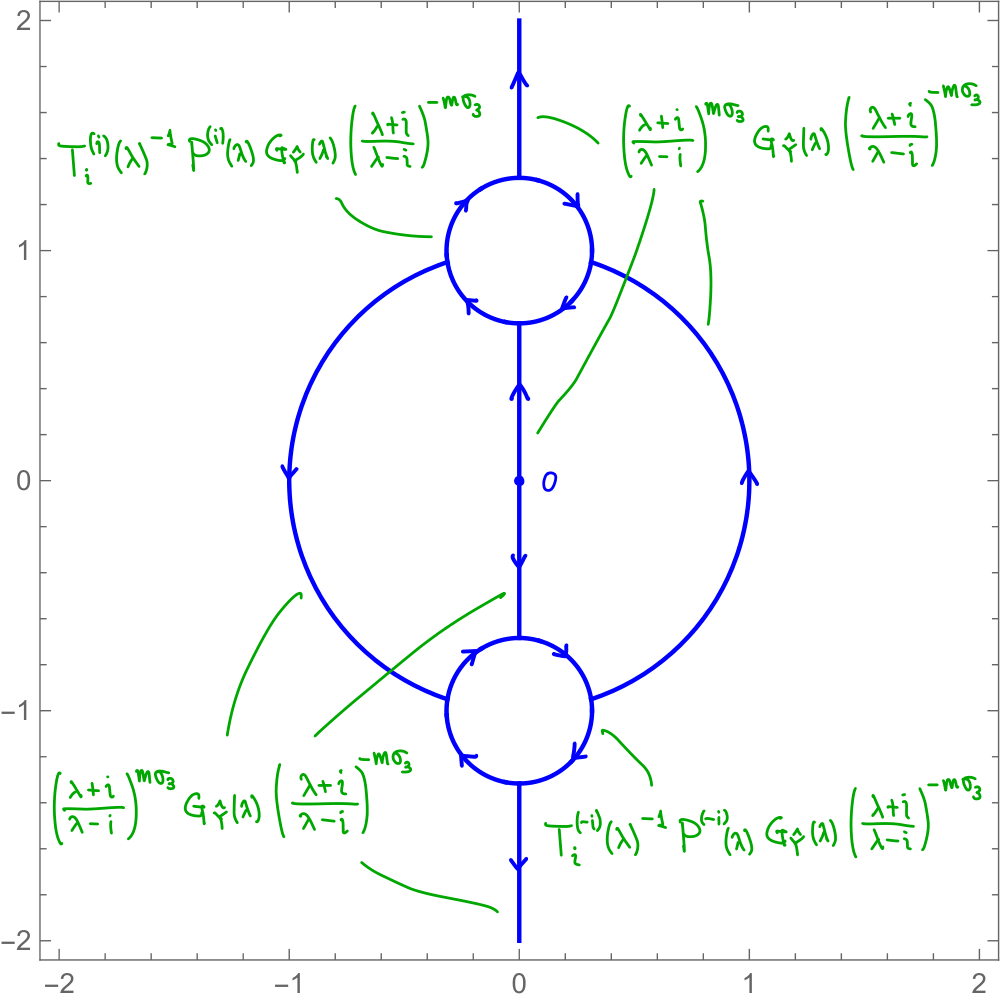}
    \caption{The contour $\Gamma_Z$ and the jump matrices $G_Z(\lambda)$. Note that $G_{\hat Y}(\lambda)$ is stated on Figure \ref{Y-hat problem pic}.}
    \label{hat Z problem pic}
\end{figure}

We look for the solution of RHP~\ref{RHP hat Z} in the form
\begin{align} 
\label{defn of Z}
    \hat Z(\lambda) = \left( I + \frac{\lambda B_1}{\lambda^2 + 1} + \frac{C_1}{\lambda^2 + 1}\right) \left( I + \frac{\lambda B_2}{\lambda^2 + 1} + \frac{C_2}{\lambda^2 + 1}\right) Z(\lambda),
\end{align}
where the matrices $B_k, C_k$, $k=1,2$, are independent of $\lambda$ and will be specified shortly, while $Z(\lambda)$ solves the following small norm Riemann-Hilbert problem.

\begin{RHP}
\label{RHP Z}
Find a $2 \times 2$ matrix function $Z(\lambda)$ such that
\begin{enumerate}
    \item $Z(\lambda)$ is analytic in $ \C \setminus \Gamma_Z$;
    \item one-sided traces $Z_{\pm}(\lambda)$ exist a.e. on \( \Gamma_Z \), belong to \( L^2(\Gamma_Z) \), and satisfy
    \[
        Z_{+}(\lambda) = Z_{-}(\lambda) G_{Z}(\lambda); 
    \]
    \item it holds that \( Z(\lambda) = I + \O(1/\lambda) \) as \( \lambda\to\infty \).
\end{enumerate}
\end{RHP}

Let us now show that RHP~\ref{RHP Z} is indeed a small norm problem and therefore has a solution for all sufficiently large \( x \) by the small norm theorem, see \eqref{gl Y singular int eq}. Since asymptotic formulae in  RHP~\ref{RHP Y loc i}(3) and RHP~\ref{RHP Y loc -i}(3) hold for each \( i\in\{1,2\} \), our choice of \( i \) does not affect the forthcoming computations (however, this choice will be important for the construction of the rational prefactor in \eqref{defn of Z}).

We start by specifying precisely the disks \( U_{\pm\ii} \). We take their radius \( \varrho \) to be \( \tfrac1{\sqrt 2}\). Then these disks intersect \( S^1 \) at points with imaginary parts \( \pm \tfrac34\). That is,
\begin{align}
\label{bound Ui}
|\Im(\lambda)| \leq \tfrac 34, \quad \lambda\in S^1\setminus(U_\ii\cup U_{-\ii})
\end{align}
(this particular choice of \( \varrho\) is motivated by the simplicity of the above bound and nothing else). It further follows from the maximum modulus principle and the triangle inequality that
\begin{equation}
\label{max on Gamma Z}
\max_{\lambda\in\Gamma_Z} \left| \frac{\lambda+\ii}{\lambda-\ii}\right|^{\pm 1} \leq \max_{\lambda\in\partial U_\ii} \left| \frac{\lambda+\ii}{\lambda-\ii}\right| \leq 1 + \frac{2}{\varrho} < e^{\frac32}.
\end{equation}

Now we are ready to estimate the deviation of the jump matrix \( G_{Z}(\lambda) \) from the identity matrix \( I \). Recall \eqref{theta on imaginary}. The form of \( G_{\hat Y}(\lambda)\) on the imaginary axis, see Figure~\ref{Y-hat problem pic}, and \eqref{max on Gamma Z} yield that
\begin{equation}
\label{G_z 1}
    \|G_{Z}-I\|_{L^\infty(\Gamma_Z\cap \ii\R)} \leq |s^\R|e^{3m-x}.
\end{equation}
Similarly, using \eqref{L2 norm estimate 2}, we obtain that
\begin{equation}
    \label{G_z 2}
    \|G_{Z}-I\|_{L^2(\Gamma_Z\cap \ii\R)} \leq C_1|s^\R| x^{-\frac14}e^{3m- x}
\end{equation}
for some absolute constant \( C_1 \) (the constants \( C_i\) appearing here are not related to identically labeled constants from the proof of Theorem~\ref{thm:1}). Meanwhile, on the unit circle, it holds that
\begin{align*}
    G_{Z}(\lambda) = \begin{pmatrix} 1+\O(|A-1|) & - x^{m+\nu} e^{-x(1+\Im(\lambda))} \big(\frac{\lambda+\ii}{\lambda-\ii} \big)^{2m} \\ -x^{m+\nu} e^{-x(1-\Im(\lambda))}  \big(\frac{\lambda-\ii}{\lambda+\ii} \big)^{2m} & 1+\O(|A-1|) \end{pmatrix},
\end{align*}
where we used \eqref{B thm2} and one needs to recall that \( 2\theta(\lambda) = -\Im(\lambda)\) for \( \lambda \in S^1 \). Thus, we get from \eqref{A is almost 1}, \eqref{bound Ui}, and \eqref{max on Gamma Z} that
\begin{equation}
    \label{G_z 3}
    \|G_{Z}-I\|_{L^\infty(\Gamma_Z\cap S^1)} \leq C_2 |s^\R|_+ x^{m+\nu} e^{3m-\frac14x}
\end{equation}
for some absolute constant \( C_2 \). Next, it readily follows from RHP~\ref{RHP Y loc i}(3) that
\[
G_{Z}(\lambda) = \left(I + \O_m \big( x^{-1} \big)\right) \left( \frac{\lambda+\ii}{\lambda-\ii} \right)^{m\sigma_3} G_{\hat Y}(\lambda) \left( \frac{\lambda-\ii}{\lambda+\ii} \right)^{m\sigma_3}, \quad \lambda\in\partial U_\ii.
\]
This means that
\[
G_{Z}(\lambda) = \left(I + \O_m \big( x^{-1} \big)\right) \begin{pmatrix}
    A & (\ii s^\R A+B^\R) e^{2x\theta(\lambda)} \big(\frac{\lambda+\ii}{\lambda-\ii}\big)^{2m} \\ 0 & 1/A \end{pmatrix}
\]
when \( \lambda-\ii \) belongs to the second quadrant (we need to replace \( \ii s^\R A \) by 0 when \( \lambda-\ii \) is in the first quadrant; \( A \) by 1 and \( B^\R \) by 0 when \( \lambda-\ii \) is in the third quadrant, and simply set the last matrix to be \( I \) when \( \lambda-\ii \) belongs to the fourth quadrant). We have that
\[
2\Re(\theta(\lambda)) = -\tfrac12\Im(\lambda)\big(1+|\lambda|^{-2}\big)<-\tfrac12\Im(\lambda)<-\tfrac12(1-\varrho)<-\tfrac17
\]
when \( \lambda \in \partial U_\ii \). Therefore, we get from \eqref{A is almost 1} that
\[
\begin{pmatrix}
    A & (\ii s^\R A+B^\R) e^{2x\theta(\lambda)} \big(\frac{\lambda+\ii}{\lambda-\ii}\big)^{2m} \\ 0 & 1/A \end{pmatrix} = I + \begin{pmatrix} \O\big( |s^\R|_+ x^{m+\nu}e^{-x}\big) & \O\big( |s^\R|_+ e^{3m-\frac17 x}\big) \\ 0 & \O\big( |s^\R|_+ x^{m+\nu}e^{-x}\big) \end{pmatrix},
\]
where the constants in \( \O(\cdot) \) above are absolute. Now, one can readily check that for sufficiently large constant \( D_2(m) \) it holds that
\[
\frac1x \geq |s^\R|_+ 
\begin{cases}
e^{3m-\frac17 x}, \\
x^{m+\nu}e^{-x},
\end{cases}
\]
whenever \( x\geq D_2(m)+8\ln|s^\R|_+ \). This estimate readily yields that
\begin{equation}
    \label{G_z 4}
    \|G_{Z}-I\|_{L^\infty(\partial U_\ii)} \leq C_3 x^{-1}
\end{equation}
for some constant \( C_3 = C_3(m) \). Clearly, an identical estimate holds on \( \partial U_{-\ii} \). By using the same reasoning that lead to \eqref{G_z 4} in \eqref{G_z 1}, \eqref{G_z 2}, \eqref{G_z 3} and gathering these four estimates together, we deduce that
\begin{align}
\label{G_Z-I estimate}
\|G_{Z}-I\|_{L^2(\Gamma_Z)\cap L^\infty(\Gamma_Z)} \leq C_4 x^{-1}
\end{align}
for some constant \( C_4 =C_4(m) \) and all \( x\geq D_2(m) + 8\ln |s^\R|_+ \). Therefore, as in \eqref{gl Y singular int eq}, the solution of RHP~\ref{RHP Z} exists for such \( x \) and has the form
\begin{align}
\label{int formula Z}
Z(\lambda) = I + \frac1{2\pi\ii} \int_{\Gamma_Z} \frac{\rho(\lambda^\prime)(G_Z(\lambda^\prime)-I)}{\lambda^\prime-\lambda}d\lambda^\prime,
\end{align}
where \( \rho(\lambda)=(\mathcal I- \mathcal C_{\Gamma_Z})^{-1}(I) \) and the operator \( \mathcal C_{\Gamma_Z} \) is defined as in \eqref{operator C Gamma} only on \( \Gamma_Z \) instead of \( \Gamma \) and with \( G_Y(\lambda) \) replaced by \( G_Z(\lambda) \).

For the use in the next subsection, let us observe that
\(
G_Z(\lambda) = \sigma_1 G_Z(-\lambda) \sigma_1,
\)
which can be checked by straightforward computation. Therefore, since \( Z(\lambda) \) is the unique solution of RHP~\ref{RHP Z} it must satisfy
\begin{equation}
    \label{symmetries of Z}
    Z(\lambda) = \sigma_1 Z(-\lambda) \sigma_1 \quad \Rightarrow \quad Z(\lambda) = \begin{pmatrix}
        z_{1}(\lambda) & z_{2}(-\lambda)\\ 
        z_{2}(\lambda) & z_{1}(-\lambda)
    \end{pmatrix}.
\end{equation}

\subsection{Rational Prefactor}
\label{rational prefactor}

It readily follows from  RHP~\ref{RHP Z} that \eqref{defn of Z} satisfies RHP~\ref{RHP hat Z}(1-3) with any matrices \( B_k,C_k\), \( k\in\{1,2\} \). In this subsection we choose \( B_k,C_k\) so that RHP~\ref{RHP hat Z}(4) is fulfilled.

Similarly to \eqref{symmetries of Z}, if there exists a unique solution of RHP~\ref{RHP hat Z},   it must possess the same symmetry. Hence, we will apriori assume that \( B_k = -\sigma_1 B_k \sigma_1 \) and \( C_k = \sigma_1 C_k \sigma_1 \). That is, we shall suppose that
\begin{align} 
\label{Bk Ck structure}
B_k = \begin{pmatrix}
    b_k & -d_k\\
    d_k & -b_k
\end{pmatrix}, \quad C_k = \begin{pmatrix}
    a_k & c_k\\
    c_k & a_k
\end{pmatrix}, \quad k\in\{1,2\}.
\end{align}
Now, we immediately see from \eqref{symmetries of Z}, \eqref{Bk Ck structure}, and relation \( T_i^{(-\ii)}(\lambda) = \sigma_1 T_i^{(\ii)}(-\lambda) \sigma_1 \) that if we can show that $\hat Z(\lambda) T_i^{(\ii)}(\lambda)^{-1}$ is analytic at $\ii$, then $\hat Z(\lambda) T_i^{(-\ii)}(\lambda)^{-1}$ is automatically analytic at \( -\ii \). To show analyticity at \( \ii \), observe first that
\begin{equation}
\label{hat Z T inverse}
\hat Z(\lambda) T_i^{(\ii)}(\lambda)^{-1} = \left( I + \frac{\lambda B_1}{\lambda^2 + 1} + \frac{C_1}{\lambda^2 + 1}\right) \left( I + \frac{\lambda B_2}{\lambda^2 + 1} + \frac{C_2}{\lambda^2 + 1}\right)Z(\lambda) \times \{*\}, 
\end{equation}
where \( \{*\} \) stands for either 
\[
\begin{pmatrix}
    1 & \displaystyle -\frac{t_{12}^{(\ii)}(\lambda;x)}{\lambda - \ii}\\
    0 & 1
\end{pmatrix} \begin{pmatrix}
    1 & 0\\
    \displaystyle -\frac{t_{21}^{(\ii)}(\lambda;x)}{\lambda - \ii} & 1
\end{pmatrix}
\quad \text{or} \quad
\begin{pmatrix}
    1 & 0\\
    \displaystyle -\frac{t_{21}^{(\ii)}(\lambda;x)}{\lambda - \ii} & 1
\end{pmatrix}\begin{pmatrix}
    1 & \displaystyle -\frac{t_{12}^{(\ii)}(\lambda;x)}{\lambda - \ii}\\
    0 & 1
\end{pmatrix} 
\]
depending on whether we are considering \( T_1^{(\ii)}(\lambda) \) or \( T_2^{(\ii)}(\lambda) \). In the former case the following lemma holds.

\begin{lem}
\label{lem:lin sys}
Recall \eqref{taus}. If \( 4(1-W_1)^2 + W_2^2 \neq 0 \), where
\[
W_1 := \tau_{12}\det\begin{pmatrix}
    z_1(\ii) & z_2(\ii)\\
    z_1'(\ii) & z_2'(\ii)
\end{pmatrix} \quad \text{and} \quad W_2 := \tau_{12}\det\begin{pmatrix}
    z_1(\ii) & z_2(\ii)\\
    z_2(\ii) & z_1(\ii)
\end{pmatrix},
\]
then there exist constant matrices $B_2$ and $C_2$ as in \eqref{Bk Ck structure} such that 
\[
Q(\lambda) := \left( I + \frac{\lambda B_2}{\lambda^2 + 1} + \frac{C_2}{\lambda^2 + 1}\right) Z(\lambda) \begin{pmatrix}
    1 & \displaystyle -\frac{t_{12}^{(\ii)}(\lambda;x)}{\lambda - \ii}\\
    0 & 1
\end{pmatrix}
\]
is analytic at $\lambda = \ii$. Moreover, we have that
\begin{equation}
\label{a2b2c2d2}
\begin{cases}
b_2 &= \displaystyle -2\tau_{12} \frac{4(1-W_1)z_1(\ii)z_2(\ii) + \ii W_2(z_1^2(\ii)+z_2^2(\ii))}{4(1 - W_1)^2 + W_2^2}, \smallskip \\
d_2 &= \displaystyle -4\tau_{12} \frac{(1-W_1)(z_1^2(\ii)+z_2^2(\ii))+\ii W_2 z_1(\ii)z_2(\ii)}{4(1 - W_1)^2 + W_2^2}, \smallskip \\
a_2 &= \displaystyle \frac{-2 W_2^2}{4(1 - W_1)^2 + W_2^2}, \quad \text{and} \quad c_2 = \frac{4\ii (1-W_1)W_2}{4(1 - W_1)^2 + W_2^2}.
\end{cases}
\end{equation}
\end{lem}
\begin{proof}
Computing the principal part of the Laurent expansion of \( Q(\lambda) \) at \( \ii \) and setting it to zero gives a system of equations
\begin{align} 
\label{linear system}
\begin{aligned}
0 & = x_1 z_1(\ii) - x_2 z_2(\ii),\\
0 & = x_3 z_1(\ii) - x_4 z_2(\ii),\\
0 & = x_1 z_2(-\ii) - x_2 z_1(-\ii) - \tau_{12} \big( z_1(\ii) + x_1 z_1'(\ii) - x_2 z_2'(\ii) + \tfrac\ii2 x_3 z_2(\ii) - \tfrac\ii2 x_4 z_1(\ii) \big),\\
0 & = x_3 z_2(-\ii) - x_4 z_1(-\ii) - \tau_{12} \big( z_2(\ii) + x_3 z_1'(\ii) - x_4 z'_2(\ii) + \tfrac\ii2 x_1 z_2(\ii) - \tfrac\ii2 x_2 z_1(\ii) \big),
\end{aligned}
\end{align}
where
\begin{align}
\label{defn x1 x2 x3 x4}
x_1 = \frac{b_2 - \ii a_2}{2}, \quad x_2 = \frac{d_2 + \ii c_2}{2},\quad x_3 = \frac{d_2 - \ii c_2}{2} \quad x_4 = \frac{b_2 + \ii a_2}{2}.
\end{align}
If \( z_2(\ii) \neq 0\), then we can express \( x_2,x_4 \) through \( x_1,x_3\), respectively, via the first two equations and plug these expressions into the last two to get 
\begin{align}
\begin{pmatrix}
 W_1 - 1 & \frac\ii2 W_2  \\ \frac\ii2 W_2 & W_1 - 1
\end{pmatrix}
\begin{pmatrix}
 x_1 \\ x_3   
\end{pmatrix} =
\begin{pmatrix}
 \tau_{12} z_1(\ii)z_2(\ii) \\ \tau_{12} z_2^2(\ii)
\end{pmatrix}
\end{align}
where we used equality \( z_1(\ii)z_1(-\ii)-z_2(\ii)z_2(-\ii)=1\) (recall that \( \det Z(\lambda) \equiv 1\)). The above system is solvable when \( (W_1 - 1)^2 + \tfrac14 W_2^2 \neq 0 \), in which case
\begin{align}
x_1 &= \tau_{12} z_2(\ii)\frac{(W_1 - 1)z_1(\ii)-\frac\ii 2W_2 z_2(\ii)}{(1-W_1)^2 + \tfrac14 W_2^2}, \smallskip \\
x_2 &= \tau_{12} z_1(\ii)\frac{(W_1 - 1)z_1(\ii)-\frac\ii 2W_2 z_2(\ii)}{(1-W_1)^2 + \tfrac14 W_2^2}, \smallskip \\
x_3 &= \tau_{12} z_2(\ii)\frac{(W_1 - 1)z_2(\ii)-\frac\ii 2W_2 z_1(\ii)}{(1-W_1)^2 + \tfrac14 W_2^2}, \smallskip \\
x_4 &= \tau_{12} z_1(\ii)\frac{(W_1 - 1)z_2(\ii)-\frac\ii 2W_2 z_1(\ii)}{(1-W_1)^2 + \tfrac14 W_2^2},
\end{align}
from which expressions for \( a_2,b_2,c_2,d_2 \) easily follow. Clearly, the case \( z_1(\ii)\neq 0 \) leads to the same equations. Moreover, \( z_1(\ii) \) and \( z_2(\ii) \) cannot be simultaneously zero because \( \det Z(\ii) = 1\).
\end{proof}

Furthermore, that
\begin{multline}
\sigma_1\left( I + \frac{\lambda B_2}{\lambda^2 + 1} + \frac{C_2}{\lambda^2 + 1}\right) Z(\lambda) \begin{pmatrix}
    1 & 0\\
    \displaystyle -\frac{t_{21}^{(\ii)}(\lambda;x)}{\lambda - \ii} & 1
\end{pmatrix}\sigma_1 =  \\ \left( I - \frac{\lambda B_2}{\lambda^2 + 1} + \frac{C_2}{\lambda^2 + 1}\right) \begin{pmatrix}
        z_{1}(-\lambda) & z_{2}(\lambda)\\ 
        z_{2}(-\lambda) & z_{1}(\lambda)
    \end{pmatrix} \begin{pmatrix}
    1 & \displaystyle -\frac{t_{21}^{(\ii)}(\lambda;x)}{\lambda - \ii} \\
    0 & 1
\end{pmatrix}.
\label{notation for Q}
\end{multline}
This immediately shows that Lemma~\ref{lem:lin sys} should be replaced by the following statement if we choose \( i=2 \) in \eqref{defn of Z hat} (of course, matrices \( \hat Z(\lambda)\) and \( Z(\lambda)\) do depend on the choice of \( i \) in \eqref{defn of Z hat}; however, since all relevant asymptotic formulae are independent of this choice, we are not indicating this dependence explicitly; in particular, \( Z(\lambda)\) in \eqref{notation for Q} is \( Z_2(\lambda)\)).

\begin{lem}
\label{lem:lin sys a}
If \( 4(1-W_1)^2 + W_2^2 \neq 0 \), where
\[
W_1 := \tau_{21}\det\begin{pmatrix}
    z_1(-\ii) & z_2(-\ii)\\
    -z_1'(-\ii) & -z_2'(-\ii)
\end{pmatrix} \quad \text{and} \quad W_2 := \tau_{21}\det\begin{pmatrix}
    z_1(-\ii) & z_2(-\ii)\\
    z_2(-\ii) & z_1(-\ii)
\end{pmatrix},
\]
then there exist constant matrices $B_2$ and $C_2$ as in \eqref{Bk Ck structure} such that 
\[
Q(\lambda) := \left( I + \frac{\lambda B_2}{\lambda^2 + 1} + \frac{C_2}{\lambda^2 + 1}\right) Z(\lambda) \begin{pmatrix}
    1 & 0 \\
    \displaystyle -\frac{t_{21}^{(\ii)}(\lambda;x)}{\lambda - \ii} & 1
\end{pmatrix}
\]
is analytic at $\lambda = \ii$. Moreover, expressions for \( a_2,c_2 \) are the same as in \eqref{a2b2c2d2} while expressions for \( b_2,d_2\) are also the same as in there except \( \tau_{12} \) needs to be replaced by \( -\tau_{21} \).
\end{lem}

For future use observe that in both cases we have that
\begin{align}
\label{1+a2+c2}
1+a_2 \pm c_2 = \frac{2(1-W_1) \pm \ii W_2}{2(1-W_1) \mp \ii W_2}.
\end{align}

Clearly, the systems of equations defining \( B_1,C_1 \) are essentially identical to the one defining \( B_2,C_2 \). Examination of the proof of Lemma~\ref{lem:lin sys} (and therefore of Lemma~\ref{lem:lin sys a}) shows that it did not use symmetry \eqref{symmetries of Z}. This symmetry was only used to justify ansatz \eqref{Bk Ck structure}. What was important in the proof is that \( \det Z(\lambda) \equiv 1\). Now, we have that
\begin{align}
\det Q(\lambda) &= \det\left( I + \frac{\lambda B_2}{\lambda^2 + 1} + \frac{C_2}{\lambda^2 + 1}\right) \\
& = 1 + \frac{2a_2+\det B_2}{\lambda^2+1} + \frac{\det C_2-\det B_2}{(\lambda^2+1)^2}.
\end{align}
One can readily verify that
\[
-2a_2 = \det C_2 = \det B_2 = \frac{4W_2^2}{(1-W_1)^2 + W_2^2}.
\]
Thus, \( \det Q(\lambda) \equiv 1 \) as desired and we immediately deduce the following.

\begin{lem}
\label{lem:lin sys2}
Let \( \tilde q_1(\lambda) \) and \( \tilde q_2(\lambda) \) be the entries of the first column of \( \sigma_1 Q(\lambda) \sigma_1 \). If \( 4(1-W_3)^2 + W_4^2 \neq 0 \), where
\[
W_3 := \tau_{21}\det\begin{pmatrix}
    \tilde q_1(\ii) & \tilde q_2(\ii)\\
    \tilde q_1'(\ii) & \tilde q_2'(\ii)
\end{pmatrix} \quad \text{and} \quad W_4 := \tau_{21}\det\begin{pmatrix}
    \tilde q_1(\ii) & \tilde q_2(\ii)\\
    \tilde q_2(\ii) & \tilde q_1(\ii)
\end{pmatrix},
\]
then there exist constant matrices $B_1$ and $C_1$ as in \eqref{Bk Ck structure} such that 
\[
\left( I + \frac{\lambda B_1}{\lambda^2 + 1} + \frac{C_1}{\lambda^2 + 1}\right) Q(\lambda) \begin{pmatrix}
    1 & 0\\
    \displaystyle -\frac{t_{21}^{(\ii)}(\lambda;x)}{\lambda - \ii} & 1
\end{pmatrix}
\]
is analytic at $\lambda = \ii$. Moreover, we have that
\begin{equation}
\label{a1b1c1d1}
\begin{cases}
b_1 &= \displaystyle 2\tau_{21} \frac{4(1-W_3)\tilde q_1(\ii)\tilde q_2(\ii) + \ii W_4(\tilde q_1^2(\ii)+\tilde q_2^2(\ii))}{4(1 - W_3)^2 + W_4^2}, \smallskip \\
d_1 &= \displaystyle 4\tau_{21} \frac{(1-W_3)(\tilde q_1^2(\ii)+\tilde q_2^2(\ii))+\ii W_4 \tilde q_1(\ii)\tilde q_2(\ii)}{4(1 - W_3)^2 + W_4^2}, \smallskip \\
a_1 &= \displaystyle \frac{-2W_4^2}{4(1-W_3)^2 + W_4^2}, \quad \text{and} \quad c_1 = \frac{4\ii (1-W_3)W_4}{4(1-W_3)^2 + W_4^2}.
\end{cases}
\end{equation}
\end{lem}

Clearly, the previous lemma addresses the case \( i=1 \) in \eqref{defn of Z hat}. The case \( i=2 \) is then covered by the following lemma.

\begin{lem}
\label{lem:lin sys2 a}
Let \( q_1(\lambda) \) and \( q_2(\lambda) \) be the entries of the first column of \(  Q(\lambda) \). If \( 4(1-W_3)^2 + W_4^2 \neq 0 \), where
\[
W_3 := \tau_{12}\det\begin{pmatrix}
    q_1(\ii) & q_2(\ii)\\
    q_1'(\ii) & q_2'(\ii)
\end{pmatrix} \quad \text{and} \quad W_4 := \tau_{12}\det\begin{pmatrix}
    q_1(\ii) & q_2(\ii)\\
    q_2(\ii) & q_1(\ii)
\end{pmatrix},
\]
then there exist constant matrices $B_1$ and $C_1$ as in \eqref{Bk Ck structure} such that 
\[
\left( I + \frac{\lambda B_1}{\lambda^2 + 1} + \frac{C_1}{\lambda^2 + 1}\right) Q(\lambda) \begin{pmatrix}
    1 & \displaystyle -\frac{t_{12}^{(\ii)}(\lambda;x)}{\lambda - \ii} \\
    0 & 1
\end{pmatrix}
\]
is analytic at $\lambda = \ii$. Moreover, expressions for \( a_1,c_1 \) are the same as in \eqref{a1b1c1d1} while expressions for \( b_1,d_1\) are also the same as in there except \( \tau_{21} \) needs to be replaced by \( -\tau_{12} \).
\end{lem}

Similarly to \eqref{1+a2+c2} we also have that
\begin{align}
\label{1+a1+c1}
1+ a_1 \pm c_1 = \frac{2(1-W_3) \pm \ii W_4}{2(1-W_3) \mp \ii W_4}.
\end{align}

\subsection{Asymptotic Analysis}

It follows from RHP~\ref{RHP Y hat}(3), \eqref{defn of Z hat}, and \eqref{defn of Z} that
\[
P_0 = \hat Y(0) = (-1)^m\hat Z(0) = (-1)^m(I+C_1)(I+C_2)Z(0).
\]
That is, it holds by the definition of \( P_0 \) in RHP~\ref{rhp original}(3) that
\begin{align}
\begin{pmatrix}
\cosh \tfrac u2 \\
\sinh \tfrac u2
\end{pmatrix} 
 = (-1)^m
\begin{pmatrix}
(1+a_1)(1+a_2) + c_1c_2 & c_1(1+a_2) + c_2(1+a_1) \\
c_1(1+a_2) + c_2(1+a_1) & (1+a_1)(1+a_2) + c_1c_2
\end{pmatrix} 
 \begin{pmatrix}
z_1(0) \\
z_2(0)
\end{pmatrix}.
\end{align}
This means that
\[
e^u = \frac{\cosh \tfrac u2+\sinh \tfrac u2}{\cosh \tfrac u2-\sinh \tfrac u2} = \frac{(1+a_1+c_1)(1+a_2+c_2)(z_1(0)+z_2(0))}{(1+a_1-c_1)(1+a_2-c_2)(z_1(0)-z_2(0))}.
\]
Respectively, we deduce from \eqref{1+a2+c2} and \eqref{1+a1+c1} that
\begin{align}
\label{formula eu}
e^u = \frac{(2(1-W_1) + \ii W_2)^2(2(1-W_3) + \ii W_4)^2(z_1(0)+z_2(0))}{(2(1-W_1) - \ii W_2)^2(2(1-W_3) - \ii W_4)^2(z_1(0)-z_2(0))}.
\end{align}

Exactly as in \eqref{gl rho_Y - I has a small norm}, \eqref{G_Z-I estimate}  yields that
\[
\|\rho-I\|_{L^2(\Gamma_Z)} \leq C_5 x^{-1} 
\]
for some constant \( C_5=C_5(m) \). Thus, similarly to \eqref{representation Y(0)}, \eqref{int formula Z} gives
\[
Z(\pm\ii) = I + \frac1{2\pi\ii} \int_{\Gamma_Z} \frac{G_Z(\lambda^\prime)-I}{\lambda^\prime \mp \ii}d\lambda^\prime + \O_m\big(x^{-1}\big) = I + \O_m\big(x^{-1}\big),
\]
where we used \eqref{G_Z-I estimate} and Cauchy-Schwarz inequality twice. That is,
\begin{align}
\label{der k zi}
\big (z_1(\lambda)-1 \big)^{(k)}_{|\lambda=\pm\ii}, \big(z_2(\lambda)\big)^{(k)}_{|\lambda=\pm\ii} = \O_m\big( x^{-1}\big)
\end{align}
in the notation of \eqref{symmetries of Z} for any \( k\geq 0 \), where we need to use the Cauchy integral formula for derivatives when \( k>0 \). Recall also that the jump matrix \( G_Z(\lambda) \) around the origin is simply \( G_Y(\lambda)\) conjugated by a bounded function. Hence, similarly to \eqref{representation Y(0)}, we can conclude that
\begin{align}
\label{z1+-z_2}
Z(0) = I + \O_m\big( x^{-1}\big) \quad \Rightarrow \quad z_1(0) \pm z_2(0) = 1 + \O_m\big(x^{-1}\big).
\end{align}

Next, recall that \( \tau_{12}\tau_{21}=mx^{-1} \) and that we must have \( x>m \) by \eqref{G_z 2} and \eqref{G_z 3} in order for the norm \eqref{G_Z-I estimate} to be small. Hence, either \( |\tau_{12}|<1 \) or \( |\tau_{21}|<1 \) and this is what determined the choice of the index \( i \) in \eqref{defn of Z hat}. Hence, it follows from \eqref{der k zi} that 
\begin{align}
\label{W1W2}
W_1 = \O_m\big(x^{-1}\big) \quad \text{and} \quad W_2 = \tau + \O_m\big(x^{-1}\big),
\end{align}
see Lemmas~\ref{lem:lin sys} and~\ref{lem:lin sys a}, where \( \tau \) is either \( \tau_{12} \) or \( \tau_{21} \) depending on which quantity has the smallest modulus. In particular, it holds that
\begin{align}
\label{asymp first det}
2(1-W_1) \pm \ii W_2 = 2 \pm \ii \tau + \O_m(x^{-1})
\end{align}
and this quantity is at least \( \frac12 \) in absolute value for all \( x \geq D_2(m) + 8 \ln |s^\R|_+ \) for an appropriately chosen \( D_2(m) \). Respectively, it follows from Lemmas~\ref{lem:lin sys} and \ref{lem:lin sys a}, as well as \eqref{W1W2} that
\begin{align}
a_2,b_2,c_2,d_2 = \O(1).
\end{align}
Let now \( Q(\lambda) \) be as in Lemma~\ref{lem:lin sys} or Lemma~\ref{lem:lin sys a}. Similarly, to the proof of those lemmas, we can write down explicit expressions for \( \tilde q_1(\ii) \), \( \tilde q_2(\ii) \), \( q_1(\ii) \) and \( q_2(\ii) \). These expressions will involve \( a_2,b_2,c_2,d_2 \) and the values of the derivatives of \( z_i(\lambda) \), \( i\in\{1,2\} \), and \( t^{(\ii)}_{12}(\lambda;x)\) or \( t^{(\ii)}_{21}(\lambda;x)\) at \( \pm\ii \). Hence, 
\[
q_1(\ii)-1, q_2(\ii) = \O_m\big(x^{-1}\big) \quad \text{and} \quad \tilde q_1(\ii)-1, \tilde q_2(\ii) = \O_m\big(x^{-1}\big)
\]
by \eqref{der k zi}. Thus, similarly to \eqref{asymp first det} and using the notation of Lemmas~\ref{lem:lin sys2} and~\ref{lem:lin sys2 a}, we get that
\begin{align}
\label{asymp second det}
2(1-W_3) \pm \ii W_4 = 2 \pm \ii\tau^* + \O_m\big(x^{-1}\big),
\end{align}
where \( \{\tau,\tau^*\} = \{\tau_{12},\tau_{21}\} \). By plugging the estimates \eqref{asymp first det} and \eqref{asymp second det} into \eqref{formula eu} and using \eqref{z1+-z_2}, we obtain that
\[
u(x) = 2\ln\left(\frac{(2 + \ii\tau_{12})(2 + \ii\tau_{21})  + \O_m(x^{-1})}{(2 - \ii\tau_{12})(2 - \ii\tau_{21}) + \O_m(x^{-1})}\right)
\]
for all \( x\geq D_2(m) + 8\ln |s^\R|_+  \). To deduce the desires asymptotics \eqref{asymptotic formula 2} it only remains to recall \eqref{symmetries of u} and that \(\tau_{12}\tau_{21} = m x^{-1} \).

\section{Proof of Theorem~\ref{thm:3}}

The proof of Theorem~\ref{thm:3} is the same as the proof of Theorem~\ref{thm:2}. The only difference is that we are able to do a deeper analysis of the local parametrices \( P^{(\pm\ii)}(\lambda) \) introduced in Section~\ref{subsec:local_par1} when \( m=0 \). In particular, our starting point for this proof is RHP~\ref{RHP Y hat}. Recall also that in this section
\[
B^\R = e^{-\varkappa x} = x^\nu e^{-x}, \quad |\nu|\leq \tfrac12.
\]

\subsection{Local Parametrices}

Recall \eqref{defn of He0}--\eqref{H0 expansion}. For each \( n\geq 1\) we set 
\begin{equation}
    \label{definition of en}
    \mathcal H_{0,n}(z) := -\frac{1}{2\ii\sqrt{\pi}}\sum_{m=0}^{n-1} \frac{(1/2)_m}{z^{2m+1}}.
\end{equation}
It readily follows from \eqref{H0 expansion} that
\[
(\mathcal H_0 - \mathcal H_{0,n})(z) = \O_n\big(z^{-2n-1}\big) \quad \text{as} \quad z\to\infty.
\]
Observe also that this difference is holomorphic in \( \C^* \) with a pole of order \( 2n-1 \) at the origin. Recall \eqref{theta to map}. We now replace \eqref{P(i)} by
\begin{align}
\label{P+-i matrix}
P_n^{(\ii)}(\lambda) := \begin{pmatrix}
    1 & 0 \\ \displaystyle -\frac{x^\nu}{A} (\mathcal H_0-\mathcal H_{0,n})(\sqrt x\phi_\ii(\lambda)) & 1
\end{pmatrix}, \quad \lambda \in U_\ii,
\end{align}
and, similarly to \eqref{P-i}, set
\begin{align}
P_n^{(-\ii)}(\lambda) := \sigma_1 P_n^{(\ii)}(-\lambda)\sigma_1, \quad \lambda\in U_{-\ii}.
\end{align}
RHP~\ref{RHP Y loc i} and RHP~\ref{RHP Y loc -i} are now replaced by the following Riemann-Hilbert problems.

\begin{RHP}
\label{RHP Y loc +-i}
Find a $2 \times 2$ matrix function $P_n^{(\pm\ii)}(\lambda)$ such that
\begin{enumerate}
    \item $P_n^{(\pm\ii)}(\lambda)$ is analytic in $U_{\pm\ii} \setminus S^1$ and \( P_n^{(\pm\ii)}(\lambda) = \O\big((\lambda\mp \ii)^{-2n+1}\big)\) as \( \lambda\to\pm\ii \);
    \item one-sided traces $P^{(\pm\ii)}_{n\pm}(\lambda)$ are continuous on \( U_{\pm\ii} \cap (S^1\setminus\{\pm\ii\}) \) and satisfy RHP~\ref{RHP Y loc i},\ref{RHP Y loc -i}(2);
    \item it holds uniformly on \( \partial U_{\pm\ii} \) that \( \displaystyle
    P_n^{(\pm\ii)}(\lambda) = I + \O_n\big(x^{\nu-n-\frac12}\big) \).
\end{enumerate}
\end{RHP}

The reason we improve RHP~\ref{RHP Y loc i},\ref{RHP Y loc -i}(3) to RHP~\ref{RHP Y loc +-i}(3) in such a simple fashion is the triangular structure of \( He_0(z) \). It is theoretically possible to replace matrices \( T_i^{(\pm\ii)}(\lambda)\) in RHP~\ref{RHP Y loc i},\ref{RHP Y loc -i}(3) by their proper modifications to achieve a better error rate than \( x^{-1} \), but practically, we found the required computations to be too complicated except in the above case \( m=0 \).


\subsection{Small Norm Problem}

Similarly to \eqref{defn of Z hat}, we look for the solution of RHP~\ref{RHP Y hat} in the form
\begin{align}
    \hat Y(\lambda) = \hat Z_n(\lambda)
    \begin{cases}
        P_n^{(\pm\ii)}(\lambda), & \lambda \in U_{\pm\ii}, \\
        I, & \text{otherwise}.
    \end{cases}
    \label{defn of Zn hat}
\end{align}
Then, \( \hat Z_n(\lambda) \) must solve the following Riemann-Hilbert problem.

\begin{figure}[ht!]
    \centering
    \includegraphics[width=8cm]{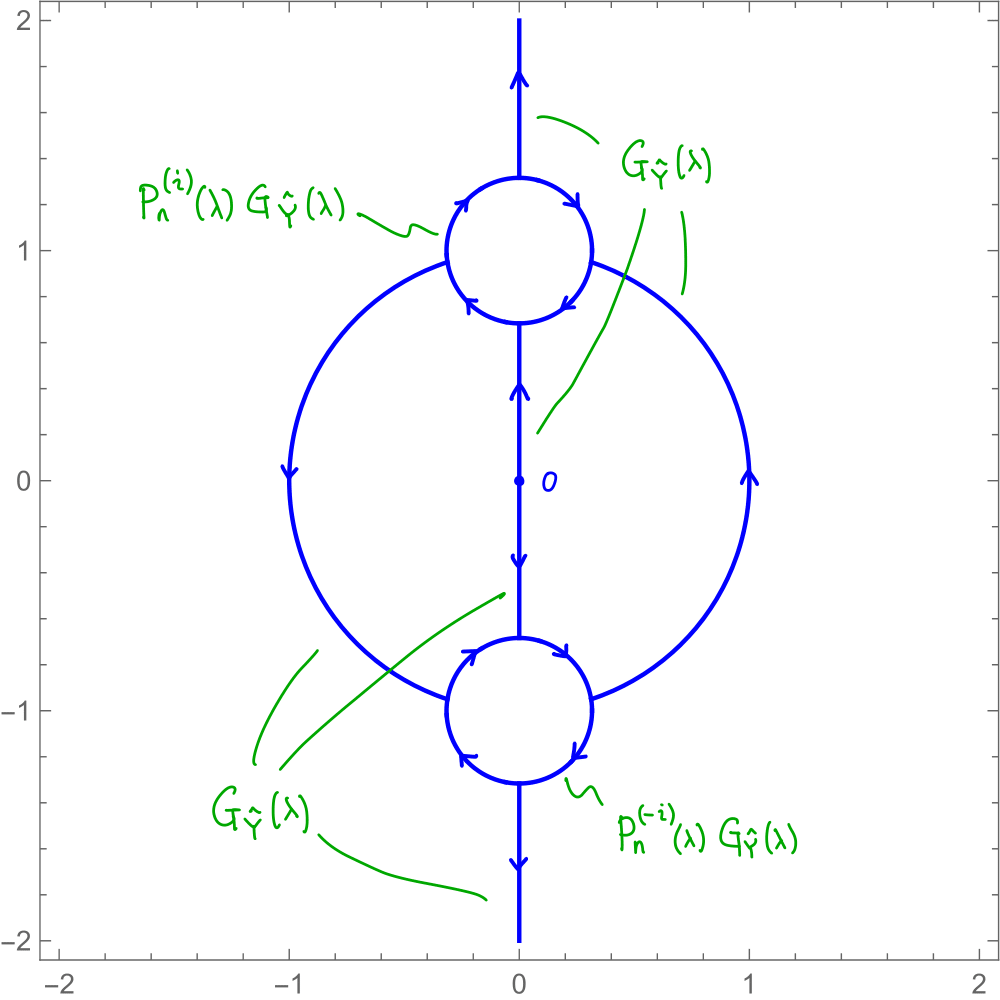}
    \caption{The contour $\Gamma_Z$ and the jump matrices $G_{\hat Z_n}(\lambda)$ for RHP~\ref{RHP Zn}. Note that $G_{\hat Y}(\lambda)$ is stated on Figure \ref{Y-hat problem pic}.}
    \label{hat Zn problem pic}
\end{figure}

\begin{RHP}
\label{RHP hat Zn}
Find a $2 \times 2$ matrix function $\hat Z_n(\lambda)$ such that
\begin{enumerate}
    \item $\hat Z_n(\lambda)$ is analytic for $\lambda \in \C \setminus (\Gamma_Z\cup\{\pm\ii\})$, where $\Gamma_Z$ is depicted on Figure~\ref{hat Zn problem pic};
    \item one-sided traces $\hat Z_{n\pm}(\lambda)$ exist a.e. on \( \Gamma_Z \), belong to \( L^2(\Gamma_Z) \), and satisfy
    \[
        \hat Z_{n+}(\lambda) = \hat Z_{n-}(\lambda) G_{\hat Z_n}(\lambda), 
    \]
    where the jump matrices \(G_{\hat Z_n}(\lambda)\) are as on Figure~\ref{hat Zn problem pic};
    \item it holds that \( \hat Z_n(\lambda) = I + \O(1/\lambda) \) as \( \lambda\to\infty \);
    \item $\hat Z_n(\lambda)$ has poles of order \( 2n-1 \) at \( \pm\ii \) with principal parts such that the products \( \hat Z_n(\lambda)P_n^{(\pm\ii)}(\lambda) \) are analytic in the respective disks \( U_{\pm\ii} \).
\end{enumerate}
\end{RHP}

We look for \( \hat Z_n(\lambda) \) in the following form:
\begin{equation}
    \label{matrix Zn}
    \hat Z_n(\lambda) = B_n(\lambda) Z_n(\lambda)D^{1-2n}(\lambda), \quad D(\lambda) := \begin{pmatrix}
        \lambda-\ii & 0 \\ 0 & \lambda+\ii \end{pmatrix},
\end{equation}
where \( B_n(\lambda) \) is a matrix polynomial of degree \( 2n-1 \) chosen so that \(B_n(\lambda)=D^{2n-1}(\lambda)  + \O(\lambda^{2n-2}) \) as \( \lambda\to\infty \) and RHP~\ref{RHP hat Zn}(4) is  while \( Z_n(\lambda) \) must solve the following Riemann-Hilbert problem.

\begin{RHP}
\label{RHP Zn}
Find a $2 \times 2$ matrix function $Z_n(\lambda)$ such that
\begin{enumerate}
    \item $Z_n(\lambda)$ is analytic for $\lambda \in \C \setminus \Gamma_Z$;
    \item one-sided traces $Z_{n\pm}(\lambda)$ exist a.e. on \( \Gamma_Z \), belong to \( L^2(\Gamma_Z) \), and satisfy
    \[
        Z_{n+}(\lambda) = Z_{n-}(\lambda) D^{1-2n}(\lambda) G_{\hat Z_n}(\lambda)D^{2n-1}(\lambda);
    \]
    \item it holds that \( Z_n(\lambda) = I + \O(1/\lambda) \) as \( \lambda\to\infty \).
\end{enumerate}
\end{RHP}

As in the case of RHP~\ref{RHP Z}, we show that RHP~\ref{RHP Zn} is a small norm problem and therefore has a solution for all sufficiently large \( x \) by the small norm theorem. To this end, notice that
\begin{align}
G_{Z_n}(\lambda) &:=  D^{1-2n}(\lambda) G_{\hat Z_n}(\lambda)D^{2n-1}(\lambda) \\
& = \left( \frac{\lambda+\ii}{\lambda-\ii} \right)^{(n-\frac12)\sigma_3} G_{\hat Z_n}(\lambda)\left( \frac{\lambda+\ii}{\lambda-\ii} \right)^{-(n-\frac12)\sigma_3}.
\end{align}
Hence, the jump matrices \( G_{Z}(\lambda) \) and \( G_{Z_n}(\lambda) \) are the same on \( (\ii\R\cup S^1 )\cap \Gamma_Z \) upon replacing \( m \) by \( n-\tfrac12 \), see Figure~\ref{hat Z problem pic}. Therefore, estimates \eqref{G_z 1}, \eqref{G_z 2}, and \eqref{G_z 3} remains valid upon replacing \( m \) by \( n-\tfrac12 \) in the exponentials and changing \( x^{m+\nu} \) to \( x^\nu \). Moreover, we can repeat all the steps leading to \eqref{G_z 4} to get that
\[
\|G_{Z_n}-I\|_{L^\infty(\partial U_{\pm\ii})} \leq cx^{\nu-n-\frac12},
\]
for some constant \( c=c(n) \). Hence, as in \eqref{G_Z-I estimate}, we see that the above estimate holds for \( L^2(\Gamma_Z)\cap L^\infty(\Gamma_Z) \)-norm as well, perhaps with a modified constant \( c \). Therefore, the solution of RHP~\ref{RHP Zn} exists for all \( x\geq D_3(n)+8\ln |s^\R|_+ \) and has the form as in \eqref{int formula Z} with \( G_Z(\lambda^\prime) \) replaced by \( G_{Z_n}(\lambda^\prime) \). Finally, as in the case of \eqref{symmetries of Z}, one can readily verify that
\begin{align}
    Z_n(\lambda) = \sigma_1 Z_n(-\lambda) \sigma_1 \quad \Rightarrow \quad Z_n(\lambda) = \begin{pmatrix}
        z_{n,1}(\lambda) & z_{n,2}(-\lambda) \\ z_{n,2}(\lambda) & z_{n,1}(-\lambda)
    \end{pmatrix}. 
    \label{Zn sigma1 symmetry}
\end{align}


\subsection{Polynomial Prefactor} 
\label{Polynomial Prefactor}

In this subsection, we construct a polynomial matrix function \( B_n(\lambda) \) such that \( \hat Z_n(\lambda)\) given by \eqref{matrix Zn} solves RHP~\ref{RHP hat Zn}. Clearly, RHP~\ref{RHP hat Zn}(1,2) are immediately satisfied. RHP~\ref{RHP hat Zn}(3) follows from our requirement that \(B_n(\lambda)=D^{2n-1}(\lambda)  + \O(\lambda^{2n-2}) \) as \( \lambda\to\infty \). Hence, we only need to satisfy RHP~\ref{RHP hat Zn}(4). It can be rewritten as
\begin{equation}
    \label{condition on analyticity 1}
    B_n(\lambda) Z_n(\lambda) D^{1-2n}(\lambda) \begin{pmatrix}
    1 & 0 \\ \displaystyle \frac{x^\nu}{A} \mathcal H_{0,n}(\sqrt x\phi_\ii(\lambda)) & 1
\end{pmatrix}
\end{equation}
must be analytic in \( U_\ii \) and
\begin{equation}
    \label{condition on analyticity 2}
    B_n(\lambda) Z_n(\lambda) D^{1-2n}(\lambda)\begin{pmatrix}
    1 & \displaystyle \frac{x^\nu}{A} \mathcal H_{0,n}(\sqrt x\phi_{-\ii}(\lambda)) \\
    \displaystyle 0 & 1
\end{pmatrix}
\end{equation}
must be analytic in \( U_{-\ii}\). Similarly to Section~\ref{rational prefactor}, we assume that
\begin{align}
    B_n (\lambda) = -\sigma_1 B_n(-\lambda) \sigma_1
    \label{Bn symmetry}
\end{align}
and observe that \eqref{condition on analyticity 2} is automatically satisfied when \eqref{condition on analyticity 1} is due to \eqref{Zn sigma1 symmetry}.

Since $B_n (\lambda)$ is a matrix polynomial of degree \( 2n-1 \) and \(B_n(\lambda)=D^{2n-1}(\lambda)  + \O(\lambda^{2n-2}) \) as $\lambda \to \infty$, it admits the following representation in terms of $D(\lambda)$:
\begin{align}
\begin{aligned}
    B_n(\lambda) &= D^{2n-1}(\lambda) + \sum_{m=1}^{2n-1} B_{n,m} D^{2n-1-m}(\lambda)\\
    &= \left( I + \sum_{m=1}^{2n-1} B_{n,m} D^{-m}(\lambda) \right) D^{2n-1}(\lambda),
    \label{Bn as a polynomial}
\end{aligned}
\end{align}
where the coefficient matrices $B_{n,m}$ are constant in $\lambda$ and, due to \eqref{Bn symmetry}, satisfy
\begin{align}
\label{coeff Bm}
    B_{n,m} = (-1)^m \sigma_1 B_{n,m} \sigma_1 \quad \Rightarrow \quad
        B_{n,m} = \begin{pmatrix}
        b_{n,m} & (-1)^m d_{n,m} \\
        d_{n,m} & (-1)^m b_{n,m}
    \end{pmatrix}.
\end{align}
Recalling \eqref{Zn sigma1 symmetry}, we can readily compute that the analyticity of the \((1,1)\)-entry of the product in \eqref{condition on analyticity 1} is equivalent to the equality of the principal parts of
\begin{equation}
\label{lhs product}
\left( \sum_{m=1}^{2n-1} \frac{b_{n,m}}{(\lambda-\ii)^m} \right) \left( z_{n,1}(\lambda) + z_{n,2}(-\lambda)\left( \frac{\lambda-\ii}{\lambda+\ii} \right)^{2n-1} \frac{x^\nu}{A} \mathcal H_{0,n}(\sqrt x\phi_\ii(\lambda)) \right)
\end{equation}
and
\begin{equation}
\label{rhs product}
\left(\sum_{m=1}^{2n-1} \frac{(-1)^md_{n,m}}{(\lambda+\ii)^m} \right) \left( -z_{n,1}(-\lambda) \frac{x^\nu}{A} \mathcal H_{0,n}(\sqrt x\phi_\ii(\lambda)) - \left(\frac{\lambda+\ii}{\lambda-\ii}\right)^{2n-1}z_{n,2}(\lambda) \right).
\end{equation}
Let \( T_n \) be the following upper triangular Toeplitz matrix:
\begin{align}
\label{definiion Tn}
    T_n := [t_{n, j - i}]_{i,j=1}^{2n-1}
\end{align}
with \( i \) being the row index and \( j \) being the column one, where \( \sum_{m=0}^{\infty} t_{n,m}(\lambda-\ii)^m \) is the Taylor expansion of the second factor in \eqref{lhs product} and $t_{n,m} = 0$ if $m < 0$.
Write \( \vec b_n=(b_{n,1},\ldots,b_{n,2n-1})^\mathsf{T}\). Then, the principal part of \eqref{lhs product} can be written as
\[
\sum_{m=1}^{2n-1} \frac{(T_n\vec b_n)_m}{(\lambda-\ii)^m},
\]
where \( (T_n\vec b_n)_m \) is the \( m \)-th entry of the vector \( T_n \vec b_n \). Moreover,
\[
\sum_{m=1}^{2n-1} \frac{(-1)^md_{n,m}}{(\lambda+\ii)^m} = \sum_{k=1}^{2n-1} (L_n\vec d_n)_k (\lambda-\ii)^{k-1} + \O\big((\lambda-\ii)^{2n-1}\big),
\]
where we set \( \vec d_n=(d_{n,1},\ldots,d_{n,2n-1})^\mathsf{T}\) and \( L_n \) was defined in \eqref{definition Ln}. Further, let $H_n$ be the following triangular Hankel matrix:
\begin{align}
H_n := [h_{n,i+j-1}]_{i,j=1}^{2n-1} \label{defn of Hn}
\end{align}
with \( i \) being the row index and \( j \) being the column one, where \(\sum_{m=1}^{2n-1} h_{n,m}(\lambda-\ii)^{-m}\) is the principal part of the second factor in \eqref{rhs product} and \( h_{n,m}=0 \) if \( m\geq 2n \). Therefore, the principal part of \eqref{rhs product} is given by
\[
\sum_{m=1}^{2n-1} \frac{(H_n L_n\vec d_n)_m}{(\lambda-\ii)^m}.
\]
That is, we arrive at the equation \( T_n \vec b_n = H_n L_n \vec d_n\). 
Similarly, the analyticity of the \((2,1)\)-entry of the product in \eqref{condition on analyticity 1} is equivalent to the equality of the principal parts of
\[
\left( \sum_{m=1}^{2n-1} \frac{d_{n,m}}{(\lambda-\ii)^m} \right) \left( z_{n,1}(\lambda) + z_{n,2}(-\lambda)\left( \frac{\lambda-\ii}{\lambda+\ii} \right)^{2n-1} \frac{x^\nu}{A} \mathcal H_{0,n}(\sqrt x\phi_\ii(\lambda)) \right)
\]
and
\[
\left(1+\sum_{m=1}^{2n-1} \frac{(-1)^mb_{n,m}}{(\lambda+\ii)^m} \right) \left( -z_{n,1}(-\lambda) \frac{x^\nu}{A} \mathcal H_{0,n}(\sqrt x\phi_\ii(\lambda)) - \left(\frac{\lambda+\ii}{\lambda-\ii}\right)^{2n-1}z_{n,2}(\lambda)\right).
\]
As in the case of \( (1,1) \)-entry, we can rewrite the above condition as a linear system. Altogether, we arrive at the equations
\[
\begin{cases}
    T_n \vec b_n &= H_n L_n \vec d_n, \\
    T_n \vec d_n & = H_n(\vec e_1 + L_n \vec b_n),
\end{cases}
\]
where \( \vec e_1=(1,0,\ldots,0)^{\mathsf T}\). The above linear system can be equivalently rewritten as
\[
    (T_n-H_nL_n) (\vec d_n + \vec b_n) = H_n \vec e_1 = (T_n+H_nL_n) (\vec d_n - \vec b_n).
\]
Thus, it is solvable if and only if the determinants of the matrices \( T_n\pm H_nL_n \) are non-zero. For future use, let us also point out that it follows from \eqref{definition Ln}, \eqref{definiion Tn}, \eqref{defn of Hn}, and Proposition~\ref{prop:det} used with \( \ell=\ii \) and \( H=\pm H_n \) that
\begin{equation}
    \label{prefactor linear system}
    1 + \sum_{m=1}^{2n-1} \ii^m (b_{n,m} \pm d_{n,m}) = \frac{\det(T_n \pm H_nL_n)}{\det(T_n \mp H_nL_n)}.
\end{equation}


\subsection{Asymptotic Analysis}

It follows from RHP~\ref{RHP Y hat}(3), \eqref{defn of Zn hat}, and \eqref{matrix Zn} that
\[
P_0 = \hat Y(0) = \hat Z_n(0) = \ii^{2n-1} B_n(0)Z_n(0) \sigma_3.
\]
We get from the top line of \eqref{Bn as a polynomial} that
\[
\ii^{2n-1} B_n(0) = \left( I + \sum_{m=1}^{2n-1}\ii^m B_{n,m} \sigma_3^m \right)\sigma_3.
\]
Thus, it holds by the definition of \( P_0 \) in RHP~\ref{rhp original}(3) that
\begin{align}
\begin{pmatrix}
\cosh \tfrac u2 \\
\sinh \tfrac u2
\end{pmatrix} 
 = 
\begin{pmatrix}
1+\sum_{m=1}^{2n-1} \ii^m b_{n,m}  & \sum_{m=1}^{2n-1} \ii^m d_{n,m} \smallskip \\
\sum_{m=1}^{2n-1} \ii^m d_{n,m} & 1+\sum_{m=1}^{2n-1} \ii^m b_{n,m}
\end{pmatrix} 
 \begin{pmatrix}
z_1(0) \\
-z_2(0)
\end{pmatrix}.
\end{align}
This means that
\[
e^u = \frac{\cosh \tfrac u2+\sinh \tfrac u2}{\cosh \tfrac u2-\sinh \tfrac u2} = \frac{(1 + \sum_{m=1}^{2n-1} \ii^m (b_{n,m}+d_{n,m}))(z_1(0)-z_2(0))}{(1 + \sum_{m=1}^{2n-1} \ii^m (b_{n,m}-d_{n,m}))(z_1(0)+z_2(0))}.
\]
Therefore, we get from \eqref{prefactor linear system} that
\begin{align}
\label{formula for u part 3}
u(x) = 2\ln\left( \frac{\det(T_n + H_nL_n)}{\det(T_n - H_nL_n)} \right) + \ln \left(\frac{z_1(0)-z_2(0)}{z_1(0)+z_2(0)}\right).    
\end{align}
Now, exactly as in \eqref{z1+-z_2} we can reach the conclusion that
\[
z_1(0) \pm z_2(0) = 1 + \O_n\big(x^{\nu-n-\frac12}\big).
\]
Hence, we now need to analyze the asymptotic behavior of the above determinants.

Recall \eqref{definition Cn}. Let us show that
\begin{equation}
    \label{expansion of en}
\mathcal H_{0,n}(\sqrt x\phi_\ii(\lambda)) = -\sum_{m=1}^{2n-1} \frac{c_{n,m}}{(\lambda-\ii)^m} + \text{holomorphic part}.
\end{equation}
Indeed, it follows from \eqref{theta to map} that
\[
\frac1{\phi_\ii^{2m+1}(\lambda)} = \frac{2^{m+\frac12}}{(\lambda-\ii)^{2m+1}} (-\ii\lambda)^{m+\frac12},
\]
where one needs to recall that we have chosen the branch of the square root such that \( (-\ii\lambda)^{\frac12}_{|\ii} = -1 \). It holds that
\begin{align}
\left((-\ii\lambda)^{m+\frac12} \right)^{(l)}_{|\ii} = (-\ii)^l l! \binom{m+\frac12}l (-\ii\lambda)^{m-l+\frac12}_{|\ii} = -(-\ii)^l l! \binom{m+\frac12}l.
\end{align}
Therefore, we get from \eqref{definition of en} that
\[
\mathcal H_{0,n}(\sqrt x\phi_\ii(\lambda)) = \sum_{m=0}^{n-1} \frac{(1/2)_m}{\sqrt{2\pi x}}\left(\frac2x\right)^m \sum_{l=0}^\infty \binom{m+\frac12}l\frac{(-\ii)^{l+1}}{(\lambda-\ii)^{2m+1-l}},
\]
which, after setting \( 2m+1-l=k \), gives \eqref{definition Cn}.

It now follows from \eqref{A is almost 1} that
\begin{equation}
    \label{expansion of en 2}
\frac{x^\nu}{A}\mathcal H_{0,n}(\sqrt x\phi_\ii(\lambda)) = - x^\nu\sum_{m=1}^{2n-1} \frac{ c_{n,m} (1 + \O(x^\nu e^{-x}))}{(\lambda-\ii)^m} + \text{hol. part}.
\end{equation}
Thus, we see from \eqref{definition Cn} that the Laurent coefficients above are of order \( x^{\nu-\tfrac12} \). As \( |\nu|\leq\tfrac12 \), they are always bounded.

Recall that \( t_{n,m} \) are the Taylor coefficients of the second factor in \eqref{lhs product} and \( h_{n,m} \) were defined as the Laurent coefficients of the principal part of the second factor in \eqref{rhs product}. Similarly to \eqref{der k zi}, we observe that
\[
(z_{n,1}(\lambda)-1)^{(k)}_{|\lambda=\pm\ii}, (z_{n,2}(\lambda))^{(k)}_{|\lambda=\pm\ii} = \O_n\big(x^{\nu-n-\frac12}\big)
\]
for any \( k\geq 0 \). Hence, it holds that
\begin{align}
    t_{n,m} &= \delta_{m,0} + \O_n\big( x^{\nu-n-\frac12} \big), \\
    h_{n,m} &= x^\nu c_{n,m} + \O_n \big( x^{\nu-n-\frac12} \big),
\end{align}
where \( \delta_{i,j} \) is the Kronecker symbol. Thus, we can write
\[
T_n \pm H_nL_n = I \pm  x^\nu C_nL_n + \O_n\big( x^{\nu-n-\frac12}\big).
\]
Denote the entries of \( C_nL_n \) by \( m_{n;i,j} \). Then,
\begin{align}
D_n^\pm & :=  \det(T_n \pm H_nL_n) = \det\left( I \pm x^\nu C_nL_n + \O_n\big( x^{\nu-n-\frac12}\big) \right) \\
& = \sum_\sigma (-1)^{\sign(\sigma)} \prod_{i=1}^{2n-1}\left( \delta_{i,\sigma(i)} \pm x^\nu m_{n;i,\sigma(i)} + \O_n\big( x^{\nu-n-\frac12} \big) \right),
\end{align}
where the sum is taken over all permutations of \( \{1,2,\ldots,2n-1\} \). In the next paragraph we explain that \( m_{n;i,j} = \O(x^{-\frac12}) \). As \( \nu\leq \tfrac12 \), this gives that
\begin{align*}
    D_n^\pm & = \sum_\sigma (-1)^{\sign(\sigma)} \prod_{i=1}^{2n-1}\big( \delta_{i,\sigma(i)} \pm x^\nu m_{n;i,\sigma(i)} \big) + \O_n\big( x^{\nu-n-\frac12} \big) \\
    & = \det(I \pm x^\nu C_nL_n) + \O_n\big( x^{\nu-n-\frac12} \big) = P_n(\pm x^\nu) + \O_n\big( x^{\nu-n-\frac12} \big).
\end{align*}
The next paragraph shows that \( P_n(\pm x^\nu) = \O_n(1)\) as \( x\to\infty \). This finishes the proof of \eqref{asymptotic formula 3} in view of \eqref{formula for u part 3}.

It only remains to estimate the size of the coefficients of \( P_n(t) \). We have that
\[
P_n(t) = \sum_\sigma (-1)^{\sign(\sigma)} \prod_{i=1}^{2n-1}\big( \delta_{i,\sigma(i)} - t m_{n;i,\sigma(i)} \big).
\]
Hence, \( p_{n,l} \), the coefficients next to \( t^l \), is a sum (with coefficients \( \pm1 \)) of products of the form \( \prod_{i\in S} m_{n;i,\sigma(i)} \), where the cardinality of \( S \subset \{1,2,\ldots,2n-1\} \) is  \( l \).  One can readily see from \eqref{definition Cn} that
\[
c_{n,k} = \O\big(x^{-\frac12}\big) \quad \text{and} \quad c_{n,2k},c_{n,2k+1} = \O\big(x^{-k-\frac12}\big)
\]
for \( k\in\{1,2,\ldots,n-1\} \). Therefore, it follows from the triangular form of \( C_n \) that
\[
m_{n;1,j} = \O\big(x^{-\frac12}\big) \quad \text{and} \quad m_{n;2k,j},m_{n;2k+1,j} = \O\big(x^{-k-\frac12}\big)
\]
independently of \( j \). Then, the worst asymptotic order of \( \prod_{i\in S} m_{n;i,\sigma(i)} \) is obtained by choosing \( S = \{1,2,\ldots,l\} \). This, in particular, implies that
\[
p_{n,1}=\O(x^{-\frac12}), \quad p_{n,2k} = \O\big(p_{n,2k-1} x^{-k-\frac12}\big) \; \text{and} \; p_{n,2k+1} = \O\big(p_{n,2k-1} x^{-2k-1}\big),
\]
from which the desired claim easily follows.

\section{Proof of Theorem~\ref{thm:4}}

Similarly to \eqref{defn of Y}, we take \( \rho = 4/x \) in RHP~\ref{rhp original} and apply the gauge transformation:
\begin{align}
    X(\lambda) := e^{x\ell\sigma_3} \hat{\Psi}\left( \frac{4\lambda}{x} \right) \begin{Bmatrix} I, & |\lambda|>1 \\ \sigma_1\sigma_3, & |\lambda|<1 \end{Bmatrix} e^{-x g(\lambda)\sigma_3},
    \label{defn of X}
\end{align}
for yet to be determined function \( g(\lambda) \) such that \( g(\lambda) =\tfrac\ii4\lambda + \ell + \O(\lambda^{-1})\) as \( \lambda\to\infty \).

Recall \eqref{kappa as elliptic}, \eqref{Elliptic integrals}, and \eqref{ks}. Throughout the proof, it will be useful for us to use the following constants:
\begin{align}
\label{Kalpha}
\mathcal K_\alpha := \KK^2 \kk^3 (\kk^\prime )^6, \quad \mathcal K_\alpha^* := \KK \sqrt{\kk} (\kk^\prime)^2, \quad \text{and} \quad \mathcal K_\alpha^{**} := \mathcal K_\alpha^*(\kk\kk^\prime)^{-\frac13}.    
\end{align}

\subsection{The $g$-function}

In what follows, we denote by $\Gamma(a_1,a_2) \subset S^1$ the subarc of the unit circle $S^1$ with endpoints $a_1$ and $a_2$ oriented from $a_1$ to $a_2$.

\subsubsection{Definition of \( a \)}

We need two preparatory steps before we can define and discuss properties of the desired function \( g(\lambda) \). Here, we make the first one. 

\begin{lem}
\label{lem:varkappa}
The function
\[
\widehat\varkappa(t) = \frac{1}{2} \int_0^{t} \sqrt{2(\cos 2\theta - \cos 2t)} \, d\theta
\]
is continuous and increasing on $( 0, \tfrac{\pi}{2})$ with range $( 0,1)$. In particular, the inverse function is also increasing and continuous with domain $( 0,1)$ and range $( 0, \tfrac{\pi}{2})$. Moreover,
\begin{align}
    \widehat\varkappa(t) = E(\sin t) - \cos^2 t \, K(\sin t) > \frac\pi4 \sin^2 t,
\end{align}
where \( K(\cdot) \) and \( E(\cdot) \) are the complete elliptic integrals of the first and second kind.
\end{lem}
\begin{proof}
The first claim is obvious. It is also easy to see that
\[
\lim_{t \to 0^+} \widehat \varkappa(t) = 0 \quad \text{and} \quad \lim_{t \to \frac{\pi}{2}^-} \widehat \varkappa(t) = 1.
\]
Hence, the inverse function has the desired properties. Using $\cos 2t = 1 - 2 \sin^2t$ gives
\begin{align}
    \widehat\varkappa(t) = \int_0^t \sqrt{\sin^2 t  - \sin^2 \theta} \, d \theta.
\end{align}
Thus, if we let $k=\sin t$ and $y = \frac{\sin \theta}{\sin t}$, then
\[
\widehat\varkappa(t) = k^2 \int_0^1 \frac{\sqrt{1 - y^2}}{\sqrt{1 - k^2y^2}} \, dy = E(\sin t) - \cos^2 t \, K(\sin t). 
\]
Finally, estimate \cite[Equation~(19.9.7)]{NIST:DLMF} gives
\[
\widehat\varkappa(t) > \frac\pi2 \frac{\sin^2 t}{\sqrt{4-\sin^2t}},
\]
which readily yields the last lower bound of the lemma.
\end{proof}

It follows from Lemma~\ref{lem:varkappa} that, given \( \varkappa\in(0,1)\), there exists unique \( \alpha\in(0,\tfrac\pi2) \) satisfying \eqref{defn of alpha}, i.e., \( \varkappa=\widehat\varkappa(\alpha) \). Moreover, \eqref{kappa as elliptic} indeed takes place. Further, set
\begin{equation}
\label{defn of V}
a := e^{\ii \alpha}, \quad V(\varkappa) := \frac{1}{4\pi} \int_{\alpha}^{\pi - \alpha} \sqrt{2(\cos 2\alpha - \cos 2\theta)} \, d\theta.
\end{equation}
Below, we often simply write \( V \) for \( V(\varkappa)\). Observe that
\begin{align*}
\pi V & = \frac{1}{2\pi} \int_{\alpha}^{\pi/2} \sqrt{2(\cos 2\alpha - \cos 2\theta)} \, d\theta = \frac1\pi \int_\alpha^{\pi/2} \sqrt{\cos^2 \alpha  - \cos^2 \theta} \, d \theta \\
& = E(\cos \alpha) - \sin^2 \alpha \, K(\cos \alpha),
\end{align*}
where we used the substitution  $y = \frac{\cos \theta}{\cos \alpha}$. 

Recall our notation \eqref{ks}. Notice that we can view \( \varkappa \) as a function of  \( \kk \), which immediately gives \( \pi V=\varkappa(\kk^\prime) \in(0,1)\). Moreover, we get from the last estimate of Lemma~\ref{lem:varkappa} that
\begin{align}
\label{kappa V estimate}
(\kk\kk^\prime)^2 < 2\varkappa \pi V. 
\end{align}
 Furthermore, Legendre's relation \cite[Equation~(19.7.1)]{NIST:DLMF} gives us
\begin{align}
\label{Legendre's relation}    
\varkappa \KK^\prime + \pi V \KK = \frac\pi2 \quad \Rightarrow  \quad xV = \frac{x-2\kappa \KK^\prime}{2\KK}.
\end{align}

\subsubsection{Definition of \( \sqrt f \)}

Here, we make the second preparatory step. Let \( f(\mu) = (\mu^2 - a^2)(\mu^2 - \overline{a}^2) \) and define 
\begin{align} \label{defn of sqrt of f}
    \sqrt{f(\mu)} = \sqrt{(\mu^2 - a^2)(\mu^2 - \overline{a}^2)} = \sqrt{\mu^4-2\Re(a^2)\mu^2+1}
\end{align}
as a function holomorphic away from
\[
\Delta := \Gamma (-\overline{a}, -a) \cup \Gamma (\overline{a},a)
\]
and such that $\sqrt{f(\mu)} < 0$ for $\mu \in (-1, 1)$ and $\sqrt{f(\mu)} > 0$ for $\mu \in (1, \infty) \cup (-\infty, -1)$, see Figure~\ref{fig:cut}.
\begin{figure}[ht!]
    \centering
    \includegraphics[width=0.5\linewidth]{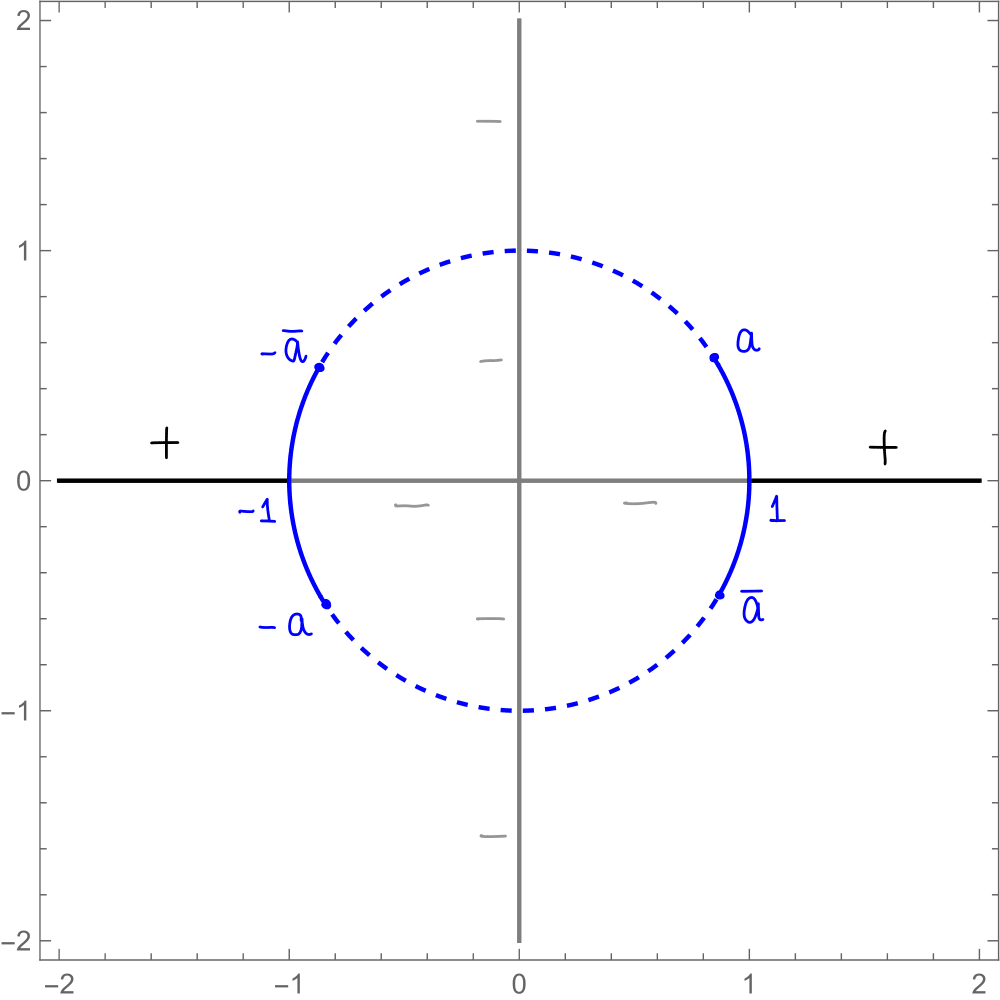}
    \caption{The branch cut for $\sqrt{f(\mu)}$ and its signs on the axes.}
    \label{fig:cut}
\end{figure}
In particular, it holds that
\begin{equation}
    \label{f_series}
    \sqrt{f(\mu)} = \mu^2 \left(1 - \Re (a^2)\mu^{-2} + \O\left(\mu^{-4} \right) \right), \quad \mu \to \infty.
\end{equation}
Observe also that \( \sqrt{f(\mu)} \) is negative on the imaginary axis and
\begin{equation}
\label{f_symmetries}
\sqrt{f(\mu)} = \sqrt{f(-\mu)} = - \mu^2\sqrt{f(\mu^{-1})} = \overline{\sqrt{f(\overline \mu)}}.
\end{equation}
Next, it clearly holds that
\[
f(\mu) = 2\mu^2\big( \Re(\mu^2)-\Re(a^2)\big), \quad \mu\in S^1.
\]
Therefore, the trace of \( \sqrt{f(\mu)} \) on \( S^1 \) from within the unit disk satisfies
\begin{equation}
\label{f_trace}
\sqrt{f(\mu)}_+ = \sqrt{2|\Re(\mu^2)-\Re(a^2)|}
\begin{cases}[r]
-\mu, & \mu\in \Gamma(\overline a,a),  \\
\ii\mu, & \mu\in \Gamma(a,-\overline a), \\
\mu, & \mu\in\Gamma(-\overline a,-a), \\
-\ii\mu, & \mu\in\Gamma(-a,\overline a),
\end{cases}
\end{equation}
where the root on the right-hand side above is arithmetic. Of course, it also holds that
\begin{equation}
\label{f_jump}
\sqrt{f(\mu)}_- = \sqrt{f(\mu)}_+
\begin{cases}[r]
-1, & \mu\in \Delta, \smallskip \\
1, & \mu\in S^1\setminus\Delta.
\end{cases}
\end{equation}

\subsubsection{Definition of \( g \)}

We are now ready to define the \( g \)-function used in \eqref{defn of X}. We set
\begin{align}
\label{defn of g}
    g(\lambda) := \frac{\ii}{4} \int_{1}^{\lambda} \frac{\sqrt{f(\mu)}}{\mu^2} \, d \mu, \quad \lambda \in \C^* \setminus S^1,
\end{align}
where the path of integration does not cross the unit circle and can wind arbitrarily around the origin or infinity because the integrated differential has double poles with zero residues at these two points by \eqref{f_series} and \eqref{f_symmetries}. 

If we formally put \( a=\ii \) in \eqref{defn of g}, which corresponds to \( \varkappa=1\), we will get that
\begin{align}
\label{g theta}
g(\lambda) = \begin{cases}[r]
    \theta(\lambda), & |\lambda|>1, \\
    -\theta(\lambda), & |\lambda|<1.
\end{cases}
\end{align}
In this case, if we multiply \( X(\lambda) \)  by \( \sigma_3\sigma_1 \) in the unit disk, then it will be equal to \( Y(\lambda)\), see \eqref{defn of Y} and \eqref{defn of X}. On the other hand, if we formally put \( a=1 \), which corresponds to \( \varkappa=0 \), then we get
\[
g(\lambda) = \frac\ii 4 \left(\lambda+\frac1\lambda \right) - \frac\ii2,
\]
which, up to trivial summand \( -\frac\ii2\), coincides with the function  \( \varphi(\lambda) \) from \eqref{defn W} further below, which is used to transform \( \hat\Psi(\lambda)\) to study the case \( \varkappa \to0 \) as \( x\to\infty \).

It readily follows from \eqref{f_symmetries} that
\begin{equation}
\label{g_symmetry}
g(\lambda) = g(\lambda^{-1}) = -\overline{g(\overline \lambda)}, \quad \lambda\in\C^*\setminus S^1.
\end{equation}
Moreover, we get from \eqref{f_series} that
\begin{equation}
\label{g_series}
g(\lambda) = \frac{\ii\lambda}{4} + \ell + \frac{\ii}{4}\frac{\Re (a^2)}{\lambda} + \O\left(\frac{1}{\lambda^3} \right), \quad \lambda \to \infty
\end{equation}
(this also gives the first three terms of the Laurent series at \( 0 \) due to \eqref{g_symmetry}), where
\[
\ell = -\frac{\ii}{4} + \frac{\ii}{4} \int_1^{\infty} \left( \frac{\sqrt{f(\mu)}}{\mu^2} - 1 \right) d\mu \in \ii \big(-\tfrac12,0\big)
\]
(the last conclusion follows from the fact that \( \sqrt{f(\mu)}>0 \) is a decreasing function of \( \Re(a^2)\in(-1,1) \) for each \( \mu\in(1,\infty) \)). The definition of \( g(\lambda) \) readily yields
\[
g_-(-1) = g(\lambda) + g(-\lambda) = \left( g(\lambda) -\frac{\ii \lambda}4 \right) + \left( g(-\lambda) +\frac{\ii \lambda}4 \right) = 2\ell,
\]
where we used \eqref{g_series} and independence from \( \lambda \) to get the last equality. Notice that the first symmetry in \eqref{g_symmetry} gives that \( g_+(-1)=g_-(-1)\). Thus,
\begin{equation}
\label{g_symmetry2}
g(\lambda) = 2\ell - g(-\lambda), \quad \lambda\in\C^*\setminus S^1.
\end{equation}

\subsubsection{Traces of \( g \) on \( S^1 \)}

In what follows, it will be important for us to know the relations between the boundary values of \( g(\lambda) \) from both sides of the unit circle. To this end, let
\begin{equation}
\label{PiOmega}
\begin{cases}
\Pi(\lambda) & := g_+(\lambda) - g_-(\lambda), \\
\Omega(\lambda) & := \ii (g_+(\lambda) + g_-(\lambda)),
\end{cases} \quad \lambda\in S^1.
\end{equation}
Recall Lemma~\ref{lem:varkappa}. Extend \( \widehat\varkappa(t) \) to the interval \(( -\tfrac{\pi}{2}, \tfrac{\pi}{2})\) as an odd function. Then, it readily follows from \eqref{f_trace} and \eqref{f_jump} that
\[
\Pi(\lambda) = \frac\ii2 \int_1^\lambda \frac{\sqrt{f(\mu)}_+}{\mu^2} \, d \mu = \frac{1}{2} \int_0^{\Arg \lambda} \sqrt{2(\cos 2\theta - \cos 2\alpha)} \, d\theta = \widehat\varkappa(\Arg\lambda)
\]
for \( \lambda\in \Gamma(\overline a,a) \), where \( \Arg(\lambda)\in(-\pi,\pi] \), the root is arithmetic, and we used $\mu = e^{\ii \theta}$ and $a = e^{\ii \alpha}$. In particular, \( \Pi(a)=\varkappa \) and \( \Pi(\overline a)=-\varkappa \) by the very definition \eqref{defn of alpha}. Using \eqref{g_symmetry2} we conclude that
\begin{equation}
\label{Pi}
\Pi(\lambda) =
\begin{cases}[r]
\widehat\varkappa(\Arg\lambda), & \lambda\in \Gamma(\overline a,a),  \\
\varkappa, & \lambda\in \Gamma(a,-\overline a), \\
-\widehat\varkappa(\Arg(-\lambda)), & \lambda\in\Gamma(-\overline a,-a), \\
-\varkappa, & \lambda\in\Gamma(-a,\overline a).
\end{cases}
\end{equation}
Observe that \( \Pi(\lambda) \) is real-valued. Further, let
\[
\widehat V(t) := \frac1{4\pi}\int_\alpha^t \sqrt{2(\cos2\alpha-\cos2\theta)}\,d\theta, \quad t\in(\alpha,\pi-\alpha).
\]
Clearly, \( \widehat V(t) \) is an increasing function of \( t \). Notice that \( V=\widehat V(\pi-\alpha)\), see \eqref{defn of V}. We readily get from \eqref{f_trace} and \eqref{f_jump} that \( \Omega(\lambda) =0 \) on \( \Gamma(\overline a,a) \) and that
\[
\Omega(\lambda) = -\frac12 \int_a^\lambda \frac{\sqrt{f(\mu)}}{\mu^2}\,d\mu = 2\pi \widehat V(\Arg\lambda)
\]
for \( \lambda \in \Gamma(a,-\overline a) \). Extend \( \widehat V(t) \) to \( (\alpha-\pi,-\alpha) \) as an even function. Then, it follows from \eqref{g_symmetry} that
\begin{equation}
\label{Omega}
\Omega(\lambda) = 2\pi
\begin{cases}[r]
0, & \lambda\in \Gamma(\overline a,a),  \\
V, & \lambda\in\Gamma(-\overline a,-a), \\
\widehat V(\Arg\lambda), & \lambda\in\Gamma(a,-\overline a)\cup \Gamma(-a,\overline a).
\end{cases}
\end{equation}
As in the case of \( \Pi(\lambda) \), the function \( \Omega(\lambda) \) is real-valued.

Relations \eqref{PiOmega}, \eqref{Pi}, and \eqref{Omega} show that \( 2g(\lambda)\) maps the unit disk onto the complement of the rectangle with vertices \( \varkappa,\varkappa-2\pi\ii V,-\varkappa-2\pi\ii V,-\varkappa\) with the unit circle being mapped into the boundary of this rectangle in one-to-one fashion. It then follows from the Argument Principle that \( g(\lambda) \) is univalent in  the unit disk and therefore in the complement of the closed unit disk as well (recall \eqref{g_symmetry}). 

\begin{lem}
\label{lem:V_ell_connection}
It holds that \(2\ii\ell = \pi V \).
\end{lem}
\begin{proof}
It readily follows from \eqref{g_series}, \eqref{PiOmega}, \eqref{Pi}, and \eqref{Omega} that
\[
\ell = \frac1{2\pi}\int_{S^1} g_-(\lambda)|d\lambda| = -\frac1{4\pi}\int_{S^1}(\Pi(\lambda)+\ii\Omega(\lambda))|d\lambda| = -\ii\alpha V - \ii\int_\alpha^{\pi-\alpha} \widehat V(t)dt.
\]
Exchanging the order of integration gives
\begin{align}
\int_\alpha^{\pi-\alpha} \widehat V(t)dt & = \frac1{4\pi}\int_\alpha^{\pi-\alpha} (\pi-\alpha-\theta)\sqrt{2(\cos2\alpha-\cos2\theta)}\,d\theta \\
& = \frac1{4\pi}\int_\alpha^{\pi-\alpha} \left(\frac\pi2-\alpha\right)\sqrt{2(\cos2\alpha-\cos2\theta)}\,d\theta = \left(\frac\pi2-\alpha\right)V,
\end{align}
where the second equality follows by symmetry (integrals of functions that are odd with respect to the midpoint of the interval of integration are zero). The claim now follows from the above two equalities. 
\end{proof}

\subsubsection{Estimates of \( \Re(g)\)}

We close this subsection with one more observation that will become useful later. 

\begin{lem}
\label{lem:Re_g}
It holds that
\[
2\Re(g(\lambda)) \geq \varkappa + s({\Arg\lambda}) \sinh^2\big(\tfrac12\ln|\lambda|\big), \quad \begin{cases}
|\lambda| \leq 1, & \Arg(\lambda)\in[\alpha,\pi-\alpha], \smallskip \\
|\lambda| \geq 1, & \Arg(\lambda)\in[-\pi+\alpha,-\alpha],
\end{cases}
\]
and
\[
2\Re(g(\lambda)) \leq -\varkappa - s({\Arg\lambda}) \sinh^2\big(\tfrac12\ln|\lambda|\big), \quad \begin{cases}
|\lambda| \geq 1, & \Arg(\lambda)\in[\alpha,\pi-\alpha], \smallskip \\
|\lambda| \leq 1, & \Arg(\lambda)\in[-\pi+\alpha,-\alpha],
\end{cases}
\]
where
\[
s(\theta):= \frac1{2^{1/4}}\begin{cases}[r]
\sin(2\theta), & |\theta|\in[0,\tfrac\pi4]\cup[\tfrac{3\pi}4,\pi], \smallskip \\
1, & |\theta|\in[\tfrac\pi4,\tfrac{3\pi}4].
\end{cases}
\]
\end{lem}

\begin{proof}
Let \( \lambda = |\lambda|e^{\ii\theta}\), where \(|\lambda|\leq 1\) and \( \theta\in[\alpha,\tfrac\pi2]\). It follows from \eqref{PiOmega} and \eqref{Pi} that
\[
2\Re(g(\lambda)) = \varkappa + \Re\left( \frac\ii 2\int_{e^{\ii\theta}}^\lambda \frac{\sqrt{f(\mu)}}{\mu^2} d\mu\right).
\]
When \( \mu = te^{\ii\theta}\) with \(0<t<1\), we have that
\begin{align*}
f(\mu) &= e^{2\ii\theta} t^2 \big( e^{2\ii\theta} t^2 - 2\Re\left( e^{2\ii \alpha} \right) + e^{-2\ii\theta} t^{-2} \big) \\
& = e^{2\ii\theta} t^2 \big(\cos(2\theta)(t^2+t^{-2})-2\cos(2\alpha)-\ii\sin(2\theta)(t^{-2}-t^2) \big).
\end{align*}
Notice that the imaginary part of the term in parentheses is negative. Recall also that \( \sqrt{f(\mu)}\) is holomorphic in the unit disk and is negative on the imaginary axis, i.e., when \( \theta=\tfrac\pi2\). Hence,
\[
\mu^{-1}\sqrt{f(\mu)} = -\sqrt{\cos(2\theta)(t^2+t^{-2})-2\cos(2\alpha)-\ii\sin(2\theta)(t^{-2}-t^2)},
\]
where the branch of the square root is now principal. In particular, the above value lies in the second quadrant and therefore its imaginary part is non-negative. Thus,
\[
2\Re(g(\lambda)) = \varkappa - \Im\left( \int_{|\lambda|}^1 \sqrt{\cos(2\theta)(t^2+t^{-2})-2\cos(2\alpha)-\ii\sin(2\theta)(t^{-2}-t^2)}\frac{dt}{2t}\right).
\]
Observe that if \( v,q>0\), then
\[
-\sqrt{u-\ii v} = p + \ii q \quad \Rightarrow \quad 2q^2 = \sqrt{u^2+v^2}-u \geq \begin{cases} v-u, \\ \displaystyle \frac{v^2}{2\sqrt{u^2+v^2}}, \end{cases}
\]
where all the square roots are principal. Let
\[
u=\cos(2\theta)(t^2+t^{-2})-2\cos(2\alpha) \quad \text{and} \quad v=\sin(2\theta)(t^{-2}-t^2).
\]
When \( \theta\in[\tfrac\pi4,\tfrac\pi2]\), in which case \( \cos(2\theta)\leq\min\{\cos(2\alpha),0\}\), we have that
\begin{align*}
2 q^2 &\geq v-u = \sin(2\theta)(t^{-2}-t^2)-\cos(2\theta)(t^2+t^{-2})+ 2\cos(2\alpha) \\ &\geq (\sin(2\theta)-\cos(2\theta))(t^{-1}-t)^2  + 2(\cos(2\alpha) - \cos(2\theta)) \geq (t^{-1}-t)^2
\end{align*}
for \( t\in[0,1] \). On the other hand, when \( \theta\in[\alpha,\tfrac\pi4] \) (if this set is non-empty), we have
\[
u^2 + v^2 \leq \big( t^2+t^{-2} + 2 \big)^2 + \big( t^{-2} -t^2 \big)^2 \leq 2\big( t+ t^{-1} \big)^4
\]
and therefore
\[
2q^2 \geq \frac{v^2}{2\sqrt2(t+t^{-1})^2} = \frac{\sin^2(2\theta)}{2\sqrt2}(t^{-1}-t)^2
\]
for \( t\in[0,1] \). Altogether, it holds that
\[
2\Re(g(\lambda)) \geq \varkappa + \frac{s(\theta)}{4}\int_{|\lambda|}^1 \left( \frac1{t^2}-1\right)dt = \varkappa + s(\theta) \sinh^2\big(\tfrac12\ln|\lambda|\big).
\]
This finishes the proof of the lemma when \( |\lambda|\leq 1\) and \( \Arg\lambda\in[\alpha,\tfrac\pi2] \). Since
\[
\Re(g(-\overline\lambda)) = \Re(g(\lambda))
\]
by \eqref{g_symmetry} and \eqref{g_symmetry2} as \( \ell \) is purely imaginary, we can extend the above estimate to \( \Arg\lambda\in[\alpha,\pi-\alpha] \). The other three regions are again obtained via symmetry relations \eqref{g_symmetry} and \eqref{g_symmetry2}.
\end{proof}

\subsection{$X$-RHP and its Deformation}

Let \( X(\lambda) \) be defined by \eqref{defn of X}. It solves the following Riemann-Hilbert problem.

\begin{RHP} 
\label{RHP X}
Find a $2 \times 2$ matrix function $X(\lambda)$ such that
\begin{enumerate}
    \item $X(\lambda)$ is analytic for $\lambda \in \C \setminus \Gamma$, where the oriented contour $\Gamma$ consists of the imaginary axis and the unit circle \( S^1 \) oriented as on Figure~\ref{X problem pic}; 
    \item one-sided traces $X_\pm(\lambda)$ exist a.e. on \( \Gamma \), belong to \( L^2(\Gamma) \), and satisfy
    \begin{align*}
        X_+(\lambda) = X_-(\lambda) G_X(\lambda), \quad \lambda \in \Gamma,
    \end{align*}
    where the jump matrices $G_X(\lambda)$ on \( \Gamma \) are as on Figure~\ref{X problem pic};
    \item it holds that\footnote{Recall \eqref{g_symmetry} and \eqref{g_series}.}
    \begin{align}
    \label{X normalization}
        X(\lambda) = \begin{cases}
            I + \O(1/\lambda)  &  \text{as} \quad \lambda \to\infty, \smallskip \\
            e^{x \ell \sigma_3} P_0 \sigma_1 \sigma_3 e^{-x \ell \sigma_3} (I + \O(\lambda)) & \text{as} \quad \lambda \to 0.
        \end{cases}
    \end{align}
\end{enumerate}
\end{RHP}

\begin{figure}[ht!]
    \centering
    \includegraphics[width=8cm]{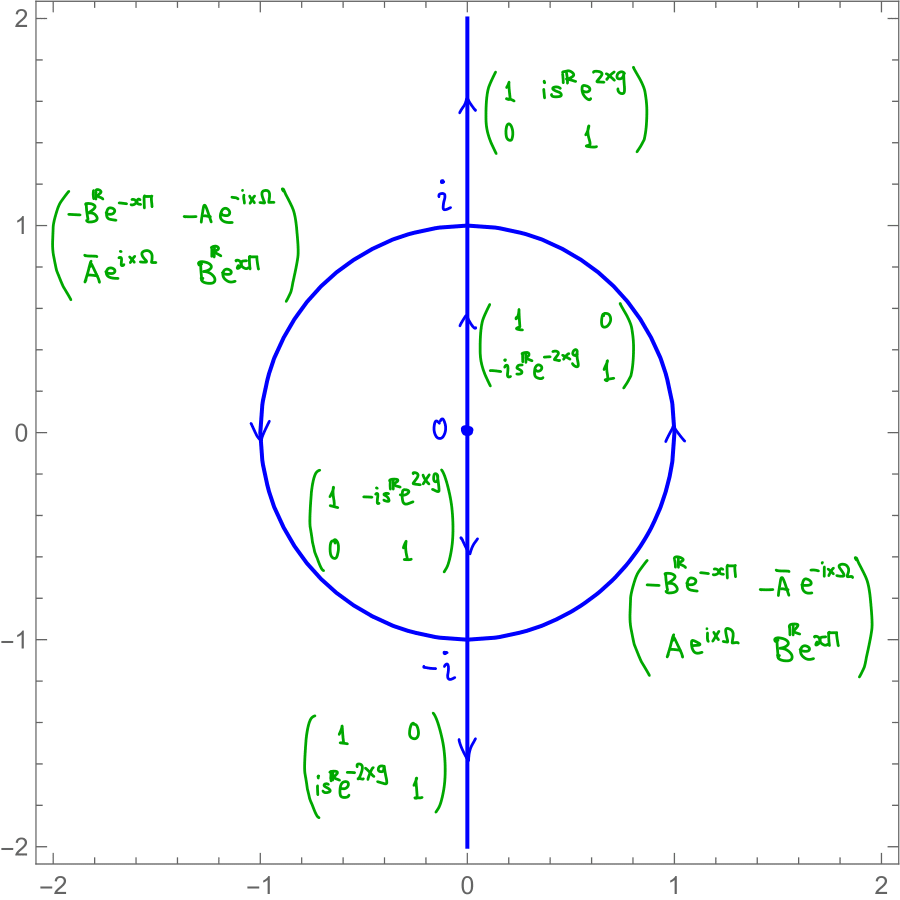}
    \caption{The jump matrices $G_X(\lambda)$ for RHP~\ref{RHP X}.}
    \label{X problem pic}
\end{figure}

Notice that on the imaginary axis away from \( \pm \ii \) the jump matrix \( G_X(\lambda)\) is exponentially close to the identity matrix as \( x\to\infty \) due to Lemma~\ref{lem:Re_g}. Furthermore, it holds on \( \Gamma(\overline a,a)\) that
\begin{equation}
\label{cut1}
G_X(\lambda) = \begin{pmatrix}
        -B^\R e^{-x \Pi(\lambda)} & -\overline A\\
        A & B^\R e^{x \Pi(\lambda)} 
    \end{pmatrix} =
    \begin{pmatrix}
        -e^{-x (\varkappa+\Pi(\lambda))} & -1+\O(|s^\R|_+e^{-\varkappa x}) \\ 1+\O(|s^\R|_+e^{-\varkappa x}) & e^{-x (\varkappa-\Pi(\lambda))} 
    \end{pmatrix},
\end{equation}
where we used \eqref{Omega} as well as \eqref{A is almost 1} to get the last equality. Observe that the diagonal entries are also exponentially close to zero away from \( \overline a,a\)  since \( |\Pi(\lambda)|<\varkappa \) by \eqref{Pi} and Lemma~\ref{lem:varkappa}. Similarly, it holds on \( \Gamma(-\overline a,-a)\) that
\begin{align}
G_X(\lambda) &= \begin{pmatrix}
        -e^{-x (\varkappa+\Pi(\lambda))} & - Ae^{-2\pi\ii xV}\\
        \overline Ae^{2\pi\ii xV} & e^{-x (\varkappa-\Pi(\lambda))} 
    \end{pmatrix} \\ &=
    \begin{pmatrix}
        -e^{-x (\varkappa+\Pi(\lambda))} & -e^{-2\pi\ii xV}+\O(|s^\R|_+e^{-\varkappa x}) \\ e^{2\pi\ii xV} + \O(|s^\R|_+e^{-\varkappa x}) & e^{-x (\varkappa-\Pi(\lambda))} 
    \end{pmatrix},
    \label{cut2}
\end{align}
where we used \eqref{Omega} and \eqref{A is almost 1}. Again, the diagonal entries are exponentially close to zero away from \( -\overline a,-a\). To discuss \( G_X(\lambda) \) on \( \Gamma(a,-\overline a)\) and \( \Gamma(-a,\overline a)\), first observe that
\[
\begin{dcases}
\Pi(\lambda)=\varkappa \quad \text{and} \quad -\ii\Omega(\lambda) = 2g_+(\lambda)-\varkappa = 2g_-(\lambda) + \varkappa, & \lambda \in \Gamma(a, -\overline a), \\
\Pi(\lambda)=-\varkappa \quad \text{and} \quad -\ii\Omega(\lambda) = 2g_+(\lambda)+\varkappa = 2g_-(\lambda)-\varkappa, & \lambda \in \Gamma(-a, \overline a),
\end{dcases}
\]
by \eqref{PiOmega}, \eqref{Pi}, and \eqref{Omega}. Therefore,
\begin{align}
G_X(\lambda) &= \begin{pmatrix}
        -(B^\R)^2 & - \overline A e^{-\ii x \Omega(\lambda)}\\
        A e^{\ii x \Omega(\lambda)} & 1 
    \end{pmatrix} \\
& =   \begin{pmatrix}
        1 & - \overline A e^{x(\varkappa+2g_-(\lambda))}\\
        0 & 1 
    \end{pmatrix}
    \begin{pmatrix}
        1 & 0\\
        A e^{x(\varkappa-2g_+(\lambda))} & 1 
    \end{pmatrix} =: S_{L_1} (\lambda) S_{R_1} (\lambda)  
\label{defn of SL1 and SU1}
\end{align}
for $\lambda \in \Gamma (a, \ii)$, where we used the identity \( |A|^2 = 1 + (B^\R)^2\). Quite similarly, it holds that
\begin{equation}    
\label{defn of SL2 and SU2}
G_X(\lambda) = \begin{pmatrix}
        1 & - A e^{x(\varkappa+2g_-(\lambda))} \\
        0 & 1 
    \end{pmatrix}
    \begin{pmatrix}
        1 & 0\\
        \overline A e^{x(\varkappa-2g_+(\lambda))} & 1 
    \end{pmatrix}  =: S_{L_2} (\lambda) S_{R_2} (\lambda)
\end{equation}
for $\lambda \in \Gamma (\ii,-\overline a)$. Next, we have that
\begin{align}
    G_X(\lambda) & = \begin{pmatrix}
        -1 & - A e^{-\ii x \Omega(\lambda)}\\
        \overline A e^{\ii x \Omega(\lambda)} & (B^\R)^2 
    \end{pmatrix}\\
    &= \begin{pmatrix}
        1 & 0\\
        -\overline A e^{x(\varkappa - 2g_-(\lambda))} & 1 
    \end{pmatrix}
    \begin{pmatrix}
        -1 & 0\\
        0 & -1
    \end{pmatrix}
    \begin{pmatrix}
        1 & A e^{x(\varkappa + 2g_+(\lambda))}\\
        0 & 1 
    \end{pmatrix} \\
    &=: S_{L_3}(\lambda) (-I) \, S_{R_3}(\lambda)
    \label{defn of SL3 and SU3}
\end{align}
for $\lambda \in \Gamma (-a, -\ii)$ and finally
\begin{align}
    G_X(\lambda)     &= \begin{pmatrix}
        1 & 0\\
        -A e^{x(\varkappa - 2g_-(\lambda))} & 1 
    \end{pmatrix}
    \begin{pmatrix}
        -1 & 0\\
        0 & -1
    \end{pmatrix}
    \begin{pmatrix}
        1 & \overline A e^{x(\varkappa + 2g_+(\lambda))}\\
        0 & 1 
    \end{pmatrix} \\
    &=: S_{L_4}(\lambda) (-I) \, S_{R_4}(\lambda)
    \label{defn of SL4 and SU4}
\end{align}
$\lambda \in \Gamma (-\ii, \bar a)$. In \eqref{defn of SL1 and SU1}, \eqref{defn of SL2 and SU2}, \eqref{defn of SL3 and SU3}, and \eqref{defn of SL4 and SU4}, all upper and lower triangular matrices have fast oscillating off-diagonal entries as $x \to \infty$.

Let $\gamma_i^{in}$ and $\gamma_i^{out}$, $i \in \{1, 2, 3, 4\}$, be arcs as depicted on Figure~\ref{tilde X problem pic}. That is, we assume that $\gamma_1^{in}\cup \gamma_2^{in}$ and $\gamma_3^{in}\cup \gamma_4^{in}$ are piecewise-smooth arcs that belong to the unit disk (except for the endpoints) oriented from \( a \) to \( -\overline a \) and \( -a \) to \( \overline a \), respectively, while $\gamma_1^{out}\cup\gamma_2^{out}$ and $\gamma_3^{out}\cup\gamma_4^{out}$ are piecewise-smooth arcs that belong to the exterior of the closed unit disk (again, except for the endpoints) oriented  from \( a \) to \( -\overline a \) and \( -a \) to \( \overline a \), respectively, and
\[
\begin{cases}
\gamma_1^{in},\gamma_1^{out} \subset \big\{\lambda:\Arg\lambda\in[\alpha,\tfrac\pi2]\big\}, & \gamma_3^{in},\gamma_3^{out} \subset \big\{\lambda:\Arg\lambda\in[-\pi+\alpha,-\tfrac\pi2]\big\}, \smallskip \\
\gamma_2^{in},\gamma_2^{out} \subset \big\{\lambda:\Arg\lambda\in[\tfrac\pi2,\pi-\alpha]\big\}, & \gamma_4^{in},\gamma_4^{out} \subset \big\{\lambda:\Arg\lambda\in[-\tfrac\pi2,-\alpha]\big\}.
\end{cases}
\]
In fact, we shall specify these arcs exactly when discussing local parametrices further below. Moreover, let $\mathcal{R}_i^{in}$ (resp. $\mathcal{R}_i^{out}$) be domains bordered by the unit circle, imaginary axis, and $\gamma_i^{in}$ (resp. $\gamma_i^{out}$), $i \in \{1, 2, 3, 4\}$, see Figure~\ref{tilde X problem pic}.

\begin{figure}[ht!]
    \centering
    \includegraphics[width=0.9\linewidth]{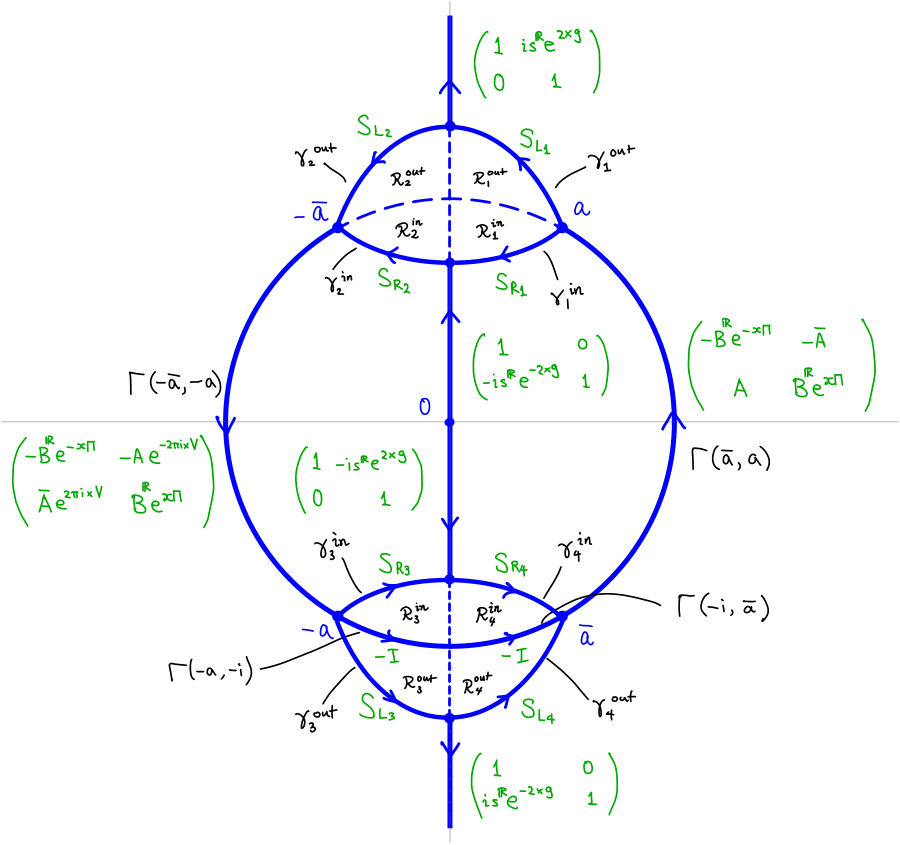}
    \caption{The jump $G_{\tilde X}(\lambda)$ and the contour \( \tilde \Gamma \) for RHP~\ref{RHP tilde X}.}
    \label{tilde X problem pic}
\end{figure}

Notice that each \( S_{R_i}(\lambda) \) naturally extends as a holomorphic matrix function into \( \mathcal{R}_i^{in} \) and each  \( S_{L_i}(\lambda) \) into \( \mathcal{R}_i^{out} \). Hence, we can define a sectionally holomorphic function $\tilde X(\lambda)$ by
\begin{align}
\tilde X(\lambda) := 
\begin{dcases}
X(\lambda) S_{L_i}, (\lambda) & \lambda \in \mathcal{R}_i^{out}, \quad i\in\{1,2,3,4\}, \\
X(\lambda) S_{R_i}^{-1}, (\lambda) & \lambda \in \mathcal{R}_i^{in}, \quad i\in\{1,2,3,4\},\\
X(\lambda), & \text{otherwise}.
\end{dcases}
\end{align}
Observe that
\begin{align}
\label{SL factorization}
S_{L_1}^{-1}(\lambda) \begin{pmatrix} 1 & \ii s^\R e^{2xg(\lambda)} \\ 0 & 1 \end{pmatrix} S_{L_2}(\lambda) = \begin{pmatrix} 1 &  \big(\ii s^\R + (\overline A-A)e^{\varkappa x}\big) e^{2xg(\lambda)} \\ 0 & 1 \end{pmatrix} = I
\end{align}
for \( \lambda \in \partial \mathcal R_1^{out} \cap \partial \mathcal R_2^{out} \), where we used \eqref{mon data}. Similarly,
\begin{equation}
\label{SU factorization}
S_{R_1}(\lambda) \begin{pmatrix} 1 & 0 \\ -\ii s^\R e^{-2xg(\lambda)} & 1 \end{pmatrix} S_{R_2}^{-1}(\lambda) = \begin{pmatrix} 1 & 0 \\ \big(-\ii s^\R + (A-\overline A)e^{\varkappa x}\big) e^{-2xg(\lambda)} & 1 \end{pmatrix} = I
\end{equation}
for \( \lambda \in \partial \mathcal R_1^{in} \cap \partial \mathcal R_2^{in} \). Moreover, the computations showing that \( \tilde X(\lambda) \) has no jump on \( \partial \mathcal R_3^{in} \cap \partial \mathcal R_4^{in} \) and \( \partial \mathcal R_3^{out} \cap \partial \mathcal R_4^{out} \) are absolutely identical. Thus, $\tilde X(\lambda)$ satisfies the following Riemann-Hilbert problem.
\begin{RHP} \label{RHP tilde X}
Find a $2 \times 2$ matrix function $\tilde X(\lambda)$ such that
\begin{enumerate}
    \item $\tilde X(\lambda)$ is analytic for $\lambda \in \C \setminus \tilde \Gamma$, where \( \tilde \Gamma \) is depicted on Figure~\ref{tilde X problem pic};
    \item one-sided traces $\tilde X_\pm(\lambda)$ exist a.e. on \( \tilde \Gamma \), belong to \( L^2(\tilde \Gamma) \), and satisfy
    \begin{align*}
        \tilde X_+(\lambda) = \tilde X_-(\lambda) G_{\tilde X}(\lambda), \quad \lambda \in \tilde \Gamma,
    \end{align*}
    where the jump matrices $G_{\tilde X}(\lambda)$ on \( \tilde \Gamma \) are as on Figure~\ref{tilde X problem pic};
    \item it holds that
    \begin{align}
    \label{Z normalization}
        \tilde X(\lambda) = \begin{cases}
            I + \O(1/\lambda)  &  \text{as} \quad \lambda \to\infty, \smallskip \\
            e^{x \ell \sigma_3} P_0 \sigma_1 \sigma_3 e^{-x \ell \sigma_3} (I + \O(\lambda)) & \text{as} \quad \lambda \to 0.
        \end{cases}
    \end{align}
\end{enumerate}
\end{RHP}

One can readily see from Lemma~\ref{lem:Re_g} that the jump matrix \( G_{\tilde X}(\lambda) \) is exponentially close to the identity on the arcs \( \gamma_i^{in},\gamma_i^{out}\) away from their endpoints.

\subsection{Global Parametrix}

If we replace \( G_{\tilde X}(\lambda) \) by \( I \) everywhere it is close to \( I \) asymptotically as \(x\to\infty\), we shall arrive at the following Riemann-Hilbert problem:
\begin{RHP}\label{RHP N}
Find a $2 \times 2$ matrix function $N(\lambda)$ such that
\begin{enumerate}
    \item $N(\lambda)$ is analytic for $\lambda \in \C \setminus \Gamma(-\overline a,a)$;
    \item one-sided traces $N_\pm(\lambda)$ exist a.e. on \( \Gamma(-\overline a,a) \), belong to \( L^2(\Gamma(- \overline a,a)) \), and satisfy
    \begin{equation}
    N_+(\lambda) = N_-(\lambda) 
    \begin{cases}[r]
    \begin{pmatrix}
        0 & - e^{-2\pi \ii x V}\\
        e^{2\pi \ii x V} & 0
    \end{pmatrix}, & \lambda \in \Gamma(-\overline a, -a),\\
    -I, & \lambda \in \Gamma(-a, \overline a),\\
    \begin{pmatrix}
        0 & -1 \\
        1 & 0
    \end{pmatrix}, & \lambda \in \Gamma(\overline a, a);
    \end{cases}
    \label{N jump}
    \end{equation}
    \item it holds that \( N(\lambda) =  I + \O(1/\lambda)  \)  as \( \lambda \to \infty \).
\end{enumerate}
\end{RHP}

This type of Riemann-Hilbert problem is well-known, see for example \cite{DKMVZ2, DKMVZ1}. We present its solution in full as the details of the construction are important to the final result and are always problem-specific. Moreover, in the process we obtain very technical asymptotic estimates of \( N(\lambda) \) needed for our proof.


\subsubsection{Riemann Surface}

Consider an elliptic curve
\begin{align}
    \S = \{ \z := (z, w) \, | \, w^2 = f(z) \}, \quad f(z) = (z^2 - a^2)(z^2 - \overline{a}^2),
\end{align}
which is a Riemann surface of genus one. We realize $\S$ as a two-sheeted ramified covering of the Riemann sphere $\overline{\C}$ consisting of two copies of $\overline{\C}\setminus\Delta$, $\Delta = \Gamma(-\overline{a}, -a) \cup \Gamma(\overline{a}, a)$, which are glued crosswise along the corresponding cuts. Let
\begin{align}
    \pi, w: \S \to \overline{\C}, \quad \pi: \z \mapsto z, \quad w: \z \mapsto w. 
\end{align}
We call $\pi$ the canonical projection and often denote points in $\S$ as a preimage $\z = \pi^{-1}(z)$ of $z \in \overline{\C}$ under $\pi$. The function $w(\z)$ is meromorphic on $\S$ with simple zeros at the four ramification points $\{\pm \a, \, \pm \overline{\a} \}$ and double poles at \( \infty^{(0)} \) and \( \infty^{(1)} \), the points on top of infinity in each of the sheets $\overline{\C} \setminus \Delta$. We label these sheets by $\S^{(0)}$ and $\S^{(1)}$ so that 
\begin{align}
    w(\z) = (-1)^k z^2 + \O(1), \quad \z \to \infty^{(k)} \in \S^{(k)}.
    \label{fix of branch}
\end{align}
In particular, we shall write \( z^{(k)}= \pi^{-1}(z)\cap \S^{(k)} \), \( k\in\{0,1\} \). 

\begin{figure}[ht!]
    \centering
    \includegraphics[width=0.5\linewidth]{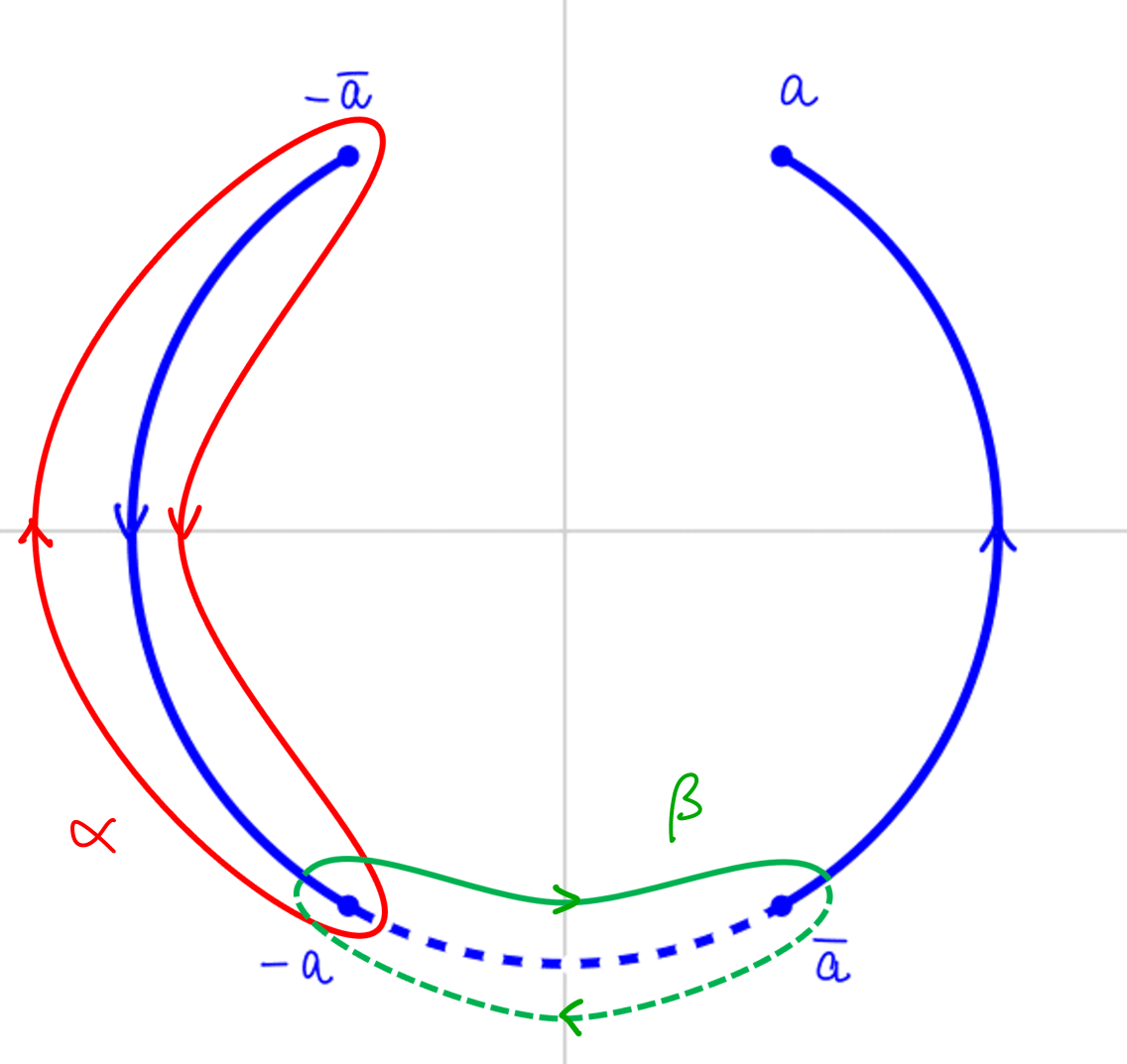}
    \caption{Homology basis for $\S$ depicted on $\S^{(0)}$ (which is identified with $\overline{\C} \setminus \Delta$).}
    \label{homology basis}
\end{figure}

We fix a homology basis, $\boldsymbol{\alpha}$-cycle and $\boldsymbol{\beta}$-cycle, on $\S$ as indicated in Figure \ref{homology basis}. Opening $\S$ via the cut along a homology basis will give the planar picture, see Figure \ref{planer model}, and we denote the resulting cut-surface by $\S_{\boldsymbol{\alpha, \beta}} := \S \setminus \{ \boldsymbol{\alpha} \cup \boldsymbol{\beta} \}$.
\begin{figure}[ht!]
    \centering
    \includegraphics[width=0.7\linewidth]{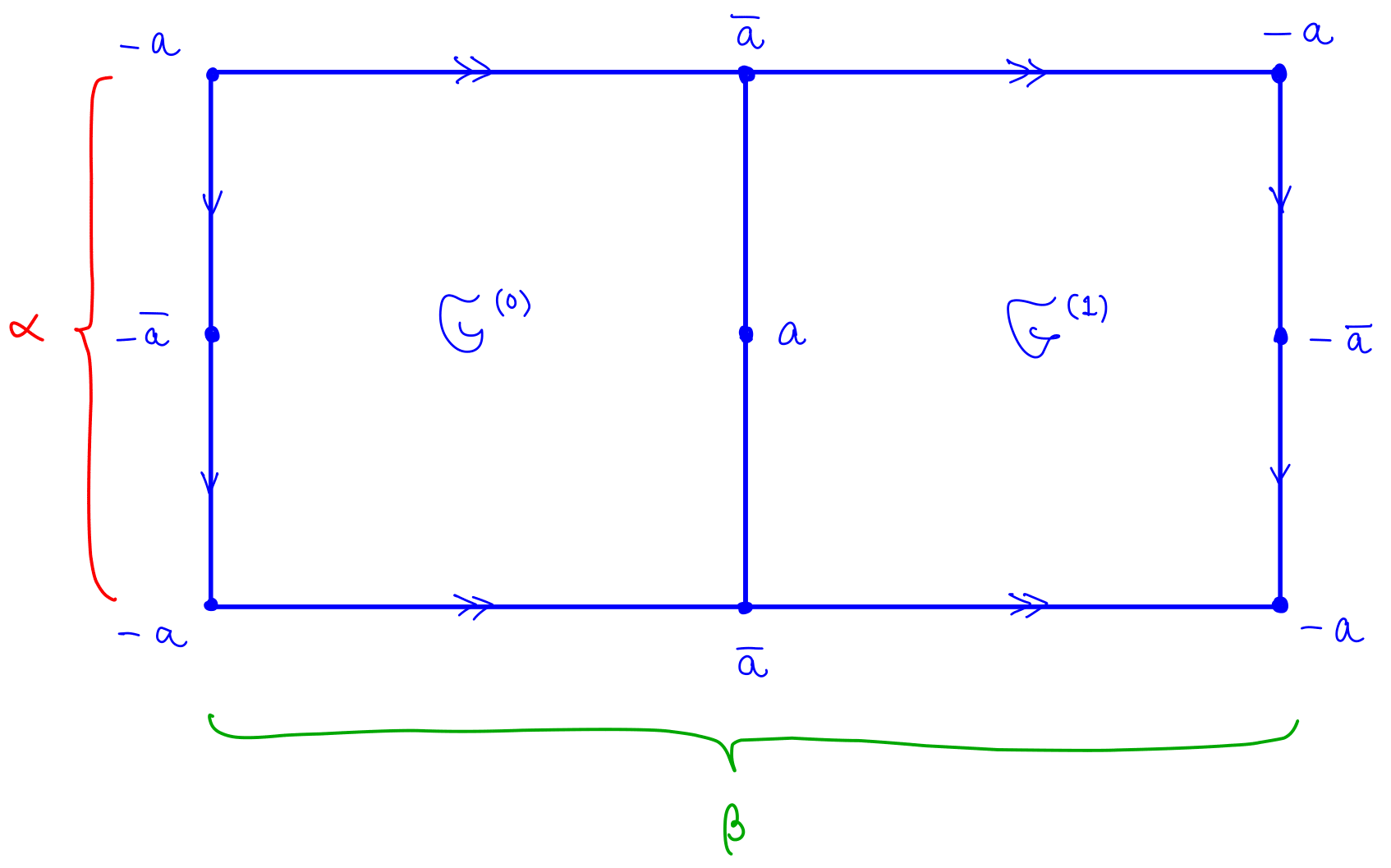}
    \caption{Planar model of $\S$ and the $\boldsymbol{\alpha}$- and $\boldsymbol{\beta}$-cycles.}
    \label{planer model}
\end{figure}


\subsubsection{Abel-Jacobi map}

Recall that the space of holomorphic differentials on any elliptic curve is one-dimensional, and we denote the unique (up to normalization) holomorphic differential on \( \S \) by
\begin{align}
    \omega = \left( \oint_{\boldsymbol{\alpha}} \frac{dz}{w(\z)} \right)^{-1} \frac{dz}{w(\z)}
    \label{holomorphic differential}
\end{align}
with normalized periods
\begin{align}
\label{periods}
    \oint_{\boldsymbol{\alpha}} \omega = 1 \quad \text{and} \quad \oint_{\boldsymbol{\beta}} \omega =: \tau.
\end{align}
According to our choice of \( \boldsymbol\alpha\) and the normalization of \( w\), we have that
\[
\oint_{\boldsymbol{\alpha}} \frac{dz}{w(\z)} = 2\int_{\Gamma(-\overline a,-a)} \frac{d\mu}{\sqrt{f(\mu)}_+} = 2\ii \int_{\pi-\alpha}^{\pi+\alpha} \frac{d\theta}{\sqrt{2(\cos(2\theta)-\cos(2\alpha))}} = 2\ii \KK,
\]
recall \eqref{Elliptic integrals} and the notation introduced in \eqref{ks}, and
\begin{align*}
\oint_{\boldsymbol{\beta}} \frac{dz}{w(\z)} &= 2\int_{\Gamma(-a,\overline a)} \frac{d\mu}{\sqrt{f(\mu)}_+} = -2 \int_{\pi+\alpha}^{2\pi-\alpha} \frac{d\theta}{\sqrt{2(\cos(2\alpha)-\cos(2\theta))}} \\
& = -2\int_0^{\tfrac\pi2-\alpha} \frac{d\theta}{\sqrt{\cos^2\alpha-\sin^2\theta}} = -2
\KK^\prime,
\end{align*}
where we used \eqref{f_trace}, see also Lemma~\ref{lem:varkappa} and its proof. In particular,
\begin{equation}
\label{defn of tau}
\tau = \ii \frac{\KK^\prime}{\KK}.
\end{equation}

By \eqref{holomorphic differential}, the elliptic integral of the first kind
\begin{align}
    \aj(\z) = \int_{\a}^{\z} \omega, \quad \z \in \S,
    \label{Abel-type map}
\end{align}
where $\a = \pi^{-1}(a)$, defines a locally single-valued but globally additively multi-valued holomorphic function. More precisely, if we restrict the path of integration to \( \S_{\boldsymbol\alpha,\boldsymbol\beta} \), then \( \aj(\z)\) is a well-defined function and it follows from \eqref{periods} that
\begin{align}
    \aj_+ (\z) - \aj_- (\z) = \begin{dcases}
        -\tau, & \z \in \boldsymbol{\alpha} \setminus \{-\a\},\\
        1, & \z \in \boldsymbol{\beta} \setminus \{-\a\}.
    \end{dcases}
    \label{AJ+ - AJ-}
\end{align}
The above relations also means that if we let \( \Lambda := \Z+\tau\Z \), the Abel-Jacobi map from \( \S \) into the Jacobi variety $\C/\Lambda$, given by
\begin{align}
    \mathfrak a: \S \to \C/\Lambda, \quad \z \mapsto \mathfrak a(\z) = \left[ \aj(\z) \right] = \left[ \int_{\a}^{\z} \omega \right],
    \label{Abel map}
\end{align}
is well-defined, where \( [s] = [s+l+\tau m:l,m\in\Z] \). Since $\S$ is a genus one Riemann surface, \eqref{Abel map} is a biholomorphism, see \cite[Theorem 21.10]{forster} or a more recent book \cite[Chapter 10]{takebe}. 

By considering \( \aj(\z)\) as a function in \( \S_{\boldsymbol\alpha,\boldsymbol\beta} \) only, it readily follows from its definition that
\begin{equation}
\label{AJ_symmetry}
\aj\big(z^{(0)}\big) + \aj\big(z^{(1)}\big) = 0, \quad z\notin \Gamma(-\overline a, a),
\end{equation}
and from \eqref{fix of branch} and \eqref{Abel-type map} that
\begin{align}
    \aj\big(z^{(0)}\big) = \aj\big(\infty^{(0)}\big) - \frac1{2\ii \KK}\frac1z + \O \left(\frac1{z^3} \right), \quad z \to \infty.
\end{align}
To compute \( \aj\big(\infty^{(0)}\big) \), observe that
\[
\aj\big(\infty^{(0)}\big) = \int_\gamma \omega = \int_{-\gamma}\omega + \frac1{2\ii \KK} \int_{\Gamma(a,-a)}\frac{d\mu}{\sqrt{f(\mu)}_-}
\]
where \( \gamma \) is any arc that connects \( \a\) to \( \infty^{(0)}\) within \(\S_{\boldsymbol\alpha,\boldsymbol\beta}^{(0)}\). Since \( \sqrt{f(\mu)} \) is an even function, see \eqref{f_symmetries}, it holds that \( \int_\gamma \omega + \int_{-\gamma}\omega =0 \) and therefore
\begin{equation}
\label{AJ_infinity}
\aj(\infty^{(0)}) = \frac1{4\ii\KK}\left(\int_{\Gamma(-\overline a,-a)}-\int_{\Gamma(-a,\overline a)} \right)\frac{d\mu}{\sqrt{f(\mu)}_-} = -\frac14 - \frac\tau4,    
\end{equation}
where we used evenness once more. Quite in the same way we also get that
\begin{equation}
\label{AJ_zero}
\aj(0^{(0)}) = \frac1{4\ii\KK}\left(\int_{\Gamma(-\overline a,-a)}-\int_{\Gamma(-a,\overline a)} \right)\frac{d\mu}{\sqrt{f(\mu)}_+} = \frac14 - \frac\tau4.
\end{equation}

\begin{figure}[ht!]
    \centering
    \includegraphics[width=0.6\linewidth]{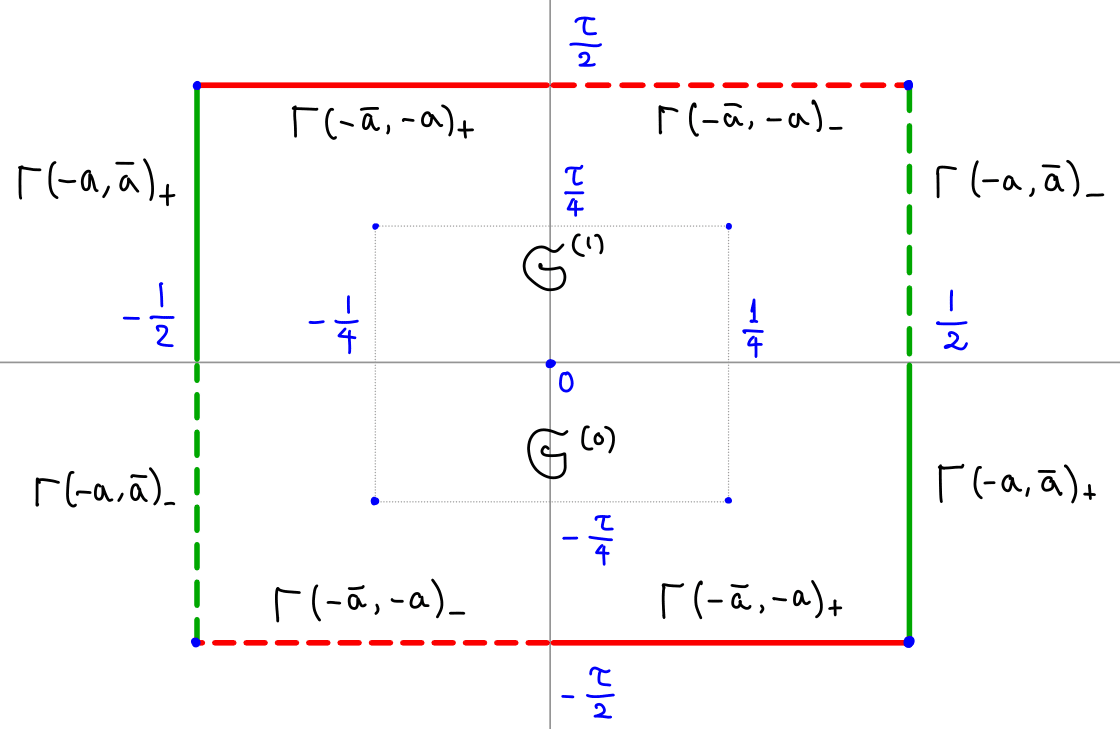}
    \caption{Image of $\S_{\boldsymbol\alpha,\boldsymbol\beta}$ under the map $\aj$. The black labels $\Gamma(-\Bar{a},\bar{a})_{\pm}$ and $\S^{(0,1)}$ denote the domain of $\aj$.}
    \label{image of abel jacobi}
\end{figure}

To close this subsection, we show that
\[
\aj(\S_{\boldsymbol\alpha,\boldsymbol\beta}) = \big\{z~|~2|\Re(z)|<1,~2|\Im(z)|<|\tau|\big\},
\]
see Figure~\ref{image of abel jacobi}. Indeed, for \( z=e^{\ii t} \in \Gamma(a,-\overline a)\), using \eqref{f_trace}, we get that
\[
\aj(z^{(0)}) = \frac1{2\ii \KK} \int_{\alpha}^t \frac{d\theta}{\sqrt{2(\cos(2\alpha)-\cos(2\theta))}}
\]
which is a negative imaginary number and is equal to \( -\tfrac\tau2 \) when \( z=-\overline a \) (this is a similar computation to the one done to derive \eqref{defn of tau}). 
When \( z = e^{\ii t} \in \Gamma(-\overline a,-a)_\pm \), we have that
\[
\aj(z^{(0)}) = -\frac\tau2 \pm \frac1{2\KK} \int_{\pi-\alpha}^t\frac{d\theta}{\sqrt{2(\cos(2\theta)-\cos(2\alpha))}}
\]
again, by \eqref{f_trace}, where the last term above is positive and reaches the value \( \tfrac 12 \) when \( z=-a \) (\(t=\pi+\alpha\)). That is, the images of \( \Gamma(-\overline a,-a)_+ \) and \( \Gamma(-\overline a,-a)_- \) are the bottom horizontal lines (solid red and dashed red, respectively) on Figure~\ref{image of abel jacobi}. Proceeding in the same fashion, one can check that \( \Gamma(-a,\overline a)_\pm \) are mapped into vertical segments \( \pm\tfrac12+[-\tfrac\tau2,0]\) (solid and dashed green lines in the lower half-plane on Figure~\ref{image of abel jacobi}) and that \( \Gamma(\overline a,a)_\pm \) are mapped into horizontal directed segments \( [\pm\tfrac12,0] \). Symmetry \eqref{AJ_symmetry} then finishes the description of \( \aj(\S_{\boldsymbol\alpha,\boldsymbol\beta}) \).


\subsubsection{Riemann Theta Function}

As standard, to describe the solution of RHP~\ref{RHP N}, we shall need the Riemann theta function associated with the period \( \tau \). It is defined by
\begin{align}
    \vartheta(z) := \sum_{k \in \Z} e^{\pi \ii k^2 \tau + 2 \pi \ii k z} = 1 + 2 \sum_{k = 1}^{\infty} e^{\pi \ii k^2 \tau} \cos (2 \pi k z),
    \label{theta func}
\end{align}
and is an entire function of $z \in\C$. It is well-known that \( \vartheta(z) \) is quasi-periodic:
\begin{align}
    &\vartheta(z + n + m \tau) = e^{-\pi \ii m^2 \tau - 2 \pi \ii m z} \vartheta(z), \quad z \in \C, \quad n, m \in \Z, \label{theta quasi periodic}
\end{align}
and that $\vartheta (z) = 0$ if and only if $[z] = \left[ \frac{1 + \tau}{2} \right]$. 

Recall, see \cite{McKMoll}, that in McKean and Moll’s notation \( \vartheta_3(z \,|\, \tau)=\vartheta(z)\) and
\begin{align} \label{other theta functions}
\begin{cases}
    \displaystyle \vartheta_0(z \,|\, \tau) \equiv \vartheta_4(z \,|\, \tau) =\vartheta_3\left( z+\tfrac{1}{2} \,\big|\, \tau \right), \smallskip \\
    \displaystyle \vartheta_1(z \,|\, \tau) = -\ii e^{\frac{\pi \ii}{4}\tau + \pi \ii z } \vartheta_3\left( z + \tfrac{1 + \tau}{2} \,\big|\, \tau \right), \smallskip \\
    \displaystyle \vartheta_2(z \,|\, \tau) =  e^{\frac{\pi \ii}{4}\tau + \pi \ii z } \vartheta_3\left( z + \tfrac{\tau}{2} \,\big|\, \tau \right).
\end{cases}
\end{align}

To solve RHP~\ref{RHP N}, we shall need the following functions
\begin{align} 
\label{defn of Theta_k(z)}
\Theta_k (\z) = e^{\pi \ii \aj(\z)} \frac{\vartheta(\aj(\z) + x V + \tfrac{\tau}{2} - \aj(\infty^{(k)}) - \tfrac{1 + \tau}{2})}{\vartheta(\aj(\z) - \aj(\infty^{(k)}) - \tfrac{1 + \tau}{2})},
\end{align}
\( \z\in\S_{\boldsymbol\alpha,\boldsymbol\beta} \), \( k\in\{0,1\} \). Clearly, \(\Theta_k (\z) \) is a meromorphic function in the domain of its definition with a single and simple pole at \( \infty^{(k)} \) (at any pole \( \z_* \) it must hold \( \mathfrak a(\z_*) = \mathfrak a(\infty^{(k)}) \); as \( \mathfrak a \) is biholomorphic, \( \z_*=\infty^{(k)}\)). It readily follows from \eqref{AJ+ - AJ-} and \eqref{theta quasi periodic} that
\[
\begin{cases}
\Theta_{k+} (\z) = e^{2\pi\ii xV}\Theta_{k-} (\z), & \z\in\boldsymbol\alpha\setminus\{\a\}, \\
\Theta_{k+} (\z) = -\Theta_{k-} (\z), & \z\in\boldsymbol\beta\setminus\{\a\}.
\end{cases}
\]
Let \( \Theta_{k,l}(z):=\Theta_k(z^{(l)})\), \( k,l\in\{0,1\}\), denote the pull-backs of the functions \( \Theta_k (\z) \) to the complex plane. Then, \( \Theta_{0,1}(z) \) and \(\Theta_{1,0}(z)\) are analytic in \( \overline\C\setminus \Gamma(-\overline a,a)\) while \( \Theta_{0,0}(z) \) and \(\Theta_{1,1}(z)\) are analytic in \( \C\setminus \Gamma(-\overline a,a)\), each having a simple pole at infinity. Of course, all functions have continuous traces on both sides of \( \Gamma(-\overline a,a) \). Moreover,
\begin{equation}
\label{theta jumps}
\begin{cases}
\Theta_{k,0\pm}(z) = e^{2\pi\ii xV}\Theta_{k,1\mp}(z), & z\in \Gamma(-\overline a,-a), \\
\Theta_{k,l+}(z) = -\Theta_{k,l-}(z), & z\in \Gamma(-a,\overline a), \\
\Theta_{k,0\pm}(z) = \Theta_{k,1\mp}(z), & z\in \Gamma(\overline a,a).
\end{cases}    
\end{equation}

We already know that \( \Theta_{0,0}(z) \) and \(\Theta_{1,1}(z)\) have simple poles at infinity. On the other hand, we get from \eqref{AJ_symmetry}, \eqref{AJ_infinity}, and \eqref{theta quasi periodic} that
\begin{equation}
\begin{cases}
\label{thetas at infinity}
\Theta_{1,0}(\infty) &= \displaystyle e^{\frac{\pi\ii}4(3\tau-1)+2\pi\ii xV} \frac{\vartheta(xV+\tfrac\tau2)}{\vartheta(0)} = e^{\frac{\pi\ii}4(2\tau-1)+\pi\ii xV} \frac{\vartheta_2(xV)}{\vartheta_3(0)}, \smallskip \\
\Theta_{0,1}(\infty) &= \displaystyle e^{\frac{\pi\ii}4(\tau+1)} \frac{\vartheta(xV+\tfrac\tau2)}{\vartheta(0)} = e^{\frac{\pi\ii}4-\pi\ii xV} \frac{\vartheta_2(xV)}{\vartheta_3(0)},
\end{cases}
\end{equation}
where we used \eqref{other theta functions} on the last steps. It is clear that the above values vanish simultaneously, and this happens if and only if \( xV=n+\tfrac12 \). Recall \eqref{Z delta}. It follows from the second part of \eqref{Legendre's relation} that
\[
\mathcal Z_\delta = \left\{(\varkappa,x)\,|\,|xV-n-\tfrac12|<\delta,\;n\in\Z\right\}.
\]
Later on, we shall make use of the following lemma.

\begin{lem} 
\label{lem: Theta at infinity bounds}
Set \( q := e^{\pi\ii\tau}=e^{-\pi|\tau|} \). Then, for each \( \delta\in(0,\tfrac12) \) and all \( (\varkappa,x)\notin \mathcal Z_\delta\) it holds that
\[
|\Theta_{0,1}(\infty)|, q^{-\frac12}|\Theta_{1,0}(\infty)| \geq \frac{4\delta}\pi \sqrt{\kk \kk^\prime} \KK.
\]
\end{lem}
\begin{proof}
When \( xV \not\in \mathcal Z_\delta \),  \( xV=n+\tfrac12+u\) for some \( n\in\Z \) and \( \delta\leq |u|\leq\tfrac12 \).  We get from \eqref{other theta functions} and \cite[Equations (20.5.1) and (20.5.3)]{NIST:DLMF} that
\[
\frac{\vartheta_2(u+\tfrac12)}{\vartheta_3(0)} = - \frac{\vartheta_1(u)}{\vartheta_3(0)} = -2 q^{\frac14} \sin(\pi u) \prod_{n=1}^\infty \frac{1-2\cos(2\pi u)q^{2n}+q^{4n}}{1+2q^{2n-1}+q ^{4n-2}}.
\]
Since \( u \) is real, it clearly holds that
\[
\prod_{n=1}^\infty \frac{1-2\cos(2\pi u)q^{2n}+q^{4n}}{1+2q^{2n-1}+q ^{4n-2}} \geq \prod_{n=1}^\infty \left( \frac{1-q^{2n}}{1+q^{2n-1}} \right)^2 = \frac{q^{-\frac14}}{\pi} \sqrt{\kk \kk^\prime} \KK,
\]
see \cite[Example~21.10]{WittakerWatson}. Since \( \sin(\pi|u|)\geq 2|u|\) when \( |u|\leq \tfrac12\) and \(|u|\geq\delta\), the claim of the lemma follows.
\end{proof}

Next, we deduce from \eqref{AJ_infinity}, \eqref{AJ_zero}, \eqref{theta quasi periodic}, and \eqref{other theta functions} that
\begin{equation}
\label{Theta of zero}
\begin{cases}
\Theta_{1,0}(0) &= \displaystyle e^{\frac{\pi\ii}4(1+3\tau)+2\pi\ii xV} \frac{\vartheta(xV+\tfrac{1+\tau}2)}{\vartheta(\tfrac12)} = e^{\frac{\pi\ii}4(3+2\tau)+\pi\ii xV} \frac{\vartheta_1(xV)}{\vartheta_0(0)}, \smallskip \\
\Theta_{0,0}(0) &= \displaystyle e^{\frac{\pi\ii}4(1-\tau)} \frac{\vartheta(xV)}{\vartheta(\tfrac\tau2)} = e^{\frac{\pi\ii}4} \frac{\vartheta_3(xV)}{\vartheta_2(0)}.
\end{cases}
\end{equation}

Finally, we estimate \( \Theta_{k,l}(z)\) on \( \Gamma(-\overline a,a) \) (recall that these functions are analytic in \( \C\setminus\Gamma(-\overline a,a) \), two analytic even at infinity and two having a simple pole there).

\begin{lem}
\label{lem:theta bounds}    
As in Lemma~\ref{lem: Theta at infinity bounds}, let \( q = e^{-\pi|\tau|} \). It holds for \( z\in\Gamma(-\overline a,a)_\pm \) that
\begin{align}
\label{01bound}
|\Theta_{0,0}(z)|,|\Theta_{0,1}(z)|,q^{-\frac12}|\Theta_{1,0}(z)|,q^{-\frac12}|\Theta_{1,1}(z)| \leq \sqrt{\frac8{\kk^\prime}}.
\end{align}
\end{lem}
\begin{proof}
As in \eqref{thetas at infinity}, we get from \eqref{defn of Theta_k(z)} combined with \eqref{AJ_infinity} and \eqref{other theta functions} that
\[
\Theta_{0,1}(z) = e^{\pi \ii \aj(z^{(1)})} \frac{\vartheta(u + x V + \tfrac{\tau}{2})}{\vartheta(u)} = e^{\frac{\pi\ii}4-\pi\ii xV} \frac{\vartheta_2(u + x V)}{\vartheta_3(u)},
\]
where \( u = \aj(z^{(1)})  - \tfrac{1 + \tau}{4} \). Thus, we obtain from \cite[Equations~(20.5.2-3)]{NIST:DLMF} that
\[
|\Theta_{0,1}(z)| = 2q^{\frac14}|\cos(\pi u+\pi xV)| \prod_{n=1}^\infty \left|\frac{1+q^{2n}\cos(2\pi u+ 2\pi xV) + q^{4n}}{1+q^{2n-1}\cos(2\pi u)+q^{4n-2}}\right|.
\]
We have explained around Figure~\ref{image of abel jacobi} that when \( z\in\Gamma(-\overline a,a)_\pm \), \(u \) belongs to the boundary of the rectangle with vertices \( -\tfrac34\pm\tfrac\tau 4,\tfrac14\pm\tfrac\tau4\). In particular, it holds that
\begin{align}
\label{bound for cosines}
|\cos(2\pi u)|,|\cos(2\pi u+2\pi xV)| \leq \cosh\left(\frac{\pi|\tau|}2\right) = \frac{q^{\frac12}+q^{-\frac12}}2
\end{align}
and similarly \( |\cos(\pi u+\pi xV)|\leq q^{-\frac14} \). Therefore, we get for \( z\in\Gamma(-\overline a,a)_\pm \) that
\begin{align*}
|\Theta_{0,1}(z)| & \leq 2\prod_{n=1}^\infty \frac{(1+q^{2n-\frac12})(1+q^{2n+\frac12})}{(1-q^{2n-1-\frac12})(1-q^{2n-1+\frac12})} \\
&= \frac{2}{1+q^{\frac12}} \prod_{n=1}^\infty \frac{1+q^{\frac12(2n-1)}}{1-q^{\frac12(2n-1)}} = \frac{2}{1+q^{\frac12}} \left(\frac{\vartheta_3(0|\tfrac\tau2)}{\vartheta_0(0|\tfrac\tau2)} \right)^{\frac12},
\end{align*}
where we used \cite[Equations~(20.5.3-4)]{NIST:DLMF} on the last step. Formulae \cite[Equations~(20.7.5),~(20.7.12), (20.7.14), and (20.9.1)]{NIST:DLMF} now give
\begin{equation}
\label{connect moduli}
\frac{\vartheta_3(0|\tfrac\tau2)}{\vartheta_0(0|\tfrac\tau2)} = \left( \frac{\vartheta_3^4(0|\tfrac\tau2)}{\vartheta_3^4(0|\tfrac\tau2)-\vartheta_2^4(0|\tfrac\tau2)} \right)^{\frac14} = \left( \frac{\vartheta_3^2(0)+\vartheta_2^2(0)}{\vartheta_3^2(0)-\vartheta_2^2(0)} \right)^{\frac12} = \frac{1+\kk}{\kk^\prime}
\end{equation}
(recall that \( \kk = \sin\alpha, \kk^\prime = \cos\alpha\) are the elliptic moduli corresponding to the nome \( q \)). Clearly, the last two displays prove \eqref{01bound}  for \( \Theta_{0,1}(z) \).

Next, the top and bottom relations in \eqref{theta jumps} immediately tell us that \( |\Theta_{0,0}(z)| \) satisfies \eqref{01bound} on \( \Delta_\pm \). Thus, we only need to estimate it on \( \Gamma(-a,\overline a) \). Due to \( 1 \)-periodicity, it is enough to consider \( z\) for which \( \aj(z^{(0)}) = \tfrac12 + (y-\tfrac14)\tau\), \( |y|\leq \tfrac14 \) (the lower-half plane part of the right vertical line on Figure~\ref{image of abel jacobi}). Notice that
\[
\Theta_{0,0}(z) = e^{\pi \ii \aj(z^{(0)})} \frac{\vartheta(u + x V)}{\vartheta(u-\tfrac\tau2)} = e^{\pi \ii (\aj(z^{(0)})-2u)} \frac{\vartheta(u + x V)}{\vartheta(u+\tfrac\tau2)} = e^{\frac{\pi \ii}4} \frac{\vartheta_3(u + x V)}{\vartheta_2(u)},
\]
where \( u = \aj(z^{(0)})  - \tfrac{1 + \tau}{4} + \tfrac\tau2 =  \tfrac14 + y\tau \) and we used \eqref{theta quasi periodic}, \eqref{other theta functions}. Thus,
\[
|\Theta_{0,0}(z)| = \frac{q^{-\frac14}}{2|\cos(\pi u)|}\prod_{n=1}^\infty \left|\frac{1+q^{2n-1}\cos(2\pi u+2\pi xV)+q^{4n-2}}{1+q^{2n}\cos(2\pi u) + q^{4n}}\right|.
\]
Observe that we still can use \eqref{bound for cosines}. On the other hand, since \( \Re(u)=\tfrac14 \), we get that
\[
|\cos(\pi u)|^2 = \frac12 \cosh(2\pi|y\tau|) \geq \frac14 e^{2\pi|y\tau|} \geq \frac14 q^{-\frac12}.
\]
Hence, it holds that
\[
|\Theta_{0,0}(z)| \leq \prod_{n=1}^\infty\frac{(1+q^{2n-1-\frac12})(1+q^{2n-1+\frac12})}{(1-q^{2n-\frac12})(1-q^{2n+\frac12})} = (1-q^{\frac12})\prod_{n=1}^\infty \frac{1+q^{\frac12(2n-1)}}{1-q^{\frac12(2n-1)}},
\]
where now we can utilize \eqref{connect moduli} to deduce \eqref{01bound}.

Let \( u = \aj(z^{(0)})+\tfrac{1+\tau}4\). Then, we get from \eqref{defn of Theta_k(z)}, \eqref{theta quasi periodic}, and \eqref{other theta functions} that
\begin{align*}
\Theta_{1,0}(z) &= e^{\pi \ii \aj(z^{(0)})} \frac{\vartheta(u + x V + \tfrac{\tau}{2}-1-\tau)}{\vartheta(u-1-\tau)} = e^{\pi \ii (\aj(z^{(0)})+\tau+2xV)} \frac{\vartheta(u + x V + \tfrac{\tau}{2})}{\vartheta(u)} \\
& = q^{\frac12} e^{-\frac{\pi \ii}4+\pi\ii xV} \frac{\vartheta_2(u + x V)}{\vartheta_3(u)}.
\end{align*}
Since \( u \) lies on the boundary of the rectangle with vertices \( -\tfrac14\pm\tfrac\tau4,\tfrac34\pm\tfrac\tau4 \), which belongs to the same horizontal strip as the rectangle in the first part of the proof, the estimate of \( q^{-\frac12}|\Theta_{1,0}(z)| \) follows exactly as the one of \( |\Theta_{0,1}(z)| \).

The validity of \eqref{01bound} for \( q^{-\frac12}|\Theta_{1,1}(z)| \) on \( \Delta_\pm \) is a consequence of the top and bottom relations in \eqref{theta jumps}. On \( \Gamma(-a,\overline a)_\pm \), due to \( 1 \)-periodicity, it is enough to consider \(z\) for which \( \aj(z^{(1)}) = \tfrac12 + (y+\tfrac14)\tau\), \( |y|\leq \tfrac14 \). Then,
\[
\Theta_{1,1}(z) = e^{\pi \ii (\aj(z^{(1)})-2u)} \frac{\vartheta(u + x V)}{\vartheta(u+\tfrac\tau2)} = q^{\frac12} e^{\frac{3\pi\ii}{4}} \frac{\vartheta_3(u + x V)}{\vartheta_2(u)}
\]
by \eqref{theta quasi periodic} and \eqref{other theta functions}, where \( u = \aj(z^{(1)})  - \tfrac{3(1 + \tau)}{4} + \tfrac\tau2 =  -\tfrac14 + y\tau \). We now can proceed exactly as in the case of \( \Theta_{0,0}(z) \).
\end{proof}


\subsubsection{Solution of ${N}$-RHP}

Let \( \gamma(\lambda) \) be a function holomorphic in $\overline{\C} \setminus \Delta$ and given by
\begin{align}
\label{defn of gamma}
    \gamma(\lambda) := \left( \frac{(\lambda - a)(\lambda + a)}{(z - \overline{a})(\lambda + \overline{a})} \right)^{\frac{1}{4}},
\end{align}
where the branch is fixed so that $\gamma (\infty) = 1$. Using any continuous determination of the argument of \( \gamma^4(\lambda) \) on the imaginary axis and the normalization $\gamma (\infty) = 1$, one can check that \( \gamma(0)=a\). We further can see that \( \gamma_+(\lambda)=\ii\gamma_-(\lambda)\), \( \lambda\in \Delta\).

Next, define holomorphic functions
\begin{align} \label{defn of A(z) and B(z)}
    A(\lambda) := \frac{\gamma(\lambda) + \gamma^{-1}(\lambda)}{2} \quad \text{and} \quad B(\lambda) := -\frac{\gamma(\lambda) - \gamma^{-1}(\lambda)}{2\ii}
\end{align}
for $\lambda\in \C \setminus \Delta$.  The functions $A(\lambda)$ and $B(\lambda)$ satisfy
\begin{itemize}
    \item \( A(\lambda)B(\lambda) \neq 0 \) for all $\lambda\in \C \setminus \Delta$;
    \item $A(\infty) = 1$ and $B(\infty) = 0$, where $\infty$ is the zero of order $2$;
    \item $A(0)=\cos\alpha=\kk^\prime$ and $B(0)=-\sin\alpha=-\kk$;
    \item $A_+(\lambda) = B_-(\lambda)$ and $A_-(\lambda) = -B_+(\lambda)$ for $\lambda \in \Delta$.
\end{itemize}
Set
\begin{equation}
\label{defn of M}
M(\lambda) := \begin{pmatrix} \Theta_{1,0}(\lambda) A(\lambda) & \Theta_{1,1}(\lambda)B(\lambda) \smallskip \\ -\Theta_{0,0}(\lambda) B(\lambda) & \Theta_{0,1}(\lambda) A(\lambda) \end{pmatrix}, \quad \lambda\in\overline\C\setminus\Gamma(-\overline{a}, a).
\end{equation}
Clearly, this is a holomorphic matrix function in the domain of its definition (the simple poles of \( \Theta_{0,0}(\lambda),\Theta_{1,1}(\lambda)\) at infinity are canceled by the double zero of \( B(\lambda)\) there). Hence, it satisfies RHP~\ref{RHP N}(1). The fact that it satisfies  RHP~\ref{RHP N}(2) can be readily checked using \eqref{theta jumps}. We further notice that
\begin{align}
    M(\infty) = \begin{pmatrix}
    \Theta_{1,0}(\infty) & 0\\
    0 & \Theta_{0,1}(\infty)
    \end{pmatrix},
    \label{defn of M-hat(infty)}
\end{align}
which has non-zero diagonal entries when \( (\varkappa,x)\not\in\mathcal Z_\delta\) for any \( \delta>0 \), see Lemma~\ref{lem: Theta at infinity bounds}. Thus, for each pair \( (\varkappa,x)\not\in\mathcal Z_\delta\), RHP~\ref{RHP N} is solved by
\begin{align}
    N(\lambda) := [M(\infty)]^{-1} M(\lambda), \quad \lambda \in \overline\C \setminus \Gamma(-\overline{a}, a).
    \label{construction of the soln of the N-hat problem}
\end{align}
The following lemma holds.

\begin{lem}
\label{lem:estimate of N}    
For any \( \epsilon\in(0,1) \) such that \( \{|\lambda-a|=\epsilon^4\} \) belongs to the first quadrant and for all \( (\varkappa,x)\not\in \mathcal Z_\delta\), it holds that
\[
|N_{i,j}(\lambda)| \leq \frac{6\pi}{\delta\epsilon} \frac{(\kk \kk^\prime)^{\frac14}}{\mathcal K_\alpha^*}, \quad i,j\in\{1,2\},
\]
when \( \dist(\lambda,\{\pm a,\pm\overline a\})\geq \epsilon^4 \), where \( N_{i,j}(\lambda) \) is the \( (i,j)\)-th entry of \( N(\lambda) \) and \( \mathcal K_\alpha^* \) was introduced in \eqref{Kalpha}. 
\end{lem}
\begin{proof}
We start with preliminary estimates of \( \gamma(\lambda) \). Assume that \( \alpha\in(0,\tfrac\pi4]\). When \( |\lambda-a|=\epsilon^4\), we get that
\[
\left| \frac{\lambda^2-a^2}{\lambda^2-\overline a^2} \right| \geq \frac{2-|\lambda-a|}{|a+\overline a|+|\lambda-a|} \left| \frac{\lambda-a}{\lambda-\overline a} \right| \geq \frac13\left| \frac{\lambda-a}{\lambda-\overline a} \right| \geq \frac{\epsilon^4}{24 \kk},
\]
where the last inequality follows from the Koebe Distortion Theorem (also known as Koebe \(\tfrac14\)-Theorem), see \cite[Theorem~1.3]{Pommerenke}, and the trivial observation that \( (\lambda-a)/(\lambda-\overline a) \) is conformal in \( \C\setminus\{\overline a\} \) and its derivative at \( a \) is equal to \( (a-\overline a)^{-1}\). An absolutely analogous estimate holds on the circle \( |\lambda+a|=\epsilon^4\). Thus, we get from the maximum modulus principle that
\[
\epsilon |\gamma(\lambda)|^{-1} \leq (24\kk)^{\frac14} \leq (24\sqrt 2 \kk \kk^\prime)^{\frac14} < 2\sqrt 2 (\kk \kk^\prime)^{\frac14}
\]
on \( \dist(\lambda,\{\pm a\})\geq \epsilon^4\), where we used the restriction \( \sqrt 2\kk^\prime = \sqrt 2\cos\alpha\geq 1 \). When \( \alpha\in[\tfrac\pi4,\tfrac\pi2)\) and \( |\lambda-a|=\epsilon^4\), we similarly have that
\[
\left| \frac{\lambda^2-a^2}{\lambda^2-\overline a^2} \right| \geq \frac{2-|\lambda-a|}{|a-\overline a|+|\lambda-a|} \left| \frac{\lambda-a}{\lambda+\overline a} \right| \geq \frac13\left| \frac{\lambda-a}{\lambda+\overline a} \right| \geq \frac{\epsilon^4}{24 \kk^\prime}.
\]
As this estimate is also valid on \( |\lambda+a|=\epsilon^4\), we see that
\[
\epsilon |\gamma(\lambda)|^{-1} \leq 2\sqrt 2 (\kk \kk^\prime)^{\frac14} 
\]
for \( \dist(\lambda,\{\pm a\})\geq \epsilon^4\) in this case. Clearly, we can deduce exactly the same estimate of \( \epsilon |\gamma(\lambda)| \) on the set \( \dist(\lambda,\{\pm \overline a\})\geq \epsilon^4\). Hence, it holds that
\[
\epsilon|A(\lambda)| \leq 2\sqrt2 (\kk \kk^\prime)^{\frac14}, \quad  \dist(\lambda,\{\pm a,\pm\overline a\})\geq \epsilon^4.
\]

Next, we factor \( A(\lambda) \) out of \( M(\lambda) \) and are left with estimating the functions
\[
\frac{\Theta_{1,0}(\lambda)}{\Theta_{1,0}(\infty)}, \;\; \frac{\Theta_{0,1}(\lambda)}{\Theta_{0,1}(\infty)}, \;\; \frac{B(\lambda)}{A(\lambda)}\frac{\Theta_{1,1}(\lambda)}{\Theta_{1,0}(\infty)}, \;\; \text{and} \;\; \frac{B(\lambda)}{A(\lambda)}\frac{\Theta_{0,0}(\lambda)}{\Theta_{0,1}(\infty)}.
\]
Notice that all these functions are analytic in \( \overline \C\setminus \Gamma(-\overline a,a)\). Hence, we can estimate them using the maximum modulus principle. It readily follows from Lemmas~\ref{lem: Theta at infinity bounds} and~\ref{lem:theta bounds} that these functions are bounded above in modulus by
\[
\max\left\{1,\max_\Delta\left|\frac BA\right|\right\} \frac1\delta\frac{\pi}{\KK} \frac1{\sqrt{2\kk} \kk^\prime}
\]
when \( (\varkappa,x)\not\in \mathcal Z_\delta\). Thus, we only need to estimate \( (B/A)(\lambda) \) on \( \Delta \). It holds that
\[
\frac{B(\lambda)}{A(\lambda)} = \ii \frac{\gamma^4(\lambda)-2\gamma^2(\lambda)+1}{\gamma^4(\lambda)-1} = -\frac{\lambda^2-\cos(2\alpha)-2\sqrt{f(\lambda)}}{\sin(2\alpha)}.
\]
When \( \lambda \in \Delta \), we have that \( \lambda^2=\cos(2t)+\ii \sin(2t)\) with \( |t|\leq \alpha\) and therefore 
\[
|\lambda^2-\cos(2\alpha)| = \sqrt{1-2\cos(2t)\cos(2\alpha)+\cos^2(2\alpha)} \leq \sqrt{1-\cos^2(2\alpha)} = \sin(2\alpha).
\]
Moreover, we get from \eqref{f_trace} that
\[
|\sqrt {f(\lambda)}| = \sqrt{2(\cos(2t)-\cos(2\alpha))} \leq \sqrt{2(1-\cos(2\alpha))} = 2\sin\alpha.
\]
Hence, the triangle inequality implies that
\[
\left|\frac{B(\lambda)}{A(\lambda)}\right| \leq \frac{\sin(2\alpha)+2\sin\alpha}{\sin(2\alpha)} \leq \frac3{\cos\alpha} = \frac3{\kk^\prime},
\]
which finishes the proof of the lemma.
\end{proof}

\subsection{Local Parametrices}

Although we have observed that all the jump matrices in RHP~\ref{RHP tilde X} converge to the ones in RHP~\ref{RHP N} as $x \to \infty$, the convergence is not uniform around \( \pm a,\pm\overline a \). To resolve this issue, one needs to construct local parametrices in some neighborhoods of these points that approximately solve RHP~\ref{RHP tilde X} locally within the neighborhoods and match the global parametrix on the boundary.

\subsubsection{Airy Parametrix}

The construction relies on the now classical Airy model parametrix. It was introduced (in a slightly different form) in \cite{DKMVZ2, DKMVZ1} and has been used ever since in countless publications. Up to a trivial prefactor and slightly different choice of the rays, we follow \cite{BDIK1}. Let
\begin{align}
\label{airy 1}
\tilde\Phi_{\Ai}(\zeta) := \sqrt{2 \pi} e^{\ii\frac{\pi}{6} } \begin{pmatrix}
        \Ai(\zeta) & \Ai(e^{-\ii\frac{2\pi}{3}} \zeta) \\
        -\ii\Ai'(\zeta) & -\ii\,e^{-\ii\frac{2\pi}{3}} \Ai'(e^{-\ii\frac{2\pi}{3}} \zeta) 
    \end{pmatrix} e^{- \ii\frac{\pi}{6} \sigma_3} \sigma_1,
\end{align}
which is an entire matrix function. It holds that
\[
\tilde\Phi_{\Ai}(\zeta) = \sqrt 2\begin{pmatrix} 1 & 0 \\ 0 & \ii \end{pmatrix}A_0(\zeta),
\]
where \( A_0(\zeta)\) is the matrix from \cite[Equation~(3.29)]{BDIK1}. Set
\[
\Phi_{\Ai}(\zeta) := \tilde\Phi_{\Ai}(\zeta) \begin{dcases}
I, & \arg\zeta\in\left(0,\tfrac{3\pi}4\right), \smallskip \\
\begin{pmatrix}
    1 & -1 \\ 0 & 1
\end{pmatrix}, & \arg\zeta\in\left(\tfrac{3\pi}4,\tfrac{5\pi}4\right), \smallskip \\
\begin{pmatrix}
    1 & -1 \\ 0 & 1
\end{pmatrix} 
\begin{pmatrix}
    1 & 0 \\ 1 & 1
\end{pmatrix}
, & \arg\zeta\in\left(\tfrac{5\pi}4,2\pi\right).
\end{dcases}
\]
It is a sectionally holomorphic matrix function whose traces satisfy
\begin{equation}
\label{Airy jumps}
    \Phi_{\Ai +}(\zeta) = \Phi_{\Ai -}(\zeta)
    \begin{dcases}
    \begin{pmatrix}
        1 & 1 \\
        0 & 1
    \end{pmatrix}, & \arg \zeta = \tfrac{3\pi}{4},\\
    \begin{pmatrix}
        1 & 0 \\
        -1 & 1
    \end{pmatrix}, & \arg \zeta = \tfrac{5\pi}{4},\\
    \begin{pmatrix}
        1 & 1 \\
        -1 & 0
    \end{pmatrix}, & \arg \zeta = 0,
    \end{dcases}
\end{equation}
where the orientation of the rays is depicted on Figure \ref{Airy figure} (in \cite{BDIK1} the chosen angles are \( \tfrac{2\pi}3\) and \( \tfrac{4\pi}3\)). 
\begin{figure}[ht!]
    \centering
    \includegraphics[width=0.4\linewidth]{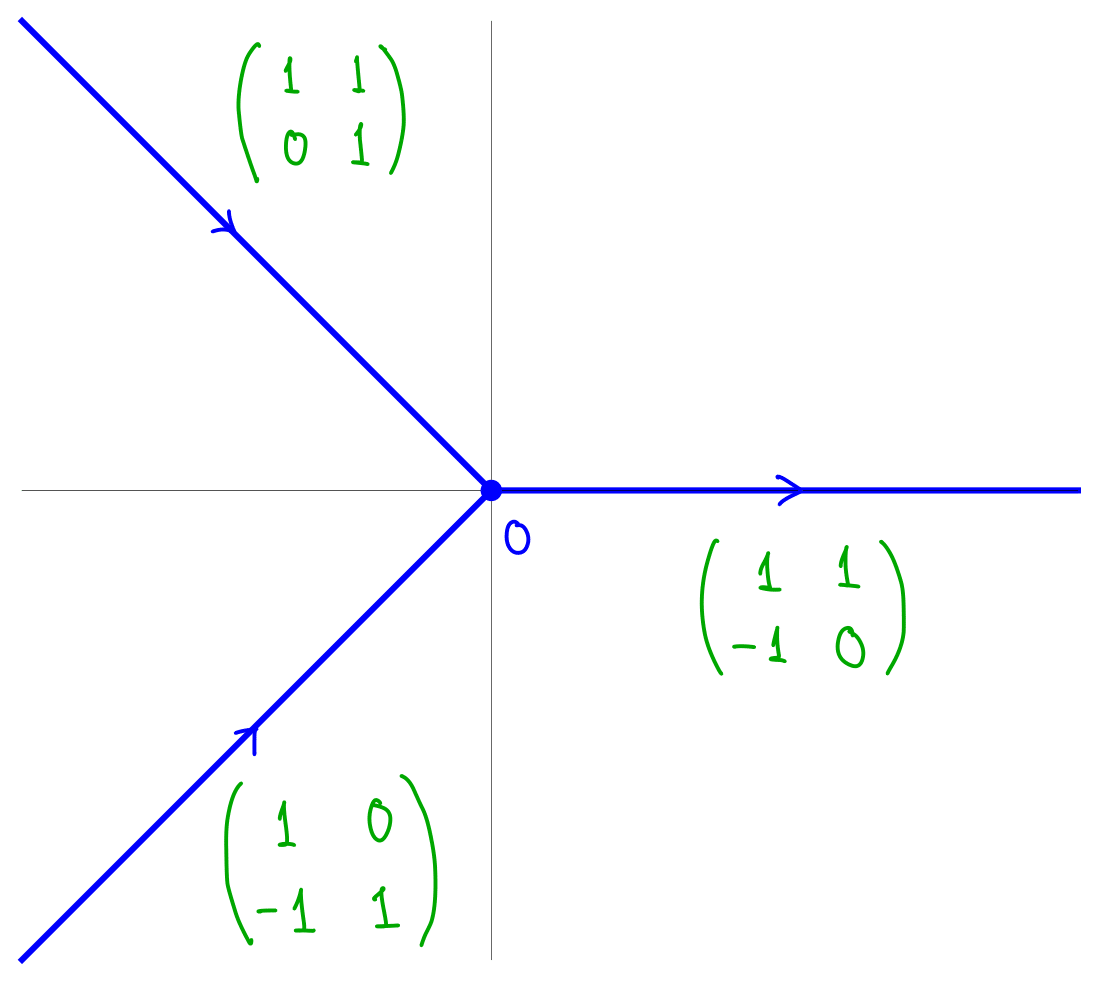}
    \caption{The jump matrices and the jump contour for $\Phi_{\Ai}(\zeta)$.}
    \label{Airy figure}
\end{figure}
Moreover, it holds that
\begin{align}
    \Phi_{\Ai}(\zeta) = \frac{\zeta^{-\frac14\sigma_3}}{\sqrt{2}} \left( \begin{pmatrix}
        1 & \ii \\ \ii & 1 \end{pmatrix} + \O\left( \zeta^{-\frac32} \right) \right) \sigma_1 \, e^{\frac{2}{3}\zeta^{\frac32} \sigma_3}
    \label{airy asymp}
\end{align}
uniformly as $\zeta \to \infty$, where \( \arg(\zeta^\gamma)= \gamma\arg\zeta \) for \( \arg\zeta\in(0,2\pi) \) (unlike similar asymptotic formula stated in \cite{DKMVZ2, DKMVZ1} where the roots are principal).

\subsubsection{Conformal Map}
\label{sss:conf map}

To use Airy parametrix from the previous subsection, we need to introduce a conformal map of a neighborhood of \( a \) (resp. \( \overline a,-a,-\overline a\)) to the \( \zeta\)-plane where \( \Phi_{\Ai}(\zeta)\) is defined. In this section, we provide a detailed construction around \( a \), understanding that the other points can be handled similarly.

Let \( \mathcal Q_1:=\{\lambda~|~\Re(\lambda),\Im(\lambda)>0\} \) be the first quadrant. For \( \lambda \in \mathcal Q_1 \setminus \Gamma(1,a) \), define
\begin{align}
\label{defn of h1}
    h_1 (\lambda) := \frac{\ii}{4} \int_a^{\lambda} \frac{\sqrt{f(\mu)}}{\mu^2} d \mu,
\end{align}
which is an analytic function in the domain of its definition. It readily follows from \eqref{f_trace} and \eqref{f_jump} that
\begin{align}
h_{1-}(\lambda) = -h_{1+}(\lambda)>0, \quad \lambda \in \Gamma(1,a). 
    \label{h_1+ and h_1-}
\end{align}
Moreover,  we obtain from \eqref{defn of g}, \eqref{PiOmega}, \eqref{Pi}, and \eqref{Omega} that
\begin{equation}
\label{g and ha}
2h_1(\lambda) = 2g(\lambda) + \begin{cases}[r]
-\varkappa, &  |\lambda|<1, \\
\varkappa, & |\lambda|>1,
\end{cases} \quad \lambda\in\mathcal Q_1.
\end{equation}

Consider \(\lambda\in\mathcal Q_1\cap \{|\lambda|<1\} \) (shaded region on the left panel of Figure~\ref{fig:step 1}). It follows from \eqref{g and ha} and the discussion after \eqref{Omega} that \( \tfrac32 xh_1(\lambda)\) bijectively maps this region into the shaded region depicted on the right panel of Figure~\ref{fig:step 1}:
\begin{figure}[ht!]
    \centering
    \includegraphics[width=0.4\linewidth]{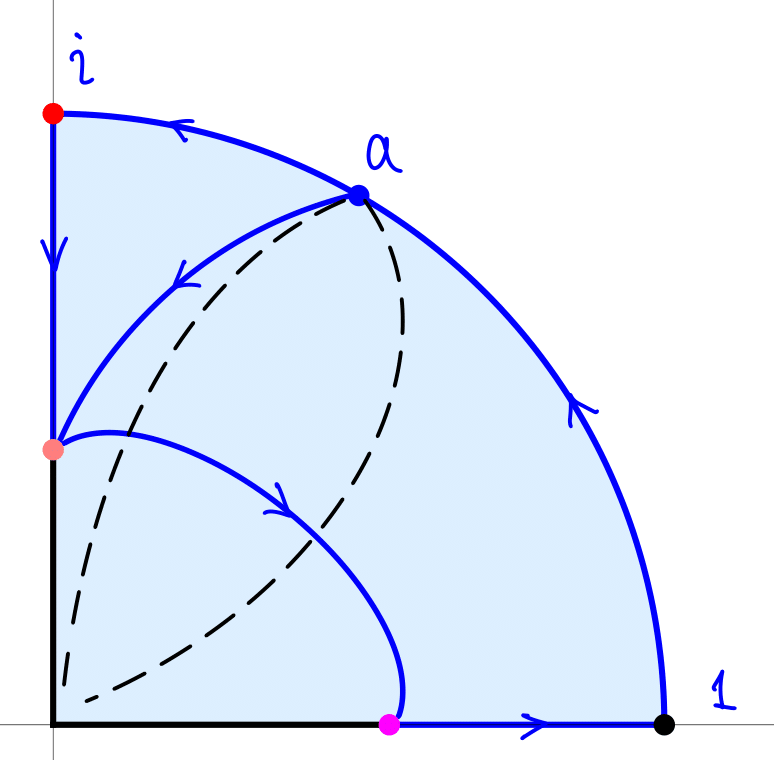}\qquad
    \includegraphics[width=0.4\linewidth]{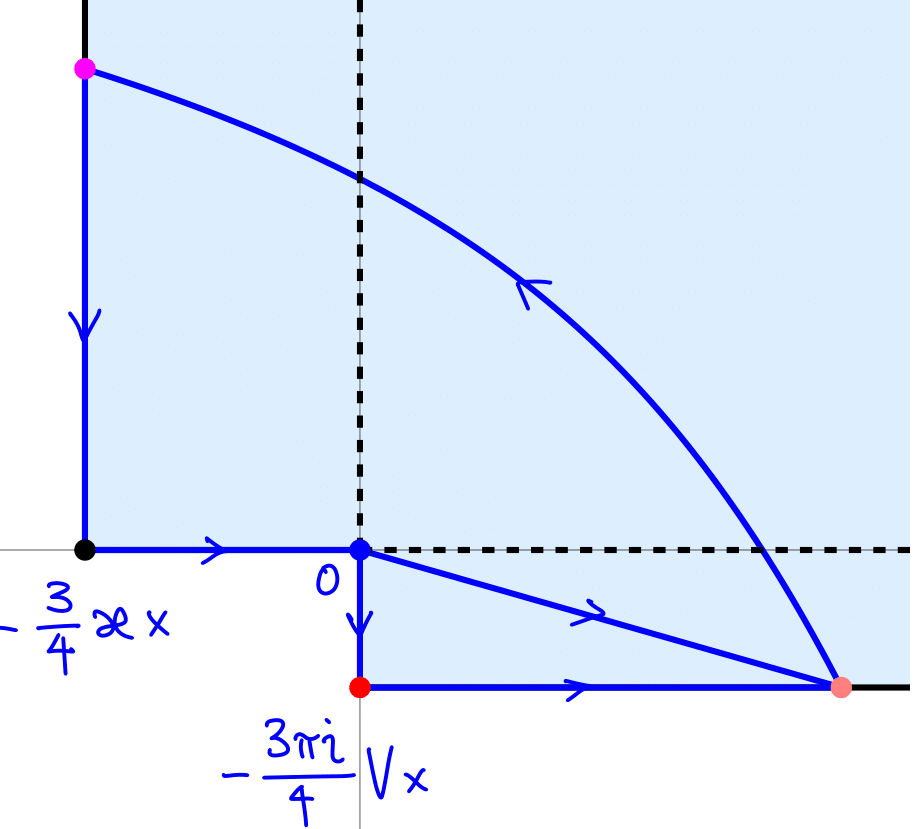}
    \caption{The interior of the unit disk in the first quadrant (left panel) and its image under \( \tfrac32 x h_1(\lambda) \) (right panel).}
    \label{fig:step 1}
\end{figure}
the arc \( \Gamma(1,a) \) is mapped into the horizontal segment \( (-\tfrac34 \varkappa x,0) \); the arc \( \Gamma(a,\ii)\) is mapped into the vertical segment \( (0,-\tfrac34 \pi\ii Vx)\); the segment \( \ii(1,0) \) is mapped into the horizontal half line \( -\tfrac34 \pi\ii Vx +(0,\infty)\); and the segment \( (0,1) \) is mapped into vertical half line \( -\tfrac34 \varkappa x + \ii (\infty,0) \). 

Next, select the cubic root of \( h_1(\lambda)\) according to the rule 
\begin{equation}
\label{arg h1}
\arg\big(h_1^{\frac13}(\lambda)\big) = \frac13 \arg(h_1(\lambda)) + \frac{2\pi}3, \quad \arg(h_1(\lambda))\in\left(-\frac\pi2,\pi\right).    
\end{equation}
In this case the map \( (\tfrac32 xh_1(\lambda))^{\frac13} \) takes \( \mathcal Q_1\cap \{|\lambda|<1\} \) bijectively onto the shaded region depicted on the left panel of Figure~\ref{fig:step 2} (the dashed lines on this figure are rays with angles \( \tfrac{2\pi}3 \) and \( \tfrac{5\pi}6 \)).
\begin{figure}[ht!]
    \centering
    \includegraphics[width=0.4\linewidth]{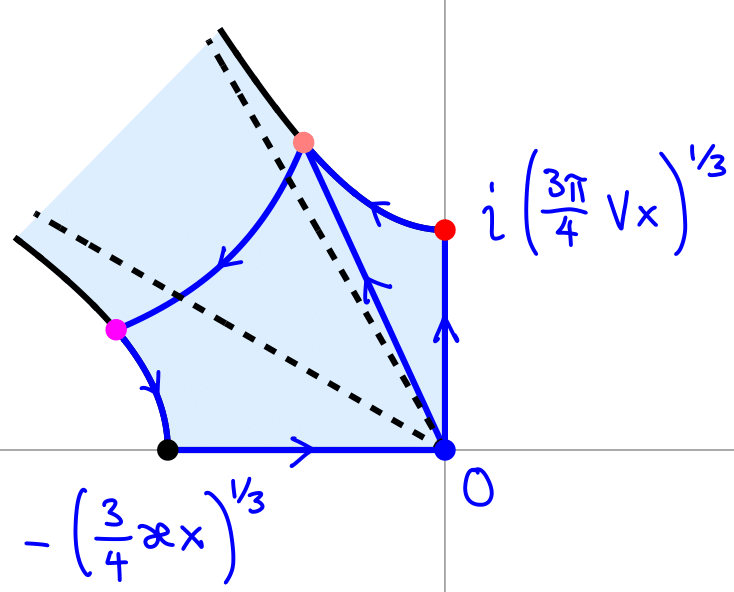}\qquad
    \includegraphics[width=0.4\linewidth]{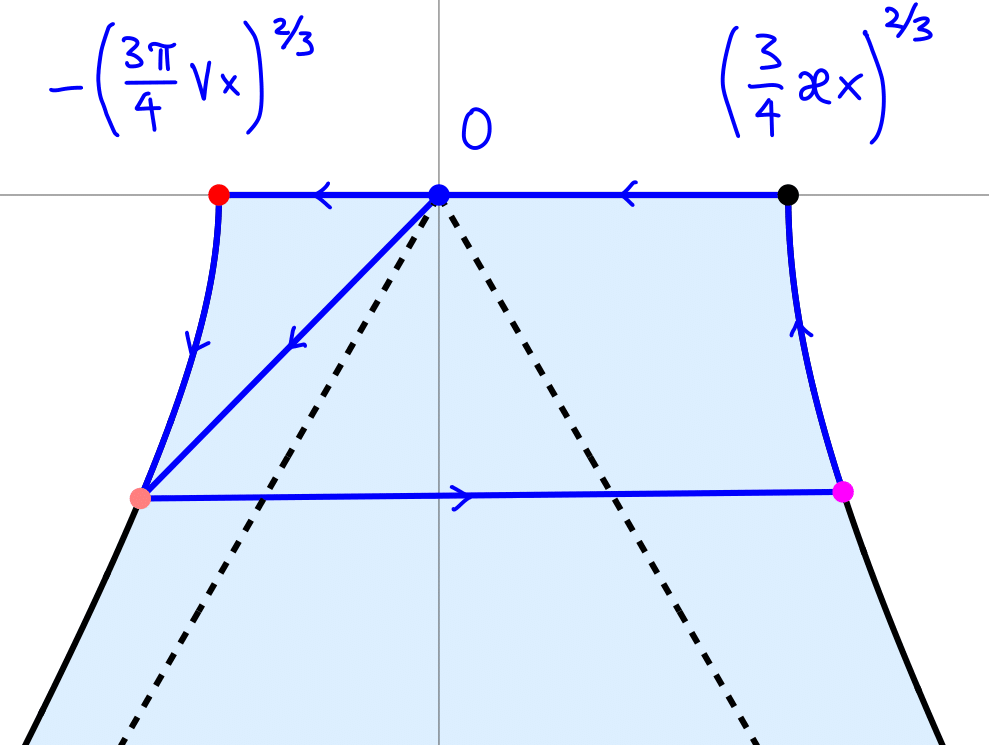}
    \caption{Images of \( \mathcal Q_1\cap \{|\lambda|<1\} \) under the maps \( (\frac32xh_1(\lambda))^{\frac13} \) (left panel) and \( \big((\tfrac32 xh_1(\lambda))^{\frac13}\big)^2 \) (right panel).}
    \label{fig:step 2}
\end{figure}
Now, we define
\begin{equation}
\label{defn of zeta 1}    
\zeta_1(\lambda) : = \big((\tfrac32 xh_1(\lambda))^{\frac13}\big)^2,    
\end{equation}
which maps \( \mathcal Q_1\cap \{|\lambda|<1\} \) conformally into the shaded region depicted on the right panel of Figure~\ref{fig:step 2} (the dashed lines are rays with angles \( \tfrac{4\pi}3 \) and \( \tfrac{5\pi}3 \)). 

Performing the same analysis in \( \mathcal Q_1 \cap \{|\lambda|>1\} \) and/or with the help of symmetry 
\[
g(\lambda)= -\overline{g(1/\overline \lambda)}
\]
see \eqref{g_symmetry}, we see that \( \zeta_1(\lambda) \) in \eqref{defn of zeta 1} extends to a conformal map of the entire first quadrant (in \eqref{arg h1} the argument of \( h_1(\lambda)\) needs to be selected in \((-2\pi,-\tfrac\pi2)\) so that it is continuous across \( \Gamma(a,\ii)\)). The set \( \zeta_1(\mathcal Q_1) \) is a strip domain consisting of the shaded region depicted on the right panel of Figure~\ref{fig:step 2}, its reflection across the real axis, and the open real interval bordering both regions. 

We now select the arc \( \gamma_1^{in},\gamma_1^{out} \) in the following way:
\[
\zeta_1(\gamma_1^{out}) = \{\arg\zeta=\tfrac{3\pi}4\} \cap \zeta_1(\mathcal Q_1) \quad \text{and} \quad \zeta_1(\gamma_1^{in}) = \{\arg\zeta=\tfrac{5\pi}4\} \cap \zeta_1(\mathcal Q_1)
\]
(notice that this definition is independent of \( x \) and that these arcs are reflections of each other across the unit circle). Further, let \( \zeta_*\neq 0 \) be the point where the ray \( \{\arg\zeta=\tfrac{3\pi}4\} \) intersects \( \zeta_1(\mathcal Q_1) \) (orange point on the right panel of Figure~\ref{fig:step 2} is \( \overline \zeta_*\)). We select domain \( \mathcal D_1 \) so that
\[
\zeta_1(\mathcal D_1) = \zeta_1(\mathcal Q_1) \cap \{-\Im(\zeta_*)<\Im(\zeta)<\Im(\zeta_*)\}.
\]
Again, this domain is independent of \( x \) and is symmetric with respect to the reflection across \( S^1 \). The boundary of \( \mathcal D_1\cap\{|\lambda|<1\} \) is depicted on the left panel of Figure~\ref{fig:step 1} (curvilinear rectangle with vertices \( 1, \ii \), orange dot on the vertical axis, and pink dot on the horizontal axis). The right panel of Figure~\ref{fig:step 1} as well as the left panel of Figure~\ref{fig:step 2} show how it transforms under the corresponding maps.

\begin{lem}
\label{lem:koebe}
It holds that
\[
|\zeta_1(\lambda)|^{\frac32} \geq \frac{3x}8 (\kk \kk^\prime)^2  \quad \text{and} \quad |\lambda| \geq d, \quad \lambda\in \partial \mathcal D_1,
\]
for some \( d>0 \) independent of \( \varkappa \) and \( x \).
\end{lem}
\begin{proof}
Consider vertical rays
\[
\ell_1 := \left\{-(\tfrac{3\pi}4 Vx)^{\frac23}(1+\ii y)~|~y>0\right\} \quad \text{and} \quad  
\ell_2:=\left\{(\tfrac34 \varkappa x)^{\frac23}(1-\ii y)~|~y>0 \right\}.
\]
Let \( \ell_1^* \) and \( \ell_2^* \) be their corresponding preimages under the map \( z\mapsto z^{2/3} \), where the cubic root is the same as in \eqref{arg h1}. It is an exercise in elementary complex analysis to show that
\[
\Im(z) > -\tfrac{3\pi}{4} Vx, \;\; z\in\ell_1^*, \quad \text{and} \quad \Re(z) > -\tfrac34 \varkappa x, \;\; z\in\ell_2^*,
\]
from which we can conclude that the curves \( \ell_1^* \) and \( \ell_2^* \) belong to the shaded region on the right panel of Figure~\ref{fig:step 1}. Respectively, \( \ell_1,\ell_2 \subset \zeta_1(\mathcal Q_1)\). It now readily follows, see Figure~\ref{fig:box}, that
\[
|\zeta_1(\lambda)|^{\frac32} \geq \frac34x \min\{\pi V,\varkappa\}, \quad \lambda \in \partial \mathcal{D}_1.
\]
Since \( \varkappa,\pi V\in(0,1)\), the first claim of the lemma follows from \eqref{kappa V estimate}. 

To prove the second claim, we first explain why there exists \( d_*>0\), independent of \( \varkappa \) and \( x \), for which
\begin{equation}
\label{zeta 1 upper bound}
|\zeta_1(\lambda)|^{\frac32} \leq xd_*, \quad \lambda \in \partial \mathcal{D}_1.
\end{equation}
Indeed, since the domains on the right panels of Figures~\ref{fig:step 1} and~\ref{fig:step 2} are connected by an appropriate branch of the map \( z\mapsto z^{\frac23} \), one can see from the both figures that the furthest points from the origin on \( \partial \zeta_1(\mathcal D_1)\) are the points lying on the horizontal line \( \{\zeta|\Im(\zeta)=-\Im(\zeta_*)\} \) (orange and pink dots on Figures~\ref{fig:step 2} and~\ref{fig:box}), one of them (the orange point) being \( \overline\zeta_* \). It can be readily seen from the right panel of Figure~\ref{fig:step 1} that
\[
|\zeta_*|^{\frac32} = \frac{3\pi Vx}4\left|\cot\left(\frac\pi 8\right) +\ii\right|.
\]
Moreover, as in the first part of the proof, a tedious but straightforward computation shows that the ray \( \{(3\varkappa x/4)^{\frac23} + re^{\frac53\pi\ii}|r>0\}\) lies outside of \( \zeta_1(\mathcal D_1)\) (this is the dash-dot ray on Figure~\ref{fig:box}). Hence, the second point of interest is no further from the origin than the point of intersection of this ray and the horizontal line \( \{\zeta|\Im(\zeta)=-\Im(\zeta_*)\} \). Thus, it is no further than
\[
\left[|\Im(\zeta_*)|^2 + \left( \left( \frac34 \varkappa x\right)^{\frac23} + \cot\left(\frac\pi3 \right) |\Im(\zeta_*)| \right)^2 \right]^{\frac12}.
\]
The last two displays clearly show the validity of \eqref{zeta 1 upper bound}. Now, we get from \eqref{zeta 1 upper bound}, \eqref{defn of zeta 1}, and \eqref{g and ha} that
\[
\frac43 d_* \geq |2h_1(\lambda)| = |2g(\lambda)-\varkappa| \geq \frac{(1-|\lambda|)^2}{2|\lambda|},
\]
where the last inequality follows from the Koebe distortion theorem \cite[Theorem~1.3]{Pommerenke} (\( (2g(\lambda)-\varkappa)^{-1}\) is univalent in \( \{|\lambda|<1\} \) with derivative \( -2\ii \) at the origin, see \eqref{g_symmetry} and \eqref{g_series}). The last display yields the second claim of the lemma.
\end{proof}

\begin{figure}[ht!]
    \centering
    \includegraphics[width=0.5\linewidth]{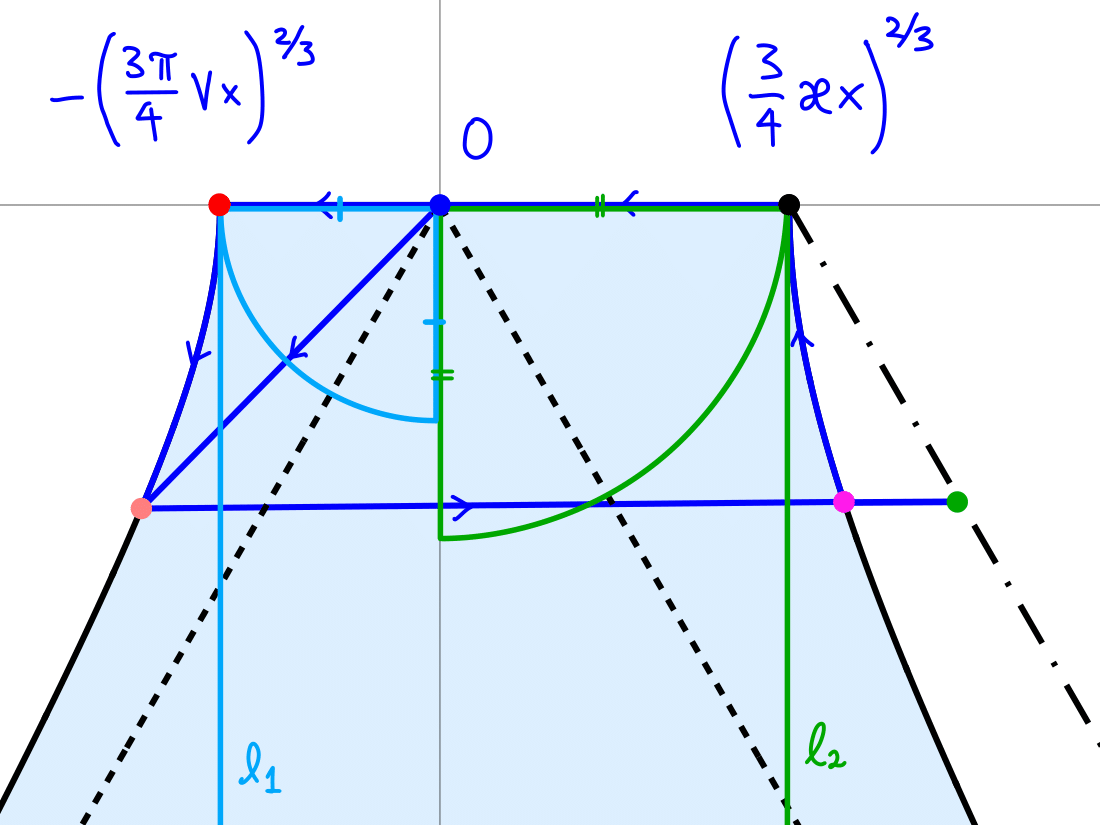}
    \caption{The vertical rays $\ell_1, \ell_2$ and the distance $|\zeta_1(\lambda)|$, \( \lambda \in \partial \mathcal{D}_1 \), which is always greater than $(\tfrac{3}{4}  \pi V x)^{\frac23}$ in the depicted case.}
    \label{fig:box}
\end{figure}

Recall \eqref{Kalpha}. We shall also need the following lemma.

\begin{lem}
\label{lem:koebe2}
Let \( R(\varkappa,x) := (\tfrac x4 (\kk \kk^\prime)^2)^{\frac23}\) and \( 0<R_* \leq R(\varkappa,x)\). It holds for all \( (\varkappa,x)\not\in \mathcal Z_\delta\) that
\[
N(\lambda) = \O\left( \frac{x^{\frac16}}{\delta \mathcal K_\alpha^{**}} \right)
\]
for \( \lambda \in \mathcal D_1 \setminus \zeta_1^{-1}\big(\{|z|< R_*\}\big) \), where \( \O(\cdot) \) is independent of \( \varkappa,x\), and \(\delta \). Moreover,
\[
N(\lambda) = \O\left( \frac{1}{\delta\mathcal K_\alpha^*}\right)
\]
for \( \lambda\in \overline{\mathcal Q_1}\setminus\mathcal D_1\) and \( (\varkappa,x)\not\in \mathcal Z_\delta\), where again \( \O(\cdot) \) is independent of \( \varkappa,x\), and \(\delta \).
\end{lem}
\begin{proof}
It follows from the first claim of Lemma~\ref{lem:koebe} that \( \zeta_1(\mathcal D_1) \)  contains a disk of radius \( R(\varkappa,x) \) centered at the origin. Hence, we can apply the Koebe Distortion Theorem \cite[Theorem~1.3]{Pommerenke} to \( \zeta_1^{-1}(z) \) on the disk of radius \( R_* \) centered at the origin, which gives
\[
\min_{|\zeta_1(\lambda)|=R_*} |\lambda-a| \geq \frac{R_*}{4|\zeta_1^\prime(a)|}.
\]
It can be readily computed using \eqref{defn of zeta 1}, \eqref{defn of h1}, and the definition of \( f(\mu) \) that
\[
|\zeta_1^\prime(a)| = \left( \frac12 {x^2} \kk \kk^\prime \right)^{\frac13}.
\]
Therefore,
\[
\min_{|\zeta_1(\lambda)|=R_*}|\lambda-a| \geq \frac{x^{-\frac23}R_* }{4(\kk \kk^\prime)^{\frac13}}.
\]
Hence, the lower bound above can serve as \( \epsilon^4 \) in Lemma~\ref{lem:estimate of N}, which gives the first claim of the lemma. The second claim is obtained analogously by replacing \( R_* \) with \( R(\varkappa,x) \) and observing that
\[
\min_{\lambda\in\partial\mathcal D_1}|\lambda-a| \geq \min_{|\zeta_1(\lambda)|=R(\varkappa,x)} |\lambda-a| \quad \text{and} \quad \frac{R(\varkappa,x)}{4|\zeta_1^\prime(a)|} = \frac18 \kk \kk^\prime. \qedhere
\]
\end{proof}

\subsubsection{Local Parametrix Near $a$}

Let \( \zeta_1(\lambda) \) be as in the previous subsection. Selecting roots of this map needs to be done in accordance with the rule stated after \eqref{airy asymp}. Since \( \Gamma(1,a)_-\) (as approached from outside of the unit disk) is mapped into \( (0,\infty)_+ \) (as approached from upper half-plane), it must hold that \( \zeta_{1-}^\gamma(\lambda)>0 \) on \( \Gamma(1,a)\). Hence, 
\begin{align}
\label{zeta1 h1}
\zeta_1^{\frac32}(\lambda)=\frac 32 xh_1(\lambda)
\end{align}
by \eqref{h_1+ and h_1-} and \( \zeta_1^{\frac14}(\lambda)\) is holomorphic in \( \mathcal D_1\setminus\Gamma(1,a) \) satisfying
\begin{equation}
\label{quarter root 1}
\zeta_{1-}^{\frac14}(\lambda) = -\ii \zeta_{1+}^{\frac14}(\lambda)\geq 0, \quad \lambda \in \Gamma(1,a).
\end{equation}

Define
\[
E_1(\lambda) := \frac1{\sqrt 2} N(\lambda) \sigma_1 \begin{pmatrix} 1 & -\ii \\ -\ii & 1 \end{pmatrix} \zeta_1^{\frac14\sigma_3}(\lambda),\quad \lambda \in \mathcal D_1.
\]
RHP~\ref{RHP N}(2) shows that it holds on \( \Gamma(1,a) \) that
\begin{equation}
\label{E+a=E-a}
E_{1+}(\lambda) = \frac1{\sqrt 2} N_-(\lambda) \begin{pmatrix} 0 & -1 \\ 1 & 0 \end{pmatrix} \sigma_1 \begin{pmatrix} 1 & -\ii \\ -\ii & 1 \end{pmatrix} \ii^{\sigma_3} \zeta_{1-}^{\frac14\sigma_3}(\lambda) = E_{1-}(\lambda).
\end{equation}
It is easy to see from \eqref{defn of gamma} and \eqref{defn of M} that \( N(\lambda) \) can have at most a quarter root singularity at \( a \). Hence, \( E_1(\lambda) \) can have at most a square root singularity there. Combining this with \eqref{E+a=E-a} means that \( E_1(\lambda) \) is analytic in \( \mathcal D_1 \). Fix \( 0<R_*<R(\varkappa,x)\). It readily follows from the first estimate of Lemma~\ref{lem:koebe2} and the maximum modulus principle for analytic functions that
\begin{equation}
\label{est E1}    
E_1(\lambda) = \O\left( \frac{x^{\frac16}}{\delta \mathcal K_\alpha^{**}}\right), \quad |\zeta_1(\lambda)|\leq R_*,
\end{equation}
for all \( (\varkappa,x)\not\in \mathcal Z_\delta\), where \( \O(\cdot) \) is independent of \( \varkappa,x\), and \(\delta \). 

Set
\begin{equation}
\label{defn of H1}
H_1(\lambda) := E_1(\lambda) \Phi_{\Ai}(\zeta_1(\lambda))e^{-\frac23\zeta_1^{\frac32}(\lambda)\sigma_3}, \quad \lambda \in \mathcal D_1\setminus\tilde \Gamma.
\end{equation}
Then,  \( H_1(\lambda) \) solves the following Riemann-Hilbert problem (this is not as much a formulation of a Riemann-Hilbert problem as an enumeration of properties of \( H_1(\lambda) \)).
\begin{RHP}
\label{RHP H1}
Find a $2 \times 2$ matrix function $H_1(\lambda)$ such that
\begin{enumerate}
    \item \( H_1(\lambda) \) is analytic in \( \lambda \in \mathcal D_1\setminus \tilde \Gamma\), where $\tilde \Gamma$ is depicted on Figure~\ref{tilde X problem pic};
    \item \( H_1(\lambda) \) has bounded traces on \( \tilde \Gamma \cap \mathcal D_1 \) that satisfy
    \[
    H_{1+}(\lambda) = H_{1-}(\lambda) \begin{dcases}
        \begin{pmatrix}
            0 & -1 \\ 1 &  e^{x(\Pi(\lambda)-\varkappa)}
        \end{pmatrix}, & \lambda\in \Gamma(1,a), \\
        \begin{pmatrix}
            1 & 0 \\ e^{x(\varkappa-2g(\lambda))} & 1 
        \end{pmatrix}, & \lambda\in\gamma_1^{in}, \\
        \begin{pmatrix}
            1 & -e^{x(\varkappa+2g(\lambda))}\\ 0 & 1 
        \end{pmatrix}, & \lambda\in\gamma_1^{out};
    \end{dcases}
    \]
    \item it holds for all \( (\varkappa,x)\notin \mathcal Z_\delta\) satisfying \eqref{D4} and uniformly on \( \partial \mathcal D_1\) that
    \[
    H_1(\lambda) = \left( I + \O\left(\frac1{\delta^2x\mathcal K_\alpha}\right) \right) N(\lambda),
    \]
    where \( \O(\cdot) \) is independent of \( \varkappa, x \), and \( \delta \);
    \item it holds for all \( (\varkappa,x)\notin \mathcal Z_\delta\) satisfying \eqref{D4} and uniformly in \( \mathcal D_1\) that
    \[
    H_1(\lambda) = \O\left( \frac{x^{\frac16}}{\delta \mathcal K_\alpha^{**}}\right),
    \]
    where \( \O(\cdot) \) is independent of \( \varkappa, x \), and \( \delta \). 
\end{enumerate}
\end{RHP}

Indeed, RHP~\ref{RHP H1}(1) is straightforward. RHP~\ref{RHP H1}(2) follows from \eqref{Airy jumps} and \eqref{g and ha} (when the orientation of the rays is reversed, the jump matrices in \eqref{Airy jumps} must be inverted). Since \( \kk,\kk^\prime<1\) and \( \kk^\prime \KK < E(\kk) \leq \tfrac\pi2\), see \cite[Equation~(19.9.8)]{NIST:DLMF}, it holds that
\[
x (\kk\kk^\prime)^2 \geq \frac 4{\pi^2} x\mathcal K_\alpha \geq c_0\frac 4{\pi^2} D_\delta \geq c_0 D_4,
\]
where we used Lemma~\ref{lem:K alpha} from Appendix~\ref{app:elliptic} further below for the penultimate estimate (\( c_0\) is an absolute constant). Thus,  with a proper choice of \( D_4 \), we see from the first claim of Lemma~\ref{lem:koebe} that \eqref{airy asymp} is applicable. Hence, we get from the definition of \( E_1(\lambda) \) and Lemma~\ref{lem:koebe} that 
\[
H_1(\lambda) =  \left( I + N(\lambda)\, \O\left(\frac1{x(\kk \kk^\prime)^2}\right)N^{-1}(\lambda) \right) N(\lambda), \quad \lambda\in\partial \mathcal D_1.
\]
Thus, RHP~\ref{RHP H1}(3) is a consequence of \eqref{kappa V estimate} and the second claim of Lemma~\ref{lem:koebe2}. To see the validity of RHP~\ref{RHP H1}(4), we first consider \( \lambda \) such that \( |\zeta_1(\lambda)|\geq R_*\). Choosing \( D_4 \) in \eqref{D4} large enough, we can select \( R_* \) sufficiently large so that
\[
\Phi_{\Ai}(\zeta) = \zeta^{-\frac14\sigma_3} \O( 1)\, e^{\frac{2}{3}\zeta^{\frac32} \sigma_3}, \quad |\zeta|>R_*,
\]
see again \eqref{airy asymp}. The desired claim now follows from the definition of \( E_1(\lambda) \) and the first claim of Lemma~\ref{lem:koebe2}. On the other hand, when \( |\zeta_1(\lambda)|\leq R_*\), we have that 
\[
\Phi_{\Ai}(\zeta),  e^{-\frac23\zeta^{\frac32}\sigma_3} = \O(1), \quad |\zeta|\leq R_*,
\]
which, combined with \eqref{est E1}, finishes the demonstration of RHP~\ref{RHP H1}(4).

\subsubsection{Local Parametrix Near $-\overline{a}$}

The remaining three constructions are extremely similar to the previous one. Let \( \mathcal Q_2:=\{\lambda~|~-\Re(\lambda),\Im(\lambda)>0\} \) denote the second quadrant. Set
\begin{align}
    h_2 (\lambda) := \frac{\ii}{4} \int_{{-\bar{a}}}^{\lambda} \frac{\sqrt{f(\mu)}}{\mu^2} d \mu, \quad \lambda\in \mathcal Q_2 \setminus \Gamma(-\overline a,-1).
\end{align}
It follows from \eqref{f_trace} and \eqref{f_jump} that
\begin{align}
h_{2-}(\lambda) = -h_{2+}(\lambda)>0, \quad \lambda \in \Gamma(-\overline a,1),
    \label{h_2+ and h_2-}
\end{align}
and from \eqref{defn of g}, \eqref{PiOmega}, \eqref{Pi}, and \eqref{Omega} that
\begin{equation}
\label{g and h - bar a}
2h_2(\lambda) = 2g(\lambda) + \begin{cases}[r]
-\varkappa+2\pi\ii V, &  |\lambda|<1, \\
\varkappa+2\pi\ii V, & |\lambda|>1,
\end{cases} \quad \lambda \in \mathcal Q_2.
\end{equation}
Moreover, we get from \eqref{g_symmetry}, \eqref{g_symmetry2}, and Lemma~\ref{lem:V_ell_connection} that
\[
g(\lambda) = -\overline{g(\overline\lambda)} = \overline{g(-\overline\lambda)} - \pi \ii V.
\]
Thus, we see from \eqref{g and ha} and \eqref{g and h - bar a} that \( h_2(\lambda) = \overline{h_1(-\overline\lambda)}\). Respectively, the desired conformal map on \( \mathcal Q_2 \) is given by
\begin{align}
    \zeta_2(\lambda) := \left( -\frac{3}{2} \, x \, h_2(\lambda) \right)^{\frac23} = \overline{\zeta_1(-\overline\lambda )}.
    \label{defn of zeta 2}
\end{align}
Strictly speaking, there is no reason to put \( - \) sign in front of \( h_2(\lambda) \) in the above formula since its effect is erased by squaring. However, since \( \zeta_2(\lambda) \) takes the positive side of \( \Gamma(-\overline a,1) \) into the positive side of \( (0,\infty)\), unlike in the case of \( \zeta_1(\lambda)\), we now need to have that \( \zeta_{2+}^\gamma(\lambda)>0 \) on \( \Gamma(-\overline a,1) \). This, in view of \eqref{h_2+ and h_2-}, implies that
\[
\zeta_2^{\frac32}(\lambda)=-\frac 32 xh_2(\lambda),
\]
which we emphasize by having the sign \( - \) in \eqref{defn of zeta 2}. Furthermore, it holds that
\begin{equation}
\label{quarter root 2}
\zeta_{2+}^{\frac14}(\lambda) = -\ii \zeta_{2-}^{\frac14}(\lambda)\geq 0, \quad \lambda \in \Gamma(-\overline a,1).
\end{equation}
We now choose the arcs \( \gamma_2^{in},\gamma_2^{out}\) and domain \(\mathcal D_2\) to be the reflections across the vertical axis of the arcs \( \gamma_1^{in},\gamma_1^{out}\) and domain \(\mathcal D_1\), respectively. Naturally, the claims of Lemma~\ref{lem:koebe} remain valid for \( \zeta_2(\lambda)\) on \( \partial \mathcal D_2\) and of Lemma~\ref{lem:koebe2} for \( N(\lambda) \) in \( \mathcal D_2 \setminus \zeta_2^{-1}\big(\{|z|< R_*\}\big) \) and \( \overline {\mathcal Q_2}\setminus \mathcal D_2\). 

Define
\[
E_2(\lambda) := \frac1{\sqrt2}N(\lambda) e^{-\pi\ii xV\sigma_3} \begin{pmatrix}
        1 & -\ii\\
        -\ii & 1
    \end{pmatrix}
    \zeta_2^{\frac{1}{4}\sigma_3}(\lambda),
\]
which is an analytic matrix function in \( \mathcal D_2 \setminus \Gamma(-\overline a,-1)\). We get from RHP~\ref{RHP N}(2) that
\[
E_{2+}(\lambda) = \frac1{\sqrt2}N_-(\lambda) e^{-\pi\ii xV\sigma_3} \begin{pmatrix}
        0 & -1\\
        1 & 0
    \end{pmatrix}  \begin{pmatrix}
        1 & -\ii\\
        -\ii & 1
    \end{pmatrix} (-\ii)^{\sigma_3} \zeta_{2-}^{\frac14\sigma_3}(\lambda) = E_{2-}(\lambda)
\]
for \( \lambda \in \Gamma(-\overline a,-1)\) and therefore \( E_2(\lambda) \) is analytic in \( \mathcal D_2 \) (again, the singularity at \( -\overline a \) is removable). Moreover, the analogue of the first claim of Lemma~\ref{lem:koebe2} now gives
\begin{equation}
\label{est E2}    
E_2(\lambda) = \O\left( \frac{x^{\frac16}}{\delta \mathcal K_\alpha^{**}}\right), \quad |\zeta_2(\lambda)|\leq R_*,
\end{equation}
for all \( (\varkappa,x)\not\in \mathcal Z_\delta\), where \( \O(\cdot) \) is independent of \( \varkappa,x\), and \(\delta \). 

Set
\begin{equation}
\label{defn of H2}
H_2(\lambda) := E_2(\lambda) \Phi_{\Ai}(\zeta_2(\lambda)) \sigma_1 e^{\frac23 \zeta_2^{\frac32}(\lambda) \sigma_3} e^{\pi\ii xV\sigma_3}, \quad \lambda \in \mathcal D_2\setminus \tilde\Gamma.
\end{equation}
Then,  \( H_2(\lambda) \) solves the following Riemann-Hilbert problem.
\begin{RHP}
\label{RHP H2}
Find a $2 \times 2$ matrix function $H_2(\lambda)$ such that
\begin{enumerate}
    \item \( H_2(\lambda) \) is analytic in \( \mathcal D_2\setminus \tilde \Gamma\);
    \item \( H_2(\lambda) \) has bounded traces on \( \tilde \Gamma \cap \mathcal D_2 \) that satisfy
    \[
    H_{2+}(\lambda) = H_{2-}(\lambda) \begin{dcases}
        \begin{pmatrix}
            0 & -e^{-2\pi\ii xV} \\ e^{2\pi\ii xV} & e^{x(\Pi(\lambda)-\varkappa)}
        \end{pmatrix}, & \lambda\in \Gamma(-\overline a,-1), \\
        \begin{pmatrix}
            1 & 0 \\ e^{x(\varkappa-2g(\lambda))} & 1 
        \end{pmatrix}, & \lambda\in\gamma_2^{in}, \\
        \begin{pmatrix}
            1 & -e^{x(\varkappa+2g(\lambda))} \\ 0 & 1 
        \end{pmatrix}, & \lambda\in\gamma_2^{out};
    \end{dcases}
    \]
    \item it holds for all \((\varkappa,x)\notin \mathcal Z_\delta\) satisfying \eqref{D4} and uniformly on \( \partial \mathcal D_2 \) that
    \[
    H_2(\lambda) = \left( I + \O\left(\frac1{\delta^2 x \mathcal K_\alpha}\right) \right) N(\lambda),
    \]
    where \( \O(\cdot) \) is independent of \( \varkappa,x\), and \( \delta \);
    \item it holds for all \((\varkappa,x)\notin \mathcal Z_\delta\) satisfying \eqref{D4} and uniformly on \( \mathcal D_2 \) that
    \[
    H_2(\lambda) = \O\left( \frac{x^{\frac16}}{\delta \mathcal K_\alpha^{**}} \right),
    \]
    where \( \O(\cdot) \) is independent of \( \varkappa,x\), and \( \delta \).
\end{enumerate}
\end{RHP}

As before, RHP~\ref{RHP H2}(1) is straightforward; RHP~\ref{RHP H2}(2) follows from \eqref{Airy jumps} and \eqref{g and h - bar a}; RHP~\ref{RHP H2}(3) is a consequence of \eqref{airy asymp}, analogues of Lemmas~\ref{lem:koebe} and~\ref{lem:koebe2} for \( \zeta_2(\lambda) \), and the fact that \( e^{\pi\ii xV} \) is a unimodular number; and RHP~\ref{RHP H2}(4) is deduced from \eqref{airy asymp}, an analogue of Lemma~\ref{lem:koebe2}, and \eqref{est E2}.


\subsubsection{Local Parametrix Near $-a$}

Denote by \( \mathcal Q_3:=\{\lambda~|~\Re(\lambda),\Im(\lambda)<0\}\) the third quadrant. Set
\begin{align}
    h_3 (\lambda) := \frac{\ii}{4} \int_{-a}^{\lambda} \frac{\sqrt{f(\mu)}}{\mu^2} d \mu, \quad \lambda\in \mathcal Q_3 \setminus \Gamma(-1,- a).
\end{align}
In this case, it follows from \eqref{f_trace} and \eqref{f_jump} that
\begin{equation}
\label{h_3+ and h_3-}    
h_{3+}(\lambda) = -h_{3-}(\lambda)>0, \quad \lambda \in \Gamma(-1,- a).
\end{equation}
Moreover, we get from \eqref{defn of g}, \eqref{PiOmega}, \eqref{Pi}, and \eqref{Omega} that
\begin{equation}
\label{g and h - a}
2h_3(\lambda) = 2g(\lambda) + \begin{cases}[r]
\varkappa+2\pi\ii V, & |\lambda|<1, \\
-\varkappa+2\pi\ii V, & |\lambda|>1,
\end{cases} \quad \lambda \in \mathcal Q_3.
\end{equation}
Symmetry relation \eqref{g_symmetry2}, Lemma~\ref{lem:V_ell_connection}, and \eqref{g and ha} now give
\[
g(\lambda) = - g(-\lambda) - \pi\ii V \quad \Rightarrow \quad h_3(\lambda) = -h_1(-\lambda).
\]
Hence, we have that
\begin{align}
    \zeta_3(\lambda) := \left( -\frac{3}{2} \, x \, h_3(\lambda) \right)^{\frac23} = \zeta_1(-\lambda).
    \label{defn of zeta 3}
\end{align}
Since \( \zeta_3(\lambda) \) maps the negative side of \( \Gamma(-1,-a) \) into the positive side of \( (0,\infty)\), it must hold that \( \zeta_{3-}^\gamma(\lambda)>0 \) on \( \Gamma(-1,-a) \). Given \eqref{h_3+ and h_3-}, we see that
\[
\zeta_3^{\frac32}(\lambda)=-\frac 32 xh_3(\lambda)
\]
and that \( \zeta_3^{\frac14}(\lambda) \) satisfies \eqref{quarter root 1} on \( \Gamma(-1,-a) \). The arcs \( \gamma_3^{in},\gamma_3^{out}\) and domain \(\mathcal D_3 \) are now chosen to be the reflections across the origin of the arcs \( \gamma_1^{in},\gamma_1^{out}\) and domain \(\mathcal D_1\), respectively. The claims of Lemma~\ref{lem:koebe} remain valid for \( \zeta_3(\lambda)\) on \( \partial \mathcal D_3\) and of Lemma~\ref{lem:koebe2} remain valid for \( N(\lambda) \) in \( \mathcal D_3 \setminus \zeta_3^{-1}\big(\{|z|< R_*\}\big) \) and \( \overline {\mathcal Q_3}\setminus \mathcal D_3\).

Define
\[
E_3(\lambda) := \frac1{\sqrt2}N(\lambda) e^{-\pi\ii xV\sigma_3} J(\lambda) \sigma_3 \begin{pmatrix}
        1 & -\ii\\
        -\ii & 1
    \end{pmatrix}
    \zeta_3^{\frac{1}{4}\sigma_3}(\lambda),
\]
which is an analytic matrix function in \( \mathcal D_1\setminus S^1\), where
\[
J(\lambda) := \begin{cases}[r] -\ii^{\sigma_3}, & |\lambda|<1, \\ \ii^{\sigma_3}, & |\lambda|>1. \end{cases} 
\]
It holds that
\[
E_{3+}(\lambda) = \frac1{\sqrt2}N_-(\lambda)(-I) e^{-\pi\ii xV\sigma_3} J_-(\lambda)(-I) \sigma_3 \begin{pmatrix}
        1 & -\ii\\
        -\ii & 1
    \end{pmatrix}
    \zeta_3^{\frac{1}{4}\sigma_3}(\lambda) = E_{3-}(\lambda)
\]
on \( \Gamma(-a,-\ii) \) and that
\[
E_{3+}(\lambda) = \frac1{\sqrt2}N_-(\lambda) e^{-\pi\ii xV\sigma_3} J_-(\lambda) \begin{pmatrix} 0 & -1 \\ 1 & 0 \end{pmatrix} \sigma_3 \begin{pmatrix}
        1 & -\ii\\
        -\ii & 1
    \end{pmatrix}
    \ii^{\sigma_3} \zeta_{3-}^{\frac{1}{4}\sigma_3}(\lambda) = E_{3-}(\lambda)
\]
on \( \Gamma(-1,-a) \). Hence, as in the previous three cases \( E_3(\lambda) \) is analytic in the disk of its definition. Again, the analogue of the first claim of Lemma~\ref{lem:koebe2} now gives
\begin{equation}
\label{est E3}    
E_3(\lambda) = \O\left( \frac{x^{\frac16}}{\delta \mathcal K_\alpha^{**}}\right), \quad |\zeta_3(\lambda)|\leq R_*,
\end{equation}
for all \( (\varkappa,x)\not\in \mathcal Z_\delta\), where \( \O(\cdot) \) is independent of \( \varkappa,x\), and \(\delta \). 

Put
\begin{equation}
\label{defn of H3}
H_3(\lambda) := E_3(\lambda) \Phi_{\Ai}(\zeta_3(\lambda)) J(\lambda) \sigma_1 \sigma_3 e^{\frac23\zeta_3^{\frac32}(\lambda)\sigma_3} e^{\pi\ii xV\sigma_3}, \quad \lambda \in \mathcal D_3\setminus \tilde\Gamma.
\end{equation}
Then, \( H_3(\lambda) \) solves the following Riemann-Hilbert problem.
\begin{RHP}
\label{RHP H3}
Find a $2 \times 2$ matrix function $H_3(\lambda)$ such that
\begin{enumerate}
    \item \( H_3(\lambda) \) is analytic in \( \mathcal D_3\setminus \tilde \Gamma\);
    \item \( H_3(\lambda) \) has bounded traces on \( \tilde \Gamma \cap \mathcal D_3 \) that satisfy
    \[
    H_{3+}(\lambda) = H_{3-}(\lambda) \begin{dcases}
        -I, & \lambda\in\Gamma(-a,-\ii), \\
        \begin{pmatrix}
            -e^{-x(\Pi(\lambda)+\varkappa)} & -e^{-2\pi\ii xV} \\ e^{2\pi\ii xV} & 0 
        \end{pmatrix}, & \lambda\in \Gamma(-1,-a), \\
        \begin{pmatrix}
            1 & e^{x(\varkappa+2g(\lambda))} \\ 0 & 1 
        \end{pmatrix}, & \lambda\in\gamma_3^{in}, \\
        \begin{pmatrix}
            1 & 0 \\ -e^{x(\varkappa-2g(\lambda))} & 1 
        \end{pmatrix}, & \lambda\in\gamma_3^{out};
    \end{dcases}
    \]
    \item it holds for all \((\varkappa,x)\notin \mathcal Z_\delta\) satisfying \eqref{D4} and uniformly on \( \partial \mathcal D_3 \) that
    \[
    H_3(\lambda) = \left( I + \O\left(\frac1{\delta^2 x \mathcal K_\alpha}\right) \right) N(\lambda),
    \]
    where \( \O(\cdot) \) is independent of \( \varkappa,x\), and \( \delta \);
    \item it holds for all \((\varkappa,x)\notin \mathcal Z_\delta\) satisfying \eqref{D4} and uniformly on \( \mathcal D_3 \) that
    \[
    H_3(\lambda) = \O\left( \frac{x^{\frac16}}{\delta \mathcal K_\alpha^{**}} \right),
    \]
    where \( \O(\cdot) \) is independent of \( \varkappa,x\), and \( \delta \).
\end{enumerate}
\end{RHP}

RHP~\ref{RHP H3}(1) is immediate; RHP~\ref{RHP H3}(2) follows from \eqref{Airy jumps} and \eqref{g and h - a} (because of the reverse orientation of the rays, the jump matrices in \eqref{Airy jumps} must be inverted); RHP~\ref{RHP H3}(3) is a consequence of \eqref{airy asymp}, the analogues of Lemmas~\ref{lem:koebe} and~\ref{lem:koebe2} for \( \zeta_3(\lambda) \), and the fact that \( e^{\pi\ii xV} \) is a unimodular number; and finally RHP~\ref{RHP H3}(4) is deduced from \eqref{airy asymp}, an analogue of Lemma~\ref{lem:koebe2}, and \eqref{est E3} exactly as before.


\subsubsection{Local Parametrix Near $\overline{a}$}

Finally, let \( \mathcal Q_4:=\{\lambda~|~\Re(\lambda),-\Im(\lambda)>0\}\) be the fourth quadrant. Set
\begin{align}
    h_4 (\lambda) := \frac{\ii}{4} \int_{\bar a}^{\lambda} \frac{\sqrt{f(\mu)}}{\mu^2} d \mu, \quad \lambda \in \mathcal Q_4\setminus\Gamma(\overline a,1).
\end{align}
In this case we get from \eqref{f_trace} and \eqref{f_jump} that
\begin{equation}
\label{h_4+ and h_4-}    
h_{4+}(\lambda) = -h_{4-}(\lambda)>0, \quad \lambda \in \Gamma(\overline a,1),
\end{equation}
and from \eqref{defn of g}, \eqref{PiOmega}, \eqref{Pi}, and \eqref{Omega} that
\begin{equation}
\label{g and h bar a}
2h_4(\lambda) = 2g(\lambda) + \begin{cases}[r]
\varkappa, & |\lambda|<1, \\
-\varkappa, & |\lambda|>1,
\end{cases} \quad \lambda \in \mathcal Q_4.
\end{equation}
Symmetry relations \eqref{g_symmetry} and \eqref{g and ha} yield
\[
g(\lambda) = - \overline{g(\overline\lambda)} \quad \Rightarrow \quad h_4(\lambda) = -\overline{h_1(\overline\lambda)}.
\]
Hence, we have that
\begin{align}
    \zeta_4(\lambda) := \left( \frac{3}{2} \, x \, h_4(\lambda) \right)^{\frac23} = \overline{\zeta_1(\overline\lambda)}.
    \label{defn of zeta 4}
\end{align}
Since \( \zeta_4(\lambda) \) maps the positive side of \( \Gamma(\overline a,1) \) into the positive side of \( (0,\infty)\), it must hold that \( \zeta_{4+}^\gamma(\lambda)>0 \) on \( \Gamma(\overline a,1) \). Thus, we get from \eqref{h_4+ and h_4-} that
\[
\zeta_4^{3/2}(\lambda) = \frac 32 xh_4(\lambda)
\]
and that \( \zeta_4^{\frac14}(\lambda) \) satisfies \eqref{quarter root 2} on \( \Gamma(\overline a,1) \). The arcs \( \gamma_4^{in},\gamma_4^{out}\) and domain \(\mathcal D_4 \) are chosen as the reflections across the horizontal axis of the arcs \( \gamma_1^{in},\gamma_1^{out}\) and domain \(\mathcal D_1\), respectively. As usual, The claims of Lemma~\ref{lem:koebe} still hold for \( \zeta_4(\lambda)\) on \( \partial \mathcal D_4 \) and of Lemma~\ref{lem:koebe2} hold for \( N(\lambda) \) in \( \mathcal D_4 \setminus \zeta_4^{-1}\big(\{|z|< R_*\}\big) \) and \( \overline {\mathcal Q_4}\setminus \mathcal D_4\).

Recall the definition of \( J(\lambda) \) in the previous subsection. Define
\[
E_4(\lambda) := \frac1{\sqrt2}N(\lambda) J^{-1}(\lambda) \sigma_3 \sigma_1 \begin{pmatrix}
        1 & -\ii\\
        -\ii & 1
    \end{pmatrix}
    \zeta_4^{\frac{1}{4}\sigma_3}(\lambda),
\]
which is an analytic matrix function in \( \mathcal D_4\setminus S^1\). We get from RHP~\ref{RHP N}(2) that
\[
E_{4+}(\lambda) = \frac1{\sqrt 2} N_-(\lambda) (-I) J_-^{-1}(\lambda)(-I) \sigma_3 \sigma_1 \begin{pmatrix} 1 & -\ii \\ -\ii & 1 \end{pmatrix} \zeta_4^{\frac14\sigma_3}(\lambda) = E_{4-}(\lambda)
\]
on \( \Gamma(-\ii,\overline a) \) and
\begin{align*}
E_{4+}(\lambda) &= \frac1{\sqrt 2} N_-(\lambda) \begin{pmatrix} 0 & -1 \\ 1 & 0 \end{pmatrix} J_+^{-1}(\lambda)\sigma_3 \sigma_1 \begin{pmatrix} 1 & -\ii \\ -\ii & 1 \end{pmatrix} (-\ii)^{\sigma_3} \zeta_{4-}^{\frac14\sigma_3}(\lambda) \\
& = \frac1{\sqrt 2} N_-(\lambda)J_-^{-1}(\lambda) \begin{pmatrix} 0 & -1 \\ 1 & 0 \end{pmatrix} \begin{pmatrix} -1 & \ii \\ \ii & -1 \end{pmatrix} \zeta_{4-}^{\frac14\sigma_3}(\lambda) = E_{4-}(\lambda)
\end{align*}
on \( \Gamma(\overline a,1) \). Hence, as in all the previous subsections, we can conclude that \( E_4(\lambda) \) is analytic in \( \mathcal D_4 \) and that the analogue of the first claim of Lemma~\ref{lem:koebe2} gives
\begin{equation}
\label{est E4}    
E_4(\lambda) = \O\left( \frac{x^{\frac16}}{\delta \mathcal K_\alpha^{**}}\right), \quad |\zeta_4(\lambda)|\leq R_*,
\end{equation}
for all \( (\varkappa,x)\not\in \mathcal Z_\delta\), where \( \O(\cdot) \) is independent of \( \varkappa,x\), and \(\delta \).

Set
\begin{equation}
\label{defn of H4}
H_4(\lambda) := E_4(\lambda) \Phi_{\Ai}(\zeta_4(\lambda)) J(\lambda) \sigma_3 e^{-\frac23\zeta_4^{\frac32}(\lambda)\sigma_3}
\end{equation}
for \( \lambda \in \mathcal D_4 \setminus\tilde \Gamma \). Then,  \( H_4(\lambda) \) solves the following Riemann-Hilbert problem.
\begin{RHP}
\label{RHP H4}
Find a $2 \times 2$ matrix function $H_4(\lambda)$ such that
\begin{enumerate}
    \item \( H_4(\lambda) \) is analytic in \( \mathcal D_4\setminus \tilde \Gamma\);
    \item \( H_4(\lambda) \) has bounded traces on \( \tilde \Gamma \cap \mathcal D_4 \) that satisfy
    \[
    H_{4+}(\lambda) = H_{4-}(\lambda) \begin{dcases}
        -I, & \lambda\in \Gamma(-\ii,\overline a), \\
        \begin{pmatrix}
            -e^{-x(\Pi(\lambda)+\varkappa)} & -1 \\ 1 & 0
        \end{pmatrix}, & \lambda\in \Gamma(\overline a,1), \\
        \begin{pmatrix}
            1 & e^{x(\varkappa+2g(\lambda))} \\ 0 & 1 
        \end{pmatrix}, & \lambda\in\gamma_4^{in}, \\
        \begin{pmatrix}
            1 & 0 \\ -e^{x(\varkappa-2g(\lambda))} & 1 
        \end{pmatrix}, & \lambda\in\gamma_4^{out};
    \end{dcases}
    \]
    \item it holds for all \((\varkappa,x)\notin \mathcal Z_\delta\) satisfying \eqref{D4} and uniformly on \( \partial \mathcal D_4 \) that
    \[
    H_4(\lambda) = \left( I + \O\left(\frac1{\delta^2 x \mathcal K_\alpha}\right) \right) N(\lambda),
    \]
    where \( \O(\cdot) \) is independent of \( \varkappa,x\), and \( \delta \);
    \item it holds for all \((\varkappa,x)\notin \mathcal Z_\delta\) satisfying \eqref{D4} and uniformly on \( \mathcal D_4 \) that
    \[
    H_4(\lambda) = \O\left( \frac{x^{\frac16}}{\delta\mathcal K_\alpha^{**}} \right),
    \]
    where \( \O(\cdot) \) is independent of \( \varkappa,x\), and \( \delta \).
\end{enumerate}
\end{RHP}

Again, RHP~\ref{RHP H4}(1) is straightforward; RHP~\ref{RHP H4}(2) follows from \eqref{Airy jumps} and \eqref{g and h bar a}; RHP~\ref{RHP H4}(3) is a consequence of \eqref{airy asymp}, analogues of Lemmas~\ref{lem:koebe} and~\ref{lem:koebe2} for \( \zeta_4(\lambda) \); RHP~\ref{RHP H4}(4) is a consequence of \eqref{airy asymp}, an analogue of Lemma~\ref{lem:koebe2}, and \eqref{est E4}.

\subsection{Asymptotic Analysis}

Let \( \mathcal D:=\mathcal D_1\cup\mathcal D_2\cup\mathcal D_3\cup \mathcal D_4\) and put \( H(\lambda):= H_i(\lambda) \), \( \lambda \in \mathcal D_i \setminus \tilde \Gamma \), \( i\in\{1,2,3,4\} \). We look for the solution of RHP~\ref{RHP tilde X} in the following form:
\begin{align}
    \tilde X(\lambda) =: R(\lambda) \begin{cases}[r]
    H(\lambda), &\lambda \in \mathcal D\setminus\tilde\Gamma, \\
    N(\lambda), & \lambda \in \C \setminus (\tilde \Gamma \cup \mathcal D).
    \end{cases}
    \label{defn of R}
\end{align}
Then, the error function $R(\lambda)$ must solve the following Riemann-Hilbert problem.
\begin{RHP}
\label{RHP R}
Find a $2 \times 2$ matrix function $R(\lambda)$ such that
\begin{enumerate}
    \item $R(\lambda)$ is analytic in $\C \setminus \Gamma_R$, where $\Gamma_R:=\tilde \Gamma \cup \partial \mathcal D$, see Figure \ref{R problem pic};
    \item one-sided traces $R_\pm(\lambda)$ exist a.e. on \( \Gamma_R \), belong to \( L^2(\Gamma_R) \), and satisfy
    \begin{align}
    \begin{aligned}
    R_+(\lambda) &= R_-(\lambda) G_R(\lambda), \quad \lambda \in \Gamma_{R},
    \end{aligned}
    \end{align}
    where
    \begin{align}
    G_R(\lambda) = \begin{cases}[r] 
        N(\lambda) H^{-1}(\lambda), & \lambda \in \partial \mathcal D\setminus (\ii\R\cup\R), \smallskip \\ 
        H_-(\lambda) H_+^{-1}(\lambda), & \lambda \in \partial \mathcal D\cap (\ii\R\cup\R), \smallskip \\
        H_-(\lambda) G_{\tilde X}(\lambda) H_+^{-1}(\lambda), & \lambda \in \tilde \Gamma \cap \mathcal D, \smallskip \\
        N_-(\lambda)G_{\tilde X}(\lambda) N_+^{-1}(\lambda), & \lambda \in \ii\R \setminus \partial \mathcal D;
    \end{cases}
    \label{R jump}
    \end{align}
    \item it holds that \( R(\lambda) =  I + \O(1/\lambda)  \)  as \( \lambda \to \infty \).
\end{enumerate}
\end{RHP}

\begin{figure}[ht!]
    \centering
    \includegraphics[width=.9\linewidth]{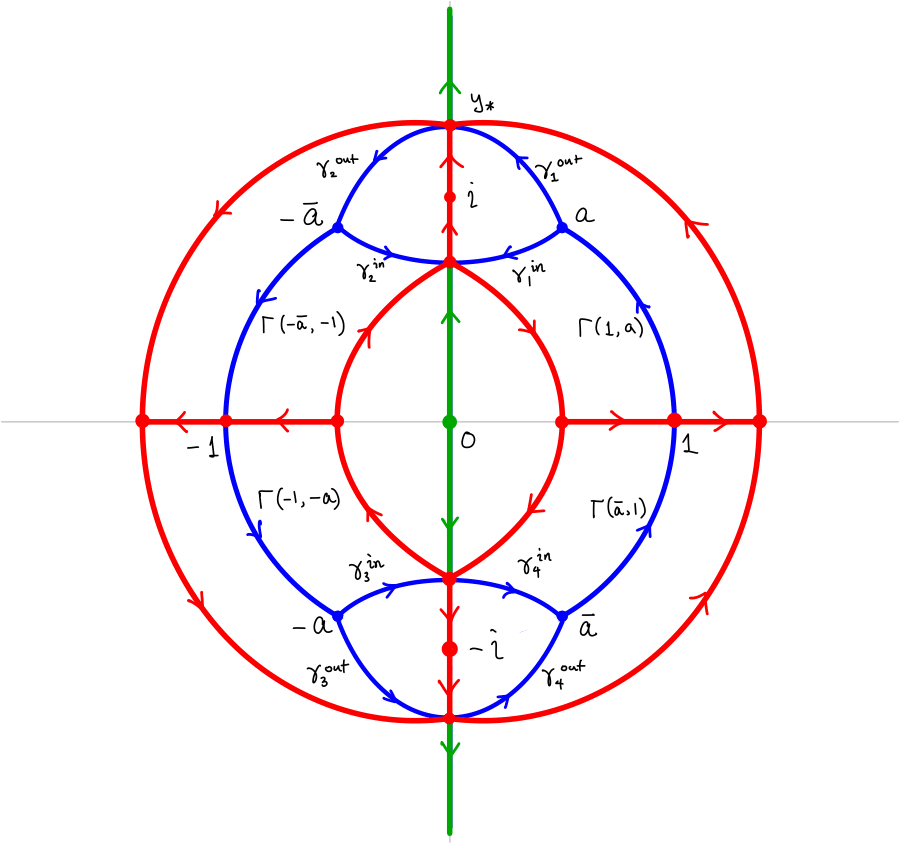}
    \caption{The contour $\Gamma_{R}$ for RHP \ref{R problem pic}.}
    \label{R problem pic}
\end{figure}

As in all the previous cases, we now discuss the size of \( G_R(\lambda) - I\) to show that RHP~\ref{RHP R} is a small norm problem. In what follows, we assume without mentioning that \( (\varkappa,x)\notin \mathcal Z_\delta\) for some \( \delta\in(0,\tfrac12)\) and that \eqref{D4} holds. All \( \O(\cdot) \) terms appearing below are independent of \( \varkappa,x \), and \( \delta \). 

It readily follows from RHP~\ref{RHP H1}(3),  RHP~\ref{RHP H2}(3),  RHP~\ref{RHP H3}(3), and RHP~\ref{RHP H4}(3) that 
 \begin{equation}
\label{G_R est 1}
G_R(\lambda) = I + \O\left(\frac1{\delta^2 x\mathcal K_\alpha}\right), \quad \lambda \in \partial \mathcal D
\end{equation}
(this is an estimate on the red parts of the contour on Figure~\ref{R problem pic}). 

Next, let \( \ii y_*\) be the common endpoint of the arcs \( \gamma_1^{out}\) and \( \gamma_2^{out}\). Then, according to Lemma~\ref{lem:Re_g} and the second claim of Lemma~\ref{lem:koebe2} together with its analogues in \( \mathcal Q_2, \mathcal Q_3 \), and \( \mathcal Q_4 \), we have on \(  \ii[ y_*,\infty) \) that
\begin{align}
G_R(\lambda) &= I  + e^{2xg(\lambda)} N(\lambda) \begin{pmatrix} 0 & \ii s^\R \\ 0 & 0 \end{pmatrix} N^{-1}(\lambda) \\
& = I + \O\left( |s^\R| \delta^{-2} (\mathcal K_\alpha^*)^{-2} e^{-x \left( \varkappa + \frac1{2} \sinh^2(\frac12\ln|\lambda|) \right)} \right).
\label{G_R est 2a}
\end{align}
In fact, one can easily see from Lemma~\ref{lem:Re_g} and the form of \( G_{\tilde X}(\lambda) \) on the imaginary axis that the above estimate remains valid on \( \ii \R \setminus \partial \mathcal D\). In particular, we have that
\begin{equation}
\label{G_R est 2}
\|G_R-I\|_{L^2(\ii\R\setminus \partial \mathcal D)\cap L^\infty(\ii\R\setminus \partial \mathcal D)} = \O\left(|s^\R| \delta^{-2}(\mathcal K_\alpha^*)^{-2} e^{-\varkappa x}\right)
\end{equation}
(this is an estimate on the green parts of the contour on Figure~\ref{R problem pic}). 

Next, it follows from \eqref{cut1}, RHP~\ref{RHP H1}(2,4) and RHP~\ref{RHP H4}(2,4) that
\begin{align}
G_R(\lambda) &= I  + H_-(\lambda) \begin{pmatrix} \O(|s^\R|_+e^{-\varkappa x}) & -e^{-x(\varkappa+\Pi(\lambda))} \\ e^{-x(\varkappa-\Pi(\lambda))}\O(|s^\R|_+e^{-\varkappa x}) & \O(|s^\R|_+e^{-\varkappa x}) \end{pmatrix} H_-^{-1}(\lambda) \\
& = I + \O\left( |s^\R|_+ \delta^{-2}(\mathcal K_\alpha^{**})^{-2} x^{\frac13} e^{-\varkappa x} \right),
\label{G_R est 3}
\end{align}
on \( \Gamma(\overline a,a) \), where we used the fact that \( -\varkappa \leq \Pi(\lambda)\leq \varkappa \) on the considered arc by \eqref{Pi}. An estimate on \( \Gamma(-\overline a,-a)\) is absolutely analogous (we also need to use unimodularity of \( e^{2\pi\ii xV}\)). On the other hand, on \( \gamma_1^{out}\) we get from the definition of \( S_{L_1}(\lambda) \) in \eqref{defn of SL1 and SU1} and RHP~\ref{RHP H1}(2,4) that
\begin{align}
G_R(\lambda) &= I  + e^{x(\varkappa+2g(\lambda))} H_{1-}(\lambda) \begin{pmatrix} 0 & 1-\overline A \\ 0 & 0 \end{pmatrix} H_{1-}^{-1}(\lambda) \\
& = I + \O\left(|s^\R|_+\delta^{-2}(\mathcal K_\alpha^{**})^{-2} x^{\frac13} e^{-\varkappa x}\right),
\label{G_R est 4}
\end{align}
where we used \eqref{A is almost 1} and the fact that \( 2\Re(g(\lambda))\leq -\varkappa \) on the considered set by Lemma~\ref{lem:Re_g}. Again, one can straightforwardly verify that the same estimate holds on all \( \gamma_i^{out} \) and \( \gamma_i^{in}\), \( i\in\{1,2,3,4\} \) (this finishes an estimate on the blue parts of the lens on Figure~\ref{R problem pic}). Combining \eqref{G_R est 1}, \eqref{G_R est 2}, \eqref{G_R est 3}, and \eqref{G_R est 4} gives 
\begin{equation}
\label{G_R est}
\|G_R-I\|_{L^2(\Gamma_R)\cap L^\infty(\Gamma_R)} = \O\left(\frac{1}{\delta^2 x\mathcal K_\alpha}\right),
\end{equation}
where we also used Lemma~\ref{lem:K alpha} of Appendix~\ref{app:elliptic}. Lemma~\ref{lem:K alpha} also shows that
\[
\delta^2 x\mathcal K_\alpha \geq c_0 \delta^2 D_\delta = c_0 D_4
\]
for some absolute constant \( c_0 \). Hence, as in the first part of this paper, we get from \eqref{G_R est} and the small norm theorem that there exists a unique solution $R(\lambda)$ of the RHP~\ref{RHP R} if we choose \( D_4 \) large enough. There is a technical difficulty here in that \( \Gamma_R \) is not a single contour but a family of contours that depend on \( (\varkappa,x) \). Therefore, to see that a single constant \( D_4 \) suffices, we need uniform boundedness of the norms of the corresponding Cauchy operators. This fact is explained in Appendix~\ref{s:co}. Then,
\[
R(\lambda) = I + \frac{1}{2 \pi \ii} \int_{\Gamma_R} \frac{\rho(\lambda')(G_R(\lambda') - I)}{\lambda' - \lambda} d\lambda', \quad \lambda \notin \Gamma_R,
\]
where $\ds \rho(\lambda') = R_-(\lambda^\prime)$ for $\lambda' \in \Gamma_R$, see \eqref{gl Y singular int eq}. Exactly as in \eqref{gl rho_Y - I has a small norm}, we get in this case that
\[
\|\rho - I\|_{L^2(\Gamma_R)} = \O\left(\frac{1}{\delta^2 x\mathcal K_\alpha}\right). 
\]
Observe that \( G_R(\lambda')-I \) vanishes exponentially at the origin by \eqref{G_R est 2a} and Lemma~\ref{lem:Re_g}. Hence, we get exactly as in the case of \eqref{representation Y(0)} that \( R(0)\) is well defined and is expressed as 
\[
R(0) = I + \frac{1}{2\pi \ii} \int_{\Gamma_{R}} \frac{G_{R}(\lambda') - I}{\lambda'} d\lambda' + \frac{1}{2\pi \ii} \int_{\Gamma_{R}} \frac{(\rho(\lambda') - I)(G_R (\lambda') - I)}{\lambda'} d\lambda'.
\]
Using the exponential vanishing of \( G_R(\lambda')-I \) at the origin once more, we see that estimate \eqref{G_R est 2} remains valid for \( (G_R(\lambda')-I)/\lambda' \). Moreover, since \( \partial \mathcal D\) is separated away from the origin independently of \( \varkappa \) and \( x \), see Lemma~\ref{lem:koebe}, estimates \eqref{G_R est 1}, \eqref{G_R est 3}, and \eqref{G_R est 4} remain true for \( (G_R(\lambda')-I)/\lambda' \) as well. Thus, we get from the Cauchy-Schwarz inequality that
\begin{equation}
\label{error rate}
R(0) = I + o(1), \quad o(1) = \O\left(\frac{1}{\delta^2 x\mathcal K_\alpha}\right).
\end{equation}

To finish the proof of Theorem~\ref{thm:4}, we observe from RHP~\ref{RHP tilde X}(3) and \eqref{defn of R} that
\[
P_0\sigma_1\sigma_3 = e^{-x\ell\sigma_3} R(0) N(0)e^{x\ell\sigma_3}.
\]
Recall that \( \ell \) is purely imaginary, see Lemma~\ref{lem:V_ell_connection}, and therefore \( e^{x\ell} \) is unimodular. We get from the definition of \( P_0 \) in RHP~\ref{rhp original}(3), \eqref{defn of M}, and \eqref{defn of M-hat(infty)} that
\[
\begin{pmatrix}
\sinh\tfrac u2 & * \\ \cosh \tfrac u2 & *    
\end{pmatrix} = \left(I + o(1)\right)
\begin{pmatrix}
\Theta_{1,0}^{-1}(\infty) & 0 \\ 0 & \Theta_{0,1}^{-1}(\infty)    
\end{pmatrix}
\begin{pmatrix}
\kk^\prime \, \Theta_{1,0}(0) & * \\ e^{2x\ell} \kk \, \Theta_{0,0}(0) & *    
\end{pmatrix},
\]
where, until the end of this subsection, the error rate \( o(1) \) is as in \eqref{error rate}. Hence, we get that
\[
e^{\pm\frac u2} = (1 + o(1)) e^{2x\ell} \kk \frac{\Theta_{0,0}(0)}{\Theta_{0,1}(\infty)} \pm (1 + o(1)) \kk^\prime \frac{\Theta_{1,0}(0)}{\Theta_{1,0}(\infty)},
\]
which, with the help of Lemma~\ref{lem:V_ell_connection}, \eqref{thetas at infinity}, and \eqref{Theta of zero}, can be rewritten as
\[
e^{\pm\frac u2} = (1+o(1)) \kk \frac{\vartheta_3(0) \vartheta_3(xV)}{\vartheta_2(0) \vartheta_2(xV)} \mp (1 + o(1)) \kk^\prime \frac{\vartheta_3(0) \vartheta_1(xV)}{\vartheta_0(0) \vartheta_2(xV)}.
\]
Writing \( e^{u} \) as the ratio of \( e^{\frac u2} \) and \( e^{-\frac u2} \) gives
\[
e^u = \frac{1 + o(1) - (1 + o(1)) \frac{\kk^\prime}{\kk} \frac{\vartheta_2(0) \vartheta_1(xV)}{\vartheta_0(0) \vartheta_3(xV)} }{1 + o(1) + (1 + o(1)) \frac{\kk^\prime}{\kk}  \frac{\vartheta_2(0) \vartheta_1(xV)}{\vartheta_0(0) \vartheta_3(xV)} }.  
\]
Since \( \kk,\kk^\prime \) are elliptic moduli corresponding to the nome \( q=e^{\pi\ii \tau}=e^{-\pi \KK^\prime/\KK} \), see \eqref{defn of tau}, we get from \cite[Equations~(22.2.2) and (22.2.7)]{NIST:DLMF} that
\[
\kk = \frac{\vartheta_2^2(0)}{\vartheta_3^2(0)} \quad \text{and} \quad \sd(2xV\KK,\kk) = \frac{\vartheta_3^2(0)}{\vartheta_0(0)\vartheta_2(0)} \frac{\vartheta_1(xV)}{\vartheta_3(xV)}.
\]
Thus, it holds that
\[
\frac{\kk^\prime}{\kk} \frac{\vartheta_2(0) \vartheta_1(xV)}{\vartheta_0(0) \vartheta_3(xV)} = \kk^\prime \sd(2xV\KK,\kk).
\]
Since \( \kk^\prime|\sd(\cdot,\kk)|\leq 1 \) on the real line, it follows from \eqref{Legendre's relation} that 
\[
e^u = \frac{1 - \kk^\prime\sd(2xV\KK,\kk) + o(1)}{1 + \kk^\prime\sd(2xV\KK,\kk) + o(1)} = \frac{1 - \kk^\prime\sd(x-2\kappa \KK,\kk) + o(1)}{1 + \kk^\prime\sd(x-2\kappa \KK,\kk) + o(1)}.
\]
Asymptotic formula \eqref{asymptotic formula 4} is now a consequence of \eqref{symmetries of u} and \eqref{error rate}.


\section{Proof of Theorem~\ref{thm:5}}

The proof of Theorem~\ref{thm:5} closely follows the proof of Theorem~\hyperref[IMY main thm]{IMY} given in \cite{IMY}. However, as compared to \cite{IMY}, we use a slightly different jump factorization, which leads to a simpler construction of the global parametrix. On the other hand, since the monodromy parameter \( B^\R \) is now a function of \( x \), we need a more precise analysis of the local parametrices. In that, we follow the line of arguments developed in \cite{Bothner17}. 

Similarly to all the previous sections, we start with a gauge transformation of RHP~\ref{rhp original}. Again, we take \( \rho=4/x \) and now set
\begin{align}
\label{defn W}
    W(\lambda) := \hat{\Psi}\left( \frac{4\lambda}{x} \right) \begin{Bmatrix}
        I, & |\lambda| > 1 \\
        \sigma_1 \sigma_3, & |\lambda| < 1
    \end{Bmatrix} e^{-x \varphi(\lambda)\sigma_3}, \quad \varphi(\lambda) := \frac{\ii}4\left(\lambda + \frac{1}{\lambda} \right).
\end{align}
As we have pointed out before, if we take \( \varkappa=0 \) in \eqref{defn of g}, then \( \varphi(\lambda) = g(\lambda)+\frac\ii2 \).

\subsection{Deformed Riemann-Hilbert Problem}

As expected, the matrix function \( W(\lambda) \) defined above solves a Riemann-Hilbert problem that is extremely similar to RHP~\ref{RHP X}:
\begin{RHP}
\label{RHP W}
Find a $2 \times 2$ matrix function $W(\lambda)$ such that
\begin{enumerate}
\item $W(\lambda)$ is analytic in $\C \setminus \Gamma$, where, as before, \( \Gamma = \ii \R \cup S^1 \), see Figure~\ref{W problem pic};
    \item one-sided traces $W_\pm(\lambda)$ exist a.e. on \( \Gamma \), belong to \( L^2(\Gamma) \), and satisfy
    \begin{align*}
        W_+(\lambda) = W_-(\lambda) G_{W}(\lambda), \quad \lambda \in \Gamma,
    \end{align*}
    where the jump matrices $G_{W}(\lambda)$ on \( \Gamma \) are as on Figure~\ref{W problem pic};
    \item it holds that
    \begin{align}
    \label{W asymptotics lambda -> 0}
        W(\lambda) = \begin{cases}
            I + \O(1/\lambda)  &  \text{as} \quad \lambda\to\infty, \smallskip \\
            P_0 \, \sigma_1 \sigma_3 (I + \O(\lambda)) & \text{as} \quad \lambda\to0.
        \end{cases}
    \end{align}
\end{enumerate}
\end{RHP}

\begin{figure}[ht!]
    \centering
    \includegraphics[width=8cm]{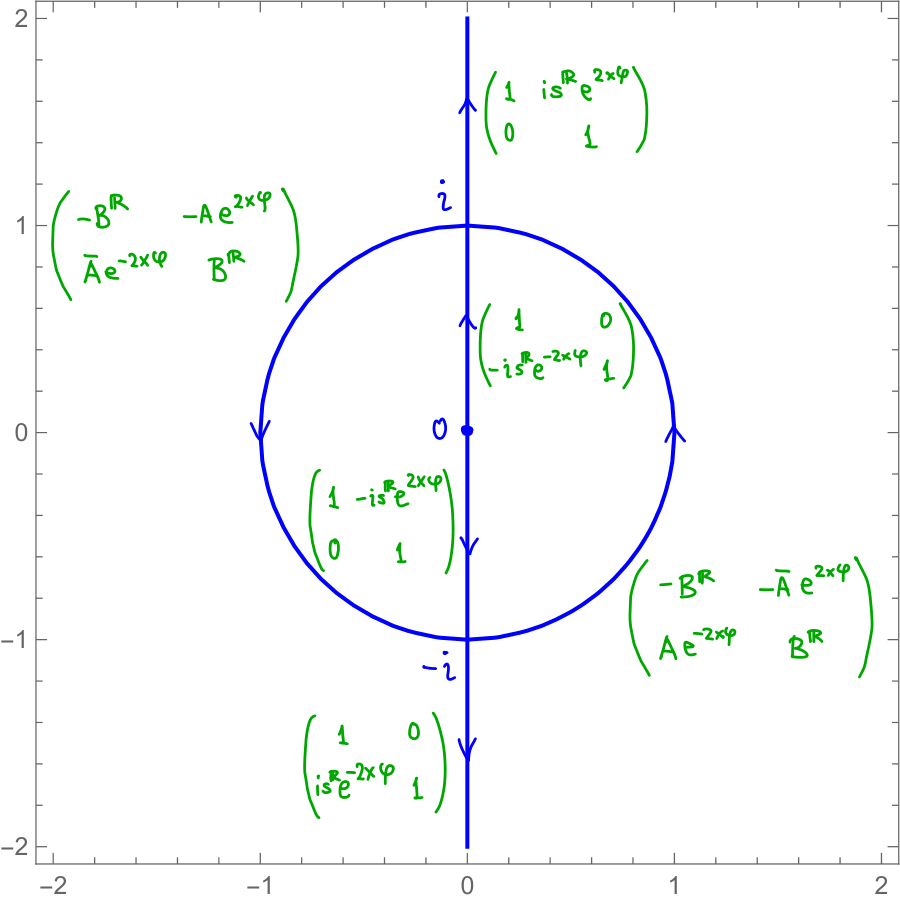}
    \caption{The jump matrices $G_{W}(\lambda)$ for RHP~\ref{RHP W}.}
    \label{W problem pic}
\end{figure}

We now proceed by deforming RHP~\ref{RHP W} using certain factorizations of the jump matrices comprising \( G_W(\lambda) \) that are quite similar to \eqref{defn of SL1 and SU1}--\eqref{defn of SL4 and SU4}. However, since \( \varphi(\lambda) \) is a rational function and therefore, unlike \( g(\lambda) \), has no branch cut, the functions \( \Pi(\lambda)\) and \( \Omega(\lambda) \) from \eqref{PiOmega} are now simply \( 0 \) and \( 2\ii\varphi(\lambda)\), respectively. This necessitates an appearance of a diagonal matrix \( e^{\varkappa x\sigma_3}\) in the formulae below, where one needs to recall that \( B^\R = e^{-\varkappa x}\). More precisely, we have that
\begin{align}
G_W(\lambda) & = \begin{pmatrix}
        -B^\R & -\overline{A} e^{2x\varphi(\lambda)}\\
        A e^{-2x\varphi(\lambda)} & B^\R
    \end{pmatrix} \\
    & = \begin{pmatrix}
        1 & -\overline{A} e^{x(\varkappa+2\varphi(\lambda))} \\
        0 & 1
    \end{pmatrix} \begin{pmatrix}
        1 & 0 \\
        Ae^{-x(\varkappa+2\varphi(\lambda))} & 1
    \end{pmatrix} e^{\varkappa x\sigma_3} =: S_{L_1}^*(\lambda)S_{R_1}^*(\lambda)e^{\varkappa x\sigma_3}
\label{SL1 star}    
\end{align}
for \( \lambda\in S^1\cap \mathcal Q_1\), where, as before, \( \mathcal Q_i \) is the \( i \)-th quadrant. In a similar way we get that
\begin{align}
G_W(\lambda) & = \begin{pmatrix}
        -B^\R & -A e^{2x\varphi(\lambda)}\\
        \overline A e^{-2x\varphi(\lambda)} & B^\R
    \end{pmatrix} \\
& = \begin{pmatrix}
        1 & -A e^{x(\varkappa+2\varphi(\lambda))} \\
        0 & 1
    \end{pmatrix} \begin{pmatrix}
        1 & 0 \\
        \overline Ae^{-x(\varkappa+2\varphi(\lambda))} & 1
    \end{pmatrix} e^{\varkappa x\sigma_3} =: S_{L_2}^*(\lambda)S_{R_2}^*(\lambda)e^{\varkappa x\sigma_3}
\end{align}
for \( \lambda\in S^1\cap \mathcal Q_2 \). As in \eqref{defn of SL3 and SU3}--\eqref{defn of SL4 and SU4}, we change the order of the triangular matrices in the lower half-plane. This way we obtain that
\begin{align}
G_W(\lambda) &= \begin{pmatrix}
        1 & 0 \\
        -\overline Ae^{x(\varkappa-2\varphi(\lambda))} & 1
    \end{pmatrix} \begin{pmatrix}
        1 & A e^{-x(\varkappa-2\varphi(\lambda))} \\
        0 & 1
    \end{pmatrix}  \big(-e^{-\varkappa x\sigma_3}\big) \\
    & =: S_{L_3}^*(\lambda)S_{R_3}^*(\lambda)\big(-e^{-\varkappa x\sigma_3}\big)
\end{align}
for \( \lambda\in S^1\cap \mathcal Q_3 \) and
\begin{align}
G_W(\lambda) &= \begin{pmatrix}
        1 & 0 \\
        -A e^{x(\varkappa-2\varphi(\lambda))} & 1
    \end{pmatrix} \begin{pmatrix}
        1 & \overline A e^{-x(\varkappa-2\varphi(\lambda))} \\
        0 & 1
    \end{pmatrix}  \big(-e^{-\varkappa x\sigma_3}\big) \\
    & =: S_{L_4}^*(\lambda)S_{R_4}^*(\lambda)\big(-e^{-\varkappa x\sigma_3}\big)
\end{align}
for \( \lambda\in S^1\cap \mathcal Q_4 \). Moreover, similarly to \eqref{SL factorization} and \eqref{SU factorization}, it follows from \eqref{mon data} that
\begin{align}
S_{L_1}^*(\lambda)^{-1} \begin{pmatrix} 1 & \ii s^\R e^{2x\varphi(\lambda)} \\ 0 & 1 \end{pmatrix} S_{L_2}^*(\lambda) = I
\end{align}
for \( \lambda\in\ii\R \) with \( \Im(\lambda)>1\) and
\begin{align}
S_{L_3}^*(\lambda)^{-1} \begin{pmatrix} 1 & 0 \\ \ii s^\R e^{-2x\varphi(\lambda)} & 1 \end{pmatrix} S_{L_4}^*(\lambda) = I
\end{align}
for \( \lambda\in\ii\R \) with \( \Im(\lambda)<-1\), as well as
\begin{align}
S_{R_1}^*(\lambda)e^{\varkappa x\sigma_3} \begin{pmatrix} 1 & 0 \\ -\ii s^\R e^{-2x\varphi(\lambda)} & 1 \end{pmatrix}e^{-\varkappa x\sigma_3} S_{R_2}^*(\lambda)^{-1} = I
\end{align}
for \( \lambda\in\ii\R \) with \( \Im(\lambda)\in(0,1)\) and
\begin{align}
S_{R_3}^*(\lambda)e^{-\varkappa x\sigma_3} \begin{pmatrix} 1 & -\ii s^\R e^{2x\varphi(\lambda)} \\ 0 & 1 \end{pmatrix}e^{\varkappa x\sigma_3} S_{R_4}^*(\lambda)^{-1} = I
\end{align}
for \( \lambda\in\ii\R \) with \( \Im(\lambda)\in(-1,0)\).

\begin{figure}[ht!]
\centering
\includegraphics[width=6cm]{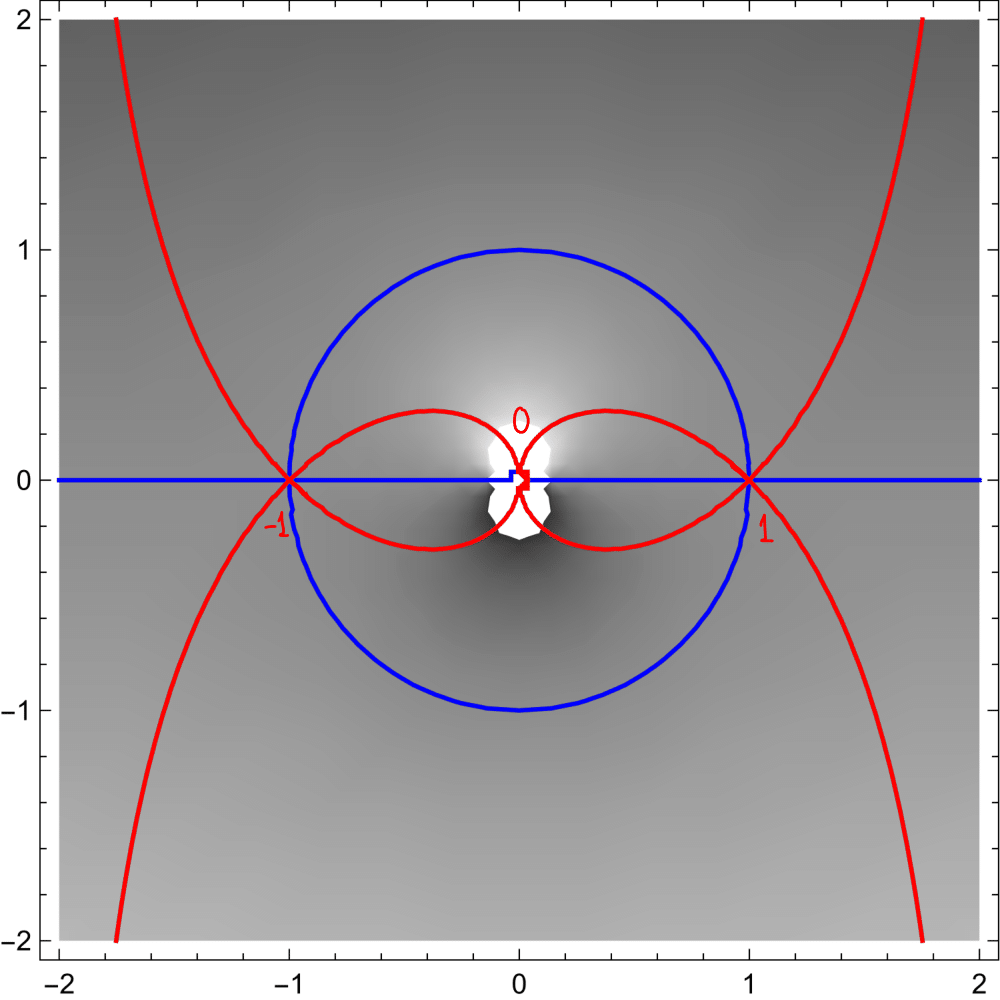}
\includegraphics[width=1cm]{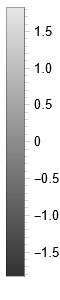}
\caption{The zero level curve $\Re(2\varphi(\lambda)) = 0$ (blue), the sign of $\Re(2\varphi(\lambda))$ (gray), and the stationary contour $\Im(2\varphi) = \pm 1$ (red).}
\label{Re 2 phi}
\end{figure}

The deformation of RHP~\ref{RHP W} is done in the regions determined by the stationary curves of \( \varphi(\lambda) \). More precisely, Let $\lambda = \xi + \ii \eta$, where $\xi, \, \eta \in \R$. Then,
\begin{align}
\label{real part varphi}
    \Re (2 \varphi (\lambda)) = \frac{\eta \left( 1 - \xi^2 - \eta^2 \right)}{2 \left( \xi^2 + \eta^2 \right)} \quad \text{and} \quad \Im (2 \varphi (\lambda)) = \frac{\xi \left( 1 + \xi^2 + \eta^2 \right)}{2 \left( \xi^2 + \eta^2 \right)}.
\end{align}
Clearly, \( \varphi(\lambda) \) is purely imaginary on the unit circle and the stationary points of $2\varphi(\lambda)$ are $\lambda = \pm 1$ with $2 \varphi ( \pm 1 ) = \pm \ii$. The stationary contour $\{\Im (2 \varphi(\lambda)) = \pm 1\}$ is shown on Figure~\ref{Re 2 phi}. We set
\[
\Gamma^* := \{\Im (2 \varphi(\lambda)) = \pm 1\} \cup [-1,1],
\]
whose orientation is shown on Figure~\ref{W tilde problem pic}. 

\begin{figure}[ht!]
\centering
\includegraphics[width=8cm]{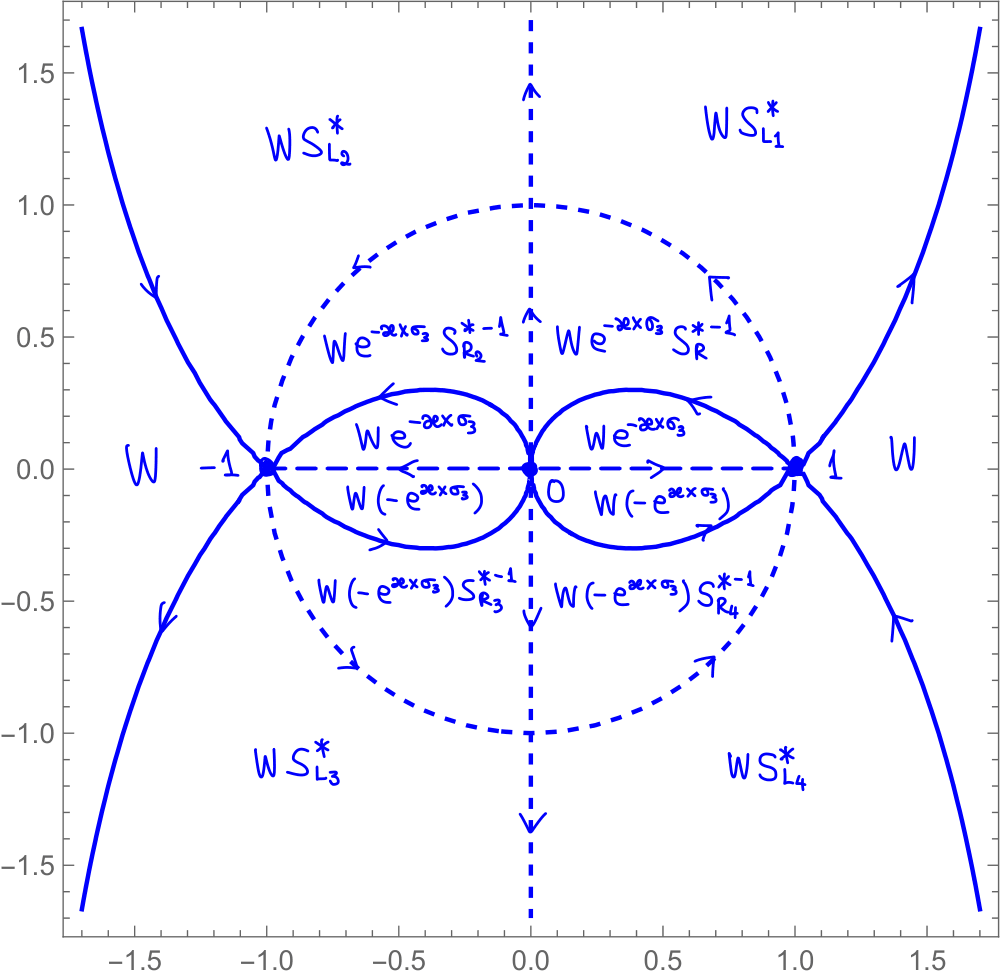}
\caption{Definition of the matrix $W^*(\lambda)$. The solid contour is the stationary contour of \( 2\varphi(\lambda)\), see Figure~\ref{Re 2 phi}.}
\label{defn of W tilde}
\end{figure}

Let the matrix function \( W^*(\lambda)\) be defined by Figure~\ref{defn of W tilde}. Then, \( W^*(\lambda) \) is the solution of the following Riemann-Hilbert problem.
\begin{RHP} 
\label{W tilde RH problem}
Find a $2 \times 2$ matrix function $W^*(\lambda)$ such that
\begin{enumerate}
    \item $W^*(\lambda)$ is analytic in $\C \setminus \Gamma^*$;
    \item one-sided traces $W_\pm^*(\lambda)$ exist a.e. on \( \Gamma^* \), belong to \( L^2(\Gamma^*) \), and satisfy
    \begin{align*}
        W^*_+(\lambda) = W^*_-(\lambda) G_{W^*}(\lambda), \quad \lambda \in \Gamma^*,
    \end{align*}
    where the jump matrices $G_{W^*}(\lambda)$ are as on Figure~\ref{W tilde problem pic};
    \item it holds that
    \begin{align}
    \label{W tilde near zero}
            W^*(\lambda) = \begin{cases}
            I + \O(1/\lambda)  &  \text{as} \quad \lambda\to\infty, \smallskip \\
            \pm P_0 \, \sigma_1\sigma_3 (I + \O(\lambda))e^{\mp\varkappa x\sigma_3} & \text{as} \quad \lambda\to0, \;\; \pm\Im(\lambda)>0.
        \end{cases}
    \end{align}
\end{enumerate}
\end{RHP}

\begin{figure}[ht!]
\centering
\includegraphics[width=8cm]{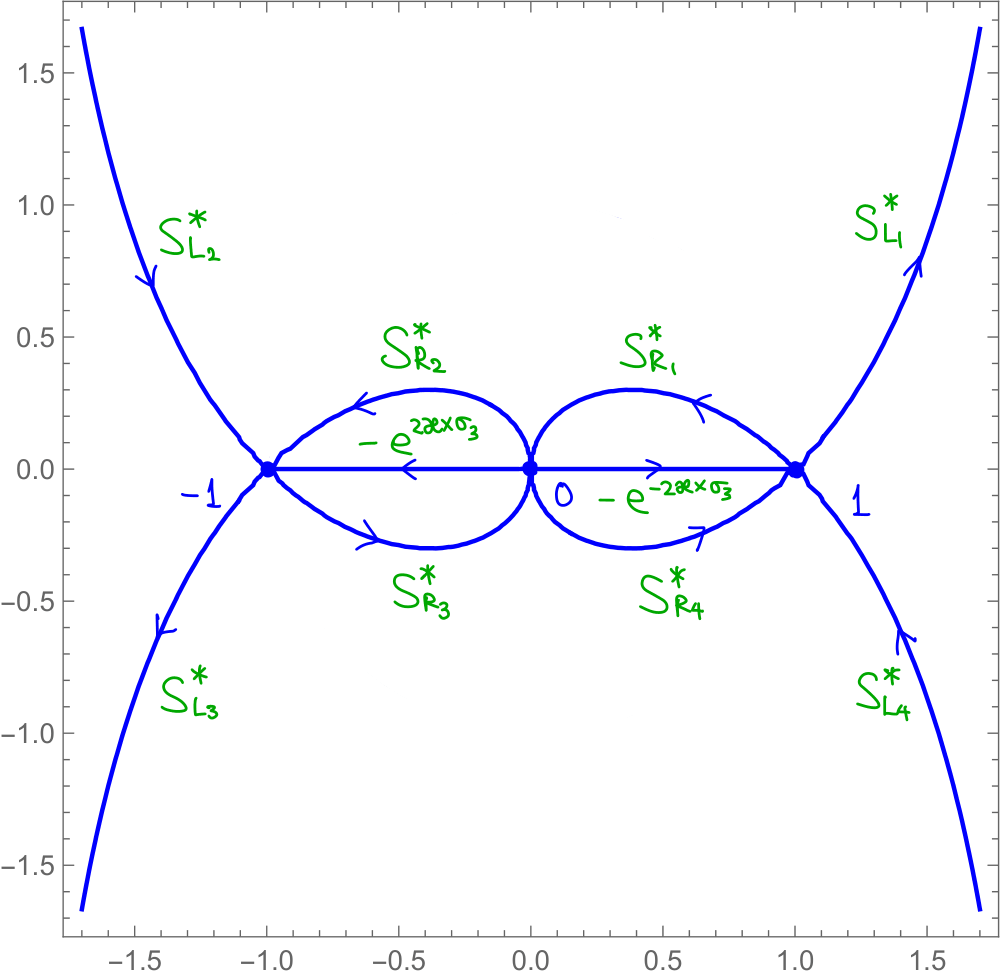}
\caption{ The contour $\Gamma^*$ and the jump matrices $G_{W^*}(\lambda)$ for RHP~\ref{W tilde RH problem}.}
\label{W tilde problem pic}
\end{figure}

RHP~\ref{W tilde RH problem}(2) readily follows from the factorizations obtained above and RHP~\ref{RHP W}(2). The limit at \( 0 \) in RHP~\ref{W tilde RH problem}(3) is simply the result of the definition of \( W^*(\lambda) \) around the origin and of RHP~\ref{RHP W}(3). The limit at infinity in RHP~\ref{W tilde RH problem}(3) is the consequence of the fact that \( S_{L_i}(\lambda)\) uniformly converges to the identity matrix as \( \lambda\to\infty \) within the domain where \( W^*(\lambda) := W(\lambda)S_{L_i}(\lambda)\) because \( \Re(\varphi(\lambda))=-\tfrac12\eta+o(1) \) as \( \lambda=\xi+\ii\eta\to\infty\).

Lastly, we point out that one can directly verify that
\begin{align} 
\label{Wtilde jump symmetry}
G_{W^*}(\lambda) = \sigma_1 G_{W^*}(-\lambda) \sigma_1, \quad \lambda\in\Gamma^*
\end{align}
(the reason we oriented \( (-1,0)\) from right to left is so that \eqref{Wtilde jump symmetry} holds everywhere on \( \Gamma^*\)).


\subsection{Global Parametrix}

It is not hard to convince oneself using \eqref{real part varphi} that the jump matrix \( G_{W^*}(\lambda)\) is approaching identity as \( x\to\infty \) everywhere on \( \Gamma^*\setminus[-1,1]\). Therefore, the leading order behavior of \( W^*(\lambda)\) is determined by a matrix function that is analytic in \( \overline \C\setminus[-1,1]\), is equal to \( I \) at infinity and whose boundary values on \( (-1,1) \) are related by the jump matrix \( - e^{-2\varkappa x\sigma_3} \) when \( (-1,1) \) has a single orientation from left to right. Clearly, the matrix function
\begin{align} 
\label{defn of Pgl}
\left( \frac{\lambda - 1}{\lambda + 1} \right)^{\nu \sigma_3}, \quad \nu = \frac{1}{2} + \ii\frac{\varkappa x}{\pi},
\end{align}
possesses all these properties, where we take the principal branch of the power function (i.e., this matrix function is positive on the real line way from \( [-1,1] \)). Observe also that
\begin{align} 
\label{Pgl asymp 0}
\lim_{\lambda\to0,\pm\Im(\lambda)>0} \left( \frac{\lambda - 1}{\lambda + 1} \right)^{\nu \sigma_3} = e^{\pm\pi\ii\nu \sigma_3}.
\end{align}

\subsection{Local Parametrices}
\label{ss:6.3}

The jump matrix \(G_{W^*}(\lambda) \) is uniformly close to the identity on \( \Gamma^*\setminus[-1,1] \) only away from \( \pm1 \). Therefore, we separately solve RHP~\ref{W tilde RH problem} locally around \( \pm1 \).  To this end, we denote by \( U_{\pm1} \) disks centered at \( \pm1 \) and, hopefully without causing any ambiguity, we denote their common radius by \( \varrho \) (we also used \( \varrho \) as a common radius of disks \( U_{\pm\ii}\) in Section~\ref{ss:3.1}). In fact, we set
\begin{align}
\label{radius rho}
\varrho = \tfrac12 (\varkappa x)_+^{-\frac12}, \quad (\varkappa x)_+=\max\{1,\varkappa x\}
\end{align}
(the constant \(\frac12\) is there simply to ensure the bound \(\varrho \leq \frac12\) for all choices \( (\varkappa,x)\)). By now, the construction of a local parametrix using parabolic cylinder functions is rather standard. A detailed explanation of this construction in our setting can be found in \cite{IMY}. Below, we only provide the most salient points.

\subsubsection{Model Local Parametrix}

Let \( Pc(z)\) be a \( 2\times2 \) matrix function that solves the following Riemann-Hilbert problem:
\begin{RHP}
\label{parabolic cylinder rhp}
Find a $2 \times 2$ matrix function $Pc(z)$ such that
\begin{enumerate}
    \item $Pc(z)$ is analytic in $\C \setminus \Gamma_{Pc}$, where $\Gamma_{Pc}$ is a contour depicted in Figure~\ref{pic for the K problem};
    \item $Pc(z)$ has continuous traces on $\Gamma_{Pc}\setminus\{0\}$ that satisfy
    \begin{align*}
        Pc_+(z) = Pc_-(z) G_{Pc}(z), \quad z \in \Gamma_{Pc}\setminus\{0\},
    \end{align*}
    where the jump matrices comprising $G_{Pc}(z)$ are as on Figure \ref{pic for the K problem}, and
    \[
    s_1 := A e^{-\ii x+\varkappa x}, \quad  s_2 := -\overline{A} e^{\ii x+\varkappa x};
    \]
    \item it holds that \( Pc(z) = (I + \O(1/z)) z^{-\nu \sigma_3} \) as \( z\to\infty \).
\end{enumerate}
\end{RHP}
\begin{figure}[htbp]
    \centering
    \includegraphics[width=7cm]{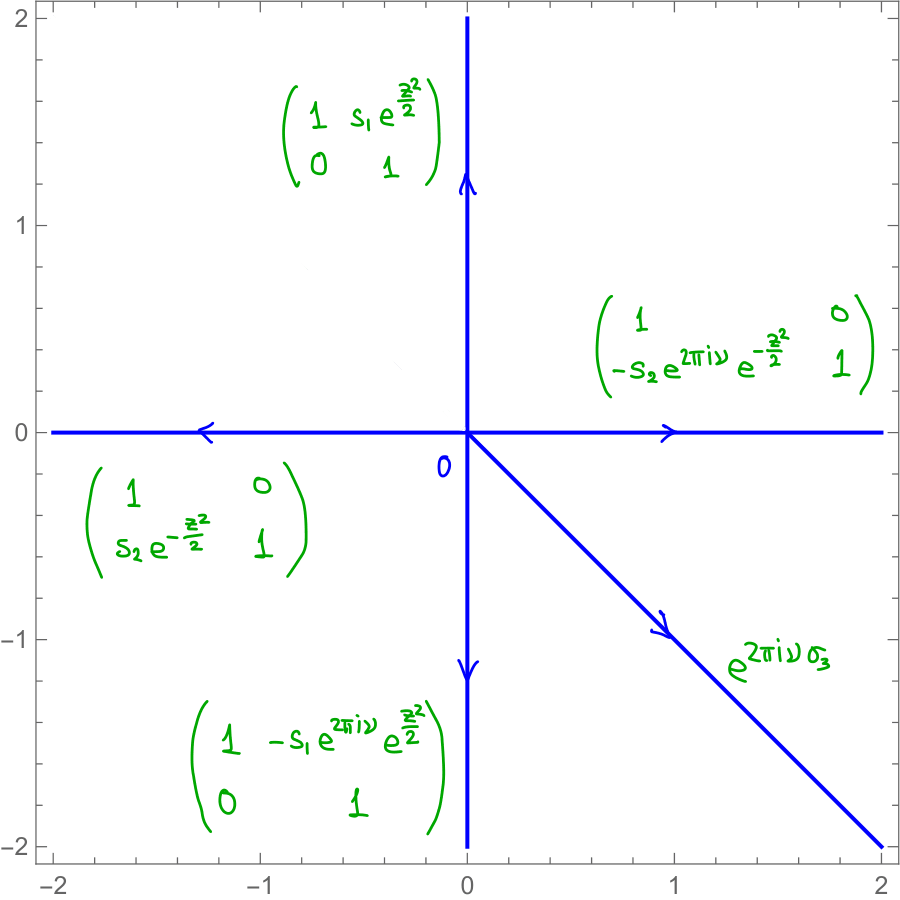}
    \caption{Contour $\Gamma_{Pc}$, which consists of the coordinate axes and the ray \( \arg(z) = -\pi /4 \), and the jump matrices $G_{Pc}(z)$ for RHP~\ref{parabolic cylinder rhp}.}
    \label{pic for the K problem}
\end{figure}

RHP~\ref{parabolic cylinder rhp} can be explicitly solved with the help of parabolic cylinder functions, see for example \cite[Section~4.2.3]{IMY} (and recall that \( B^\R=e^{-\varkappa x}\)). Its explicit form is also given further below in Appendix~\ref{app:PCF}.  There, we explain that the asymptotics of \( Pc(z) \) in RHP~\ref{parabolic cylinder rhp}(3) can be improved to the following formula:
\begin{equation}
\label{matrix K}
\begin{pmatrix}
\displaystyle \sum_{s=0}^{n-1} \frac{a_s^{(-\nu)}}{(-z^2)^s} + \O_n\left(\frac{a_n^{(-\nu)}}{z^{2n}}\right) &  \displaystyle -\frac\alpha z\sum_{s=0}^{n-1} \frac{a_s^{(\nu-1)}}{z^{2s}} + \alpha\O_n\left(\frac{a_n^{(\nu-1)}}{z^{2n+1}}\right) \medskip \\
 \displaystyle \frac \beta z\sum_{s=0}^{n-1}  \frac{a_s^{(-\nu-1)}}{(-z^2)^s} + \beta \O_n\left(\frac{a_n^{(-\nu-1)}}{z^{2n+1}}\right) & \displaystyle \sum_{s=0}^{n-1} \frac{a_s^{(\nu)}}{z^{2s}} + \O_n\left(\frac{a_n^{(\nu)}}{z^{2n}}\right)
\end{pmatrix}
z^{-\nu\sigma_3}
\end{equation}
for \( |z|\geq \max\{5,|\nu+1|\}\) and any non-negative integer \( n \) (in fact, we can and will make different choices for different entries of this matrix), where \( \O_n(\cdot) \) terms are independent of \( \nu \) and \( z \),
\begin{align} 
\label{defn of beta}
a_s^{(\nu)} = \frac{(-1)^s(-\nu)_{2s}}{2^ss!}, \quad \alpha = -e^{-2\pi \ii\nu}\frac{\ii}{s_2}\frac{\sqrt{2\pi}}{\Gamma(\nu)}, \quad \text{and} \quad \alpha\beta = - \nu.
\end{align}
Moreover, upon recalling \eqref{mon data} and the definition of \( \nu \) in \eqref{defn of Pgl} as well as using the identity \( \pi |\Gamma \left(\frac12+ \ii \frac{\varkappa x}{\pi} \right)|^{-2} = \cosh(\varkappa x)\), we obtain that
\begin{align}
\alpha & = e^{\frac{3}{2} \varkappa x} e^{-\ii(x + \arg(p\Gamma(\nu)) + \frac\pi 2)}.
\label{definition of alpha}
\end{align}

Recall the definition of \( \varphi(\lambda) \) in \eqref{defn W}. To set up the correspondence between \( \lambda \) and \( z \) planes, we introduce the following conformal mapping on $U_1$:
\begin{align}
\label{Z <--> lambda}
    z_1(\lambda) := e^{\frac{\pi \ii}{2}} \sqrt{4 \varphi(\lambda) - 2\ii} = e^{\frac{3\pi \ii}{4}}\frac{\lambda - 1}{\sqrt{\lambda}},
\end{align}
where we fix the branches of the square roots so that \(z_1(\lambda) = e^{\frac{3\pi \ii}{4}} (\lambda - 1)\left( 1 + \O \left(\lambda - 1 \right) \right)\) as \(\lambda \to 1\).
Recall that the contour \( \Gamma_{W^*} \cap U_1 \) consists of $(1-\varrho,1]$ and the curves on which the imaginary part of \( 2\varphi(\lambda) - \ii \) vanishes. Hence, one can easily see that \( x^{1/2} z_1(\lambda )\) maps it into \( \Gamma_{Pc} \) for any \( x>0 \) with the segment $(1-\varrho,1]$ mapped into the ray \( \arg(z) = -\frac\pi4 \).

\subsubsection{Local Parametrix around \( 1 \)}

We now define a local parametrix \( P^{(1)}(\lambda) \) by setting
\begin{align} \label{defn of P1}
P^{(1)}(\lambda) := f(\lambda; x)^{\sigma_3} \sigma_1 Pc\left(x^{\frac12} z_1(\lambda)\right) \sigma_1, \quad \lambda \in U_1,
\end{align}
where \( f(\lambda; x)^{\sigma_3} \) is a holomorphic prefactor given by
\begin{align} 
\label{defn of f}
    f(\lambda; x) := \left( \frac{\lambda - 1}{\lambda + 1} \right)^{\nu} \left( x^{\frac12} z_1(\lambda) \right)^{-\nu} = x^{-\frac12\nu} e^{-\frac{3\pi \ii}{4}\nu} \frac{\lambda^{\frac12\nu}}{(\lambda+1)^\nu}
\end{align}
and all the branches above are principal. Then, \( P^{(1)}(\lambda) \) is a solution of the following Riemann-Hilbert problem.
\begin{RHP} 
\label{RHP P1}
Find a $2 \times 2$ matrix function $P^{(1)}(\lambda)$ such that
\begin{enumerate}
    \item $P^{(1)}(\lambda)$ is analytic in $U_1 \setminus \Gamma^*$;
    \item one-sided traces $P^{(1)}_\pm(\lambda)$ are continuous on \( U_1\cap\Gamma^* \) and satisfy
    \begin{align*}
        P^{(1)}_+(\lambda) = P^{(1)}_+(\lambda) G_{W^*}(\lambda);
    \end{align*}
    \item it holds for \( \lambda\in \partial U_1\) that
\[
P^{(1)}(\lambda) = T^{(1)}(\lambda)\left( I + \frac1x \begin{pmatrix} O\big( \varrho^{-6}\big) &  \O \big( \varrho^{-3} \big) \medskip \\  \O\big( \varrho^{-7} \big)  & \O\big( \varrho^{-6} \big) \end{pmatrix} \right)\left( \frac{\lambda-1}{\lambda+1} \right)^{\nu\sigma_3},
\]
where the constants in \( \O(\cdot) \) terms are absolute, and
\[
T^{(1)}(\lambda) = \begin{pmatrix}
        1 & 0 \\
        \displaystyle \frac{f^*(\lambda;x)}{\lambda-1} & 1
    \end{pmatrix}, \quad f^*(\lambda;x) = -\frac{\alpha x^{-\frac12}}{f^2(\lambda;x)} \frac{\lambda - 1}{z_1(\lambda)}.
\]
\end{enumerate}
\end{RHP}

RHP~\ref{RHP P1}(1) is obvious and RHP~\ref{RHP P1}(2) is straightforward to check. Thus, only RHP~\ref{RHP P1}(3) requires an explanation. First of all, we get from \eqref{D5} and \eqref{radius rho} that
\[
x^{\frac12} \varrho \geq \tfrac12 D_5^{\frac16} x^{\frac13}.
\]
Since \( |z_1(\lambda)|\sim \varrho \) on \( \partial U_1 \) and we are interested in \( Pc(x^{\frac12} z_1(\lambda)) \), the error formula \eqref{matrix K} applies on \( \partial U_1 \). It is easy to see from \eqref{matrix K} that
\[
f^{\sigma_3}\sigma_1 Pc(z;\alpha,\beta) = \sigma_1 Pc\big(z;\alpha f^{-2},\beta f^2\big) f^{-\sigma_3},
\]
where we write \( Pc(z)=Pc(z;\alpha,\beta) \) to emphasize the dependence on the parameters \( \alpha,\beta\). Clearly, it also holds that
\[
f(\lambda;x)^{-\sigma_3} \big(x^{\frac12}z_1(\lambda)\big)^{-\nu\sigma_3} \sigma_1 = \sigma_1\left( \frac{\lambda-1}{\lambda+1} \right)^{\nu\sigma_3}.
\]
Given \eqref{radius rho}, one can readily check that
\[
\max\left\{\big|a_1^{(\nu)}\big|,\big|a_1^{(\nu-1)}\big|,\big|a_1^{(-\nu)}\big| \right\} \leq (\varkappa x)_+^2 = (2\varrho)^{-4}.
\]
Then, by using \eqref{matrix K} with \( n=0 \) in (2,1) entry and \( n = 1 \) in all other entries, we get that
\[
P^{(1)}(\lambda) = \left( T^{(1)}(\lambda) + \begin{pmatrix} \O\big(x^{-1}\varrho^{-6}\big) &  \O \big( \beta f^2 x^{-\frac12}\varrho^{-1} \big) \medskip \\  \O\big(\alpha f^{-2}x^{-\frac32}\varrho^{-7} \big)  & \O\big( x^{-1}\varrho^{-6} \big) \end{pmatrix} \right)\left( \frac{\lambda-1}{\lambda+1} \right)^{\nu\sigma_3}
\]
for \( \lambda\in \partial U_1 \), where the constants in \( \O(\cdot) \) terms are absolute. Next, we get from the definition of \( \nu \) in \eqref{defn of Pgl} and from \eqref{defn of f} that
\[
|f(x,\lambda)|^2 = x^{-\frac12}e^{\frac32\varkappa x-\frac1\pi \varkappa x(\arg(\lambda)-2\arg(1+\lambda))} \left| \frac{\sqrt\lambda} {\lambda+1} \right|.
\]
When \( \lambda\in\partial U_1\), i.e., \( \lambda = 1 + \varrho e^{\ii\theta} \), we have that
\begin{align}
\arg(\lambda)-2\arg(1+\lambda) & = \arctan\left( \frac{\varrho \sin\theta}{1 + \varrho \cos\theta}\right) - 2\arctan\left( \frac{\varrho \sin\theta}{2 + \varrho \cos\theta}\right) \\ & = -\frac{\varrho^2\sin\theta\cos\theta}{(1+\varrho\cos\theta)(2+\varrho\cos\theta)} + \O(\varrho^3) = \O(\varrho^2),
\end{align}
where the constant in \( \O(\cdot) \) term is absolute. Furthermore, we get from \eqref{defn of beta} and \eqref{definition of alpha} that \( |\alpha| = e^{\frac32\varkappa x}\) and \( |\beta| = |\nu|e^{-\frac32\varkappa x} \). Hence,
\[
P^{(1)}(\lambda) = \left( T^{(1)}(\lambda) + \frac1x \begin{pmatrix} O\big(\varrho^{-6}\big) &  \O \big(\varrho^{-3} \big) \medskip \\  \O\big(\varrho^{-7} \big)  & \O\big(\varrho^{-6} \big) \end{pmatrix} \right)\left( \frac{\lambda-1}{\lambda+1} \right)^{\nu\sigma_3}
\]
for \( \lambda\in \partial U_1 \), where we used \eqref{radius rho}. We also have that \( |f^*(\lambda;x)| \) is bounded above by an absolute constant on \( \partial U_1 \). Hence, the \( (2,1) \) entry of \( T^{(1)}(\lambda) \) is of size \( \varrho^{-1} \). This observation and the above asymptotic formula immediately yield RHP~\ref{RHP P1}(3).

\subsubsection{Local Parametrix around \( -1 \)}

It readily follows from \eqref{Wtilde jump symmetry} that a desired local parametrix in \( U_{-1} \) can be defined as
\begin{align} 
\label{defn of P-1}
P^{(-1)}(\lambda) := \sigma_1 P^{(1)}(-\lambda) \sigma_1, \quad \lambda \in U_{-1}.
\end{align}
Then, \( P^{(-1)}(\lambda) \) is a solution of the following Riemann-Hilbert problem.
\begin{RHP} 
\label{RHP P-1}
Find a $2 \times 2$ matrix function $P^{(-1)}(\lambda)$ such that
\begin{enumerate}
    \item $P^{(-1)}(\lambda)$ is analytic in $\in U_{-1} \setminus \Gamma^*$;
    \item one-sided traces $P^{(-1)}_\pm(\lambda)$ are continuous on \( U_{-1}\cap\Gamma^* \) and satisfy
    \begin{align*}
        P^{(-1)}_+(\lambda) = P^{(-1)}_+(\lambda) G_{W^*}(\lambda);
    \end{align*}
    \item it holds for \( \lambda\in \partial U_{-1} \) that
\[
P^{(-1)}(\lambda) = T^{(-1)}(\lambda)\left( I + \frac1x \begin{pmatrix} O\big( \varrho^{-6}\big) &  \O \big( \varrho^{-7} \big) \medskip \\  \O\big( \varrho^{-3} \big)  & \O\big( \varrho^{-6} \big) \end{pmatrix} \right)\left( \frac{\lambda-1}{\lambda+1} \right)^{\nu\sigma_3},
\]
where the constants in \( \O(\cdot) \) terms are absolute, and \( T^{(-1)}(\lambda) = \sigma_1 T^{(1)}(-\lambda)\sigma_1 \).
\end{enumerate}
\end{RHP}

\subsection{Small Norm Problem}

We look for a solution of RHP~\ref{W tilde RH problem} in the form
\begin{equation}
\label{defn of Tstar}
    W^* (\lambda) = V^*(\lambda)\begin{dcases}
        T^{(\pm1)}(\lambda)^{-1}P^{(\pm1)}(\lambda), & \lambda \in U_{\pm 1}, \\
        \left( \frac{\lambda - 1}{\lambda + 1} \right)^{\nu \sigma_3}, & \text{otherwise}.
    \end{dcases}
\end{equation}
Then, \( V^*(\lambda)\) must solve the following meromorphic Riemann-Hilbert problem.

\begin{RHP}
\label{V star rhp}
Find a $2 \times 2$ matrix function $V^*(\lambda)$ such that
\begin{enumerate}
    \item $V^*(\lambda)$ is analytic in $\C \setminus (\Gamma_V \cup \{\pm 1\})$, where \( \Gamma_V\) is depicted on Figure~\ref{pic for V problem};
    \item one-sided traces $V^*_\pm(\lambda)$ exist a.e. on \( \Gamma_V \), belong to \( L^2(\Gamma_V) \), and satisfy
    \begin{align*}
        V^*_+(\lambda) = V^*_-(\lambda) G_{V^*}(\lambda), \quad \lambda \in \Gamma_V,
    \end{align*}
    where
    \begin{align}
        G_{V^*}(\lambda) = \begin{dcases}
              \left( \frac{\lambda - 1}{\lambda + 1} \right)^{\nu \sigma_3} G_{W^*}(\lambda) \left( \frac{\lambda - 1}{\lambda + 1} \right)^{-\nu \sigma_3}, & \lambda\in \Gamma_V \setminus(\partial U_1\cup\partial U_{-1}), \\
             T^{(\pm1)}(\lambda)^{-1}P^{(\pm1)}(\lambda)\left( \frac{\lambda - 1}{\lambda + 1} \right)^{-\nu \sigma_3}, & \lambda \in \partial U_{\pm1};
        \end{dcases}
    \end{align}
    \item it holds that
    \begin{align*}
        V^*(\lambda) = \begin{cases}
            I + \O(1/\lambda) & \text{as} \quad \lambda\to\infty,\\
            P_0 \, \sigma_1 \sigma_3 \, \ii^{-\sigma_3} (I + \O(\lambda)) & \text{as} \quad \lambda \to 0;
        \end{cases}
    \end{align*}
        \item  $V^*(\lambda) = (V^{*(1)}(\lambda), V^{*(2)}(\lambda))$ has simple poles at $\lambda = \pm 1$ with residues
    \begin{align}
        \underset{\lambda = 1}{\res} \, V^{*(1)}(\lambda) &= -\alpha^* \, V^{*(2)}(1) \label{residue relation 1},\\
        \underset{\lambda = -1}{\res} \, V^{*(2)}(\lambda) &= \alpha^* \, V^{*(1)}(-1), \label{residue relation 2}
    \end{align}
    where \( \alpha^* := -f^*(1;x) \).
\end{enumerate}
\end{RHP}

\begin{figure}[htbp]
    \centering
    \includegraphics[width=6cm]{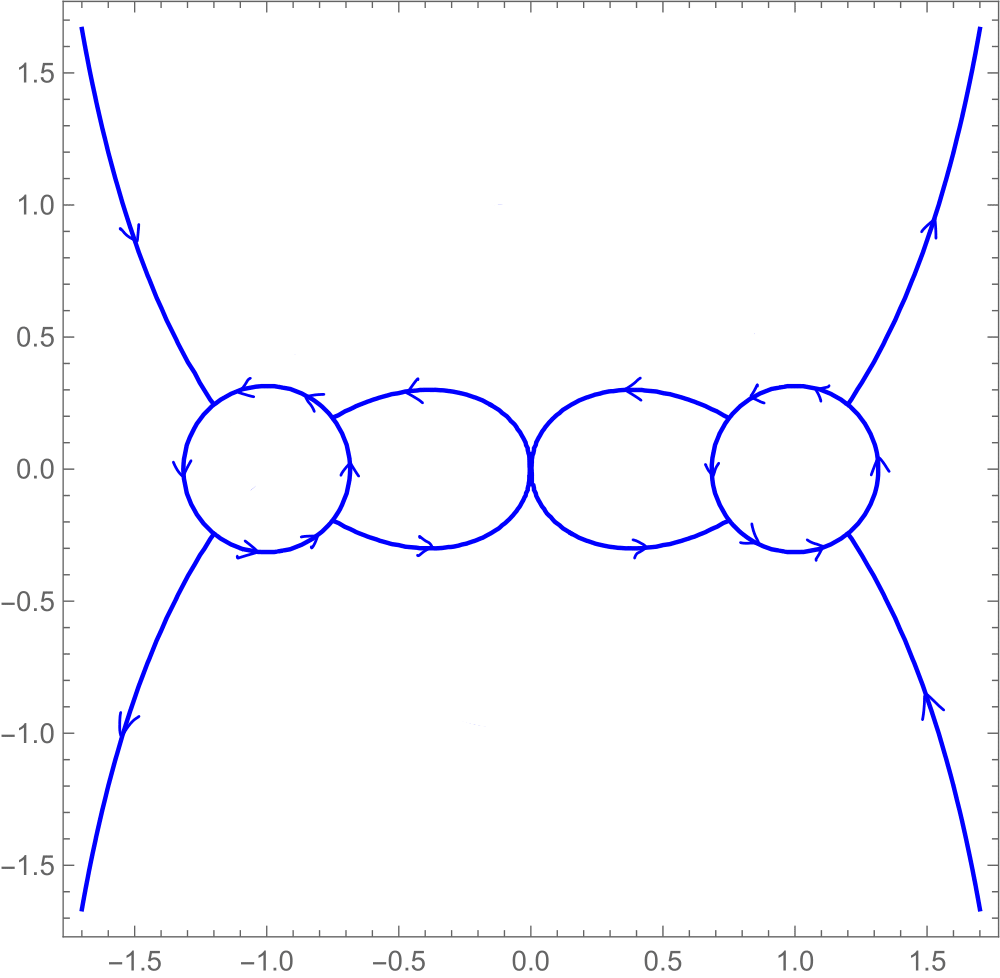}
    \caption{The contour $\Gamma_V$ for RHP~\ref{V star rhp}.}
    \label{pic for V problem}
\end{figure}

Following dressing technique of \cite{BothnerIts}, we look for a solution of RHP~\ref{V star rhp} in the form
\begin{align}
    V^*(\lambda) = (\lambda I + B) V(\lambda) \begin{pmatrix}
        \lambda - 1 & 0 \\
        0 & \lambda + 1
    \end{pmatrix}^{-1},
    \label{defn of matrix V}
\end{align}
where the matrix $B$ will be specified shortly and $V(\lambda)$ solves the following (small norm) Riemann-Hilbert problem.
\begin{RHP}
\label{RHP V}
Find a $2 \times 2$ matrix function $V(\lambda)$ such that
\begin{enumerate}
    \item $V(\lambda)$ is analytic in $\C \setminus \Gamma_V$;
    \item one-sided traces $V_\pm(\lambda)$ exist a.e. on \( \Gamma_V \), belong to \( L^2(\Gamma_V) \), and satisfy
    \begin{align*}
        V_+(\lambda) = V_-(\lambda) G_V(\lambda), \quad \lambda \in \Gamma_L,
    \end{align*}
    where
    \begin{align}
        G_V(\lambda) = \begin{pmatrix}
        \lambda - 1 & 0 \\
        0 & \lambda + 1
    \end{pmatrix}^{-1} G_{V^*}(\lambda)
    \begin{pmatrix}
        \lambda - 1 & 0 \\
        0 & \lambda + 1
    \end{pmatrix};
    \end{align}
    \item it holds that \( V(\lambda) = I + \O(1/\lambda) \) as \( \lambda\to\infty \).
\end{enumerate}
\end{RHP}

Indeed, let us show that this is a small norm problem. For \( \lambda\in\partial U_1\), we readily get from RHP~\ref{V star rhp}(2) and RHP~\ref{RHP P1}(3) that
\begin{align}
G_V(\lambda) & = I + \frac1x \begin{pmatrix} O\big( \varrho^{-6}\big) &  \frac{\lambda+1}{\lambda-1}\O \big( \varrho^{-3} \big) \medskip \\  \frac{\lambda-1}{\lambda+1} \O\big( \varrho^{-7} \big)  & \O\big( \varrho^{-6} \big) \end{pmatrix}  \\ 
& = I + \frac1x \begin{pmatrix} O\big( \varrho^{-6}\big) &  \O \big( \varrho^{-4} \big) \medskip \\  \O\big( \varrho^{-6} \big)  & \O\big( \varrho^{-6} \big) \end{pmatrix} = I + \O\left(\big(x\varrho^6\big)^{-1}\right),
\label{GV circles}
\end{align}
where the constants in \( \O(\cdot) \) terms are absolute. Clearly, the same estimate holds on \( \partial U_{-1} \).

Next, let us consider the unbounded subarc of \( \Gamma_V \) that lies in the first quadrant. On it,  \( G_{W^*}(\lambda)=S_{L_1}^*(\lambda) \), see Figures~\ref{W tilde problem pic} and~\ref{pic for V problem}. We shall call it \( \Gamma_{L_1} \). There, we have that
\[
G_V(\lambda) = I -\overline A e^{x(\varkappa + 2\varphi(\lambda))} \left( \frac{\lambda - 1}{\lambda + 1} \right)^{2\nu-1} \begin{pmatrix} 0 & 1 \\ 0 & 0 \end{pmatrix}
\]
by \eqref{SL1 star}. Recall \eqref{real part varphi}. One can easily check that the unbounded arc of the stationary contour to which \( \Gamma_{L_1} \) belongs can be parametrized as
\begin{align}
\label{param Gamma L1}
\lambda=\xi+\ii\eta, \quad \xi = s+1, \quad \eta = s\sqrt{\frac{1+s}{1-s}}, \quad s\in(0,1).
\end{align}
These relations easily yield that
\[
\Im(2\varphi(\lambda)) = 1 \quad \text{and} \quad \Re(2\varphi(\lambda)) = -\frac{s^2}{\sqrt{1-s^2}}  
\]
as well as that
\[
\frac{\lambda - 1}{\lambda + 1} = \frac{s+\ii s\sqrt{1-s^2}}{2-s^2} = \frac{s}{\sqrt{2-s^2}} \exp\left\{ \ii \arctan\big(\sqrt{1-s^2} \big)\right\}.
\]
Therefore,
\[
\left| e^{x(\varkappa + 2\varphi(\lambda))} \left( \frac{\lambda - 1}{\lambda + 1} \right)^{2\nu-1} \right| = \exp\left\{\varkappa x\left(1-\frac2\pi \arctan\big(\sqrt{1-s^2} \big)\right) - \frac{xs^2}{\sqrt{1-s^2}}\right\}.
\]
Recall that \( |\lambda-1|\geq \varrho \) for \( \lambda\in\Gamma_{L_1} \). Hence, it is sufficient for our purposes to consider \( s\in(\frac14 \varrho,1) \). In this case, due to monotonicity in \( s \), we get that
\begin{align}
\left| e^{x(\varkappa + 2\varphi(\lambda))} \left( \frac{\lambda - 1}{\lambda + 1} \right)^{2\nu-1} \right| & \leq \exp\left\{ \frac{\varkappa x}2 - \frac{x\varrho^2}{16}\right\} \leq \exp\left\{\frac1{8\varrho^2} - \frac{x\varrho^2}{16} \right\} \\
& = \exp\left\{- \frac{x\varrho^2}{16} \left( 1 - \frac{2}{x\varrho^4} \right) \right\} \leq \exp\left\{ -\frac{x\varrho^2}{32}\right\},
\end{align}
where the last inequality is valid because
\[
\frac{2}{x\varrho^4} = 2^5 \frac{(\varkappa x)_+^2}x \leq 2^5D_5^{-\frac23}x^{-\frac13} \leq \frac12
\]
by \eqref{D5} and \eqref{radius rho} with the last inequality being true when \( D_5 \) is taken to be sufficiently large. Moreover, since
\begin{align}
\label{param Gamma L1 diff}
|d\lambda| = \frac{\sqrt{2-s^2}}{\sqrt{1+s}}\frac{ds}{(1-s)^{3/2}}
\end{align}
on \( \Gamma_{L_1} \), we similarly have that
\begin{multline}
\int_{\Gamma_{L_1}} \left| e^{x(\varkappa + 2\varphi(\lambda))} \left( \frac{\lambda - 1}{\lambda + 1} \right)^{2\nu-1} \right|^2 |d\lambda| \leq \sqrt 2 e^{\varkappa x}\int_0^1 \exp\left\{- \frac{x\varrho^2}{16\sqrt{1-s}}\right\}\frac{ds}{(1-s)^{3/2}} \\
 = \frac{32\sqrt 2}{x\varrho^2} \exp\left\{ \varkappa x - \frac{x\varrho^2}{16}\right\} \leq \frac{32\sqrt 2}{x\varrho^2} \exp\left\{ -\frac{x\varrho^2}{32}\right\},
\end{multline}
where, again, the last inequality holds if we take \( D_5 \) sufficiently large. As \( |A|^2\leq 2 \) by \eqref{mon data} and \( x\varrho^2>1 \) by \eqref{D5} and \eqref{radius rho}, we get that
\begin{align}
\label{est on Gamma L1}
\|G_V-I\|_{L^2(\Gamma_{L_1})\cap L^\infty(\Gamma_{L_1})} = \O\left(e^{-\frac1{32} x\varrho^2}\right)
\end{align}
where the constant in \( \O(\cdot) \) is absolute and the estimate holds for all \( (\varkappa,x) \) satisfying \eqref{D5} as long as \( D_5 \) is sufficiently large.

Let \( \Gamma_{R_1} \) be subarc of \( \Gamma_V \) on which \( G_{W^*}(\lambda) = S_{R_1}^*(\lambda) \), see Figures~\ref{W tilde problem pic} and~\ref{pic for V problem}. On it, we have that
\[
G_V(\lambda) = I + A e^{-x(\varkappa + 2\varphi(\lambda))} \left( \frac{\lambda - 1}{\lambda + 1} \right)^{1-2\nu} \begin{pmatrix} 0 & 0 \\ 1 & 0 \end{pmatrix}.
\]
Similarly to \( \Gamma_{S_1} \), \( \Gamma_{R_1} \) can be parametrized as
\begin{align}
\label{param Gamma V}
\lambda=\xi+\ii\eta, \quad \xi = s+1, \quad \eta = (-s)\sqrt{\frac{1+s}{1-s}}, \quad s\in(-1,0).
\end{align}
It now holds that
\[
\Im(2\varphi(\lambda)) = 1 \quad \text{and} \quad \Re(2\varphi(\lambda)) = \frac{s^2}{\sqrt{1-s^2}}  
\]
as well as that
\[
\frac{\lambda - 1}{\lambda + 1} = \frac{s-\ii s\sqrt{1-s^2}}{2-s^2} = \frac{|s|}{\sqrt{2-s^2}} \exp\left\{ \ii \left(\pi-\arctan\big(\sqrt{1-s^2} \big)\right)\right\}.
\]
Therefore, we again get that
\begin{multline}
\label{est GV}
\left| e^{-x(\varkappa + 2\varphi(\lambda))} \left( \frac{\lambda - 1}{\lambda + 1} \right)^{1-2\nu} \right| = \\ \exp\left\{\varkappa x\left(1-\frac2\pi \arctan\big(\sqrt{1-s^2} \big)\right) - \frac{xs^2}{\sqrt{1-s^2}}\right\}  \leq \exp\left\{ -\frac{x\varrho^2}{32}\right\}
\end{multline}
(we do not need a separate \( L^2\) estimate as \( \Gamma_{R_1} \) is of finite length). Since the estimates in the remaining three quadrants are essentially identical, we can now conclude that
\begin{align}
\label{GV-I is small}
\|G_V-I\|_{L^2(\Gamma_V)\cap L^\infty(\Gamma_V)} = \O\left(\big(x\varrho^6\big)^{-1}\right) = \O\left(D_5^{-1}\right)
\end{align}
by \eqref{D5} and \eqref{radius rho}, where the constant in \( \O(\cdot) \) is absolute. That is, RHP~\ref{RHP V} is indeed a small norm problem and respectively is uniquely solvable for \( D_5 \) sufficiently large. Here, one needs to again observe that \( \Gamma_V \) is in fact dependent on \( (\varkappa x)_+\). Hence, to see that a single constant \( D_5 \) is sufficient (and respectively the constant in \( O(\cdot) \) term in \eqref{sing rho - I has a small norm} is absolute) we use uniform boundedness of the norms of the corresponding Cauchy operators. For the explanation of this fact see again Appendix~\ref{s:co}.

\subsection{Linear Prefactor}

The goal of this subsection is to find matrix \( B \) so that \eqref{defn of matrix V} solves RHP~\ref{V star rhp}. This matrix is algebraically determined by residue conditions \eqref{residue relation 1} and \eqref{residue relation 2}. It readily follows from \eqref{defn of matrix V} that
\begin{align*}
    \left(V^{*(1)}(\lambda), V^{*(2)}(\lambda)\right) = (\lambda I + B) \left( \frac{1}{\lambda - 1}V^{(1)}(\lambda), \frac{1}{\lambda + 1} V^{(2)}(\lambda)\right),
\end{align*}
where we denote by \( V^{(1)}(\lambda) \) and \( V^{(2)}(\lambda) \) the first and second columns of \( V(\lambda) \), respectively. Recall respective definitions of \( G_{V^*}(\lambda) \) and \( G_V(\lambda) \) in RHP~\ref{V star rhp}(2) and RHP~\ref{RHP V}(2). Symmetry relations \eqref{Wtilde jump symmetry} and the connection formula \eqref{defn of P-1} yield that
\[
G_V(\lambda) = \sigma_1 G_V(-\lambda) \sigma_1.
\]
As a solution of a small norm Riemann-Hilbert problem, \( V(\lambda) \) is unique and therefore \( V(\lambda) = \sigma_1 V(-\lambda) \sigma_1 \). Write \( V(\lambda) = I + E(\lambda) \). Then,
\begin{equation}
    \label{symmerries of E}
    E(\lambda) = \sigma_1 E(-\lambda) \sigma_1.
\end{equation}
Write \( E(\lambda) = (E^{(1)}(\lambda),E^{(2)}(\lambda))\). Then, we get from \eqref{residue relation 1} and \eqref{residue relation 2} that
\begin{align}
\label{solve for B}
\begin{dcases}
    (I + B) \left( \mathbf{e}_1 + E^{(1)}(1) \right) = - \tfrac12 \alpha^* (I + B) \left( \mathbf{e}_2 + E^{(2)}(1) \right), \\
    (-I + B) \left( \mathbf{e}_2 + E^{(2)}(-1) \right) = -\tfrac12\alpha^* (-I + B) \left( \mathbf{e}_1 + E^{(1)}(-1) \right),
\end{dcases}
\end{align}
where \(\mathbf{e}_1\) and \( \mathbf{e}_2 \) are the standard coordinate vectors in \( \R^2 \). Thus, symmetry \eqref{symmerries of E} at \( \lambda=1\) now implies that \( \)
\[
b_{11} = - b_{22} \quad \text{and} \quad b_{12} = - b_{21},  \quad B = \begin{pmatrix}
        b_{11} & b_{12}\\
        b_{21} & b_{22}
    \end{pmatrix}.    
\]
It follows from \eqref{definition of alpha}, \eqref{Z <--> lambda}, and \eqref{defn of f} that
\begin{align}
\label{alpha star}
\alpha^* = -f^*(1;x) = -2\ii e^{-\ii(x - \kappa \ln x + \phi)} \quad \Rightarrow \quad |\alpha^*|=2,
\end{align}
where $\phi = -\kappa \ln 4 + \arg (p\,\Gamma(\nu))$ and one needs to recall that \( \kappa\pi=\varkappa x \). Hence, we can further deduce that
\begin{align}
\label{bees}
    b_{11} = - b_{22} = - \frac{4 + (\alpha^*)^2 + o(1)}{4 - (\alpha^*)^2 + o(1)}, \quad b_{12} = -b_{21} = \frac{ 4\alpha^* + o(1)}{4 - (\alpha^*)^2 + o(1)}
\end{align}
provided the denominator above is non-zero, where \( o(1) \) is as big as the size of \( E(1) \), which is estimated in the next subsection. 

\subsection{Asymptotic Analysis}

By the small norm theorem, the unique solution $V(\lambda)$ of RHP~\ref{RHP V} is given by
\begin{align}
    V(\lambda) &= I + \frac{1}{2 \pi \ii} \int_{\Gamma_V} \frac{\rho(\lambda') (G_V (\lambda') - I)}{\lambda' - \lambda} d\lambda' \\
    & =  I + \frac{1}{2 \pi \ii} \int_{\Gamma_V} \frac{G_V (\lambda') - I}{\lambda' - \lambda} d\lambda' + \frac{1}{2 \pi \ii} \int_{\Gamma_V} \frac{(\rho(\lambda') -I ) (G_V (\lambda') - I)}{\lambda' - \lambda} d\lambda'
     \label{sing singular integral equation}
\end{align}
for \( \lambda \notin \Gamma_L \), where $\displaystyle{\rho(\lambda') := V_-(\lambda^\prime)}$, $\lambda' \in \Gamma_V$, which itself is a solution of the corresponding integral equation. Exactly as in \eqref{gl rho_Y - I has a small norm}, the small norm theorem and \eqref{GV-I is small} imply that
\begin{align}
\label{sing rho - I has a small norm}
    \|\rho - I\|_{L^2(\Gamma_L)} = \O\left(\big(x\varrho^6\big)^{-1}\right),
\end{align}
where the constant in \( \O(\cdot) \) is absolute. Let, for brevity, \( J_V(\lambda) = G_V (\lambda) - I\). Now, observe that
\begin{align}
\left| \frac1{2\pi\ii}\int_{\Gamma_V} \frac{J_V (\lambda')}{\lambda' - 1} d\lambda' \right|  \leq & \frac1{2\pi} \int_0^{2\pi}|J_V (1+\varrho e^{\ii t})|dt + \frac\varrho{2\pi} \int_0^{2\pi}|J_V (-1+\varrho e^{\ii t})| dt \\
& + \left( \frac1{2\pi} \int_{\Gamma^*\cap \Gamma_V} \frac{|d\lambda'|}{|\lambda'-1|^2} \right)^{1/2}\left( \frac1{2\pi}\int_{\Gamma^*\cap \Gamma_V} |J_V(\lambda')|^2 |d\lambda'| \right)^{1/2},
\end{align}
where \( \Gamma^*\cap \Gamma_V = \Gamma_V\setminus \partial(U_1\cup U_{-1})\). Using \eqref{param Gamma L1} and \eqref{param Gamma L1 diff} one can easily show that the integral of \( |\lambda-1|^{-2} \) is of order \( \varrho^{-1} \). Hence, we get from \eqref{GV circles} and similarly to \eqref{est on Gamma L1} we can conclude that
\[
\left| \frac1{2\pi\ii}\int_{\Gamma_V} \frac{J_V (\lambda')}{\lambda' - 1} d\lambda' \right| = \O\left(\big(x\varrho^6\big)^{-1}\right) + \O\left(e^{-\frac1{32} x\varrho^2}\right) = \O\left(\big(x\varrho^6\big)^{-1}\right).
\]
It is not hard to see that the second integral in \eqref{sing singular integral equation} can be estimated similarly with the help of Cauchy-Schwarz inequality and of \eqref{sing rho - I has a small norm} to yield
\begin{align}
\label{size of E1}
E(1) = V(1) - I = \O\left(\big(x\varrho^6\big)^{-1}\right).
\end{align}

Similarly, we now need to estimate the size of \( V(0) \).
Using \eqref{param Gamma V} and similar parametrizations of the remaining three parts of \( \Gamma_V \) around the origin, one can verify that \( \Gamma_V \) is tangential to the imaginary axis at the origin and that \( J_V(\lambda) \) vanishes exponentially there. Indeed, we can see from \eqref{param Gamma V} and \eqref{est GV} that
\[
\left| \frac1\lambda e^{-x(\varkappa + 2\varphi(\lambda))} \left( \frac{\lambda - 1}{\lambda + 1} \right)^{1-2\nu} \right| \leq \sqrt{\frac2\xi}\exp\left\{ x^{\frac13} - \frac{x}{2\sqrt\xi}\right\} \leq 2\sqrt 2 \exp\left\{x^{\frac13} -x\right\}
\]
for \( \lambda=\xi+\ii\eta \in \Gamma_{R_1} \) for all \( \xi\in(0,\frac14) \). Since \( \varrho\leq1 \) and \( x\varrho^6\geq D_5\) by \eqref{D5} and \eqref{radius rho}, it holds that
\[
e^{x^{\frac13} -x} \leq e^{-\frac12 x} \leq e^{-\frac12 (x\varrho^6)} \leq \big(x\varrho^6\big)^{-1}.
\]
In particular, we can see that from the above vanishing and Plemelj-Sokhotski formulae that \( V_+(0)=V_-(0) \), which means that \( V(0) \) is well defined. Since the estimates of the integrals in \eqref{sing singular integral equation} away from  the origin can be estimated with the help of Cauchy-Schwarz inequality exactly as was done for \( V(1) \), we deduce that
\begin{align}
\label{size of V0}
V(0) = I + \O\left(\big(x\varrho^6\big)^{-1}\right).
\end{align}

It follows from RHP~\ref{V star rhp}(3) and \eqref{defn of matrix V} that
\begin{align*}
    P_0 &= -B V(0)\ii^{\sigma_3} \sigma_1.
\end{align*}
We now get from the definition of \( P_0 \) in RHP~\ref{rhp original}(3) that
\[
e^u = \frac{\cosh(\frac12 u)+\sinh(\frac12 u)}{\cosh(\frac12 u)-\sinh(\frac12 u)} = \frac{(1+o(1)) b_{12} - (1+o(1)) b_{11}}{(1+o(1)) b_{12} + (1+o(1)) b_{11}},
\]
where $o(1)=\O\left(\big(x\varrho^6\big)^{-1}\right)$ by \eqref{size of V0}. Observe that the denominators in both expressions of \eqref{bees} are the same. Hence, we get from \eqref{alpha star} that
\begin{align}
\label{penultimate formula}
e^u = \frac{4\alpha^* + 4+(\alpha^*)^2+o(1)}{4\alpha^* - 4 - (\alpha^*)^2+o(1)} = \frac{1 - \sin(x-\kappa \ln x + \phi) + o(1)}{1+ \sin(x-\kappa \ln x + \phi) + o(1)},
\end{align}
where \( o(1)=\O\big(\big(x\varrho^6\big)^{-1}\big) \) by \eqref{size of E1}, provided the common denominator in \eqref{bees} is non-zero. Taking logarithms on both sides of \eqref{penultimate formula} finishes the proof of \eqref{asymptotic formula 5} when \( \varepsilon=1 \). Assume now that \( \varepsilon=-1 \). Recall that the transformation \( (A, s^{\R}, B^{\R}) \mapsto (-\iota A, -s^{\R}, \iota B^{\R}) \) corresponds to \( p \mapsto - p \) by \eqref{p parametrization}. Respectively, we get from \eqref{kappa phi} that \( \phi \mapsto \phi+\pi \). Therefore, in this case we get from \eqref{symmetries of u} and \eqref{penultimate formula} that
\begin{align}
u\big(x,(A, s^{\R}, B^{\R})\big) & = -\ln  \left( \frac{1 - \sin(x-\kappa \ln x + \phi + \pi) + o(1)}{1+ \sin(x-\kappa \ln x + \phi +\pi) + o(1)} \right) \\
& = \ln  \left( \frac{1 - \sin(x-\kappa \ln x + \phi) + o(1)}{1+ \sin(x-\kappa \ln x + \phi) + o(1)} \right),
\end{align}
which finishes the proof of Theorem~\ref{thm:5}.


\appendix

\section{On Cauchy Operators}
\label{s:co}

Let \( \Sigma \) be a locally rectifiable oriented contour. Given \( \tau\in\Sigma \) and \( \epsilon>0 \), let \( \Sigma(\tau,\epsilon):=\Sigma\cap\{|s-\tau|<\epsilon\} \). Then, for any \( f\in L^1(\Sigma) \), the limit
\[
(\mathcal S_\Sigma f)(\tau) = \lim_{\epsilon\to0} \int_{\Sigma\setminus\Sigma(\tau,\epsilon)} \frac{f(s)}{\tau-s}ds 
\]
exists for almost every \( \tau\in\Sigma \), see \cite[Theorem~4.14]{BottcherKarlovich}. In the main text we denoted by \( \mathcal C \) the Cauchy operator
\[
(Cf)(\lambda) := \frac1{2\pi\ii} \int_\Sigma \frac{f(s)}{s-\lambda}ds
\]
and by \( \mathcal C_\pm \) the boundary Cauchy operators
\[
(\mathcal C_\pm f)(\tau) := \lim_{\lambda\to\tau\in\Sigma^\pm}(Cf)(\lambda),
\]
where \( \Sigma^- \) and \( \Sigma^+ \) are the left-hand and right-hand sides of \( \Sigma \) according to the chosen orientation. It is well-known (Plemelj-Sokhotski formulae) that
\[
2(\mathcal C_\pm f)(\tau) = \pm f(\tau) + (\mathcal S_\Sigma f)(\tau),
\]
see \cite[Equation~(5.72)]{BottcherKarlovich}. Therefore, \( L^2(\Sigma) \)-boundedness of \( \mathcal C_\pm \) operators is equivalent to \( L^2(\Sigma) \)-boundedness of the singular integral operator \( \mathcal S_\Sigma \).

Contour \( \Sigma \) is called a Carleson contour or Alhfors-David contour if
\[
C_\Sigma := \sup_{\tau\in\Sigma}\sup_{\epsilon>0}\epsilon^{-1}|\Sigma(\tau,\epsilon)|<\infty,
\]
where \( |\Sigma(\tau,\epsilon)| \) is the arclength of \( \Sigma(\tau,\epsilon) \). It is a deep and by now classical result of David that \( \mathcal S_\Sigma \) is a bounded operator from \( L^p(\Sigma) \) for any \( p\in(1,\infty)\) into itself if and only if \( \Sigma \) is a Carleson contour, see \cite[Theorem~4.17]{BottcherKarlovich}. Unfortunately, we could not locate in the literature a quantitative version of this result stating that the norm of \( \mathcal S_\Sigma \) depends only on \( p \) and the constant \( C_\Sigma \). Therefore, we need to provide some explanations.

Assume for simplicity that \( \Sigma=\cup_{n=1}^N \Sigma_n \), where each \( \Sigma_n \) is a smooth closed Jordan arc (in particular, a Carleson arc) and any two arcs \( \Sigma_n \) and \( \Sigma_m \), \( n\neq m \), are either disjoint or have only one point in common which is an endpoint for both arcs. Then, the operator \( \mathcal S_\Sigma \) can be viewed as a matrix operator acting on \( \oplus L^p(\Sigma_n) \) with entries
\[
\mathcal S_{nm} = \chi_n \mathcal S_\Sigma \chi_m : L^p(\Sigma_n) \to L^p(\Sigma_m),
\]
where \( \chi_n \) is the characteristic function of \( \Sigma_n \). Since any two arcs \( \Sigma_n \) and \( \Sigma_m \) can be embedded into an unbounded Carleson arc, say \( \Sigma_{nm} \), the norm of \( \mathcal S_\Sigma \) depends on \( N \) and the norms of \( \mathcal S_{\Sigma_{nm}} \) (\cite[Section~4.5]{BottcherKarlovich} is devoted to a more detailed explanation of this reduction in full generality of arbitrary Carleson contours). It is stated in lecture notes \cite[Remark following Theorem~2.24]{uDeift} that if \( \Sigma \) is an unbounded Carleson arc, then
\[
\|\mathcal S_\Sigma \|_{L^p(\Sigma)\to L^p(\Sigma)} \leq \phi_p(C_\Sigma),
\]
where \( \phi_p(t)\geq0 \) is a continuous, increasing function with \( \phi_p(0)\geq0 \) independent of \( \Sigma \). The primary reference in \cite{uDeift} is \cite{BottcherKarlovich}. As we already mentioned, \cite[Corollary~5.42]{BottcherKarlovich} unfortunately is not stated in this quantitative way (unlike the case of Lipschitz graphs, see \cite[Theorem~5.29]{BottcherKarlovich}) and only careful examination of the proof shows that the above claim is indeed correct. Even though the above discussion was concerned only with operators on functions, it can be naturally extended to matrix functions as well.

In the ucoming considerations, we use the following general principle: if \( \sigma_\alpha(t) \), \( t\in[a,b] \), are parametrizations of a family of bounded arcs depending on \( \alpha \) and there exists a constant \( \delta>0 \) such that
\[
|\arg(\sigma_\alpha^\prime(t_1))-\arg(\sigma_\alpha^\prime(t_0))| \leq \tfrac\pi3 \quad \Rightarrow \quad \Re e^{\ii \arg(\sigma_\alpha^\prime(t_1))- \ii \arg(\sigma_\alpha^\prime(t_0))} \geq \tfrac12
\]
for all \( |t_1-t_0|\leq \delta \), then these arcs have uniformly bounded Carleson constants because
\[
\frac{\int_{t_0}^{t_1} |\sigma_\alpha^\prime(t)|dt}{\big| \int_{t_0}^{t_1} \sigma_\alpha^\prime(t)dt \big|} \leq \frac{\int_{t_0}^{t_1} |\sigma_\alpha^\prime(t)|dt}{\big| \int_{t_0}^{t_1} |\sigma_\alpha^\prime(t)| \Re e^{\ii \arg(\sigma_\alpha^\prime(t))- \ii \arg(\sigma_\alpha^\prime(t_0))} dt \big|} \leq 2,
\]
where, trivially, the top integral on the left represents the arclength between \( \sigma(t_0) \) and \( \sigma(t_1) \) while the integral on the bottom represents the distance between these points (it is also easy to show that for bounded arcs it is enough to consider \( t_0,t_1\) that are sufficiently close to each other to determine finiteness of the Carleson constant). Clearly, the above property holds as long as \( \{\arg(\sigma_\alpha^\prime(t))\} \) is a uniformly equicontinuous family of functions on the parametrizing interval. Furthermore, by Arzel\`a-Ascoli theorem, the latter property holds if \( \arg(\sigma_\alpha^\prime(t)) \) converges uniformly to \( \arg(\sigma_{\alpha_0}^\prime(t)) \) when \( \alpha\to\alpha_0\).

We needed bounds on Cauchy operators on four contours: \( \Gamma \) in the proof of Theorem~\ref{thm:1}, see Figure~\ref{Y problem pic}; \( \Gamma_Z \) in the proof of Theorems~\ref{thm:2} and~\ref{thm:3}, see Figure~\ref{hat Z problem pic}; \( \Gamma_R \) in the proof of Theorem~\ref{thm:4}, see Figure~\ref{R problem pic}; and \( \Gamma_V \) in the proof of Theorem~\ref{thm:5}, see Figure~\ref{pic for V problem}. One can readily compute the Carleson constant \( C_\Sigma \) for a circle, which is \( \pi \) (hence, Carleson constant of any subrac of a circle is at most \( \pi \)). Since \( \Gamma \) and \( \Gamma_Z \) separate into straight line segments and subarcs of circles, they are clearly Carleson contours. \( \Gamma_V \) consists of a stationary contour of \( 2\varphi(\lambda) \), see Figure~\ref{defn of W tilde}, and circles of variable radii around \( \pm1 \) (the part of the stationary contour inside the circles must be removed). It is clear that this stationary contour is a Carleson contour (one can even use the explicit expression for \( \varphi(\lambda)\) to estimate the Carleson constant). Hence, the Carleson constant of \( \Gamma_V \) is bounded above independently of the radii of the circles.

It only remains to discuss \( \Gamma_R \). It consists of the imaginary axis, subarcs of the unit circle, subintervals of the real line, arcs \( \gamma_i^{in},\gamma_i^{out} \), and arcs connecting the endpoints of \( \gamma_i^{in},\gamma_i^{out} \) to the real line, see Figure~\ref{R problem pic}. Clearly, only the latter two groups need to be studied. In fact, we only need to explain why the two arcs of \( \Gamma_R\cap \mathcal Q_1 \cap \{|\lambda|<1\} \) have uniformly bounded Carleson constants as the remaining arcs are obtained by various reflections of these two. Recall Section~\ref{sss:conf map}. By construction, the conformal maps
\[
\zeta(\lambda;\alpha) := h_1(\lambda)^{\frac23}, 
\] 
parametrized by \( \alpha\in(0,\frac\pi2) \), form a continuous family. Thus, we only need to study the limits as \( \alpha\to0 \) and \( \alpha\to\frac\pi2\). We concentrate on the latter as the former can be described similarly. These maps take the arcs connecting \( \gamma_1^{in} \) to the real line into horizontal segments. Hence, we can use
\[
\sigma_\alpha(t) = \zeta^{-1}(t;\alpha)
\]
as functions parametrizing these arcs (here, superscript \( -1 \) denotes the inverse function). Since \( V\to0 \) when \( \alpha\to \frac\pi2\), these horizontal segments converge to \( [0,2^{-\frac23}] \) (clearly, the above argument about finiteness of Carleson constants can be easily modified to include changing but converging parametrizing intervals). Observe that
\[
\arg(\sigma_\alpha^\prime(t)) = -\arg(\zeta^\prime(\sigma_\alpha(t);\alpha)).
\]
Thus, we need to show that the above right-hand sides converge uniformly as \( \alpha \to \frac\pi2\). Recall that the functions \( h_1(\lambda) \) themselves converge uniformly, including their derivatives, to \( -\theta(\lambda) - \frac12  \) by \eqref{g and ha} and \eqref{g theta}. From this, we readily deduce the desired conclusion.

The above argument needs a slight modification for the arcs \(\gamma_1^{in} \) as they collapse into a point. Thus, we first need to perform some affine transformations (which, of course, do not change Carleson constants). The functions \( V^{-1}h_1(\lambda) \) from \eqref{defn of h1} map \( \gamma_1^{in} \) into a fixed segment \( e^{-\frac{\pi\ii}8}[0,\frac \pi2\csc(\frac\pi 8)] \). Thus, we parametrize \( \gamma_1^{in} \) by the inverses of these functions. One can readily compute using \eqref{defn of V} that \( (\frac\pi2-\alpha)^2/V \) has a limit as \( \alpha\to\frac\pi2\). Moreover,
\[
\frac1V h_1(\lambda) = \frac{(\overline a+a)^2}V \frac\ii4\int_0^z \frac{\sqrt{((\overline a+a)u+a-\overline a)((\overline a+a)u+2a)}}{((\overline a+a)u+a)^2}\sqrt{u(u+1)}du,
\]
where \( z=\frac{\lambda-a}{\overline a+a} \) and \(\overline a+a=2\sin(\frac\pi2-\alpha)\). When \( \alpha\to\frac\pi2 \), the above functions converge uniformly (including their derivatives) to
\[
c\int_0^z\sqrt{u(u+1)}du
\]
for some constant \( c \), where branch of the square root is chosen so that the integrand is analytic in the vertical strip \( -1<\Re(u)<0 \) (notice that under the above substitution, \( -\overline a \) corresponds to \( -1 \)). Now, we can again appeal to uniform convergence of the parametrizations to conclude uniform boundedness of the Carleson constants.

\section{A Curious Determinantal Identity}

In this appendix, we prove a proposition that was used to derive formula \eqref{prefactor linear system}. The proof below uses a combinatorial identity
\begin{align}
    \label{bimonial identity}
    \sum_{k=0}^m\binom{m+n}k x^k y^{m-k} & = \sum_{k=0}^m \binom{-n}{k} (-x)^k(x+y)^{m-k} \\
    & = \sum_{k=0}^m (-1)^{n-1}\binom{-k-1}{n-1} x^k(x+y)^{m-k},
\end{align}
where the first equality can be found, for instance, in \cite[Equation (5.19)]{MR1397498}, and the second one is a straightforward consequence of the definition of the generalized binomial coefficient used with integer parameters.

\begin{prop}
\label{prop:det}
Let \( T=[t_{j-i}]_{i,j=1}^N \) and \( H=[h_{i+j-1}]_{i,j=1}^N \) be a Toeplitz and a Hankel matrix, respectively, with \( t_k=0 \) when \( k<0 \) and \( h_k=0 \) when \( k>N \) (in all matrices, \( i \) stands for the row index and \( j \) for the column one). Further, let
\[
L = \left[\left(\frac\ell2\right)^{i+j-1}(-1)^{i-1}\binom{-j}{i-1}\right]_{i,j=1}^N
\]
for some constant \( \ell \). Assume that \( \det(T-HL)\neq 0 \). If \( \vec v=(v_1,\ldots,v_N)^\mathsf{T} \) is the solution of \( (T-HL)\vec v = H\vec e_1\), where \( \vec e_1 = (1,0,\ldots,0)^\mathsf{T} \), then
\[
1 + \sum_{i=1}^N v_i\ell^i = \frac{\det(T+HL)}{\det(T-HL)}.
\]
\end{prop}
\begin{proof}
Let \( V \) be the matrix whose first row is given by \( \big(\ell,\ell^2,\ldots,\ell^N\big) \) and whose remaining entries are all zero. Then
\begin{align*}
1 + \sum_{i=1}^N v_i\ell^i &= \vec e_1^\mathsf{T}\big(I + V (T-HL)^{-1} H\big)\vec e_1 \\
 &= \det\big(I + V (T-HL)^{-1} H\big) =\frac{\det(T+H(V-L))}{\det(T-HL)}.
\end{align*}
Thus, we need to show that \( \det(T+H(V-L))=\det(T+HL) \), i.e., that the matrix \( T+HL \) can be obtained by row and column operations from \( T+H(V-L) \). We shall perform the following operations:
\begin{itemize}
    \item for \( i \) from \( N \) to \( 2 \) (starting with the last row and ending with the second row) add \( \ell \) multiple of the \( i \)-th row to the \( (i-1)\)-st one;
    \item for \( j \) from \( N-1 \) to \( 1 \) (starting with the penultimate column and ending with the first one) subtract \( \ell \) multiple of the \( j \)-th column from the \( (j+1)\)-st one.
\end{itemize}
Set \( L_{i,j} := I + \ell E_{i,j} \), where \( E_{i,j} \) is the matrix with all zero entries except for \( (i,j) \)-th entry, which is \( 1 \). Since \(  I - \ell E_{i,j} = L_{i,j}^{-1} \), performing the suggested operations, is of course, equivalent to the multiplication on the left by \( L_{1,2} \cdots L_{N-1,N} \) and on the right by \( L_{N-1,N}^{-1} \cdots L_{1,2}^{-1} \). That is, we would like to show that
\[
L_{1,2} \cdots L_{N-1,N} (T+H(V-L)) L_{N-1,N}^{-1} \cdots L_{1,2}^{-1} = T+HL.
\]

It is straightforward to verify that
\[
\begin{cases}
L_{1,2} \cdots L_{N-1,N} T &= T L_{1,2} \cdots L_{N-1,N}, \\ 
L_{1,2} \cdots L_{N-1,N} H &= H L_{N,N-1}\cdots L_{2,1}.
\end{cases}
\]
For instance, when \( N=3 \), one can readily see that
\begin{multline*}
\begin{pmatrix} 1&\ell&0 \\ 0&1&0 \\ 0&0&1 \end{pmatrix} \begin{pmatrix} 1&0&0 \\ 0&1&\ell \\ 0&0&1 \end{pmatrix} T = \begin{pmatrix}
t_0 & t_1 + \ell t_0 & t_2 + \ell t_1 + \ell^2 t_0 \\
0 & t_0 & t_1 + \ell t_0 \\
0 & 0 & t_0
\end{pmatrix} \\ = T \begin{pmatrix} 1&\ell&0 \\ 0&1&0 \\ 0&0&1 \end{pmatrix} \begin{pmatrix} 1&0&0 \\ 0&1&\ell \\ 0&0&1 \end{pmatrix}.
\end{multline*}
and
\begin{multline*}
\begin{pmatrix} 1&\ell&0 \\ 0&1&0 \\ 0&0&1 \end{pmatrix}\begin{pmatrix} 1&0&0 \\ 0&1&\ell \\ 0&0&1 \end{pmatrix} H = \begin{pmatrix} h_1 + \ell h_2 + \ell^2 h_3 & h_2 + \ell h_3 & h_3 \\ h_2 + \ell h_3 & h_3 & 0 \\ h_3 & 0 & 0 \end{pmatrix} \\ = H \begin{pmatrix} 1&0&0 \\ 0&1&0 \\ 0&\ell&1 \end{pmatrix} \begin{pmatrix} 1&0&0 \\ \ell&1&0 \\ 0&0&1 \end{pmatrix}.
\end{multline*}
Hence, to prove the proposition, we need to show that
\[
L_{N,N-1}\cdots L_{2,1} (V-L) L_{N-1,N}^{-1} \cdots L_{1,2}^{-1} = L.
\]
Observe that
\[
L_{N,N-1}\cdots L_{2,1} V  = \big[\ell^{i+j-1}\big]_{i,j=1}^N.
\]
Therefore, the desired equality can be rewritten as
\[
\big[\ell^{i+j-1}\big]_{i,j=1}^N = L_{N,N-1}\cdots L_{2,1} L + L L_{1,2}\cdots L_{N-1,N}.
\]
The first summand on the right-hand side of the equality above represents successive additions of the \( (i-1) \)-st row multiplied by \( \ell \) to the \( i\)-th one starting with \( i=2 \). Thus, its \( (i,j) \)-th entry is given by
\begin{align*}
    \sum_{k=1}^i \ell^{i-k}\left(\frac\ell2\right)^{k+j-1}(-1)^{k-1}\binom{-j}{k-1} & = \frac{\ell^{i+j-1}}{2^j} \sum_{k=0}^{i-1} \left(-\frac12 \right)^k \binom{-j}{k} \\
    & = \left(\frac\ell2\right)^{i+j-1}\sum_{k=0}^{i-1} \binom{i+j-1}{k},
\end{align*}
where we used the equality of the first and middle terms in \eqref{bimonial identity} with \( x=y=\tfrac12\), \( m=i-1 \), and \( n=j \). On the other hand, the second summand represents successive additions of the \( (j-1) \)-st column multiplied by \( \ell \) to the \( j\)-th one starting with \( j=2 \). Hence, its \( (i,j) \)-th entry is given by
\begin{align*}
    \sum_{k=1}^j \ell^{j-k}\left(\frac\ell2\right)^{k+i-1}(-1)^{i-1}\binom{-k}{i-1} & = \frac{\ell^{i+j-1}}{2^i} \sum_{k=0}^{j-1} (-1)^{i-1} \left(\frac12\right)^k \binom{-k-1}{i-1} \\
    & = \left(\frac\ell2\right)^{i+j-1} \sum_{k=0}^{j-1}\binom{i+j-1}{k},
\end{align*}
where we used the equality of the first and last terms in \eqref{bimonial identity} with \( x=y=\tfrac12\), \( m=j-1 \), and \( n=i \). Since
\[
\sum_{k=0}^{i-1} \binom{i+j-1}{k} + \sum_{k=0}^{j-1}\binom{i+j-1}{i+j-1-k} = \sum_{k=0}^{i+j-1}\binom{i+j-1}{k} = 2^{i+j-1},
\]
the claim of the proposition follows.
\end{proof}

\section{On Elliptic Moduli}
\label{app:elliptic}


Recall \eqref{Elliptic integrals}. It is well known, see \cite[(19.5.1-2)]{NIST:DLMF}, that
\begin{equation}
\label{seriesK}
K(k) = \frac{\pi}{2} \sum_{m=0}^{\infty}\frac{{\left(\tfrac{1}{2}\right)_{m}}{\left (\tfrac{1}{2}\right)_{m}}}{m!\;m!}{k}^{2m} \quad \text{and} \quad E(k) = \frac{\pi}{2} \sum_{m=0}^{\infty}\frac{{\left(-\tfrac{1}{2}\right)_{m}}{\left (\tfrac{1}{2}\right)_{m}}}{m!\;m!}{k}^{2m},
\end{equation}
where the series converge for \( |k|<1 \). Recall also \eqref{kappa as elliptic} and \eqref{ks}.

\begin{lem}
\label{lem:k varkappa}
We have that
\[
\varkappa \leq \kk^2\KK \quad \text{and} \quad 1-\varkappa \leq (\kk^\prime)^2 \KK.
\]
Moreover, it holds as \( \varkappa\to 0 \) that
\[
\kk = 2\left(\frac\varkappa\pi\right)^{\frac12} - \frac12\left(\frac\varkappa\pi\right)^{\frac32} - \frac{5}{16} \left(\frac\varkappa\pi\right)^{\frac52} + \O \left( \varkappa^{\frac72} \right).
\]
\end{lem}
\begin{proof}
By Lemma~\ref{lem:varkappa}, we have that
\[
\varkappa = \kk^2 \KK + (E(\kk)-K(\kk)).
\]
Since \( E(k)\leq \tfrac\pi2 \leq K(k)\), see \eqref{seriesK} and \cite[Equation (19.9.1)]{NIST:DLMF}, the first claim of the lemma follows. On the other hand, we have that
\[
1-\varkappa = 1-E(\kk ) + (\kk^\prime)^2 \KK.
\]
Since \( E(k) \geq 1\), see again \cite[Equation (19.9.1)]{NIST:DLMF}, we get the second claimed bound. Moreover, we readily get from \eqref{seriesK} that
\[
\varkappa = \frac{\pi}{2} \sum_{m=1}^\infty \left( \frac{ \left(\tfrac12\right)_m}{m!} \right)^2 \frac{2m}{(2m - 1)^2} \kk^{2m}.
\]
By the Lagrange inversion theorem, one can write $\kk^2$ as a power series in $\varkappa$. A direct computation yields
\[
\kk^2 = 4\frac{\varkappa}{\pi}  - 2\left(\frac{\varkappa}{\pi}\right)^2 - \left(\frac{\varkappa}{\pi}\right)^3 + \O \left( \varkappa^4 \right).
\]
Since $\sqrt{1 - y} = 1 - \tfrac{1}{2}y - \tfrac{1}{8} y^2 + \O( y^3)$, the last claim of the lemma follows. 
\end{proof}

Recall \eqref{D4} and \eqref{Kalpha}.

\begin{lem}
\label{lem:K alpha}
Assuming \eqref{D4} and \( x\geq e \), there exists an absolute constant \( c_0 \) such that
\[
c_0 D_\delta \leq x\mathcal K_\alpha \leq |s^\R|_+^{-1}e^{\varkappa x} \min\left\{ (\mathcal K_\alpha^*)^2, x^{-1/3} (\mathcal K_\alpha^{**})^2 \right\}.
\]
\end{lem}
\begin{proof}
Let \( k_*\in(0,1) \) be such that \( \ln x < \tfrac1{20}x\) for \( x\geq \tfrac4{k_*}\). It follows from \cite[Equation~(19.9.2)]{NIST:DLMF} that
\[
\KK \leq \frac54 \ln\frac 4{\kk^\prime} \quad \Rightarrow \quad \KK \leq \frac 1{4\kk^\prime}, \quad \kk^\prime \in(0,k_*].
\]
Recall that \( x\geq e\) and \( D_\delta\geq 4D_4 \geq 1 \) (we can always assume that  \( D_4\geq \frac14 \)). By using these estimates, the upper estimate in \eqref{D4}, Lemma~\ref{lem:k varkappa}, and second inequality in the above display, we get that
\[
\frac13\ln \frac1{x} \leq \frac13\ln \left(\frac{D_\delta \ln x}{x}\right) \leq \ln(1-\varkappa) \leq 2\ln \kk^\prime + \ln \KK \leq \ln \frac{\kk^\prime}4.
\]
We now get from Lemma~\ref{lem:k varkappa}, the upper bound in \eqref{D4}, the first estimate of \( \KK \) above, and the preceding inequality that
\[
x \mathcal K_\alpha \geq x \kk^3 \frac{(1-\varkappa)^3}{\KK} \geq \kk ^3 \frac{D_\delta \ln x}{\KK} \geq \frac{4\kk^3}{5} \frac{D_\delta\ln x}{\ln(4/\kk^\prime)} \geq \frac{12 (1-k_*^2)^{\frac32}}{5} D_\delta.
\]
On the other hand, when \( \kk^\prime \in [k_*,1)\), we get from Lemma~\ref{lem:k varkappa}, an estimate \( \KK\geq \tfrac \pi2 \), see \eqref{seriesK}, and the lower bound in \eqref{D4} that
\[
x \mathcal K_\alpha \geq (\kk^\prime)^6 \sqrt{\KK} \, x\varkappa^{\frac32} \geq (\kk^\prime)^6 \sqrt{\frac\pi2} x\varkappa^{\frac32} \geq k_*^6 \sqrt{\frac\pi2} D_\delta,
\]
which finishes the proof of the estimate \( c_0 D_\delta \leq x\mathcal K_\alpha \).

It is a simple computation using \eqref{Kalpha} to verify that the second desired inequality is equivalent to showing that
\[
(\kk\kk^\prime)^2 \leq  x^{-1} e^{\varkappa x-\ln |s^\R|_+} \quad \text{and} \quad (\kk\kk^\prime)^{\frac83} \leq x^{-\frac43}e^{\varkappa x-\ln |s^\R|_+}.
\]
Since \( \kk\kk^\prime\leq \tfrac12\), the lower bound in \eqref{D4} yields the desired inequalities (recall that we always assume that \( x\geq1\) and we always can adjust the constant \( D_4 \) as it was simply a fixed constant sufficiently large for our purposes). 
\end{proof}

Let \( \sd(z,k) \) be the Jacobian elliptic function we used in Theorem~\ref{thm:4}.

\begin{lem} 
\label{lem:sd}
It holds as $k\to 0^+$ for all real \( z \) that
\[
\sd(z, k) = \sin \left( z - \frac{k^2 z}{4} \right) + \O \left( k^4 z \right).
\]
\end{lem}
\begin{proof}
We get from \eqref{seriesK} that
\[
K(k) = \frac{\pi}{2} \left( 1 + \frac{k^2}{4} + \O \big( k^4 \big) \right) \quad \text{and} \quad K(k^\prime) = -\ln k + 2 \ln 2 + \O \big( k^2 \ln k \big)
\]
as $k \to 0$. Hence, it holds that
\begin{align}
q^{1/2} &= \exp\left\{-\frac\pi2\frac{K(k^\prime)}{K(k)}\right\} = \exp\left\{\ln k -2\ln2 + \O \big( k^2 \ln k \big)\right\} \\
& = \frac k4\big(1 + \O(k^2\ln k) \big) \quad \text{as} \quad k \to 0
\end{align}
(\( q \) is the nome corresponding to the elliptic modulus \( k \); when \( k\in(0,1) \), it holds that \( q\in(0,1)\)). In particular, we have that
\[
\frac{2\pi q^{\frac12}}{K(k)\, k\,k^{\prime}} = 1 + \O(k^2\ln k) \quad \text{as} \quad k \to 0^+.
\]
It is known, see \cite[Equation (22.11.5)]{NIST:DLMF}, that
\begin{align}
\sd\left(z,k\right) &= \frac{2\pi}{K(k)\, k\,k^{\prime}}\sum_{n=0}^{\infty}%
\frac{(-1)^{n}q^{n+\frac{1}{2}}\sin\left((2n+1)\zeta\right)}{1+q^{2n+1}}, \quad \zeta = \frac{\pi z}{2K(k)},
\end{align}
when $q \in (0,1)$. Hence, we get that
\begin{align}
\sd\left(z,k\right) &= \left( 1 + \O \left( k^2 \ln k \right) \right) \left[ \sin \left( \frac{\pi z}{2K(k)} \right) + \O(q) \right]\\
&= \sin \left( \frac{\pi z}{2K(k)} \right) + \O\left(k^2 \ln k\right) \quad \text{as} \quad k \to 0^+.
\end{align}
One can readily compute from \eqref{seriesK} that
\[
\frac{\pi}{2K(k)} = 1 - \frac{1}{4}k^2 + \O\big(k^4\big)  \quad \text{as} \quad k \to 0^+.
\]
Trivially,
\[
\sin \left( \frac{\pi z}{2K(k)} \right) = \sin \left( z - \frac{k^2 z}{4} \right) \cos \left( \O\big(k^4z\big) \right) - \cos \left( z - \frac{k^2 z}{4} \right) \sin \left( \O\big(k^4z\big) \right).
\]
To finish the proof of the lemma, it remains to observe that
\[
\sin \left( \O\big(k^4z\big) \right) = \O\big(k^4z\big) \quad \text{and} \quad \cos \left( \O\big(k^4z\big) \right) = 1 + \O \left( k^4 z \right). \qedhere
\]
\end{proof}

The following theorem allowed us to connect the behavior of solutions in Theorems~\ref{thm:4} and~\ref{thm:5}.

\begin{thm}
\label{thm:sd-sine}
As in the main text, let \( \pi\kappa = \varkappa x = \tfrac12\ln(|p(\varkappa,x)|^2-1)\). Assume that \( \pi\kappa = D x^{1/3} \) for some \( D>0 \). It holds that
\begin{align}
\kk^\prime \sd\left(x-2\kappa \KK^\prime,\kk\right) = \varepsilon\sin( x - \kappa \ln x + \phi ) + \O \left( x^{-1/3} D^2 \ln x\right)
\end{align}
as $x \to \infty$, where, as in \eqref{kappa phi}, $\phi = - \kappa \ln 4 + \arg \left( p \Gamma \left(\frac12 + \ii \kappa \right) \right)$, and \( \O(\cdot) \) term is independent of \( D \) and \( x \).
\end{thm}
\begin{proof}
Let us write \( z=x-2\kappa \KK^\prime \). Since \( \varkappa = \pi D x^{-\frac23}\), we see from Lemma~\ref{lem:k varkappa} that \( \kk \to 0 \) and it holds that
\[
\kk^2 = \frac4\pi \varkappa + \O(\varkappa^2) = \O(D x^{-\frac23}). 
\]
We also get from \cite[Equation~(19.12.1)]{NIST:DLMF} that
\(
-2\KK^\prime = 2\ln \kk -2 \ln 4 + \O \big( \kk^2 \ln \kk \big).
\)
Therefore, we have that
\[
\frac{\kk^2 z}4 = \kappa + \O(x^{-\frac13} D^2\ln x) \quad \text{and} \quad \kk^4 z = \O(x^{-\frac13}D^2).
\]
If we define \( \phi_* \) via the relation
\[
z- \frac{\kk^2 z}4 = x - \kappa \ln x + \phi_*, \quad 
\]
then it follows from Lemma~\ref{lem:sd} that
\[
\sd\left(x-2\kappa \KK^\prime,\kk\right) = \sin( x - \kappa \ln x + \phi_* ) + \O( x^{-\frac13} D^2)
\]
as \( x\to\infty \). Moreover, it holds that
\begin{align}
\phi_* &= \kappa \ln x + 2\kappa\ln\kk - 2\kappa\ln 4 - \kappa + \O\big(x^{-\frac13} D^2\ln x\big) \\
& = -\kappa\ln 4 + \kappa \ln \kappa - \kappa + \O\big(x^{-\frac13} D^2\ln x\big)
\end{align}
as \( x\to\infty \). Next, we choose
\begin{align} \label{recall of arg p Gamma}
\arg \left( p \Gamma \left(\frac12 + \ii \kappa \right) \right) = \arg(p) + \Im \left( \log \Gamma \left( \frac12 + \ii \kappa \right) \right).
\end{align}
By \eqref{asymp regime}, one has that
\[
p = \varepsilon e^{\varkappa x} - \frac{\ii s^\R}{2} + \O \left( e^{-\varkappa x} \right) \quad \Rightarrow \quad  \arg(p) = \frac{1-\varepsilon}2\pi + \O(e^{-\pi\kappa})
\]
On the other hand, \cite[Equation (5.11.1) and Table 24.2.1]{NIST:DLMF} give that
\begin{align}
\log \Gamma(s) = \left(s - \frac{1}{2}\right) \log s - s + \frac{1}{2} \log (2 \pi) + \frac{1}{12 s} + \O \left( s^{-3} \right), \quad s \to \infty.
\end{align}
Applying this formula with \(s=\tfrac12+\ii\kappa \) after some computation gives
\[
\Im \left( \log \Gamma \left( \frac12 + \ii \kappa \right) \right) = \kappa \ln \kappa - \kappa + \frac{1}{24 \kappa} + \O \left( \kappa^{-3} \right).
\]
Altogether, we have shown that there is a choice of the argument such that
\[
\arg \left( p \Gamma \left(\frac12 + \ii \kappa \right) \right) = \frac{1-\varepsilon}2\pi + \kappa \ln \kappa - \kappa + \O \left( \kappa^{-1} \right).
\]
Hence, \( \phi_* = - \frac{1-\varepsilon}2\pi + \phi + \O(x^{-\frac13} D^2\ln x) \) and therefore 
\[
\sin( x - \kappa \ln x + \phi_* ) = \varepsilon\sin( x - \kappa \ln x + \phi ) + \O\big(x^{-\frac13} D^2\ln x\big)
\]
as \( x\to\infty \). Since \( \kk^\prime = \sqrt{1-\kk^2} = 1 + \O(\kk^2) = 1 + \O(D x^{-\frac23})\), this finishes the proof of the theorem.
\end{proof}

\section{On Parabolic Cylinder Parametrix}
\label{app:PCF}

Given a parabolic cylinder function \( D_\mu(w) \), it can be written as a Kummer \( U(a,b,z) \) confluent hypergeometric function in the following way
\[
D_\mu(w) = 2^{\mu/2} U\big(-\tfrac\mu2,\tfrac12,\tfrac{w^2}2\big)e^{-\frac14w^2},
\]
see \cite[Equation~(12.7.14)]{NIST:DLMF}. Therefore, it admits a series expansion of the form
\begin{align}
\label{expansion-1}
D_\mu(w) = w^\mu e^{-\frac14w^2} \left( \sum_{s=0}^{n-1} \frac{a_s^{(\mu)}}{w^{2s}} + R_n(w;\mu)\right), \quad a_s^{(\mu)} := \frac{(-1)^s(-\mu)_{2s}}{2^ss!},
\end{align}
for any \( n\geq 0 \), see \cite[Equation~(13.7.4)]{NIST:DLMF}. We shall now use notation and material of \cite[\textsection13(ii)]{NIST:DLMF}. Assume that
\[
\sigma := \frac{|1+2\mu|}{|w|^2} \leq \frac12.
\]
Set \( \rho := \frac14|\mu^2+\mu+1|+\frac{\sigma(4+\sigma)}{4(1-\sigma)^2} \). Then, for \( \Re(w)\geq0 \), it holds that
\[
|R_n(w;\mu)| \leq \frac{2C_n(\mu)}{1-\sigma} \frac{|a_n^{(\mu)}|}{|w|^{2n}} \exp\left\{\frac{4C_1(\mu)}{1-\sigma} \frac{\rho}{|w|^2}\right\},
\]
where constants \( C_n(\mu) \) satisfy
\[
|C_n(\mu)| \leq \left( \sqrt\pi\frac{\Gamma(\frac{n+2}2)}{\Gamma(\frac{n+1}2)} + \frac{2\sigma n}{1+\sqrt{1-4\sigma^2}} \right)\frac{2^{n/2}}{(1+\sqrt{1-4\sigma^2})^{n/2}}.
\]
Clearly, it holds that \( |C_n(\mu)|\leq C_n^\prime \), where \( C_n^\prime \) is independent of \( \mu \) and \( w \). Moreover, we have that \( 4\rho \leq |\mu|(|\mu|+1) + 10 \). Therefore, if we suppose that
\begin{align}
\label{expansion-2}
|w| \geq \max\{5,|\mu|\} \quad \text{and} \quad \Re(w)\geq0,
\end{align}
then \( \sigma\leq \frac12 \) and there exist constants \( A_n \), independent of \( \mu \) and \( w \), such that
\begin{align}
\label{expansion-3}    
|R_n(w;\mu)| \leq A_n \frac{|a_n^{(\mu)}|}{|w|^{2n}}.
\end{align}

Now, the model parametrix \( Pc(z)e^{\frac14z^2\sigma_3}\), used in Section~\ref{ss:6.3}, is explicitly given by
\[
 \begin{cases}
    \begin{pmatrix} D_{-\nu}\left(ze^{\frac{\pi \ii}2} \right )e^{\frac{\pi \ii\nu }2} & -\alpha D_{\nu-1}(z) \\ \ii\beta D_{-\nu-1}\left(ze^{\frac{\pi \ii}2}\right )e^{\frac{\pi \ii\nu }2} & D_{\nu}(z) \end{pmatrix}, & \arg(z)\in\left(-\frac\pi4,0\right), \medskip \\
    \begin{pmatrix} D_{-\nu}\left(ze^{-\frac{\pi \ii}2} \right )e^{-\frac{\pi \ii\nu }2} & -\alpha D_{\nu-1}(z) \\ -\ii\beta D_{-\nu-1}\left(ze^{-\frac{\pi \ii}2}\right )e^{-\frac{\pi \ii\nu }2} & D_{\nu}(z) \end{pmatrix}, & \arg(z)\in\left(0,\frac\pi2\right), \medskip \\
    \begin{pmatrix} D_{-\nu}\left(ze^{-\frac{\pi \ii}2} \right )e^{-\frac{\pi \ii\nu }2} & \alpha D_{\nu-1}\left(ze^{-\pi \ii} \right)e^{\pi \ii \nu} \\ -\ii\beta D_{-\nu-1}\left(ze^{-\frac{\pi \ii}2}\right )e^{-\frac{\pi \ii\nu }2} & D_{\nu}\left(ze^{-\pi \ii} \right)e^{\pi \ii \nu} \end{pmatrix}, & \arg(z)\in\left(\frac\pi2,\pi\right), \medskip \\
    \begin{pmatrix} D_{-\nu}\left(ze^{-\frac{3\pi \ii}2} \right )e^{-\frac{3\pi \ii\nu }2} & \alpha D_{\nu-1}\left(ze^{-\pi \ii} \right)e^{\pi \ii \nu} \\ \ii\beta D_{-\nu-1}\left(ze^{-\frac{3\pi \ii}2}\right )e^{-\frac{3\pi \ii\nu }2} & D_{\nu}\left(ze^{-\pi \ii} \right)e^{\pi \ii \nu} \end{pmatrix}, & \arg(z)\in\left(\pi,\frac{3\pi}2\right), \medskip \\
    \begin{pmatrix} D_{-\nu}\left(ze^{-\frac{3\pi \ii}2} \right )e^{-\frac{3\pi \ii\nu }2} & -\alpha D_{\nu-1}\left(ze^{-2\pi \ii} \right)e^{2\pi \ii \nu} \\ \ii\beta D_{-\nu-1}\left(ze^{-\frac{3\pi \ii}2}\right )e^{-\frac{3\pi \ii\nu }2} & D_{\nu}\left(ze^{-2\pi \ii} \right)e^{2\pi \ii \nu} \end{pmatrix}, & \arg(z)\in\left(\frac{3\pi}2,\frac{7\pi}{4}\right),
\end{cases}
\]
where \( \nu = \frac12+\ii\frac{\varkappa x}\pi \). In the above formula, \( \arg(z)\in(-\frac\pi4,\frac{7\pi}4) \) and multiplication of \( z \) by \( e^{\ii\phi}\) represents the rotation of the argument of \( z \) by \( \phi \). Hence, all the arguments of the parabolic cylinder functions above always belong to the right half-plane. Thus, assuming that \( |z|\geq \max\{5,|\nu+1|\} \), \eqref{matrix K} does indeed hold.






\bibliographystyle{alpha}
\bibliography{reference.bib}

\end{document}